\documentclass[pdflatex,sn-mathphys-num]{sn-jnl}

\usepackage{graphicx}%
\usepackage{physics}
\usepackage{multirow}%
\usepackage{amsmath,amssymb,amsfonts}%
\usepackage{amsthm}%
\usepackage{mathrsfs}%
\usepackage[title]{appendix}%
\usepackage{xcolor}%
\usepackage{textcomp}%
\usepackage{manyfoot}%
\usepackage{booktabs}%
\usepackage{algorithm}%
\usepackage{algorithmicx}%
\usepackage{algpseudocode}%
\usepackage{listings}%
\usepackage{comment}

\newcommand{\Sym}{\mathrm{Sym}}

\newcommand{\id}{I}

\newcommand{\GL}{\mathrm{GL}}
\newcommand{\U}{\mathrm{U}}

\newcommand{\CC}{\mathbb{C}}
\newcommand{\RR}{\mathbb{R}}

\newcommand{\cL}{\mathcal{L}}
\newcommand{\cS}{\mathcal{S}}

\newcommand{\cT}{\mathcal{T}}

\newcommand{\End}{\mathrm{End}}
\newcommand{\Comm}{\mathrm{Comm}}
\DeclareMathOperator{\Hom}{Hom}

\theoremstyle{thmstyleone}%
\newtheorem{theorem}{Theorem}
\newtheorem{proposition}[theorem]{Proposition}%
\newtheorem{lemma}[theorem]{Lemma}
\newtheorem{corollary}[theorem]{Corollary}

\theoremstyle{thmstyletwo}%
\newtheorem{example}{Example}%
\newtheorem{remark}{Remark}%

\theoremstyle{thmstylethree}%
\newtheorem{definition}{Definition}%

\begin{document}

\title{Admissible Subspaces for the Collective Haar Twirl: Schur-Weyl Classification and Heredity under Tensor-Factor Reduction}


\author*[1,2]{\fnm{Sayantan} \sur{Ganguly}}\email{sayantan.ganguly.90@tcgcrest.org}

\affil*[1]{\orgdiv{Institute for Advancing Intelligence (IAI)}, \orgname{TCG CREST}, \orgaddress{\street{EP \& GP Block, Sector - V}, \city{Kolkata}, \postcode{700091}, \state{West Bengal}, \country{India}}}

\affil[2]{\orgname{Ramakrishna
Mission Vivekananda Educational and Research Institute (RKMVERI)}, \orgaddress{\city{Belur}, \state{West Bengal}, \country{India}}}

\abstract{
Quantum private broadcasting (\textsf{QPB}) asks whether a single quantum message can be encrypted once, broadcast to several recipients, and decrypted locally by each recipient while the joint ciphertext reveals no information about the plaintext. With a counterexample,  Broadbent, Gonz\'alez-Guill\'en, and Schuknecht showed that the full space \(\mathcal{H}^{\otimes t}_{d}\) cannot be considered as a message space for \textsf{QPB}, and a symmetric subspace \(\mathrm{Sym}^{t}(\mathcal{H}_{d})\) of a Hilbert space \(\mathcal{H}^{\otimes t}_{d}\) provides a natural plaintext space for this task and introduced symmetric unitary \(t\)-designs. They further asked for a structural understanding of other plaintext subspaces supporting private broadcast.

Using Schur-Weyl decomposition we obtain necessary and sufficient conditions for
classification of all subspaces
\(
    \mathcal S\subset \mathcal H_d^{\otimes t}
\)
for which the collective Haar-twirl \(\mathcal T^{(t)}\) is a replacement channel. Namely, the Schur-Weyl block weights must be independent of the encoded state, and, within every active block \(V_\lambda\otimes W_\lambda\), the reduced state on the multiplicity space \(W_\lambda\) must be fixed. The classification permits extension to multiple Schur-Weyl blocks; we give explicit two-block and general multi-block constructions. We also identify the classification as the Schur-Weyl form of the Knill-Laflamme conditions for a complementary channel of the collective Haar twirl.

We next introduce a canonical reduction of plaintext subspaces under deletion of a recipient. We define and characterize privacy heredity, design heredity and \textsf{QPB}-heredity for exact restricted designs. The symmetric tower is fully \textsf{QPB}-hereditary, whereas the antisymmetric tower is fully privacy-hereditary but is design-hereditary from level \(s\) to \(s-1\) when
\(
    s\leq(d+1)/2
\).
For canonical carriers in an arbitrary Schur-Weyl block, hereditary reduction is governed by the multiplicity-free branching rule along the Young graph together with an inclusion of Peter-Weyl coefficient spaces.
}

\keywords{collective Haar twirl, private quantum subspaces, Schur-Weyl duality, quantum private broadcasting, complementary channels, Knill-Laflamme conditions, branching rules, matrix-coefficient spaces}


\pacs[MSC Classification]{81P47, 81P94, 81P73, 81P45, 81R05, 20C30, 20G05, 05E10,43A77}

\maketitle

\section{Introduction}\label{sec:introduction}
In this paper, we primarily study an open problem that originates in quantum private broadcasting (\textsf{QPB}), introduced by Broadbent, Gonz\'alez-Guill\'en, and Schuknecht [Section VII,\cite{BroadbentGonzalezGuillenSchuknecht2022}]. In their protocol, A single quantum plaintext is encrypted under a classical key $k$, broadcast to $t$ recipients, and decrypted locally by every recipient with the same key, while the joint ciphertext is required to reveal no information about the plaintext. In the basic broadcasting setting, the plaintexts arise from repeated pure states
\(
    \ket{\psi}^{\otimes t},
\)
and hence lie in the symmetric subspace
\(
    \operatorname{Sym}^{t}(\mathcal H_d)
    \subset
    \mathcal H_d^{\otimes t}
\)[Section VI, \cite{BroadbentGonzalezGuillenSchuknecht2022}]. The collective Haar twirl is defined by 
\begin{equation}
\label{eq:intro-Haar-twirl}
    \mathcal T^{(t)}(X)
    :=
    \int_{\mathrm U(d)}
        U^{\otimes t}
        X
        (U^\dagger)^{\otimes t}
        \,\mathrm d\mu_{\mathrm{Haar}}(U),
    \qquad
    X\in\mathcal L(\mathcal H_d^{\otimes t}),
\end{equation}
where
\(
    \mathcal H_d\cong\mathbb C^d
\)
and
\(
    \mu_{\mathrm{Haar}}
\)
is normalized Haar measure on \(\mathrm U(d)\). As a completely positive, trace-preserving, idempotent map, \(\mathcal T^{(t)}\) is the canonical conditional expectation onto the commutant of the collective \(\mathrm U(d)\)-action. The first central structural question of this paper is the following:

\begin{quote}
We ask for which nonzero subspaces \(\mathcal S\subset \mathcal H_d^{\otimes t}\) there exists a state \(\sigma_{\mathcal S} \in \mathcal{D}(\mathcal{H}^{\otimes t}_{d})\) such that
\begin{equation}\label{eq:admissibility-intro}
  \mathcal T^{(t)}\!\left(\mathcal J_{\mathcal S}(X)\right)
  =\operatorname{Tr}(X)\sigma_{\mathcal S}
  \qquad\text{for every }X\in\mathcal D(\mathcal S).
\end{equation}
Here \(\mathcal J_{\mathcal S}(X)=\iota_{\mathcal S}X
\iota_{\mathcal S}^{\dagger}\), with
\(\iota_{\mathcal S}:\mathcal S\hookrightarrow
\mathcal H_d^{\otimes t}\) the inclusion isometry.
\end{quote}
We call such a subspace \emph{Haar-twirl admissible}; in the terminology of private quantum channels, they are exactly the private admissible subspaces of $\mathcal T^{(t)}$. Although the question is motivated by quantum cryptography, Equation~\eqref{eq:admissibility-intro} is a purely operator-theoretic condition on a compact-group twirling channel.

The Haar-twirling channel $\mathcal T^{(t)}$ is a replacement channel on $\operatorname{Sym}^t(\mathcal H_d)$. This led Broadbent et al.\ to
introduce \emph{symmetric unitary $t$-designs} [Section VI,\cite{BroadbentGonzalezGuillenSchuknecht2022} ], which reproduce Haar moments only on $\operatorname{Sym}^t(\mathcal H_d)$ and therefore allow shorter keys than full unitary $t$-designs, and to ask whether designs can be relaxed in
the same way for other plaintext subspaces~[Section VII,\cite{BroadbentGonzalezGuillenSchuknecht2022}].

The symmetric-subspace guarantee does not extend automatically to unrestricted plaintexts in the full tensor power. This motivates the classification of other subspaces on which collective encryption can provide perfect privacy. A second structural question arises when a recipient leaves: whether the plaintext space, the privacy guarantee, and the same key ensemble descend to the remaining $t-1$ recipients.

For every encryption of the form $\mathrm{Enc}_k(X)=U_k^{\otimes t}
\mathcal J_{\mathcal S}(X)(U_k^{\dagger})^{\otimes t}$, where \(X \in \mathcal{D}(\mathcal{S})\), with local decryption $\mathrm{Dec}_k(Y)=U_k^{\dagger}YU_k$, the averaged ciphertext channel satisfies $\mathcal T^{(t)}\circ\Phi^{(t)}_{\mathfrak U,w}=\mathcal T^{(t)}$.
We first show that every Haar-twirl admissible subspace supports a \textsf{QPB} scheme (Theorem~\ref{thm:A}) and that, conversely, every plaintext subspace supporting a perfectly private collective scheme of the form considered here is Haar-twirl admissible. The classification of Haar-twirl admissible subspaces is therefore same as the classification of plaintext spaces acceptable for quantum private broadcasting.

\subsection{Main results}
\label{sec:main-results}

We state the main results in the notation of
Section~\ref{sec:preliminaries}. 

Let
\(
  \Lambda(d,t):=\{\lambda\vdash t:\ell(\lambda)\leq d\}.
\)
We fix, once and for all, a unitary Schur--Weyl decomposition
\begin{equation}\label{eq:schur-weyl-decomposition}
  \mathcal H_d^{\otimes t}
  =\bigoplus_{\lambda\in\Lambda(d,t)}
   \iota_\lambda(V_\lambda\otimes W_\lambda),
\end{equation}
where \(V_\lambda\) is the irreducible polynomial
\(\mathrm U(d)\)-module of highest weight \(\lambda\), and
\(W_\lambda\) is the Specht module of shape \(\lambda\). with isometric inclusions $\iota_\lambda:V_\lambda\otimes W_\lambda \longrightarrow \mathcal H_d^{\otimes t}$ and
orthogonal block projectors $\Pi_\lambda=\iota_\lambda\iota_\lambda^{\dagger}$. A finite weighted ensemble $(\mathfrak U,w)$, $\mathfrak U=\{U_k:k\in\mathcal K\}\subset\mathrm U(d)$, with moment channel $\Phi^{(t)}_{\mathfrak U,w}(X)=\sum_{k\in\mathcal K}w_kU_k^{\otimes t}X(U_k^{\dagger})^{\otimes t}$,
and $\mathfrak A_t:=\operatorname{im}\mathcal T^{(t)}$ denotes the image of the Haar-twirl. Each of the theorem stated below are the stated collectively here and proved later; the precise statements are cited after it.

\subsubsection*{Private plaintext spaces of the collective Haar-twirl}

Our first result shows that Haar-twirl admissibility is the right notion for collective encryption, and it classifies all admissible subspaces.

\begin{theorem}\label{thm:A}
For a nonzero subspace $\mathcal S \subset \mathcal H_d^{\otimes t}$ the following are equivalent.
\begin{enumerate}
\item There is a finite weighted ensemble $(\mathfrak U,w)$ for which the
  $t$-recipient \textsf{QPB} scheme with
  $\mathrm{Enc}^{(t)}_k(X)=U_k^{\otimes t}\mathcal J_{\mathcal S}(X)(U_k^{\dagger})^{\otimes t}$
  and $\mathrm{Dec}_k(Y)=U_k^{\dagger}YU_k$ is perfectly correct and perfectly
  private on $\mathcal S$.
\item $\mathcal S$ is Haar-twirl admissible, i.e.\
  $\mathcal T^{(t)}\circ\mathcal J_{\mathcal S}=\mathcal R_{\sigma_{\mathcal S},\mathcal S}$ for
  some $\sigma_{\mathcal S}\in\mathcal D(\mathcal H_d^{\otimes t})$.
\item There exist numbers $c_\lambda\ge0$ with $\sum_{\lambda\in\Lambda(d,t)}c_\lambda=1$
  and, for every $\lambda$ with $c_\lambda>0$, a state
  $\omega_\lambda\in\mathcal D(W_\lambda)$ such that, for all
  $|\psi\rangle,|\phi\rangle\in\mathcal S$,
  \begin{align}
    \langle\phi|\Pi_\lambda|\psi\rangle&=c_\lambda\langle\phi|\psi\rangle,
      \label{eq:main-weight}\\
    \operatorname{Tr}_{V_\lambda}\bigl(\iota_\lambda^{\dagger}|\psi\rangle\langle\phi|\iota_\lambda\bigr)
      &=c_\lambda\langle\phi|\psi\rangle\,\omega_\lambda .\label{eq:main-marginal}
  \end{align}
\item $\Pi_{\mathcal S}A\Pi_{\mathcal S}\in\mathbb C\,\Pi_{\mathcal S}$ for every
  $A\in\mathfrak A_t$.
\item $\mathcal S$ is exactly correctable for a complementary channel of
  $\mathcal T^{(t)}$.
\end{enumerate}
In this case the replacement state is
\begin{equation}\label{eq:main-sigma}
  \sigma_{\mathcal S}=\sum_{\lambda:\,c_\lambda>0}c_\lambda\,\iota_\lambda
  \Bigl[\frac{I_{V_\lambda}}{\dim_{\mathbb C}V_\lambda}\otimes\omega_\lambda\Bigr]\iota_\lambda^{\dagger}.
\end{equation}
Moreover, for every finite weighted ensemble $(\mathfrak U,w)$ and every
$\sigma\in\mathcal D(\mathcal H_d^{\otimes t})$,
\[
  \bigl\|\mathcal T^{(t)}\circ\mathcal J_{\mathcal S}-\mathcal R_{\mathcal T^{(t)}(\sigma),\mathcal S}\bigr\|_{1\to1}
  \le\bigl\|\Phi^{(t)}_{\mathfrak U,w}\circ\mathcal J_{\mathcal S}-\mathcal R_{\sigma,\mathcal S}\bigr\|_{1\to1},
\]
and the same inequality holds in the diamond norm.
\end{theorem}

Condition~\eqref{eq:main-weight} says that every observable surviving the twirl has scalar compression to $\mathcal S$, and resolving it block by block gives~\eqref{eq:main-weight}--\eqref{eq:main-marginal}.
The final inequality shows that an $\varepsilon$-private collective scheme forces $\mathcal T^{(t)}$ itself to be $\varepsilon$-close to a replacement channel on $\mathcal S$.

The equivalence (2)$\Leftrightarrow$(3) and formula~\eqref{eq:main-sigma} are Theorem~\ref{thm:qpb-classification}. The implication (2)$\Rightarrow$(1) is Proposition~\ref{prop:qpb-from-restricted-design}, applied to an exact unitary $(\mathcal S,t)$-design, which always exists. The implication (1)$\Rightarrow$(2) and the inequality are Proposition~\ref{prop:necessity}.
The equivalences (2)$\Leftrightarrow$(4)$\Leftrightarrow$(5) are
Proposition~\ref{prop:classification-KL}, which specializes the privacy-correctability complementarity of~\cite{KretschmannKribsSpekkens2008} (Theorem~\ref{thm:privacy-complementarity}).

Two consequences are immediate. For every $\lambda$ with $c_\lambda>0$, the
restricted block projection $\Pi_\lambda|_{\mathcal S}$ satisfies
$(\Pi_\lambda|_{\mathcal S})^{\dagger}(\Pi_\lambda|_{\mathcal S})=c_\lambda I_{\mathcal S}$ and is
therefore injective (Corollary~\ref{cor:active-block-dimension-bound}). A full
isotypic block $\iota_\lambda(V_\lambda\otimes W_\lambda)$ is Haar-twirl admissible if
and only if $\dim_{\mathbb C}W_\lambda=1$; hence, for $t\ge2$, the only admissible
full blocks are $\operatorname{Sym}^t(\mathcal H_d)$ and, when $t\le d$,
$\bigwedge^t\mathcal H_d$ (Corollary~\ref{cor:active-block-injectivity}).

\subsubsection*{Multi-block Haar-twirl admissible subspaces}

Admissibility does not force a subspace into a single isotypic block. Our second result describes all Haar-twirl admissible subspaces explicitly.

\begin{theorem}[Multi-block normal form]\label{thm:B}
Let $n\ge1$ and $\emptyset\neq\Lambda_0\subseteq\Lambda(d,t)$, and let
$c_\lambda>0$ with $\sum_{\lambda\in\Lambda_0}c_\lambda=1$ and $\omega_\lambda\in\mathcal D(W_\lambda)$ with $r_\lambda:=\operatorname{rank}\omega_\lambda$
be given for $\lambda\in\Lambda_0$. An $n$-dimensional Haar-twirl admissible subspace with active blocks $\Lambda_0$, weights $c_\lambda$ and multiplicity
states $\omega_\lambda$ exists if and only if
\[
  n\,r_\lambda\le\dim_{\mathbb C}V_\lambda \quad \text{where}\; \lambda\in\Lambda_0.
\]
Every such subspace is $\mathcal S=J(\mathbb C^n)$ for an isometry
\[
  J|x\rangle=\sum_{\lambda\in\Lambda_0}\sqrt{c_\lambda}\,\iota_\lambda
  \bigl(U_\lambda\otimes I_{W_\lambda}\bigr)\bigl(|x\rangle\otimes|\Omega_\lambda\rangle\bigr),
  \qquad|x\rangle\in\mathbb C^n,
\]
where $|\Omega_\lambda\rangle\in\mathcal K_\lambda\otimes W_\lambda$ is a minimal
purification of $\omega_\lambda$ with $\dim_{\mathbb{C}}\mathcal K_\lambda=r_\lambda$ and
$U_\lambda:\mathbb C^n\otimes\mathcal K_\lambda \longrightarrow  V_\lambda$ is an isometry; conversely,
every isometry of this form has Haar-twirl admissible range, and
$\|\Pi_\lambda\psi\|^2=c_\lambda\|\psi\|^2$ for all $\psi\in J(\mathbb C^n)$. In
particular, the maximum dimension is
\[
  n_{\max}=\min_{\lambda\in\Lambda_0}
  \Bigl\lfloor\frac{\dim_{\mathbb C}V_\lambda}{\operatorname{rank}\omega_\lambda}\Bigr\rfloor .
\]
\end{theorem}

Theorem~\ref{thm:B} is Theorem~\ref{thm:general-multiblock-construction}
together with Corollary~\ref{cor:exact-maximal-dimension-fixed-active-data}. If $|\Lambda_0|\ge2$, every nonzero vector of $\mathcal S$ has a nonzero component in every active block,
so no nonzero vector lies in a single isotypic component.
Proposition~\ref{prop:two-block-construction} is the two-block case with pure multiplicity states.

\subsubsection*{Canonical reduction and two heredity criteria}

When a recipient leaves, a plaintext subspace must be reduced to $s-1$
tensor factors. Since a subspace cannot itself be partially traced, we use
the canonical reduction
\[
  \operatorname{Red}_j(\mathcal S):=\operatorname{supp}\bigl(\operatorname{Tr}_j\Pi_{\mathcal S}\bigr),
\]
which equals the span of all one-factor contractions of vectors of $\mathcal S$ (Lemma~\ref{lem:Red-support}). Partial-trace covariance gives $\mathcal T^{(s-1)}(\operatorname{Tr}_j\rho)=\operatorname{Tr}_j\sigma_{\mathcal S}$ for every $\rho\in\mathcal D(\mathcal S)$ (Proposition~\ref{prop:automatic-marginal-descent}), but this concerns only reduced states of $\mathcal S$, not all states on $\operatorname{Red}_j(\mathcal S)$. We therefore distinguish \emph{privacy
heredity} ($\operatorname{Red}_j(\mathcal S)$ is Haar-twirl admissible), a property of the Haar channel, from \emph{design heredity} (every exact unitary
$(\mathcal S,s)$-design is, with the same unitaries and weights, an exact unitary $(\operatorname{Red}_j(\mathcal S),s-1)$-design), a property of the moment constraints imposed on a finite ensemble; their conjunction is \textsf{QPB}-heredity (Definition~\ref{qpb-heriditary}).

\begin{theorem}[Heredity criteria]\label{thm:C}
Let $s\ge2$, let $\mathcal S\subseteq\mathcal H_d^{\otimes s}$ be nonzero, fix $j\in\{1,\dots,s\}$ and put $\mathcal R_j:=\operatorname{Red}_j(\mathcal S)$.
\begin{enumerate}
\item If $\mathcal S$ is Haar-twirl admissible, then $\mathcal S$ is privacy-hereditary in coordinate $j$ if and only if there exist positive operators $\tau_\mu\in\mathcal L(W_\mu)$, $\mu\in\Lambda(d,s-1)$, with $\sum_\mu\operatorname{Tr}\tau_\mu=1$ and
  \[
    \operatorname{Tr}_{V_\mu}\Bigl[\bigl(\iota^{(s-1)}_\mu\bigr)^{\dagger}
    \mathcal J_{\mathcal R_j}(Y)\,\iota^{(s-1)}_\mu\Bigr]=\operatorname{Tr}(Y)\,\tau_\mu
    \qquad(Y\in\mathcal L(\mathcal R_j)),
  \]
  that is, if and only if $\mathcal R_j$ satisfies condition~(3) of
  Theorem~\ref{thm:A} at level $s-1$ with $c_\mu=\operatorname{Tr}\tau_\mu$.
\item $\mathcal S$ is design-hereditary in coordinate $j$ if and only if
  $\mathfrak C_{s-1}(\mathcal R_j)\subseteq\mathfrak C_s(\mathcal S)$.
\item Neither of these properties implies the other.
\item If $\mathcal S$ is \textsf{QPB}-hereditary in coordinate $j$ and
  $(\mathfrak U,w)$ is an exact unitary $(\mathcal S,s)$-design, then the
  $(s-1)$-recipient scheme on $\mathcal R_j$ with
  $\mathrm{Enc}^{(s-1)}_k(\rho)=U_k^{\otimes(s-1)}\rho(U_k^{\dagger})^{\otimes(s-1)}$
  and the same $\mathrm{Dec}_k$ is perfectly correct and perfectly private,
  with replacement state $\sigma_{\mathcal R_j}=\operatorname{Tr}_j\sigma_{\mathcal S}$.
\end{enumerate}
\end{theorem}

Part~(1) is Theorem~\ref{thm:privacy-heredity}. Part~(2) is Theorem~\ref{thm:design-heredity-criterion}, where $\mathfrak C_r(\cdot)$ is the finite-dimensional space of functions on $\mathrm U(d)$ defined in~\eqref{eq:coefficient-space}; its necessity part constructs an exact restricted design that fails to descend, by perturbing Haar measure orthogonally to $\mathfrak C_s(\mathcal S)$ and applying Carath\'eodory's theorem. Their conjunction is
Corollary~\ref{cor:qpb-heredity}. Part~(3) is shown by Example~\ref{ex:top-exterior-failure} (privacy without design), and Example~\ref{ex:design-without-privacy} shows that design heredity does not imply privacy heredity. Part~(4) is Theorem~\ref{thm:QPB-hereditary-security} together with Proposition~\ref{prop:replacement-state-partial-trace}: the reduced protocol uses the same unitary ensemble, the same weights and hence the same key set.

\subsubsection*{The symmetric and exterior towers}

\begin{theorem}\label{thm:D}
\begin{enumerate}
\item For every $s\ge2$ and $j\in\{1,\dots,s\}$,
  \[
  \operatorname{Red}_j\bigl(\operatorname{Sym}^s(\mathcal H_d)\bigr)=\operatorname{Sym}^{s-1}(\mathcal H_d),
  \]
  and, for $1\le r\le s-1$, every exact unitary
  $(\operatorname{Sym}^s(\mathcal H_d),s)$-design is also an exact unitary
  $(\operatorname{Sym}^r(\mathcal H_d),r)$-design. Hence the symmetric tower is fully \textsf{QPB}-hereditary.
\item For $2\le s\le d$ and $j\in\{1,\dots,s\}$,
  $\operatorname{Red}_j\bigl(\bigwedge^s\mathcal H_d\bigr)=\bigwedge^{s-1}\mathcal H_d$, so
  the exterior tower is fully privacy-hereditary; and $\bigwedge^s\mathcal H_d$
  is design-hereditary from level $s$ to level $s-1$ if and only if
  \[
    s\le\frac{d+1}{2}.
  \]
  Consequently, $\bigwedge^s\mathcal H_d$ is \textsf{QPB}-hereditary if and only
  if $s\le(d+1)/2$.
\end{enumerate}
\end{theorem}

Part~(1) is Theorem~\ref{thm:symmetric-tower-heredity}. Its design part was
first proved in~[Lemma~VI.3,\cite{BroadbentGonzalezGuillenSchuknecht2022}]; we recover it from an
equivariant-lifting criterion (Proposition~\ref{prop:covariant-lifting}),
which isolates the mechanism that fails for exterior powers. Part~(2) is
Theorems~\ref{thm:antisymmetric-privacy-heredity}
and~\ref{thm:antisymmetric-design-heredity}. The threshold comes from the
multiplicity-free decomposition of $\mathrm{SU}(d)$-modules
\[
\mathcal L\Bigl(\textstyle\bigwedge^s\mathcal H_d\Bigr)\cong\bigoplus_{r=0}^{\min\{s,d-s\}}L_{\varpi_r+\varpi_{d-r}}:
\] the inclusion of coefficient spaces in
Theorem~\ref{thm:C}(2) holds when
$\min\{s-1,d-s+1\}\le\min\{s,d-s\}$. For $s=d$ the top exterior power is a
one-dimensional representation on which every ensemble, even the identity
alone, is an exact restricted design, although it is not an exact design on
$\bigwedge^{d-1}\mathcal H_d$ (Example~\ref{ex:top-exterior-failure}).

\subsubsection*{Canonical carriers and deletion of an arbitrary recipient}

For $\lambda\in\Lambda(d,s)$ and a unit vector $|\eta\rangle\in W_\lambda$, the
\emph{canonical carrier}
\[
  \mathcal S_{\lambda,\eta}:=\iota^{(s)}_\lambda\bigl(V_\lambda\otimes\operatorname{span}_{\mathbb{C}}\{|\eta\rangle\}\bigr)
\]
is Haar-twirl admissible. Let $q_{\mu\leftarrow\lambda}:W_\lambda\longrightarrow W_\mu$ be the branching projections and
$\mathcal B(\lambda,\eta)=\{\mu\nearrow\lambda:q_{\mu\leftarrow\lambda}|\eta\rangle\neq0\}$ the active branch set. For $j\in\{1,\dots,s\}$ let $\mathsf J^{(j)}:=\sum_{q\neq j}(q\,j)\in\mathbb C[\mathfrak S_s]$; thus $\mathsf J^{(s)}=\mathsf J_s$ is the Jucys-Murphy element.

\begin{theorem}[Canonical carriers]\label{thm:E}
\begin{enumerate}
\item \[
        \operatorname{Red}_s(\mathcal S_{\lambda,\eta})
       =\bigoplus_{\mu\in\mathcal B(\lambda,\eta)}\iota^{(s-1)}_\mu
       \bigl(V_\mu\otimes\operatorname{span}_{\mathbb{C}}\{q_{\mu\leftarrow\lambda}|\eta\rangle\}\bigr).
      \]
  Hence $\mathcal S_{\lambda,\eta}$ is privacy-hereditary in coordinate $s$ if
  and only if $|\mathcal B(\lambda,\eta)|=1$; and if $\mathcal B(\lambda,\eta)=\{\mu\}$,
  it is design-hereditary in coordinate $s$ if and only if
  $\operatorname{Coeff}(V_\mu\otimes\overline{V_\mu})\subseteq
  \operatorname{Coeff}(V_\lambda\otimes\overline{V_\lambda})$.
\item For every $j\in\{1,\dots,s\}$, $\mathcal S_{\lambda,\eta}$ is
  privacy-hereditary in coordinate $j$ if and only if $|\eta\rangle$ is an
  eigenvector of $W_\lambda(\mathsf J^{(j)})$.
\item If $\lambda$ is a rectangle, every canonical carrier
  $\mathcal S_{\lambda,\eta}$ is privacy-hereditary in every coordinate. If
  $\lambda=(2,1)$, no canonical carrier is; each Young--Yamanouchi carrier is
  privacy-hereditary in coordinate $3$ but not in coordinates $1$ and $2$.
\end{enumerate}
\end{theorem}

Part~(1) is Theorem~\ref{thm:canonical-carrier-branching}: contracting the
last factor leaves the standard subgroup $\mathfrak S_{s-1}$, and the
multiplicity-free branching rule determines the reduced space. Following a
Young--Yamanouchi vector along its path in the Young graph gives towers of
canonical carriers that are fully privacy-hereditary under deletion of the
last factor (Corollary~\ref{cor:young-yamanouchi-heredity}). Parts~(2)
and~(3) are Theorem~\ref{thm:any-recipient} and
Corollary~\ref{cor:rectangular}: deleting factor $j$ is equivalent, after a
relabelling of the remaining factors, to deleting the last factor of
$\mathcal S_{\lambda,\eta_j}$ with $|\eta_j\rangle=W_\lambda((j\,s))|\eta\rangle$, and
the branch summands are the eigenspaces of $W_\lambda(\mathsf J_s)$. Since any
recipient may leave a broadcast, part~(3) shows that the choice of carrier
matters: rectangular shapes, which include the symmetric and exterior
towers, are robust under every deletion, while Young--Yamanouchi carriers
are in general robust only under deletion of the last recipient.

\subsection{Related work}
\label{subsec:related-work}
Quantum private broadcast was introduced by Broadbent,
Gonz\'alez-Guill\'en, and Schuknecht
\cite{BroadbentGonzalezGuillenSchuknecht2022}.
They compared independent quantum one-time-pad encryption, ordinary unitary \(t\)-designs, and symmetric unitary \(t\)-designs, and showed that exploiting the symmetric plaintext sector can substantially reduce the classical key length. They also raised the two questions
that are central here: the extension of symmetric-design ideas to more general plaintext subspaces and the reduction of the circuit depth needed to realize symmetric designs.

\medskip
The Knill-Laflamme conditions give the fundamental recovery-independent characterization of exactly correctable quantum subspaces \cite{KnillLaflamme1997}. Quantum error correction and subsystem formulations were developed subsequently; see, for example, \cite{KribsLaflammePoulin2005}. Kretschmann, Kribs, and Spekkens established the complementarity between private codes for a channel and correctable codes for a complementary channel \cite{KretschmannKribsSpekkens2008}.
Private quantum subsystems and Kraus-operator analogues of the Knill-Laflamme conditions were further investigated in \cite{JochymOConnorKribsLaflammePlosker2013,
JochymOConnorKribsLaflammePlosker2014}.

Private codes for group-averaging channels also arise in private reference-frame communication, where irreducible representation factors provide decoherence-full subsystems~\cite{BartlettRudolphSpekkens2004}. Church, Kribs, Pereira, and Plosker characterized private pure states of trace-preserving conditional expectations by trace vectors relative to the output state~\cite{ChurchKribsPereiraPlosker2011}.
Our classification resolves this condition in Schur-Weyl decomposition and gives an explicit normal form. The subsequent results study canonical tensor-factor reduction and the separate descent requirements for privacy and restricted designs.

\medskip
Unitary designs provide finite or efficiently samplable substitutes for Haar randomness at the level of finitely many moments. Structural and constructive aspects of exact and approximate designs were developed in, among others, \cite{GrossAudenaertEisert2007,dankert2009exact,RoyScott2009}.
Local random circuits as approximate polynomial designs were studied by Brand\~ao, Harrow, and Horodecki \cite{brandao2016local}.
In the present paper we use restricted designs only through their moment channels, which makes the distinction between full designs and designs restricted to a Haar-twirl admissible subspace particularly transparent.

\subsection{Organization of the paper}
\label{subsec:paper-organization}

The remainder of the paper is organized as follows.

In Section~\ref{sec:preliminaries} we develop the representation-theoretic and quantum-information preliminaries used throughout the paper. We review Schur-Weyl duality, the corresponding isotypic projectors, the blockwise form of the collective Haar twirl, symmetric and antisymmetric tensor powers, unitary designs, and the Knill-Laflamme theorem.  We also state the branching characterization, Pieri rule, and Peter-Weyl theorems needed later in the hereditary analysis.

In Section~\ref{sec:qpb-admissibility} we formulate \textsf{QPB} on a general plaintext subspace
\(
    \mathcal S\subset \mathcal H_d^{\otimes t}
\)
and separate intrinsic Haar privacy from approximation by a finite unitary ensemble.

In subsections~\ref{subsec:qpb-classification} and ~\ref{subsubsec:general-multi-block-construction}, we derive the structural classification of Haar-twirl admissible subspaces. We prove Theorem~\ref{thm:qpb-classification}, derive its Knill-Laflamme interpretation in Proposition~\ref{prop:classification-KL}, classify admissible full Schur-Weyl components, and construct explicit two-block and general multi-block families in Proposition~\ref{prop:two-block-construction} and Theorem~\ref{thm:general-multiblock-construction}.

In subsection~\ref{subsec:hereditary-qpb} we develop canonical reduction under deletion of a recipient. We distinguish privacy heredity from design heredity and prove the exact criteria Theorems~\ref{thm:privacy-heredity}--\ref{thm:design-heredity-criterion} and Corollary~\ref{cor:qpb-heredity}. We then analyze the symmetric and antisymmetric towers and, using the branching and Pieri rules together with Peter-Weyl theory, obtain the general canonical-carrier classification of Theorem~\ref{thm:canonical-carrier-branching}.

Finally, in Section~\ref{sec:conclusion} we summarize the structural results and formulate several open problems. Technical proofs concerning the blockwise Haar twirl, auxiliary representation-theoretic identities are presented in the appendices.

\section{Preliminaries}\label{sec:preliminaries}

\subsection{Notation and operator embeddings}
Let \(d,t\geq 1\) be positive integers. The set of positive integers \(\{1,\ldots,d\}\) is defined by the notation \([d]\). All Hilbert spaces are finite-dimensional and complex.  For Hilbert spaces
\(\mathcal X\) and \(\mathcal Y\), let
\(\mathcal L(\mathcal X,\mathcal Y)\) be the space of linear maps from
\(\mathcal X\) to \(\mathcal Y\), write
\(\mathcal L(\mathcal X):=\mathcal L(\mathcal X,\mathcal X)\), and let
\[
  \mathcal D(\mathcal X)
  :=\{\rho\in\mathcal L(\mathcal X):\rho\geq0,
       \ \operatorname{Tr}\rho=1\}.
\]
The identity operator on \(\mathcal X\) is \(I_{\mathcal X}\); the identity
superoperator on \(\mathcal L(\mathcal X)\) is
\(\operatorname{id}_{\mathcal L(\mathcal X)}\).  We use
\(\|X\|_1=\operatorname{Tr}\sqrt{X^\dagger X}\) and, for a linear map
\(\Delta\),
\(
  \|\Delta\|_{1\to1}
  :=\sup_{X\neq0}\frac{\|\Delta(X)\|_1}{\|X\|_1}.
\)
The diamond norm is is
\(
\|\Phi\|_{\diamond}
:=
\sup_{r\geq 1}
\left\|
\Phi\otimes
\operatorname{id}_{\mathcal L(\mathbb C^{r})}
\right\|_{1\to 1}.
\)
If \( \Phi \) is a completely positive trace-preserving (CPTP) map, then the following inequality holds:
\begin{equation}
\|\Phi\|_{1 \rightarrow 1} \leq \|\Phi\|_\diamond.
\end{equation}

If \(\mathcal S\subseteq\mathcal X\) is a subspace, then
\(\iota_{\mathcal S}:\mathcal S\hookrightarrow\mathcal X\) denotes its
inclusion isometry,
\(\Pi_{\mathcal S}=\iota_{\mathcal S}\iota_{\mathcal S}^\dagger\), and
\begin{equation}\label{eq:operator-embedding}
  \mathcal J_{\mathcal S}:\mathcal L(\mathcal S)\longrightarrow
  \mathcal L(\mathcal X),
  \qquad
  \mathcal J_{\mathcal S}(X)
  :=\iota_{\mathcal S}X\iota_{\mathcal S}^\dagger.
\end{equation}
Thus a channel \(\Phi \colon \mathcal L(\mathcal X) \longrightarrow \mathcal L(\mathcal Y)\)
restricted to inputs supported on \(\mathcal S\) is the composition map \(\Phi\circ\mathcal J_{\mathcal S}\). For \(\sigma\in\mathcal D(\mathcal Y)\), the replacement channel with domain \(\mathcal L(\mathcal S)\) is
\begin{equation}\label{eq:replacement-channel}
  \mathcal R_{\sigma,\mathcal S}(X)
  :=\operatorname{Tr}(X)\sigma.
\end{equation}

Throughout, \(U^\dagger\) denotes the Hilbert-space adjoint and is used in conjugation channels. The notation \(\overline U\) is reserved for the entrywise-conjugate representation; for example, the vectorized conjugation representation is \(U\otimes\overline U\).

The symmetric subspace of $\mathcal{H}^{\otimes t}_{d}$ is denoted by $\Sym^{t}(\mathcal{H}_{d})$. The antisymmetric subspace of \(\mathcal{H}^{\otimes t}_{d}\) is denoted by \(\bigwedge^{t}(\mathcal{H}_{d})\). Let \(\GL_{d}(\mathbb{C})\) denote the general linear group of invertible \(d \times d\) matrices, \(\U(d)\) be the group of all \(d \times d\) unitary matrices with entries in \(\mathbb{C}\) and \(\mathfrak{S}_{n}\) be the symmetric group defined on a set \([n]\) of \(n\) distinct elements. The notation \(\lambda \vdash t\) denotes that \(\lambda\) is a partition of \(t\), and \(\ell(\lambda)\) is its number of nonzero parts. We denote \(\Lambda(d,t)\) to be the set of \(\lambda \vdash t\) with \(\ell(\lambda) \leq d\), and when \(\mathcal{S}\) is fixed, \(\Lambda_{0}(d,t):=\{\lambda \in \Lambda(d,t): c_{\lambda} >0\}\). We call \(\Lambda_{0}(d,t)\) as an active set. For a compact group \(\mathrm{G}\), we denote the set of all continuous functions from \(\mathrm{G}\) to \(\mathbb{C}\) by \(C(\mathrm{G})\).

\subsection{Schur-Weyl duality and Schur-Weyl decomposition}\label{sec:Schurweyl}
To establish the framework for our constructions, we briefly review the necessary background in the representation theory of finite groups and compact Hausdorff topological groups; standard references include \cite{serre1977linear,fulton2013representation,simon1996representations}.
Throughout this paper, all vector spaces are finite-dimensional complex vector spaces. If \(K\) is a compact topological group, every representation of \(K\) is understood to be continuous. Since every finite-dimensional complex representation of a finite group, and every finite-dimensional continuous complex representation of a compact group, is unitarizable and completely reducible, we shall work throughout with unitary realizations of the representations under consideration.

Let \(V\) be a finite-dimensional complex vector space and \((V,\Gamma_{\mathrm{G}})\) be a representation of a compact Hausdorff topological group \(\mathrm{G}\). For an element \(g \in \mathrm{G}\), we denote \(\Gamma_{G}(g)\) to be the associated linear operator in \(\mathrm{GL}(V)\).

\subsubsection{Action of general linear group and the symmetric group on the tensor product space.}\label{sec:gpaction}

Fix integers \(d,t\geq1\). Let 
\(
    \mathcal H^{\otimes t}_{d}:=(\mathbb C^d)^{\otimes t},
\)
and let \(\{\ket{e_1},\ldots,\ket{e_d}\}\) be the standard orthonormal
basis of \(\mathbb C^d\). Then
\(
 \mathcal B_{\mathcal{H}^{\otimes t}_{d}}
 :=\{\ket{e_{j_1}}\otimes\cdots\otimes\ket{e_{j_t}}:
          (j_1,\ldots,j_t)\in[d]^t\}
\)
is an orthonormal basis of \(\mathcal H^{\otimes t}_{d}\), and therefore
\(\dim_{\mathbb{C}}\mathcal H^{\otimes t}_{d}=d^t\).

\smallskip
 Given a finite-dimensional vector space \(V\) over an algebraically closed field \(\mathbb{F}\) with \(\mathrm{char}(\mathbb{F})=0\), the Schur-Weyl duality and the Schur-Weyl is defined based on the action of the general linear group \(\GL(V)\) and the action of the symmetric group \(\mathfrak{S}_{t}\) on the vector space \(V^{\otimes t}\).

\smallskip 
In this paper, we consider the finite-dimensional vector space \(\mathbb{C}^{d}\) over the field \(\mathbb{C}\), and we first describe the action of \(\GL((\mathbb{C}^{d}))\cong \GL_{d}(\mathbb{C})\) and \(\mathfrak{S}_{t}\) on the Hilbert space \(\mathcal{H}^{\otimes t}_{d}=(\mathbb{C}^{d})^{\otimes t}\). Then, for the purpose of this paper, we restrict the action of \(\GL_{d}(\mathbb{C})\) to \(\U(d)\), and consider the Schur-Weyl duality based on the action of \(\U(d)\) and \(\mathfrak{S}_{t}\) on the Hilbert space \(\mathcal{H}^{\otimes t}_{d}\).
 
\smallskip
\noindent \(\bullet\) \textbf{The general linear group action.} For each $U \in \GL_{d}(\mathbb{C})$, define the linear map
\begin{align}\label{eq:glaction}
\Gamma_{\GL_{d}(\mathbb{C})} \colon &\GL_{d}(\mathbb{C}) \longrightarrow \GL\bigl((\CC^d)^{\otimes t}\bigr)\; \text{defined by} \\ \notag
&U \longmapsto \Gamma_{\GL_{d}(\mathbb{C})}(U):=U^{\otimes t} 
\end{align}
The action of \(\Gamma_{U}\) on a basis vector is defined by:
\[
U^{\otimes t} \bigl( |e_{j_1}\rangle \otimes \cdots \otimes |e_{j_t}\rangle \bigr) = (U|e_{j_1}\rangle) \otimes \cdots \otimes (U|e_{j_t}\rangle).
\]
The map defined in Equation \ref{eq:glaction} is a group homomorphism and hence a representation.

\begin{remark}
    When we consider the representation of \(\U(d)\) on \((\CC^d)^{\otimes t}\), the action of \(\U(d)\) is similar to the previous group action of \(\GL_{d}(\mathbb{C})\) on \((\CC^d)^{\otimes t}\), that is, \(\Gamma_{U(d)}(U):=U^{\otimes t}\).

\end{remark}

\medskip
\noindent \(\bullet\) \textbf{The Permutation Action.}
For each permutation $\sigma \in \mathfrak{S}_t$, define the linear map $\Gamma_{\mathfrak{S}_{t}}(\sigma) \in \mathrm{GL}\bigl((\CC^d)^{\otimes t}\bigr)$ by its action on basis vectors:
\begin{equation}\label{eq:permaction}
\Gamma_{\mathfrak{S}_{t}}(\sigma)\bigl(|e_{j_1}\rangle \otimes |e_{j_2}\rangle \otimes \cdots \otimes |e_{j_t}\rangle\bigr) := |e_{j_{\sigma^{-1}(1)}}\rangle \otimes |e_{j_{\sigma^{-1}(2)}}\rangle \otimes \cdots \otimes |e_{j_{\sigma^{-1}(t)}}\rangle.
\end{equation}
The operator $\Gamma_{\mathfrak{S}_{t}}(\sigma)$ moves the tensor factor in position $k$ to position $\sigma(k)$:
the content of the $k$-th slot in the output is $|e_{j_{\sigma^{-1}(k)}}\rangle$, which was originally the $\sigma^{-1}(k)$-th input slot.

\begin{remark}\label{convention}
There are two common conventions in the literature:
\begin{enumerate}
\item $\Gamma_{\mathfrak{S}_{t}}(\sigma)(|v_1\rangle \otimes \cdots \otimes |v_t\rangle) = |v_{\sigma^{-1}(1)}\rangle \otimes \cdots \otimes |v_{\sigma^{-1}(t)}\rangle$,
\item $\Gamma_{\mathfrak{S}_{t}}(\sigma)(|v_1\rangle \otimes \cdots \otimes |v_t\rangle) = |v_{\sigma(1)}\rangle \otimes \cdots \otimes |v_{\sigma(t)}\rangle$.
\end{enumerate}
Convention (a) gives a \emph{left} representation: $\Gamma_{\mathfrak{S}_{t}}(\sigma \circ \tau) = \Gamma_{\mathfrak{S}_{t}}(\sigma) \Gamma_{\mathfrak{S}_{t}}(\tau)$. Convention (b) gives a \emph{right} representation or equivalently a left representation of the opposite group. We use convention (a) throughout. 

With this convention, it is easy to check that $\Gamma_{\mathfrak{S}_{t}} \colon \mathfrak{S}_t \longrightarrow \GL((\CC^d)^{\otimes t})$ is a group homomorphism.
\end{remark}

\medskip
\begin{proposition}\label{prop:commute}
For all $U \in \GL_{d}(\mathbb{C})$ and $\sigma \in \mathfrak{S}_t$, the action \(\GL_{d}(\mathbb{C})\) and \(\mathfrak{S}_{t}\) on $(\CC^d)^{\otimes t}$ commute:
\[
\Gamma_{\GL_{d}(\mathbb{C})}(U) \circ \Gamma_{\mathfrak{S}_{t}}(\sigma) = \Gamma_{\mathfrak{S}_{t}}(\sigma) \circ \Gamma_{\GL_{d}(\mathbb{C})}(U).
\]
\end{proposition}

\noindent \(\bullet\) \textbf{Partitions, Young diagrams, and dimensions.}
\begin{definition}
A \emph{partition} of a positive integer $t$ is a weakly decreasing sequence of positive integers $\lambda = (\lambda_1, \lambda_2, \ldots, \lambda_\ell)$ with $\lambda_1 \geq \lambda_2 \geq \cdots \geq \lambda_\ell > 0$ and $|\lambda| := \lambda_1 + \cdots + \lambda_\ell = t$. We denote it by $\lambda \vdash t$. The \emph{length} $\ell(\lambda)$ is the number of parts. We denote the set of partitions of $t$ with at most $d$ parts by:
\[
\Lambda(d, t) := \{ \lambda \vdash t : \ell(\lambda) \leq d \}.
\]
\end{definition}

A partition \(\lambda\vdash t\) is identified with its Young diagram. A
\emph{Young tableau} of shape \(\lambda\) is a filling of the boxes of
\(\lambda\) by \(1,\ldots,t\). It is \emph{standard} if its entries increase
strictly along every row and every column. We write
\[
f^\lambda:=\#\operatorname{SYT}(\lambda)
\]
for the number of standard Young tableaux of shape \(\lambda\).

\subsubsection{Irreducible representations indexed by partitions}

The partitions of \(t\) index both the irreducible representations of
\(\mathfrak{S}_t\) and the homogeneous polynomial irreducible
representations of \(\GL_d(\CC)\) of degree \(t\).

\begin{enumerate}
\item \textbf{The symmetric group.}
For each partition \(\lambda\vdash t\), let \(S^\lambda\) denote the Specht module of shape \(\lambda\). The modules
\(
\{S^\lambda:\lambda\vdash t\}
\)
form a complete set of pairwise non-isomorphic irreducible complex representations of \(\mathfrak{S}_t\).

More concretely, if \(T\) is a Young tableau of shape \(\lambda\), with row
stabilizer \(R_T\) and column stabilizer \(C_T\), the associated Young
symmetrizer is
\[
c_T:=
\left(\sum_{\sigma\in R_T}\sigma\right)
\left(\sum_{\tau\in C_T}\operatorname{sgn}(\tau)\tau\right)
\in\CC[\mathfrak{S}_t].
\]
Up to the standard choice of left versus right convention, the Specht module
is realized as the left ideal
\[
S^\lambda \cong \CC[\mathfrak{S}_t]c_T.
\]

Let
\(
 [\lambda]:=\{(i,j)\in\mathbb N^2:
          1\leq i\leq\ell(\lambda),\ 1\leq j\leq\lambda_i\}
\) be the Young diagram of \(\lambda\).

For a box \(x=(i,j)\in[\lambda]\), its hook length is
\[
 h_\lambda(i,j)
 :=\lambda_i-j+\lambda_j'-i+1,
\]
where \(\lambda_j':=|\{r:\lambda_r\geq j\}|\) is the \(j\)-th part of
the conjugate partition.

Over \(\mathbb C\), the irreducible representations of \(\mathfrak S_t\)
are indexed by partitions \(\lambda\vdash t\). We write
\(W_\lambda:=S^\lambda\) for the corresponding Specht module. Its
dimension is
\begin{equation}
 \dim_{\mathbb{C}} W_\lambda
 =\frac{t!}{\prod_{(i,j)\in[\lambda]}h_\lambda(i,j)}.
 \label{eq:hook-length-formula-corrected}
\end{equation}

\medskip
\item \textbf{The general linear and unitary groups.}
For \(\lambda\in\Lambda(d,t)\), let
\(
\mathbb{S}_\lambda(\CC^d)
\)
denote the Schur module of shape \(\lambda\). It is the irreducible
polynomial \(\GL_d(\CC)\)-representation of highest weight \(\lambda\) and
homogeneous degree \(t\). Its restriction to \(\U(d)\) is irreducible.

Let
\(
V_\lambda:=\mathbb{S}_\lambda(\CC^d).
\)
Its dimension is given by the Weyl dimension formula
\[
\dim_{\mathbb{C}} V_\lambda
=
\prod_{1\leq i<j\leq d}
\frac{\lambda_i-\lambda_j+j-i}{j-i}.
\]
\end{enumerate}

\begin{definition}[Removable box and predecessor]
\label{def:removable-prelim}
A box
\(
    b\in[\lambda]
\)
is \emph{removable} if
\(
    [\lambda]\setminus\{b\}
\)
is again a Young diagram.  If
\(
    \mu\vdash s-1
\)
is obtained by removing one removable box from \(\lambda\), we write
\(
    \mu\nearrow\lambda
\)
or equivalently
\(
    \lambda\searrow\mu.
\)
\end{definition}

\begin{definition}[Young graph]
\label{def:Young-graph-prelim}
The \emph{Young graph} is the graded graph whose vertices at level
\(s\) are the partitions of \(s\), with an edge
\(
    \mu\nearrow\lambda
\)
whenever \(\lambda\) is obtained from \(\mu\) by adding a single
box.
\end{definition}

For \(2\leq r\leq s\), define the Jucys--Murphy element
\begin{equation}
\label{eq:JM-prelim}
    J_r
    :=
    \sum_{q=1}^{r-1}(q\,r)
    \in
    \mathbb C[\mathfrak{S}_s],
    \qquad
    J_1:=0.
\end{equation}

\begin{theorem}[Young--Yamanouchi basis]
\label{thm:YY-prelim}
For every
\(
    \lambda\vdash s
\),
the Specht module \(W_\lambda\) admits an orthonormal basis
\begin{equation}
\label{eq:YY-basis-prelim}
    \left\{
       |T\rangle:
       T\in\operatorname{SYT}(\lambda)
    \right\}
\end{equation}
such that
\begin{equation}
\label{eq:JM-diagonal-prelim}
    J_r|T\rangle
    =
    c_T(r)|T\rangle,
    \qquad
    1\leq r\leq s.
\end{equation}
We call
\eqref{eq:YY-basis-prelim}
the Young--Yamanouchi, or Young seminormal, basis.
\end{theorem}

If \(T\in\operatorname{SYT}(\lambda)\), the largest entry \(s\)
necessarily occupies a removable box.  Deleting that box gives a
standard tableau
\(
    T\!\downarrow_{s-1}
\)
of a uniquely determined shape
\(
    \mu\nearrow\lambda.
\)

The occurrence of the same partition \(\lambda\) on both sides of
Schur-Weyl duality is explained by the fact that the two actions below are
mutual commutants.

The image algebras corresponding to \(\GL_{d}(\mathbb{C})\) and \(\mathfrak{S}_{t}\) are defined by
\begin{align}
&\mathcal{A}_{\GL_{d}(\mathbb{C})}
:= 
\operatorname{span}_{\CC}\{\Gamma_{\GL_{d}(\mathbb{C})}(U):U\in \GL_{d}(\mathbb{C})\},
\\
&\mathcal{A}_{\mathfrak{S}_{t}}
:=\Gamma_{\mathfrak{S}_{t}}\bigl(\CC[\mathfrak{S}_t]\bigr)
=\operatorname{span}_{\CC}\{\Gamma_{\mathfrak{S}_{t}}(\sigma):\sigma\in\mathfrak{S}_t\}.
\end{align}
For a subalgebra \(\mathcal{A}\subseteq\End(\mathcal{H}^{\otimes t}_{d})\), let
\(
\Comm(\mathcal{A})
:=
\{T\in\End(\mathcal{H}^{\otimes t}_{d}):TA=AT
\text{ for every }A\in\mathcal{A}\}.
\)

\begin{theorem}[Schur-Weyl duality]
\label{thm:SW-duality}
The two image algebras are full mutual commutants:
\[
\Comm(\mathcal{A}_{\GL_{d}(\mathbb{C})})=\mathcal{A}_{\mathfrak{S}_{t}}=\operatorname{End}_{\GL_{d}(\mathbb{C})}(\mathcal{H}_{d}^{\otimes t}),
\qquad
\Comm(\mathcal{A}_{\mathfrak{S}_{t}})=\mathcal{A}_{\GL_{d}(\mathbb{C})}=\operatorname{End}_{\mathfrak{S}_t}(\mathcal{H}_{d}^{\otimes t}).
\]
An operator
\(T\in\operatorname{End}(\mathcal H_d^{\otimes t})\)
commutes with \(U^{\otimes t}\) for every \(U\in\mathrm U(d)\)
if and only if
\(
T\in
\operatorname{span}_{\mathbb C}
\{\Gamma_{\mathfrak S_t}(\sigma):\sigma\in\mathfrak S_t\}.
\)
Conversely, \(T\) commutes with
\(\Gamma_{\mathfrak S_t}(\sigma)\) for every
\(\sigma\in\mathfrak S_t\) if and only if
\(
T\in
\operatorname{span}_{\mathbb C}
\{U^{\otimes t}:U\in\mathrm U(d)\}.
\)
\end{theorem}

\begin{remark}\label{rem:AW-double-centralizer}
When we consider the actions of the groups \(\U(d)\) and \(\mathfrak{S}_{t}\) on \(\mathcal{H}^{\otimes t}_{d}\), the Schur-Weyl duality similarly holds
\[
\Comm(\mathcal{A}_{\U(d)})=\mathcal{A}_{\mathfrak{S}_{t}}=\operatorname{End}_{\U(d)}(\mathcal{H}_{d}^{\otimes t}),
\qquad
\Comm(\mathcal{A}_{\mathfrak{S}_{t}})=\mathcal{A}_{\U(d)}=\operatorname{End}_{\mathfrak{S}_t}(\mathcal{H}_{d}^{\otimes t}).
\]
Since \(\CC[\mathfrak{S}_t]\) is semisimple by Maschke's theorem, its image \(\mathcal{A}_{\mathfrak{S}_{t}}\) is also semisimple. The Artin--Wedderburn theorem and the double-centralizer theorem then imply that
\(
\Comm\bigl(\Comm(\mathcal{A}_{\mathfrak{S}_{t}})\bigr)
=
\mathcal{A}_{\mathfrak{S}_{t}}.
\)
Thus the second commutant equality follows from the first together with semisimplicity.
\end{remark}

\subsubsection{From Schur-Weyl duality to the Schur-Weyl decomposition}
\label{subsec:SW-decomposition}

Let
\[
\mathcal{A}_{\mathfrak{S}_{t}}
:=
\operatorname{span}_{\CC}
\{\Gamma_{\mathfrak{S}_{t}}(\sigma):\sigma\in\mathfrak{S}_t\},
\qquad
\mathcal{A}_{\U(d)}
:=
\operatorname{span}_{\CC}
\{U^{\otimes t}:U\in\U(d)\}.
\]
By Theorem~\ref{thm:SW-duality} and Remark \ref{rem:AW-double-centralizer},
\[
\Comm(\mathcal{A}_{\mathfrak{S}_{t}})
=
\mathcal{A}_{\U(d)},
\qquad
\Comm(\mathcal{A}_{\U(d)})
=
\mathcal{A}_{\mathfrak{S}_{t}},
\]
where the commutants are taken inside
\(\End(\mathcal{H}^{\otimes t}_{d})\).

We recall the finite-dimensional double-centralizer theorem. Let \(A\subseteq\End(W)\) be a finite-dimensional semisimple algebra and let
\(
B:=\End_A(W).
\)
If \(\{U_\alpha\}_{\alpha\in I}\) is a complete set of pairwise
non-isomorphic simple \(A\)-modules occurring in \(W\), then
\[
W\cong\bigoplus_{\alpha\in I}U_\alpha\otimes M_\alpha,
\qquad
M_\alpha:=\Hom_A(U_\alpha,W),
\]
as an \(A\)-\(B\)-bimodule. Moreover,
\[
B\cong\bigoplus_{\alpha\in I}\End(M_\alpha),
\qquad
\End_B(W)=A,
\]
and the nonzero multiplicity spaces \(M_\alpha\) are pairwise
non-isomorphic irreducible \(B\)-modules.

Since \(\CC[\mathfrak{S}_t]\) is semisimple by Maschke's theorem, its image
\(\mathcal{A}_{\mathfrak{S}_{t}}\) is a semisimple algebra. Applying the
double-centralizer theorem to
\[
A=\mathcal{A}_{\mathfrak{S}_{t}},
\qquad
B=\Comm(\mathcal{A}_{\mathfrak{S}_{t}})
=
\mathcal{A}_{\U(d)},
\]
gives
\[
\mathcal{H}^{\otimes t}_{d}
\cong
\bigoplus_{\lambda}
S^\lambda\otimes
\Hom_{\mathfrak{S}_t}
\bigl(S^\lambda,\mathcal{H}^{\otimes t}_{d}\bigr).
\label{eq:SW-isotypic}
\]
Here the sum ranges over those partitions \(\lambda\vdash t\) for which
\(S^\lambda\) occurs in \(\mathcal{H}^{\otimes t}_{d}\).

The multiplicity space
\(
\Hom_{\mathfrak{S}_t}
\bigl(S^\lambda,\mathcal{H}^{\otimes t}_{d}\bigr)
\)
is an irreducible \(\mathcal{A}_{\U(d)}\)-module, hence an irreducible
\(\U(d)\)-module. It is naturally identified with the Schur module
\(
\mathbb{S}_\lambda(\CC^d).
\)
Furthermore,
\[
\mathbb{S}_\lambda(\CC^d)\neq 0
\quad\Longleftrightarrow\quad
\ell(\lambda)\leq d.
\]
If \(\ell(\lambda)>d\), then the Young diagram of \(\lambda\) has a column of length \(d+1\), and the corresponding column antisymmetrization vanishes because
\(
\bigwedge^{d+1}\CC^d=0.
\)

Consequently, as a representation of
\(\U(d)\times\mathfrak{S}_t\),
\begin{equation}
(\CC^d)^{\otimes t}
\cong
\bigoplus_{\lambda\in\Lambda(d,t)}
\big(\mathbb{S}_\lambda(\CC^d)\otimes S^\lambda \big)= \bigoplus_{\lambda\in\Lambda(d,t)} \big(V_{\lambda} \otimes W_{\lambda} \big)
\label{eq:SW-decomp}
\end{equation}
Each summand occurs exactly once. Thus,
\eqref{eq:SW-decomp} is multiplicity-free as a representation of the
product group \(\U(d)\times\mathfrak{S}_t\), though not generally after
forgetting either factor.

Taking dimensions yields
\[
d^t
=
\sum_{\lambda\in\Lambda(d,t)}
\dim_{\mathbb{C}} V_\lambda \cdot  \dim_{\mathbb{C}} W_\lambda.
\]

Moreover, after choosing invariant Hermitian inner products on the
irreducible factors, there exists a unitary
\((\U(d)\times\mathfrak{S}_t)\)-equivariant isomorphism
\begin{equation}
\Phi:\mathcal{H}^{\otimes t}_{d}
\xrightarrow{\ \cong\ }
\bigoplus_{\lambda\in\Lambda(d,t)}
V_\lambda\otimes W_\lambda,
\label{eq:schur-weyl-decomposition}
\end{equation}
where
\[
V_\lambda\cong\mathbb{S}_\lambda(\CC^d),
\qquad
W_\lambda\cong S^\lambda.
\]
Under \(\Phi\), the two actions have the block forms
\begin{align}
\Phi \Gamma_{\U(d)}(U)\Phi^{-1}
&=
\bigoplus_{\lambda\in\Lambda(d,t)}
\left(V_\lambda(U)\otimes I_{W_\lambda}\right),
\qquad U\in\U(d),                                      \label{eq:U-block-action}\\
\Phi \Gamma_{\mathfrak{S}_{t}}(\sigma)\Phi^{-1}
&=
\bigoplus_{\lambda\in\Lambda(d,t)}
\left(I_{V_\lambda}\otimes W_\lambda(\sigma)\right),
\qquad \sigma\in\mathfrak{S}_t,                        \label{eq:S-block-action}
\end{align}
where
\[
\Gamma_{\U(d),V_\lambda} \colon \U(d)\longrightarrow \GL(V_\lambda)
\]
and 
\[
\Gamma_{\mathfrak{S}_{t},W_\lambda} \colon \mathfrak{S}_{t} \longrightarrow\GL(W_\lambda)
\]
are finite-dimensional irreducible representations, and with abuse of notations, we denote \(\Gamma_{\U(d),V_\lambda}(U)=V_{\lambda}(U)\) and \(\Gamma_{\mathfrak{S}_{t},W_\lambda}(\sigma)=W_{\lambda}(\sigma)\) for notational simplicity.

\smallskip
The decomposition
\eqref{eq:schur-weyl-decomposition} is multiplicity-free as a representation of \(\U(d)\times\mathfrak{S}_t\). Upon restriction to either factor of
\(\U(d)\times\mathfrak{S}_t\), the decomposition becomes
\begin{align}
\mathcal{H}^{\otimes t}_{d}\big|_{\U(d)}
&\cong
\bigoplus_{\lambda\in\Lambda(d,t)}
V_\lambda^{\oplus \dim_{\mathbb{C}} W_\lambda},                         \label{eq:U-isotypic}\\
\mathcal{H}^{\otimes t}_{d}\big|_{\mathfrak{S}_t}
&\cong
\bigoplus_{\lambda\in\Lambda(d,t)}
W_\lambda^{\oplus \dim_{\mathbb{C}} V_\lambda}.                         \label{eq:S-isotypic}
\end{align}
Thus \(W_\lambda\) is the multiplicity space of \(V_\lambda\) in the
\(\U(d)\)-representation \(\mathcal{H}^{\otimes t}_{d}\), whereas \(V_\lambda\) is
the multiplicity space of \(W_\lambda\) in the
\(\mathfrak{S}_t\)-representation \(\mathcal{H}^{\otimes t}_{d}\).

\begin{remark}
For every \(\lambda\in\Lambda(d,t)\), \(U\in\U(d)\), and
\(\sigma\in\mathfrak{S}_t\),
\[
\bigl(V_\lambda(U)\otimes I_{W_\lambda}\bigr)
\bigl(I_{V_\lambda}\otimes W_\lambda(\sigma)\bigr)
=
V_\lambda(U)\otimes W_\lambda(\sigma)
=
\bigl(I_{V_\lambda}\otimes W_\lambda(\sigma)\bigr)
\bigl(V_\lambda(U)\otimes I_{W_\lambda}\bigr).
\]
\end{remark}

\begin{theorem}\label{thm:finite-character-orthogonality}
Let \(\mathrm{G}\) be a finite group, and let
\(\{\chi_\lambda\}_{\lambda\in\widehat{\mathrm{G}}}\) be its irreducible complex
characters. With respect to the Hermitian inner product on the space of
class functions,
\[
\langle \varphi,\psi\rangle_{\mathrm{G}}
:=
\frac{1}{|\mathrm{\mathrm{G}}|}
\sum_{g\in \mathrm{G}}\varphi(g)\overline{\psi(g)},
\]
the irreducible characters form an orthonormal basis. In particular,
\begin{equation}
\frac{1}{|\mathrm{G}|}
\sum_{g\in \mathrm{G}}
\chi_\lambda(g)\overline{\chi_\mu(g)}
=
\delta_{\lambda\mu}.
\label{eq:row-character-orthogonality}
\end{equation}

Moreover, for all \(g,h\in \mathrm{G}\),
\begin{equation}
\sum_{\lambda\in\widehat{\mathrm{G}}}
\chi_\lambda(g)\overline{\chi_\lambda(h)}
=
\begin{cases}
|C_{\mathrm{G}}(g)|, & \text{if \(g\) and \(h\) are conjugate in \(\mathrm{G}\)},\\
0, & \text{otherwise},
\end{cases}
\label{eq:column-character-orthogonality}
\end{equation}
where
\[
C_{\mathrm{G}}(g):=\{x\in \mathrm{G}:xg=gx\}.
\]
Since
\(\overline{\chi_\lambda(h)}=\chi_\lambda(h^{-1})\),
\[
\sum_{\lambda\in\widehat{\mathrm{G}}}
\chi_\lambda(g)\chi_\lambda(h^{-1})
=
\begin{cases}
|C_{\mathrm{G}}(g)|, & \text{if \(g\) and \(h\) are conjugate in \(\mathrm{G}\)},\\
0, & \text{otherwise}.
\end{cases}
\]
\end{theorem}

The following standard result is due to Serre; see
\cite[Theorem~8, \S2.6]{serre1977linear}.

\begin{theorem}[Canonical isotypic decomposition and character projectors~\cite{serre1977linear,fulton2013representation,sagan2001symmetric}]
\label{thm:serre-canonical-projection}
Let
\[
\Gamma_{\mathrm{G}}\colon \mathrm{G}\longrightarrow\GL(V)
\]
be a finite-dimensional complex representation of a finite group \(\mathrm{G}\).
Let \(W_1,\ldots,W_k\) be representatives of the distinct irreducible
\(\mathrm{G}\)-modules occurring in \(V\), and let \(\chi_i\) denote the character of
\(W_i\).

For \(1\leq i\leq k\), define the \(W_i\)-isotypic component of \(V\) by
\[
V^{(i)}
:=
\sum_{\substack{L\subseteq V\\L\cong W_i}}L.
\]
Then the following hold.

\begin{enumerate}[(i)]
\item The subspaces \(V^{(i)}\) are canonical \(\mathrm{G}\)-submodules of \(V\), and
\[
V=\bigoplus_{i=1}^{k}V^{(i)},
\]
where
\[
V^{(i)}
\cong
W_i\otimes\Hom_{\mathrm{G}}(W_i,V).
\]
Thus the isotypic decomposition is independent of any choice of a
decomposition of \(V\) into irreducible subrepresentations.

\item The \(\mathrm{G}\)-equivariant projection onto \(V^{(i)}\) is
\begin{equation}
\Pi_i
=
\frac{\dim_{\mathbb{C}} W_i}{|\mathrm{G}|}
\sum_{g\in \mathrm{G}}
\overline{\chi_i(g)}\,\Gamma_{\mathrm{G}}(g)
=
\frac{\dim_{\mathbb{C}} W_i}{|\mathrm{G}|}
\sum_{g\in \mathrm{G}}
\chi_i(g^{-1})\,\Gamma_{\mathrm{G}}(g).
\label{eq:finite-group-isotypic-projector}
\end{equation}
In particular,
\[
\Pi_i^2=\Pi_i,
\qquad
\operatorname{im}(\Pi_i)=V^{(i)},
\qquad
\Pi_i\Pi_j=0 \quad (i\neq j).
\]

\item If \(V\) is equipped with a \(\mathrm{G}\)-invariant Hermitian inner product,
then \(\Pi_i\) is the orthogonal projection of \(V\) onto \(V^{(i)}\).
\end{enumerate}
\end{theorem}

\begin{theorem}[Character projector for the \(\mathfrak{S}_t\)-action~\cite{sagan2001symmetric}]
\label{thm:St-character-projector}
Let
\(
\mathcal H^{\otimes t}_{d}:=(\CC^d)^{\otimes t},
\)
and fix \(\lambda\vdash t\) with \(\ell(\lambda)\leq d\). Let
\(\chi_{W_\lambda}\) be the irreducible character of the Specht module
\[
W_\lambda\cong S^\lambda.
\]
Then the orthogonal projection of \(\mathcal H^{\otimes t}_{d}\) onto the
\(W_\lambda\)-isotypic component
\(
\Phi^{-1}(V_\lambda\otimes W_\lambda),
\)
is
\begin{equation}\label{eq:charprojector}
\Pi_\lambda^{(\mathfrak{S}_t)}
=
\frac{\dim_{\mathbb{C}} W_\lambda}{t!}
\sum_{\sigma\in\mathfrak{S}_t}
\chi_{W_\lambda}(\sigma^{-1})\Gamma_{\mathrm{\mathfrak{S}_{t}}}(\sigma).
\end{equation}
Since every irreducible character of \(\mathfrak{S}_t\) is real-valued,
\[
\chi_{W_\lambda}(\sigma^{-1})
=
\chi_{W_\lambda}(\sigma),
\]
and hence
\[
\Pi_\lambda^{(\mathfrak{S}_t)}
=
\frac{\dim_{\mathbb{C}} W_\lambda}{t!}
\sum_{\sigma\in\mathfrak{S}_t}
\chi_{W_\lambda}(\sigma)\Gamma_{\mathrm{\mathfrak{S}_{t}}}(\sigma).
\]
\end{theorem}

\begin{theorem}[Character projector for the \(\U(d)\)-action~\cite{fulton2013representation,sagan2001symmetric}]
\label{thm:Ud-character-projector}
Fix \(\lambda\vdash t\) with \(\ell(\lambda)\leq d\), and let
\[
\chi_{V_\lambda}(U)
:=
\operatorname{Tr}(V_\lambda(U))
\]
be the character of the irreducible \(\U(d)\)-module
\[
V_\lambda\cong\mathbb{S}_\lambda(\CC^d).
\]
Let \(d\mu_{\mathrm{Haar}}\) denote normalized Haar probability measure on
\(\U(d)\). Then the orthogonal projection of
\(\mathcal H^{\otimes t}_{d}=(\CC^d)^{\otimes t}\) onto the
\(V_\lambda\)-isotypic component
\(
\Phi^{-1}(V_\lambda\otimes W_\lambda),
\)
is
\begin{align}
\Pi_\lambda^{(\U(d))}
&=
\dim_{\mathbb{C}}(V_\lambda)
\int_{\U(d)}
\chi_{V_\lambda}(U^{-1})\,
U^{\otimes t}\,
d\mu_{\mathrm{Haar}}(U) \notag\\
&=
\dim_{\mathbb{C}}(V_\lambda)
\int_{\U(d)}
\overline{\chi_{V_\lambda}(U)}\,
U^{\otimes t}\,
d\mu_{\mathrm{Haar}}(U).
\label{eq:Ud-character-projector-conjugate}
\end{align}
\end{theorem}

\begin{corollary}~\cite{fulton2013representation,sagan2001symmetric}\label{cor:equalityofprojectors}
For every \(\lambda\in\Lambda(d,t)\),
\[
\Pi_\lambda^{(\mathfrak{S}_t)}
=
\Pi_\lambda^{(\U(d))}.
\]
Both operators are the orthogonal projection of
\(\mathcal H^{\otimes t}_{d}\) onto the Schur-Weyl summand
\[
\mathcal H_\lambda
:=
\Phi^{-1}(V_\lambda\otimes W_\lambda)
\subseteq\mathcal H^{\otimes t}_{d}.
\]
Consequently,
\begin{align}\label{eq:projsum}
I_{\mathcal H^{\otimes t}_{d}}
=
\sum_{\lambda\in\Lambda(d,t)}\Pi_\lambda,
\qquad
\Pi_\lambda\Pi_\mu
=
\delta_{\lambda\mu}\Pi_\lambda.
\end{align}
\end{corollary}

\begin{corollary}\label{cor:schur-integral}
Let
\[
\Gamma_{\U(d),V_\lambda} \colon \U(d)\longrightarrow\U(V_\lambda)
\]
be a finite-dimensional irreducible unitary representation, and let
\[
V_\lambda(U):=\Gamma_{\U(d),V_\lambda}(U).
\]
Let \(d\mu_{\mathrm{Haar}}\) be normalized Haar probability measure on
\(\U(d)\). Then, for every \(A\in\End(V_\lambda)\),
\begin{equation}
\int_{\U(d)}
V_\lambda(U)A V_\lambda(U)^\dagger\,
d\mu_{\mathrm{Haar}}(U)
=
\frac{\Tr(A)}{\dim V_\lambda}I_{V_\lambda}.
\label{eq:schur-integral}
\end{equation}
\end{corollary}

\begin{proof}
See Appendix~\ref{cor:schur-integral}.
\end{proof}

\subsection{Symmetric, antisymmetric, and irreducible subspaces}
\label{subsec:symmetric-antisymmetric-irrep-subspaces}

\subsubsection{The symmetric subspace}
\label{subsubsec:symmetric-subspace-corrected}

We use the permutation action \(\Gamma_{\mathfrak{S}_{t}}(\sigma)\) from
\eqref{eq:permaction}. The \(t\)-fold symmetric tensor power
is
\begin{equation}
 \operatorname{Sym}^t(\mathbb C^d)
 :=\{\psi\in\mathcal H^{\otimes t}_{d}\colon\Gamma_{\mathfrak{S}_{t}}(\sigma)\psi=\psi
                            \text{ for every }\sigma\in\mathfrak S_t\}.
 \label{eq:symmetric-subspace-definition-corrected}
\end{equation}
Define the symmetrizer
\begin{equation}
 \Pi_{\mathrm{sym}}^{(d,t)}
 :=\frac{1}{t!}\sum_{\sigma\in\mathfrak S_t}\Gamma_{\mathfrak{S}_{t}}(\sigma).
 \label{eq:symmetrizer-corrected}
\end{equation}

\begin{proposition}~\cite{fulton2013representation,simon1996representations}
\label{prop:symmetrizer-corrected}
The operator \(\Pi_{\mathrm{sym}}^{(d,t)}\) is the orthogonal projector
onto the symmetric subspace. In addition,
\begin{equation}
 \dim_{\mathbb{C}}\big(\operatorname{Sym}^t(\mathbb C^d)\big)
 =\binom{d+t-1}{t}.
 \label{eq:symmetric-dimension-corrected}
\end{equation}
\end{proposition}

In Schur-Weyl notation, \(W_{(t)}\) is the one-dimensional trivial
representation of \(\mathfrak S_t\), and
\begin{equation}
 \operatorname{Sym}^t(\mathbb C^d)
 =\mathcal H_{(t)}
 \cong V_{(t)}\otimes W_{(t)}
 \cong V_{(t)}.
 \label{eq:symmetric-schur-weyl-block}
\end{equation}
Formula \eqref{eq:symmetrizer-corrected} is the character
projector \eqref{eq:charprojector} for \(\lambda=(t)\).

If \(\mathcal D(\operatorname{Sym}^t(\mathbb C^d))\) denotes the density
operators supported on the symmetric subspace, then
\[
 \mathcal D(\operatorname{Sym}^t(\mathbb C^d))
 \subseteq\operatorname{span}_{\mathbb R}
 \{(\ket{\phi}\!\bra{\phi})^{\otimes t}:
             \phi\in\mathbb C^d,\ \|\phi\|=1\}.
\]
The coefficients in such a real linear combination need not be nonnegative.

\subsubsection{The antisymmetric subspace}
\label{subsubsec:antisymmetric-subspace-corrected}

Let \(\operatorname{sgn}:\mathfrak S_t \longrightarrow \{\pm1\}\) be the sign
character. The \(t\)-fold antisymmetric tensor power is
\begin{equation}
 \operatorname{Alt}^t(\mathbb C^d)
 :=\{\psi\in\mathcal H^{\otimes t}_{d}:\Gamma_{\mathfrak{S}_{t}}(\sigma)\psi
       =\operatorname{sgn}(\sigma)\psi
       \text{ for every }\sigma\in\mathfrak S_t\}.
 \label{eq:antisymmetric-subspace-definition}
\end{equation}
It is canonically identified with the exterior power
\(\bigwedge^t\mathbb C^d\). Define the antisymmetrizer
\begin{equation}
 \Pi_{\mathrm{alt}}^{(d,t)}
 :=\frac{1}{t!}\sum_{\sigma\in\mathfrak S_t}
       \operatorname{sgn}(\sigma)\Gamma_{\mathfrak{S}_{t}}(\sigma).
 \label{eq:antisymmetrizer}
\end{equation}

\begin{proposition}~\cite{fulton2013representation,simon1996representations}
\label{prop:antisymmetrizer}
The operator \(\Pi_{\mathrm{alt}}^{(d,t)}\) is the orthogonal projector
onto the antisymmetric subspace. Moreover,
\begin{equation}
 \dim_{\mathbb{C}}\operatorname{Alt}^t(\mathbb C^d)
 =
 \begin{cases}
   \binom dt,&1\leq t\leq d,\\
   0,&t>d.
 \end{cases}
 \label{eq:antisymmetric-dimension}
\end{equation}
When \(t\leq d\), define, for \(1\leq i_1<\cdots<i_t\leq d\),
\begin{equation}
\begin{split}
 \ket{e_{i_1}\wedge\cdots\wedge e_{i_t}}
 :={}&\frac{1}{\sqrt{t!}}\sum_{\sigma\in\mathfrak S_t}
       \operatorname{sgn}(\sigma){}\times
       \ket{e_{i_{\sigma^{-1}(1)}}}\otimes\cdots\otimes
       \ket{e_{i_{\sigma^{-1}(t)}}}.
\end{split}
 \label{eq:normalized-wedge-vector}
\end{equation}
Then an orthonormal basis is
\begin{equation}
 \left\{\ket{e_{i_1}\wedge\cdots\wedge e_{i_t}}:
       1\leq i_1<\cdots<i_t\leq d\right\}.
 \label{eq:wedge-orthonormal-basis}
\end{equation}
\end{proposition}

If \(t\leq d\), then \(W_{(1^t)}\) is the one-dimensional sign
representation and \(V_{(1^t)}\cong\bigwedge^t\mathbb C^d\). Hence
\begin{equation}
 \operatorname{Alt}^t(\mathbb C^d)
 =\mathcal H_{(1^t)}
 \cong V_{(1^t)}\otimes W_{(1^t)}
 \cong\bigwedge^t\mathbb C^d.
 \label{eq:antisymmetric-schur-weyl-block}
\end{equation}
Formula \eqref{eq:antisymmetrizer} is the character projector
\eqref{eq:charprojector} for \(\lambda=(1^t)\). If
\(t>d\), then \((1^t)\notin\Lambda(d,t)\), consistently with the vanishing
of the exterior power.

\begin{theorem}[Symmetric-group branching rule~\cite{sagan2001symmetric,fulton2013representation}]
\label{thm:branching-prelim}
For every
\(
    \lambda\vdash s
\),
the restriction of \(W_\lambda\) to the natural subgroup
\(
    \mathfrak{S}_{s-1}\subset \mathfrak{S}_s
\)
is multiplicity-free:
\begin{equation}
\label{eq:branching-prelim}
    W_\lambda
    \downarrow_{\mathfrak{S}_{s-1}}
    \cong
    \bigoplus_{\mu\nearrow\lambda}
       W_\mu .
\end{equation}
Under the Young--Yamanouchi basis, the summand associated with
\(
    \mu\nearrow\lambda
\)
is spanned by those
\(
    |T\rangle
\)
for which
\[
    \operatorname{shape}
    \left(
       T\!\downarrow_{s-1}
    \right)
    =
    \mu.
\]
\end{theorem}

Because
\eqref{eq:branching-prelim}
is multiplicity-free, every
\(W_\mu\)-summand is canonical as an \(\mathfrak{S}_{s-1}\)-isotypic subspace of the restricted module. We therefore fix the corresponding orthogonal branching projection
\begin{equation}
\label{eq:q-branch-prelim}
    q_{\mu\leftarrow\lambda}
    :
    W_\lambda
    \longrightarrow
    W_\mu,
    \qquad
    \mu\nearrow\lambda.
\end{equation}

\begin{remark}
\label{rem:tableau-path-prelim}
A standard Young tableau
\(
    T\in\operatorname{SYT}(\lambda)
\)
can be realized as a path
\[
    (1)
    =
    \lambda^{(1)}
    \nearrow
    \lambda^{(2)}
    \nearrow
    \cdots
    \nearrow
    \lambda^{(s)}
    =
    \lambda
\]
in the Young graph, where
\[
    \lambda^{(r)}
    =
    \operatorname{shape}
    \left(
       T\!\downarrow_r
    \right).
\]
This interpretation will be used in
Corollary~\ref{cor:young-yamanouchi-heredity}.
\end{remark}

For \(\lambda\vdash s\), \(W_\lambda\) denotes the irreducible
Specht module of shape \(\lambda\); see, e.g.,
\cite{sagan2001symmetric,fulton2013representation}.

\subsection{Schur functors and the Pieri rule}
\label{subsec:Pieri-prelim}

For a partition \(\lambda\) with
\(
    \ell(\lambda)\leq d
\),
we write
\(
    \mathbb S_\lambda(\mathcal H_d)
\)
for the corresponding Schur functor.  In the notation of
Schur-Weyl duality,
\(
    V_\lambda
    =
    \mathbb S_\lambda(\mathcal H_d).
\)

\begin{definition}[Horizontal and vertical strips]
\label{def:strips-prelim}
Let
\(
    \lambda\subseteq\nu
\)
be Young diagrams.

The skew diagram
\(
    \nu/\lambda
\)
is a \emph{horizontal \(r\)-strip} if it contains \(r\) boxes and
no two of them lie in the same column.

It is a \emph{vertical \(r\)-strip} if it contains \(r\) boxes and
no two of them lie in the same row.
\end{definition}

\begin{theorem}[Pieri rule~\cite{fulton2013representation,macdonald1998symmetric}]
\label{thm:Pieri-prelim}
Let
\(
    \lambda
\)
be a partition with
\(
    \ell(\lambda)\leq d
\).

For every \(r\geq0\),
\begin{equation}
\label{eq:Pieri-sym-prelim}
    V_\lambda
    \otimes
    \operatorname{Sym}^{r}(\mathcal H_d)
    \cong
    \bigoplus_{\substack{
       \nu:\,
       \nu/\lambda
       \text{ is a horizontal }r\text{-strip}\\
       \ell(\nu)\leq d
    }}
       V_\nu,
\end{equation}
and
\begin{equation}
\label{eq:Pieri-alt-prelim}
    V_\lambda
    \otimes
    \bigwedge\nolimits^{r}\mathcal H_d
    \cong
    \bigoplus_{\substack{
       \nu:\,
       \nu/\lambda
       \text{ is a vertical }r\text{-strip}\\
       \ell(\nu)\leq d
    }}
       V_\nu.
\end{equation}
Every irreducible summand in
\eqref{eq:Pieri-sym-prelim}
and
\eqref{eq:Pieri-alt-prelim}
occurs with multiplicity one.
\end{theorem}

The case \(r=1\) will be used repeatedly.

\begin{corollary}[One-box Pieri rule~\cite{fulton2013representation}]
\label{cor:one-box-Pieri-prelim}
For
\(
    \mu\vdash s-1
\),
\begin{equation}
\label{eq:one-box-Pieri}
    \mathcal H_d\otimes V_\mu
    \cong
    \bigoplus_{\substack{
       \lambda\vdash s\\
       \mu\nearrow\lambda\\
       \ell(\lambda)\leq d
    }}
       V_\lambda,
\end{equation}
and
\begin{equation}
\label{eq:contraction-Hom-Pieri}
    \dim_{\mathbb C}
    \operatorname{Hom}_{\mathrm U(d)}
    \left(
       \mathcal H_d^*
       \otimes V_\lambda,
       V_\mu
    \right)
    =
    \begin{cases}
    1,
       &\mu\nearrow\lambda,
    \\[1mm]
    0,
       &\text{otherwise}.
    \end{cases}
\end{equation}
\end{corollary}

\begin{proof}
The decomposition
\eqref{eq:one-box-Pieri}
is
Theorem~\ref{thm:Pieri-prelim}
with \(r=1\).
Using the natural tensor--Hom adjunction,
\begin{align}
    \operatorname{Hom}_{\mathrm U(d)}
    \left(
       \mathcal H_d^*\otimes V_\lambda,
       V_\mu
    \right)
    &\cong
    \operatorname{Hom}_{\mathrm U(d)}
    \left(
       V_\lambda,
       \mathcal H_d\otimes V_\mu
    \right),
\end{align}
and Schur's lemma together with the multiplicity-one statement in
\eqref{eq:one-box-Pieri}
gives
\eqref{eq:contraction-Hom-Pieri}.
\end{proof}

\subsubsection{Exterior-power duality and fundamental weights}
\label{subsubsec:exterior-duality-prelim}

Let
\[
    \det:
    \mathrm U(d)
    \longrightarrow
    \mathrm U(1)
\]
denote the determinant character.

\begin{proposition}[Duality of exterior powers~\cite{fulton2013representation,GoodmanWallach1998}]
\label{prop:exterior-duality-prelim}
For
\(
    0\leq r\leq d
\),
there is a canonical \(\mathrm U(d)\)-equivariant isomorphism
\begin{equation}
\label{eq:exterior-duality-prelim}
    \left(
       \bigwedge\nolimits^{r}\mathcal H_d
    \right)^*
    \cong
    \bigwedge\nolimits^{d-r}\mathcal H_d
    \otimes
    (\det)^{-1}.
\end{equation}
\end{proposition}

Let
\(
    \omega_1,\ldots,\omega_{d-1}
\)
denote the fundamental dominant weights of
\(\mathrm{SU}(d)\), and set
\(
    \omega_0=\omega_d:=0.
\)
For a dominant weight \(\gamma\), we write
\(
    L_\gamma
\)
for the corresponding irreducible \(\mathrm{SU}(d)\)-module.

\subsubsection{The Cauchy decomposition}
\label{subsubsec:Cauchy-prelim}

\begin{theorem}[Cauchy identity for symmetric powers~\cite{fulton2013representation,macdonald1998symmetric}]\label{thm:Cauchy-prelim}
Let \(X\) and \(Y\) be finite-dimensional complex vector spaces.
Then, as a
\(
    \mathrm{GL}(X)\times\mathrm{GL}(Y)
\)-module,
\begin{equation}
\label{eq:Cauchy-prelim}
    \operatorname{Sym}^{s}(X\otimes Y)
    \cong
    \bigoplus_{\substack{
       \lambda\vdash s\\
       \ell(\lambda)\leq
          \min\{\dim X,\dim Y\}
    }}
       \mathbb S_\lambda(X)
       \otimes
       \mathbb S_\lambda(Y).
\end{equation}
Every summand occurs with multiplicity one.
\end{theorem}

We shall also use the standard vanishing criterion
\begin{equation}
\label{eq:Schur-functor-vanishing}
    \mathbb S_\lambda(\mathbb C^R)
    =
    \{0\}
    \qquad
    \text{if and only if }
    \ell(\lambda)>R.
\end{equation}

For \(s=2\), the Cauchy identity reduces to
\begin{equation}
\label{eq:Cauchy-two-prelim}
\begin{split}
    \operatorname{Sym}^{2}(X\otimes Y)
    \cong&
    \left[
       \operatorname{Sym}^{2}(X)
       \otimes
       \operatorname{Sym}^{2}(Y)
    \right]
    \oplus
    \left[
       \bigwedge\nolimits^{2}X
       \otimes
       \bigwedge\nolimits^{2}Y
    \right].
\end{split}
\end{equation}

\subsection{Peter--Weyl theory and matrix-coefficient spaces}
\label{subsec:Peter-Weyl-prelim}

Let \(\mathrm{G}\) be a compact group and let
\[
    \pi_{\mathrm{G}}\colon \mathrm{G}\longrightarrow\mathrm U(V)
\]
be a finite-dimensional unitary representation.

\begin{definition}[Matrix coefficient and coefficient space~\cite{folland2016course,BrockerTomDieck1985}]
\label{def:coefficient-space-prelim}
For \(x,y\in V\), define
\begin{equation}
\label{eq:matrix-coefficient-prelim}
    f_{x,y}^{\pi_{\mathrm{G}}}(g)
    :=
    \langle x|\pi_{\mathrm{G}}(g)y\rangle,
    \qquad
    g\in G.
\end{equation}
The \emph{coefficient space} of \(\pi_{\mathrm{G}}\) is
\begin{equation}
\label{eq:Coeff-prelim}
    \operatorname{Coeff}(\pi_{\mathrm{G}})
    :=
    \operatorname{span}_{\mathbb C}
    \left\{
       f_{x,y}^{\pi}:
       x,y\in V
    \right\}
    \subseteq
    C(\mathrm{G}).
\end{equation}
\end{definition}

\begin{theorem}[Schur orthogonality~\cite{BrockerTomDieck1985}]
\label{thm:Schur-orthogonality-prelim}
Let
\[
    \pi_\alpha \colon \mathrm{G} \longrightarrow \mathrm U(V_\alpha),
    \qquad
    \pi_\beta \colon \mathrm{G} \longrightarrow \mathrm U(V_\beta)
\]
be irreducible unitary representations of a compact group \(\mathrm{G}\).
For orthonormal bases of \(V_\alpha,V_\beta\),
\begin{equation}
\label{eq:Schur-orthogonality-prelim}
\begin{split}
    \int_\mathrm{G}
       (\pi_\alpha(g))_{ij}
       \overline{
          (\pi_\beta(g))_{k\ell}
       }
       \,\mathrm d\mu_{\mathrm{Haar}}(g)
    =
    \frac{
       \delta_{\alpha\beta}
       \delta_{ik}
       \delta_{j\ell}
    }{
       \dim_{\mathbb C}V_\alpha
    }.
\end{split}
\end{equation}
Consequently,
\begin{equation}
\label{eq:Coeff-orthogonal-prelim}
    \operatorname{Coeff}(\pi_\alpha)
    \perp
    \operatorname{Coeff}(\pi_\beta)
\end{equation}
in
\(
    L^2(\mathrm{G},\mu_{\mathrm{Haar}})
\)
whenever
\(
    \pi_\alpha\not\cong\pi_\beta
\).
\end{theorem}

\begin{theorem}[Peter-Weyl~\cite{folland2016course}]
\label{thm:Peter-Weyl-prelim}
Let
\(
    \widehat {\mathrm{G}}
\)
be a complete set of equivalence classes of irreducible
finite-dimensional unitary representations of the compact group
\(\mathrm{G}\).  Then
\begin{equation}
\label{eq:Peter-Weyl-L2-prelim}
    L^2(\mathrm{G},\mu_{\mathrm{Haar}})
    =
    \widehat{\bigoplus}_{
       \alpha\in\widehat{\mathrm{G}}
    }
       \operatorname{Coeff}(\pi_\alpha).
\end{equation}
In particular, coefficient spaces corresponding to inequivalent
irreducible representations are mutually orthogonal.
\end{theorem}

For the applications in this paper we shall use the following
finite-dimensional consequence.

\begin{proposition}[Coefficient spaces of reducible representations~\cite{BrockerTomDieck1985}]
\label{prop:Coeff-direct-sum-prelim}
Suppose
\[
    \rho
    \cong
    \bigoplus_{\alpha}
       \pi_\alpha\otimes I_{\mathbb C^{m_\alpha}},
\]
where the \(\pi_\alpha\)'s are pairwise inequivalent irreducible
representations.  Then
\begin{equation}
\label{eq:Coeff-reducible-prelim}
    \operatorname{Coeff}(\rho)
    =
    \bigoplus_{\alpha:m_\alpha>0}
       \operatorname{Coeff}(\pi_\alpha).
\end{equation}
Thus multiplicity does not create additional matrix-coefficient
functions.
\end{proposition}

\begin{proposition}[Conjugation representation~\cite{BrockerTomDieck1985}]
\label{prop:conjugation-coeff-prelim}
Let
\(
    \pi_{\mathrm{G}} \colon \mathrm{G} \longrightarrow \mathrm U(V)
\)
be a unitary representation and let
\[
    \operatorname{Ad}_{\pi_{\mathrm{G}}}(g)(X)
    :=
    \pi_{\mathrm{G}}(g) X \pi_{\mathrm{G}}(g)^\dagger,
    \qquad
    X\in\mathcal L(V).
\]
Under the canonical identification
\[
    \mathcal L(V)
    \cong
    V\otimes\overline V,
\]
one has
\begin{equation}
\label{eq:Ad-tensor-prelim}
    \operatorname{Ad}_{\pi_{\mathrm{G}}}
    \cong
    \pi_{\mathrm{G}}\otimes\overline{\pi_{\mathrm{G}}}.
\end{equation}
Consequently,
\begin{equation}
\label{eq:Coeff-Ad-prelim}
    \operatorname{Coeff}
    (
       \operatorname{Ad}_{\pi_{\mathrm{G}}}
    )
    =
    \operatorname{Coeff}
    (
       \pi_{\mathrm{G}}\otimes\overline{\pi_{\mathrm{G}}}
    ).
\end{equation}
\end{proposition}

\subsection{Unitary \(t\)-designs}
In this section, we review some basic definitions and results on unitary \(t\)-designs.

\begin{definition}[\cite{kaznatcheev2010structure}]
Let \(\mathfrak{U}\) be a finite subset of \(\U(d)\). A  function \(w:\mathfrak{U} \longrightarrow (0,1]\) is said to be a weight function on \(\mathfrak{U}\) if for all \(U \in \mathfrak{U}\), \(w(U) > 0\), and \(\sum_{U \in \mathfrak{U}} w(U)=1\).
\end{definition}

\noindent For our construction, we consider the tensor product definition of unitary \(t\)-designs. 

\begin{definition}[Weighted tensor product definition of Unitary \(t\)-designs {\cite{kaznatcheev2010structure}}]\label{def:utd}
Let \(\mathfrak{U} \subset \U(d)\) be a finite set of unitary operators on a \(d\)-dimensional complex Hilbert space \(\mathbb{C}^d\), and let \(w: \mathfrak{U} \longrightarrow (0,1]\) be a probability distribution over \(\mathfrak{U}\), i.e.,
\[
\sum_{U \in \mathfrak{U}} w(U) = 1.
\]
Then the tuple \((\mathfrak{U}, w)\) is said to be an \emph{exact unitary \(t\)-design} if
\begin{equation}\label{eq:exact-unitary-design-moment}
\sum_{U\in\mathfrak U}
   w(U)\,
   U^{\otimes t}\otimes\overline{U}^{\otimes t}
=
\int_{\U(d)}
   U^{\otimes t}\otimes\overline{U}^{\otimes t}
   \,\mathrm d\mu_{\mathrm{Haar}}(U),
\end{equation}
i.e.,
\begin{equation}\label{eq:exact-unitary-design-channel}
    \Phi_{\mathfrak U,w}^{(t)}(X)
    =
    \mathcal T^{(t)}(X)
    \qquad
    \text{for every }X\in\mathcal L(\mathcal H_d^{\otimes t}),
\end{equation}
where the integral is taken with respect to the normalized Haar measure \(\mu_{\mathrm{Haar}}\) on the unitary group \(\U(d)\). That is, the tuple \((\mathfrak{U}, w)\) reproduces the \(t\)-th moment operator of the Haar distribution over \(\U(d)\).
\end{definition}

\begin{proposition}[\cite{kaznatcheev2010structure}]
Every unitary \(t\)-design is also a \((t-1)\)-design.
\end{proposition}

The following theorem proves the existence of unitary \(t\)-designs.

\begin{theorem}[Seymour--Zaslavsky~\cite{seymour1984averaging}]\label{thm:seymour-zaslavsky}
Let $X$ be an infinite path-connected topological space, and let $\mu$ be a Borel probability measure on $X$ with full support, i.e., $\operatorname{supp}(\mu)=X$. Let $f \colon X \longrightarrow \mathbb{C}^{n \times m}$ be continuous and $\mu$-integrable componentwise.

Then there exists an integer $N_0 \in \mathbb{N}$ such that, for every integer $N \ge N_0$, there exists a finite subset $Y \subset X$ with $|Y| = N$ satisfying
\begin{equation}
    \frac{1}{|Y|} \sum_{x \in Y} f(x)
    =
    \int_{X} f(x) \, d\mu(x).
\end{equation}
\end{theorem}

\begin{corollary}[\cite{kaznatcheev2010structure}]
    A minimum unitary \(t\)-design \((\mathfrak{U},w)\) exists for all \(d,t \geq 1\).
\end{corollary}

For standard equivalent formulations and results, see \cite{kaznatcheev2010structure,dankert2009exact,RoyScott2009,seymour1984averaging}.

\begin{definition}[Restricted unitary design]
\label{def:restricted-unitary-design}
Let \(\mathcal S\subseteq\mathcal H_d^{\otimes t}\) be nonzero and let
\(\varepsilon\geq0\).  The ensemble \((\mathfrak U,w)\) is an
\(\varepsilon\)-approximate unitary \((\mathcal S,t)\)-design in induced
trace norm if
\begin{equation}\label{eq:restricted-design}
  \left\|
    \bigl(\Phi_{\mathfrak U,w}^{(t)}-\mathcal T^{(t)}\bigr)
    \circ\mathcal J_{\mathcal S}
  \right\|_{1\to1}
  \leq\varepsilon.
\end{equation}
It is exact when \(\varepsilon=0\).  Replacing \(\|\cdot\|_{1\to1}\) by
\(\|\cdot\|_\diamond\) gives the corresponding side-information-secure
approximation.
\end{definition}

When \(\mathcal S=\operatorname{Sym}^t(\mathcal H_d)\),
Definition~\ref{def:restricted-unitary-design} is the symmetric unitary
\(t\)-design notion of
\cite{BroadbentGonzalezGuillenSchuknecht2022}.  We will use their exact
descent result: an exact symmetric unitary \(t\)-design is an exact symmetric
unitary \((t-1)\)-design
\cite[Lemma~VI.3]{BroadbentGonzalezGuillenSchuknecht2022}.

Finally, we use without proof the Knill-Laflamme characterization of exactly correctable subspaces \cite{KnillLaflamme1997} and the complementarity between private subspaces of a channel and correctable subspaces of a complementary channel \cite{KretschmannKribsSpekkens2008}. Their precise specialization to the collective Haar twirl is recorded where it is applied.

\section{Admissible subspaces for Quantum Private Broadcasting(\textsf{QPB})}
\label{sec:qpb-admissibility}
In this section, for classification purposes, first we generalize and rewrite [Definition 3.1. \cite{BroadbentGonzalezGuillenSchuknecht2022}],a \(t\)-recipient \textsf{QPB}
scheme for general admissible subspaces \(\mathcal{S} \subset \mathcal{H}^{\otimes t}_{d}\).

\begin{definition}[Restricted channel distance]
\label{def:restricted-channel-distance}
Let \(\mathcal S\subset \mathcal X\) be nonzero and let
\(\Phi,\Psi \colon \mathcal L(\mathcal X) \longrightarrow \mathcal L(\mathcal Y)\) be linear.
Define
\begin{equation}\label{eq:restricted-channel-norms}
  \|\Phi-\Psi\|_{1\to1;\mathcal S}
  :=\| (\Phi-\Psi)\circ\mathcal J_{\mathcal S}\|_{1\to1},
  \qquad
  \|\Phi-\Psi\|_{\diamond;\mathcal S}
  :=\| (\Phi-\Psi)\circ\mathcal J_{\mathcal S}\|_\diamond.
\end{equation}
\end{definition}

\begin{definition}[Quantum private broadcast scheme]
\label{def:qpb-scheme}
Let \(0\leq\delta\leq 1\) and \(\varepsilon\geq0\). Let \(\mathcal S\subset \mathcal H_d^{\otimes t}\) be nonzero, and let \((\mathcal K,w)\) be a finite probability space. A \(t\)-recipient \textsf{QPB}
scheme on \(\mathcal S\) consists, for every \(k\in\mathcal K\), of CPTP
maps
\[
  \operatorname{Enc}^{(t)}_k:
    \mathcal L(\mathcal S)\longrightarrow
    \mathcal L(\mathcal H_d^{\otimes t}),
  \qquad
  \operatorname{Dec}_k:
    \mathcal L(\mathcal H_d) \longrightarrow\mathcal L(\mathcal H_d).
\]
It is \(\delta\)-correct on \(\mathcal S\) if
\begin{equation}\label{eq:qpb-correctness}
  \left\|
    \operatorname{Dec}_k^{\otimes t}\circ
    \operatorname{Enc}^{(t)}_k-\mathcal J_{\mathcal S}
  \right\|_\diamond
  \leq1-\delta
  \qquad(k\in\mathcal K).
\end{equation}
It has \(\varepsilon\)-indistinguishable ciphertexts in induced trace norm if
there exists \(\sigma\in\mathcal D(\mathcal H_d^{\otimes t})\) such that
\begin{equation}\label{eq:qpb-privacy}
  \left\|
   \mathbb{E}_{k\in \mathcal{K}}\left[\operatorname{Enc}^{(t)}_{k}\right]
    -\mathcal R_{\sigma,\mathcal S}
  \right\|_{1\to1}
  \leq\varepsilon.
\end{equation}
If the diamond norm is used in \eqref{eq:qpb-privacy}, the definition includes security against adversaries holding arbitrary side information. Perfect correctness and perfect privacy correspond to \(\delta=1\) and \(\varepsilon=0\), respectively.
\end{definition}

The plaintext in the generalized model is an already prepared joint state $\rho\in\mathcal D(\mathcal S)$. Correctness means recovery of this joint state by the product of the local decryption maps. In
particular, if $\Tr_{\widehat j}$ is the partial trace over all coordinates except $j$, collective encryption satisfies
\[
 \operatorname{Dec}_k\!\left(
  \Tr_{\widehat j}\operatorname{Enc}_k^{(t)}(\rho)\right)
 =\Tr_{\widehat j}\mathcal J_{\mathcal S}(\rho).
\]
For arbitrary $\mathcal S$, these local marginals need not be equal, nor need any one marginal contain the full logical state. Thus the generalized model studies private collective encryption with local decryption of a distributed plaintext. For an already available repeated pure plaintext $|\psi\rangle^{\otimes t}$ it recovers the usual same-message broadcast interpretation.

\begin{definition}[Haar-twirl admissible subspace]
\label{def:haar-admissible}
A nonzero subspace
\(\mathcal S\subseteq\mathcal H_d^{\otimes t}\) is
\emph{Haar-twirl admissible} if there is a state
\(\sigma_{\mathcal S}\in\mathcal D(\mathcal H_d^{\otimes t})\) such that
\begin{equation}\label{eq:haar-admissible}
  \mathcal T^{(t)}\circ\mathcal J_{\mathcal S}
  =\mathcal R_{\sigma_{\mathcal S},\mathcal S}
  \quad\text{on }\mathcal L(\mathcal S).
\end{equation}
The state \(\sigma_{\mathcal S}\) is uniquely determined and is called the
replacement state of \(\mathcal S\).
\end{definition}

The adjective ``admissible'' will always refer to
Definition~\ref{def:haar-admissible}; it is therefore a property of the Haar channel and not of a chosen encryption ensemble.

\begin{proposition}[\textsf{QPB} from a restricted design]
\label{prop:qpb-from-restricted-design}
Let \(\mathcal S\subseteq\mathcal H_d^{\otimes t}\) be Haar-twirl
admissible with replacement state \(\sigma_{\mathcal S}\), and let
\((\mathfrak U,w)\), \(\mathfrak U=\{U_k:k\in\mathcal K\}\), be an
\(\varepsilon\)-approximate unitary \((\mathcal S,t)\)-design.  Define
\begin{align}
  \operatorname{Enc}^{(t)}_k(X)
    &:=U_k^{\otimes t}\mathcal J_{\mathcal S}(X)
       (U_k^\dagger)^{\otimes t},
       \label{eq:qpb-encryption}\\
  \operatorname{Dec}_k(Y)
    &:=U_k^\dagger YU_k.
       \label{eq:qpb-decryption}
\end{align}
Then the resulting scheme is perfectly correct and satisfies
\begin{equation}\label{eq:qpb-design-security}
  \left\|
    \mathbb{E}_{k\in \mathcal{K}}\left[\operatorname{Enc}^{(t)}_{k}\right]
    -\mathcal R_{\sigma_{\mathcal S},\mathcal S}
  \right\|_{1\to1}
  \leq\varepsilon.
\end{equation}
If the design approximation is measured in diamond norm, the same estimate
holds in diamond norm.  In particular, an exact restricted design gives
perfect privacy, including against arbitrary side information.
\end{proposition}

\begin{proof}
For every \(k \in \mathcal{K}\),
\(
  \operatorname{Dec}_k^{\otimes t}\circ
  \operatorname{Enc}^{(t)}_k
  =\mathcal J_{\mathcal S}.
\)
Moreover,
\[
  \mathbb{E}_{k\in \mathcal{K}}\left[\operatorname{Enc}^{(t)}_{k}\right]
  =\Phi_{\mathfrak U,w}^{(t)}\circ\mathcal J_{\mathcal S},
  \qquad
  \mathcal R_{\sigma_{\mathcal S},\mathcal S}
  =\mathcal T^{(t)}\circ\mathcal J_{\mathcal S}.
\]
Thus \eqref{eq:qpb-design-security} is exactly
\eqref{eq:restricted-design}; the diamond-norm assertion is identical.
\end{proof}

By Lemma~\ref{lem:blockwise-collective-Haar-twirl}, the classification of admissible subspaces reduces to determining when the block weights and the multiplicity-space marginals in [Lemma~\ref{lem:blockwise-collective-Haar-twirl}, \eqref{eq:appendix-twirl-direct-sum-form}] are independent of the input. 

The scheme of Proposition~\ref{prop:qpb-from-restricted-design} uses collective encryption. Conversely, every collective scheme that is private on $\mathcal S$ forces $\mathcal S$ to be Haar-twirl admissible.

\begin{proposition}[Collective privacy forces admissibility]\label{prop:necessity}
Let $\mathcal S\subseteq\mathcal H_d^{\otimes t}$ be nonzero, let $(\mathfrak U,w)$ be a
finite weighted ensemble and let $\sigma\in\mathcal D(\mathcal H_d^{\otimes t})$. Then
\[
  \bigl\|\mathcal T^{(t)}\circ\mathcal J_{\mathcal S}-\mathcal R_{\mathcal T^{(t)}(\sigma),\mathcal S}\bigr\|_{1\to1}
  \le\bigl\|\Phi^{(t)}_{\mathfrak U,w}\circ\mathcal J_{\mathcal S}-\mathcal R_{\sigma,\mathcal S}\bigr\|_{1\to1},
\]
and the same inequality holds with $\|\cdot\|_{\diamond}$ in place of
$\|\cdot\|_{1\to1}$. In particular, if the scheme of Proposition~\ref{prop:qpb-from-restricted-design} built from $(\mathfrak U,w)$ is perfectly private on $\mathcal S$ with replacement state $\sigma$, then $\mathcal S$ is Haar-twirl admissible and $\sigma_{\mathcal S}=\mathcal T^{(t)}(\sigma)$.
\end{proposition}

\begin{proof}
For $V\in\mathrm U(d)$, right invariance of $\mu_{\mathrm{Haar}}$ gives
$\mathcal T^{(t)}\bigl(V^{\otimes t}Y(V^{\dagger})^{\otimes t}\bigr)=\mathcal T^{(t)}(Y)$;
hence $\mathcal T^{(t)}\circ\Phi^{(t)}_{\mathfrak U,w}=\mathcal T^{(t)}$. Moreover
$\mathcal T^{(t)}\circ\mathcal R_{\sigma,\mathcal S}=\mathcal R_{\mathcal T^{(t)}(\sigma),\mathcal S}$, and
$\mathcal T^{(t)}(\sigma)\in\mathcal D(\mathcal H_d^{\otimes t})$ because $\mathcal T^{(t)}$ is a
channel. Therefore
\[
  \mathcal T^{(t)}\circ\mathcal J_{\mathcal S}-\mathcal R_{\mathcal T^{(t)}(\sigma),\mathcal S}
  =\mathcal T^{(t)}\circ\bigl(\Phi^{(t)}_{\mathfrak U,w}\circ\mathcal J_{\mathcal S}-\mathcal R_{\sigma,\mathcal S}\bigr).
\]
By convexity of the trace norm and its unitary invariance,
\[
  \bigl\|\mathcal T^{(t)}(Y)\bigr\|_1\le\int_{\mathrm U(d)}
  \bigl\|U^{\otimes t}Y(U^{\dagger})^{\otimes t}\bigr\|_1\,\mathrm d\mu_{\mathrm{Haar}}(U)
  =\|Y\|_1
\]
for every $Y$, which gives the $1\to1$ inequality. The same argument applied
to $\mathcal T^{(t)}\otimes\mathrm{id}_{\mathcal L(\mathbb C^r)}$, which is an average of
the unitary channels $Z\mapsto(U^{\otimes t}\otimes I)Z(U^{\otimes t}\otimes I)^{\dagger}$,
gives the diamond-norm inequality. If the scheme is perfectly private, the
right-hand side vanishes, and so does the left-hand side.
\end{proof}

\begin{remark}\label{rem:operational}
Combined with Proposition~\ref{prop:qpb-from-restricted-design} and the existence of
exact unitary $(\mathcal S,t)$-designs, Proposition~\ref{prop:necessity} shows that a nonzero subspace supports a perfectly correct and perfectly private collective \textsf{QPB} scheme if and only if it is Haar-twirl admissible. It also shows that an $\varepsilon$-private collective scheme on $\mathcal S$ forces $\mathcal T^{(t)}$ to be $\varepsilon$-close to a replacement channel on $\mathcal S$, which motivates the stability problem in Section~\ref{sec:conclusion}.
\end{remark}

\subsection{Classification of Haar-twirl admissible subspaces for \((\mathcal{S},t)-\mathsf{QPB}\)}\label{subsec:qpb-classification}
For an operator \(X\in\mathcal L(\mathcal{H}^{\otimes t}_{d})\), its
\(\lambda\)-diagonal Schur block is the operator
\[
    X_{\lambda\lambda}
    :=
    \iota_\lambda^\dagger X\iota_\lambda
    \in\mathcal L(V_\lambda\otimes W_\lambda).
\]

If \(\mathcal S\subset \mathcal{H}^{\otimes t}_{d}\) is a nonzero subspace, we identify an
operator on \(\mathcal S\) with its extension by zero to
\(\mathcal S^\perp\).  Under this identification,
\begin{equation}
\label{eq:density-operators-supported-on-S}
    \mathcal D(\mathcal S)
    :=
    \left\{
       \rho\in\mathcal L(\mathcal{S}):
       \rho\geq0,\ 
       \operatorname{Tr}(\rho)=1
    \right\}.
\end{equation}
Recall that \(\mathcal S\) is Haar-twirl admissible if there exists a
density operator \(\sigma\in\mathcal D(\mathcal{H}^{\otimes t}_{d})\), independent of the
input state, such that
\begin{equation}
\label{eq:qpb-admissibility-definition-recalled}
    \mathcal T^{(t)}(\mathcal J_{\mathcal S}(\rho))=\sigma
    \qquad
    \text{for every }\rho\in\mathcal D(\mathcal S).
\end{equation}

We now state our first main result on the classification of Haar-twirl admissible subspaces. 
\begin{theorem}[Classification of Haar-twirl admissible subspaces]
\label{thm:qpb-classification}
Let \(\mathcal S\subset \mathcal{H}^{\otimes t}_{d}\) be a nonzero linear subspace.  Then
\(\mathcal S\) is Haar-twirl admissible if and only if the following two conditions hold simultaneously.

\begin{enumerate}
\item\label{cond1}
For every \(\lambda\in\Lambda(d,t)\), there exists a constant
\(c_\lambda\geq0\), independent of \(\psi\) and \(\phi\), such that
\begin{equation}
\label{eq:condA}
    \langle\phi|\Pi_\lambda|\psi\rangle
    =
    c_\lambda\langle\phi|\psi\rangle
    \qquad
    \text{for all }|\psi\rangle,|\phi\rangle\in\mathcal S,
\end{equation}
and these constants satisfy
\begin{equation}
\label{eq:c-lambda-normalization}
    \sum_{\lambda\in\Lambda(d,t)}c_\lambda=1.
\end{equation}

\item\label{cond2}
For every \(\lambda\in\Lambda(d,t)\) for which \(c_\lambda>0\), there
exists a density operator \(\omega_\lambda\in\mathcal D(W_\lambda)\) such
that
\begin{equation}
\label{eq:condB}
\begin{split}
    &\operatorname{Tr}_{V_\lambda}
    \!\left(
       \iota_\lambda^\dagger
       |\psi\rangle\!\langle\phi|
       \iota_\lambda
    \right)                                                  =
    c_\lambda\langle\phi|\psi\rangle\,\omega_\lambda
    \qquad
    \text{for all }|\psi\rangle,|\phi\rangle\in\mathcal S.
\end{split}
\end{equation}
After identifying \(\operatorname{im}\Pi_\lambda\cong V_\lambda\otimes W_\lambda\), the left-hand side of \eqref{eq:condB} is
\(
    \operatorname{Tr}_{V_\lambda}
    \!\left(
       \Pi_\lambda|\psi\rangle\!\langle\phi|\Pi_\lambda
    \right).
\)
\end{enumerate}

When these conditions hold, the replacement state is
\begin{equation}
\label{eq:sigma}
    \sigma
    =
    \sum_{\substack{\lambda\in\Lambda(d,t)\\c_\lambda>0}}
    \iota_\lambda
    \left(
       c_\lambda
       \frac{I_{V_\lambda}}{\dim_{\mathbb C}V_\lambda}
       \otimes\omega_\lambda
    \right)
    \iota_\lambda^\dagger.
\end{equation}

In particular, suppose that
\begin{equation}
\label{eq:single-positive-block}
    c_{\lambda_0}=1,
    \qquad
    c_\lambda=0
    \quad\text{for every }\lambda\neq\lambda_0.
\end{equation}
Then Conditions~\ref{cond1} and~\ref{cond2} are equivalent to
\begin{enumerate}
\item
\(
    \mathcal S\subseteq V_{\lambda_0}\otimes W_{\lambda_0}
\), under the Schur-Weyl identification \eqref{eq:schur-weyl-decomposition};

\item
there exists \(\omega_{\lambda_0}\in\mathcal D(W_{\lambda_0})\) such that
\begin{equation}
\label{eq:single-block-fixed-marginal}
    \operatorname{Tr}_{V_{\lambda_0}}
    \!\left(|\psi\rangle\!\langle\phi|\right)
    =
    \langle\phi|\psi\rangle\,\omega_{\lambda_0}
    \qquad
    \text{for all }|\psi\rangle,|\phi\rangle\in\mathcal S.
\end{equation}
\end{enumerate}
Moreover, if \(\mathcal S\) is Haar-twirl admissible and contains a nonzero repeated product vector
\(
    |\varphi\rangle^{\otimes t}
\), then \eqref{eq:single-positive-block} is forced with
\(
    \lambda_0=(t)
\), the one-row partition of \(t\).
\end{theorem}

\begin{proof}
For \(\lambda\in\Lambda(d,t)\) and
\(|\psi\rangle,|\phi\rangle\in\mathcal{H}^{\otimes t}_{d}\), define the operator
\begin{equation}
\label{eq:B-lambda-definition}
    B_\lambda(\psi,\phi)
    :=
    \operatorname{Tr}_{V_\lambda}
    \!\left(
       \iota_\lambda^\dagger
       |\psi\rangle\!\langle\phi|
       \iota_\lambda
    \right)
    \in\mathcal L(W_\lambda).
\end{equation}
The map \(B_\lambda(\,\cdot\,,\,\cdot\,)\) is linear in its first
argument and conjugate-linear in its second argument.  Preservation of total
trace under the partial trace gives
\begin{align}
\operatorname{Tr}_{W_\lambda}B_\lambda(\psi,\phi)
&=
\operatorname{Tr}_{V_\lambda\otimes W_\lambda}
\!\left(
   \iota_\lambda^\dagger
   |\psi\rangle\!\langle\phi|
   \iota_\lambda
\right)
\notag\\
&=
\operatorname{Tr}
\!\left(
   |\psi\rangle\!\langle\phi|
   \iota_\lambda\iota_\lambda^\dagger
\right)
\notag\\
&=
\operatorname{Tr}
\!\left(
   |\psi\rangle\!\langle\phi|\Pi_\lambda
\right)
\notag\\
&=
\langle\phi|\Pi_\lambda|\psi\rangle.
\label{eq:trace-of-B-lambda}
\end{align}
In particular,
\begin{equation}
\label{eq:B-diagonal-trace}
    \operatorname{Tr}_{W_\lambda}B_\lambda(\psi,\psi)
    =
    \langle\psi|\Pi_\lambda|\psi\rangle
    =
    \|\Pi_\lambda\psi\|^2.
\end{equation}

\smallskip
To prove the sufficiency first, assume Conditions~\ref{cond1} and~\ref{cond2}. Observe that the
operator \(\sigma\) in \eqref{eq:sigma} is positive semidefinite, because each \(\omega_\lambda\) is positive semidefinite and each
\(c_\lambda\geq0\). Moreover,
\begin{align}
\operatorname{Tr}(\sigma)
&=
\sum_{\substack{\lambda\in\Lambda(d,t)\\c_\lambda>0}}
c_\lambda
\operatorname{Tr}_{V_\lambda\otimes W_\lambda}
\!\left(
   \frac{I_{V_\lambda}}{\dim_{\mathbb C}V_\lambda}
   \otimes\omega_\lambda
\right)
\notag\\
&=
\sum_{\substack{\lambda\in\Lambda(d,t)\\c_\lambda>0}}
c_\lambda
\frac{\operatorname{Tr}_{V_\lambda}(I_{V_\lambda})}{\dim_{\mathbb C}V_\lambda}
\operatorname{Tr}_{W_\lambda}(\omega_\lambda)
\notag\\
&=
\sum_{\substack{\lambda\in\Lambda(d,t)\\c_\lambda>0}}
c_\lambda \cdot 1 \cdot 1 \\
&=1.
\label{eq:sigma-has-unit-trace}
\end{align}
Consequently, \(\sigma\in\mathcal D(\mathcal{H}^{\otimes t}_{d})\).

\smallskip
Let \(\rho\in\mathcal D(\mathcal S)\) be arbitrary, and let
\(r_\rho:=\operatorname{rank}(\rho)\).  Since
\(\operatorname{supp}(\rho)\subseteq\mathcal S\), the finite-dimensional
spectral theorem gives
\begin{equation}
\label{eq:rho-spectral}
    \rho
    =
    \sum_{j=1}^{r_\rho}p_j|\psi_j\rangle\!\langle\psi_j|,
    \qquad
    p_j>0,
    \qquad
    \sum_{j=1}^{r_\rho}p_j=1,
\end{equation}
where
\(
    |\psi_1\rangle,\ldots,|\psi_{r_\rho}\rangle
\)
are orthonormal vectors belonging to \(\mathcal S\).

Fix \(\lambda\in\Lambda(d,t)\).  If \(c_\lambda>0\), then
Condition~\ref{cond2}, applied with \(\phi=\psi=\psi_j\), gives
\begin{equation}
\label{eq:B-positive-block-eigenvector}
    B_\lambda(\psi_j,\psi_j)
    =
    c_\lambda
    \langle\psi_j|\psi_j\rangle\omega_\lambda
    =
    c_\lambda\omega_\lambda.
\end{equation}
Using \eqref{eq:rho-spectral}, linearity of compression, and linearity of the
partial trace, we therefore obtain
\begin{align}
\operatorname{Tr}_{V_\lambda}
\!\left(
   \iota_\lambda^\dagger\rho\iota_\lambda
\right)
&=
\operatorname{Tr}_{V_\lambda}
\!\left(
   \iota_\lambda^\dagger
   \left(
      \sum_{j=1}^{r_\rho}p_j
      |\psi_j\rangle\!\langle\psi_j|
   \right)
   \iota_\lambda
\right)
\notag\\
&=
\sum_{j=1}^{r_\rho}p_j
\operatorname{Tr}_{V_\lambda}
\!\left(
   \iota_\lambda^\dagger
   |\psi_j\rangle\!\langle\psi_j|
   \iota_\lambda
\right)
\notag\\
&=
\sum_{j=1}^{r_\rho}p_j B_\lambda(\psi_j,\psi_j)
\notag\\
&=
\sum_{j=1}^{r_\rho}p_jc_\lambda\omega_\lambda
\notag\\
&=
c_\lambda
\left(\sum_{j=1}^{r_\rho}p_j\right)
\omega_\lambda
\notag\\
&=
c_\lambda\omega_\lambda.
\label{eq:positive-block-of-rho}
\end{align}

Suppose instead that \(c_\lambda=0\).  Taking \(\phi=\psi\) in
Condition~\ref{cond1} gives, for every \(|\psi\rangle\in\mathcal S\),
\begin{equation}
\label{eq:zero-c-implies-zero-projection}
    \|\Pi_\lambda\psi\|^2
    =
    \langle\psi|\Pi_\lambda|\psi\rangle
    =
    c_\lambda\|\psi\|^2
    =0.
\end{equation}
It follows that
\(
    \Pi_\lambda|\psi\rangle=0
\)
for every \(|\psi\rangle\in\mathcal S\).  In particular,
\(
    \iota_\lambda^\dagger|\psi_j\rangle=0
\)
for every eigenvector in \eqref{eq:rho-spectral}, and hence
\begin{equation}
\label{eq:zero-block-of-rho}
    \iota_\lambda^\dagger\rho\iota_\lambda
    =0,
    \qquad
    \operatorname{Tr}_{V_\lambda}
    \!\left(
       \iota_\lambda^\dagger\rho\iota_\lambda
    \right)
    =0.
\end{equation}

Substituting \eqref{eq:positive-block-of-rho} and
\eqref{eq:zero-block-of-rho} into the Haar-twirl formula
[Lemma~\ref{lem:blockwise-collective-Haar-twirl},\eqref{eq:appendix-twirl-direct-sum-form}], we obtain
\begin{align}
\mathcal T^{(t)}(\rho)
&=
\sum_{\lambda\in\Lambda(d,t)}
\iota_\lambda
\left[
   \frac{I_{V_\lambda}}{\dim_{\mathbb C}V_\lambda}
   \otimes
   \operatorname{Tr}_{V_\lambda}
   \!\left(
      \iota_\lambda^\dagger\rho\iota_\lambda
   \right)
\right]
\iota_\lambda^\dagger
\notag\\
&=
\sum_{\substack{\lambda\in\Lambda(d,t)\\c_\lambda>0}}
\iota_\lambda
\left(
   \frac{I_{V_\lambda}}{\dim_{\mathbb C}V_\lambda}
   \otimes c_\lambda\omega_\lambda
\right)
\iota_\lambda^\dagger
\notag\\
&=
\sum_{\substack{\lambda\in\Lambda(d,t)\\c_\lambda>0}}
\iota_\lambda
\left(
   c_\lambda
   \frac{I_{V_\lambda}}{\dim_{\mathbb C}V_\lambda}
   \otimes\omega_\lambda
\right)
\iota_\lambda^\dagger
\notag\\
&=\sigma.
\label{eq:sufficiency-final}
\end{align}
Because \(\rho\in\mathcal D(\mathcal S)\) was arbitrary,
\(\mathcal S\) is Haar-twirl admissible.

\medskip
Conversely, suppose that \(\mathcal S\) is Haar-twirl admissible with replacement state \(\sigma\in\mathcal D(\mathcal{H}^{\otimes t}_{d})\).  

Thus
\begin{equation}
\label{eq:necessity-replacement-assumption}
    \mathcal T^{(t)}(\rho)=\sigma
    \qquad
    \text{for every }\rho\in\mathcal D(\mathcal S).
\end{equation}
In particular, for every unit vector \(|\psi\rangle\in\mathcal S\),
\begin{align}
\sigma
&=
\mathcal T^{(t)}
\!\left(|\psi\rangle\!\langle\psi|\right)
\notag\\
&=
\sum_{\lambda\in\Lambda(d,t)}
\iota_\lambda
\left[
   \frac{I_{V_\lambda}}{\dim_{\mathbb C}V_\lambda}
   \otimes B_\lambda(\psi,\psi)
\right]
\iota_\lambda^\dagger.
\label{eq:necessity-twirled-pure-state}
\end{align}
Equation \eqref{eq:necessity-twirled-pure-state} shows, in particular, that
\(\sigma\) is block diagonal with respect to the Schur-Weyl decomposition.
For every \(\mu\in\Lambda(d,t)\), define its \(\mu\)-block by
\begin{equation}
\label{eq:sigma-mu-definition}
    \sigma_\mu
    :=
    \iota_\mu^\dagger\sigma\iota_\mu
    \in\mathcal L(V_\mu\otimes W_\mu).
\end{equation}
Multiplying \eqref{eq:necessity-twirled-pure-state} on the left by \(\iota_\mu^\dagger\), on the right by \(\iota_\mu\), and using
\begin{equation*}
    \iota_\mu^\dagger\iota_\lambda
    =
    \begin{cases}
       I_{V_{\lambda} \otimes W_{\lambda}},&\mu=\lambda,\\
       \mathbf{0},&\mu\neq\lambda,
    \end{cases}
\end{equation*}
we obtain
\begin{equation}
\label{eq:sigma-mu-fixed-B}
    \sigma_\mu
    =
    \frac{I_{V_\mu}}{\dim_{\mathbb C}V_\mu}
    \otimes B_\mu(\psi,\psi)
    \qquad
    \text{for every unit }|\psi\rangle\in\mathcal S.
\end{equation}

Define
\begin{equation}
\label{eq:tau-mu-definition}
    \tau_\mu
    :=
    \operatorname{Tr}_{V_\mu}(\sigma_\mu)
    \in\mathcal L(W_\mu).
\end{equation}
Since \(\sigma_\mu\geq0\), positivity of the partial trace implies
\(\tau_\mu\geq0\).  Applying \(\operatorname{Tr}_{V_\mu}\) to
\eqref{eq:sigma-mu-fixed-B} gives
\begin{align}
\tau_\mu
&=
\operatorname{Tr}_{V_\mu}
\!\left(
   \frac{I_{V_\mu}}{\dim_{\mathbb C}V_\mu}
   \otimes B_\mu(\psi,\psi)
\right)
\notag\\
&=
\operatorname{Tr}_{V_\mu}
\!\left(
   \frac{I_{V_\mu}}{\dim_{\mathbb C}V_\mu}
\right)
B_\mu(\psi,\psi)
\notag\\
&=
\frac{\operatorname{Tr}_{V_\mu}(I_{V_\mu})}{\dim_{\mathbb C}V_\mu}
B_\mu(\psi,\psi)
\notag\\
&=
1 \cdot B_\mu(\psi,\psi)
\notag\\
&=
B_\mu(\psi,\psi).
\label{eq:B-diagonal-equals-tau}
\end{align}
Thus the operator \(B_\mu(\psi,\psi)\) is independent of the choice of unit vector \(|\psi\rangle\in\mathcal S\).

Let
\begin{equation}
\label{eq:clambda-def}
    c_\mu
    :=
    \operatorname{Tr}_{W_\mu}(\tau_\mu).
\end{equation}
Then \(c_\mu\geq0\). Moreover,
\begin{align}
c_\mu
&=
\operatorname{Tr}_{W_\mu}(\tau_\mu)
\notag\\
&=
\operatorname{Tr}_{V_\mu\otimes W_\mu}(\sigma_\mu)
\notag\\
&=
\operatorname{Tr}
\!\left(
   \iota_\mu\sigma_\mu\iota_\mu^\dagger
\right)
\notag\\
&=
\operatorname{Tr}(\Pi_\mu\sigma\Pi_\mu)
\notag\\
&=
\operatorname{Tr}(\Pi_\mu\sigma),
\label{eq:c-mu-as-block-weight}
\end{align}
Consequently,
\begin{align}
\sum_{\mu\in\Lambda(d,t)}c_\mu
&=
\sum_{\mu\in\Lambda(d,t)}
\operatorname{Tr}(\Pi_\mu\sigma)
\notag\\
&=
\operatorname{Tr}
\!\left[
   \left(
      \sum_{\mu\in\Lambda(d,t)}\Pi_\mu
   \right)\sigma
\right]
\notag\\
&=
\operatorname{Tr}(I_{\mathcal{H}^{\otimes t}_{d}}\sigma)
\notag\\
&=
\operatorname{Tr}(\sigma)
=1.
\label{eq:c-mu-sum-one}
\end{align}

We now compute the same block weight directly from an
arbitrary unit vector \(|\psi\rangle\in\mathcal S\). Combining
\eqref{eq:sigma-mu-fixed-B}, multiplicativity of trace on tensor products, preservation of total trace under partial trace, and \eqref{eq:B-diagonal-trace}, we obtain
\begin{align}
\operatorname{Tr}(\Pi_\mu\sigma)
&=
\operatorname{Tr}_{V_\mu\otimes W_\mu}(\sigma_\mu)
\notag\\
&=
\operatorname{Tr}_{V_\mu\otimes W_\mu}
\!\left(
   \frac{I_{V_\mu}}{\dim_{\mathbb C}V_\mu}
   \otimes B_\mu(\psi,\psi)
\right)
\notag\\
&=
\frac{\operatorname{Tr}_{V_\mu}(I_{V_\mu})}{\dim_{\mathbb C}V_\mu}
\operatorname{Tr}_{W_\mu}B_\mu(\psi,\psi)
\notag\\
&=
1 \cdot
\operatorname{Tr}_{W_\mu}B_\mu(\psi,\psi)
\notag\\
&=
\operatorname{Tr}_{V_\mu\otimes W_\mu}
\!\left(
   \iota_\mu^\dagger
   |\psi\rangle\!\langle\psi|
   \iota_\mu
\right)
\notag\\
&=
\operatorname{Tr}
\!\left(
   \Pi_\mu|\psi\rangle\!\langle\psi|\Pi_\mu
\right)
\notag\\
&=
\operatorname{Tr}
\!\left(
   \Pi_\mu^2|\psi\rangle\!\langle\psi|
\right)
\notag\\
&=
\operatorname{Tr}
\!\left(
   \Pi_\mu|\psi\rangle\!\langle\psi|
\right)
\notag\\
&=
\langle\psi|\Pi_\mu|\psi\rangle
\notag\\
&=
\langle\psi|\Pi_\mu^\dagger\Pi_\mu|\psi\rangle
\notag\\
&=
\|\Pi_\mu|\psi\rangle\|^2.
\label{eq:block-weight-complete-calculation}
\end{align}
In particular,
\begin{equation}
\label{eq:diagonal-condition-A}
    \langle\psi|\Pi_\mu|\psi\rangle
    =
    c_\mu
    \qquad
    \text{for every unit }|\psi\rangle\in\mathcal S.
\end{equation}

It remains to derive the off-diagonal identities in Conditions~\ref{cond1}
and~\ref{cond2}.  For arbitrary \(|z\rangle\in\mathcal S\), homogeneity of
\(B_\mu\) and \eqref{eq:B-diagonal-equals-tau} imply
\begin{equation}
\label{eq:B-diagonal-homogeneous}
    B_\mu(z,z)
    =
    \|z\|^2\tau_\mu.
\end{equation}
The equality is trivial for \(z=0\); if \(z\neq0\), then
\begin{align*}
B_\mu(z,z)
&=
B_\mu
\!\left(
   \|z\|\frac{z}{\|z\|},
   \|z\|\frac{z}{\|z\|}
\right)                                                       \\
&=
\|z\|^2
B_\mu
\!\left(
   \frac{z}{\|z\|},
   \frac{z}{\|z\|}
\right)                                                       \\
&=
\|z\|^2\tau_\mu.
\end{align*}

We use the complex polarization identity in an operator-valued form.  For
\(|\psi\rangle,|\phi\rangle\in\mathcal S\), sesquilinearity gives
\begin{align}
B_\mu(\psi+\phi,\psi+\phi)
&=
B_\mu(\psi,\psi)
+B_\mu(\psi,\phi)
+B_\mu(\phi,\psi)
+B_\mu(\phi,\phi),
\label{eq:polarization-expansion-one}\\
B_\mu(\psi-\phi,\psi-\phi)
&=
B_\mu(\psi,\psi)
-B_\mu(\psi,\phi)
-B_\mu(\phi,\psi)
+B_\mu(\phi,\phi),
\label{eq:polarization-expansion-two}\\
B_\mu(\psi+i\phi,\psi+i\phi)
&=
B_\mu(\psi,\psi)
-iB_\mu(\psi,\phi)
+iB_\mu(\phi,\psi)
+B_\mu(\phi,\phi),
\label{eq:polarization-expansion-three}\\
B_\mu(\psi-i\phi,\psi-i\phi)
&=
B_\mu(\psi,\psi)
+iB_\mu(\psi,\phi)
-iB_\mu(\phi,\psi)
+B_\mu(\phi,\phi).
\label{eq:polarization-expansion-four}
\end{align}
Subtracting \eqref{eq:polarization-expansion-two} from
\eqref{eq:polarization-expansion-one}, adding \(i\) times
\eqref{eq:polarization-expansion-three}, and subtracting \(i\) times
\eqref{eq:polarization-expansion-four}, we obtain
\begin{equation}
\label{eq:operator-valued-polarization}
\begin{split}
B_\mu(\psi,\phi)
=\frac14\Bigl(&
B_\mu(\psi+\phi,\psi+\phi)
-B_\mu(\psi-\phi,\psi-\phi)                                  \\
&+iB_\mu(\psi+i\phi,\psi+i\phi)
-iB_\mu(\psi-i\phi,\psi-i\phi)
\Bigr).
\end{split}
\end{equation}
Substituting \eqref{eq:B-diagonal-homogeneous} into
\eqref{eq:operator-valued-polarization} yields
\begin{align}
B_\mu(\psi,\phi)
&=
\frac14\Bigl(
   \|\psi+\phi\|^2
   -\|\psi-\phi\|^2
   +i\|\psi+i\phi\|^2
   -i\|\psi-i\phi\|^2
\Bigr)\tau_\mu
\notag\\
&=
\langle\phi|\psi\rangle\tau_\mu.
\label{eq:B-off-diagonal-tau}
\end{align}
The last equality is the scalar complex polarization identity, with the
convention that the inner product is linear in its second
argument.

Taking the trace over \(W_\mu\) in
\eqref{eq:B-off-diagonal-tau}, and using
\eqref{eq:trace-of-B-lambda} and
\(
    \operatorname{Tr}_{W_\mu}(\tau_\mu)=c_\mu
\), gives
\begin{align}
\langle\phi|\Pi_\mu|\psi\rangle
&=
\operatorname{Tr}_{W_\mu}B_\mu(\psi,\phi)
\notag\\
&=
\langle\phi|\psi\rangle
\operatorname{Tr}_{W_\mu}(\tau_\mu)
\notag\\
&=
c_\mu\langle\phi|\psi\rangle.
\label{eq:condition-A-proved}
\end{align}
This proves Condition~\ref{cond1}.

If \(c_\mu>0\), define
\begin{equation}
\label{eq:omega-mu-definition}
    \omega_\mu
    :=
    \frac{1}{c_\mu}\tau_\mu.
\end{equation}
Since \(\tau_\mu\geq0\) and
\(
    \operatorname{Tr}_{W_\mu}(\tau_\mu)=c_\mu
\), it follows that
\begin{equation}
\label{eq:omega-is-density-operator}
    \omega_\mu\geq0,
    \qquad
    \operatorname{Tr}_{W_\mu}(\omega_\mu)=1.
\end{equation}
Thus \(\omega_\mu\in\mathcal D(W_\mu)\), and
\eqref{eq:B-off-diagonal-tau} becomes
\begin{equation}
\label{eq:condition-B-proved}
    B_\mu(\psi,\phi)
    =
    c_\mu\langle\phi|\psi\rangle\omega_\mu,
\end{equation}
which is Condition~\ref{cond2}.

If \(c_\mu=0\), then \(\tau_\mu\geq0\) and
\(
    \operatorname{Tr}(\tau_\mu)=0
\), so \(\tau_\mu=0\).  Therefore
\eqref{eq:sigma-mu-fixed-B} and
\eqref{eq:B-diagonal-equals-tau} give \(\sigma_\mu=0\).  If
\(c_\mu>0\), then
\begin{equation}
\label{eq:sigma-positive-block-form}
    \sigma_\mu
    =
    \frac{I_{V_\mu}}{\dim_{\mathbb{C}}(V_{\mu})}\otimes\tau_\mu
    =
    c_\mu
    \frac{I_{V_\mu}}{\dim_{\mathbb{C}}(V_{\mu})}\otimes\omega_\mu.
\end{equation}
Since \(\sigma\) is block diagonal, summing
\eqref{eq:sigma-positive-block-form} over the positive-weight blocks proves
the replacement-state formula \eqref{eq:sigma}.

\smallskip
\noindent \textbf{The single-block case.}
Assume \eqref{eq:single-positive-block}.  For every
\(|\psi\rangle\in\mathcal S\) and every \(\lambda\neq\lambda_0\),
Condition~\ref{cond1} gives
\begin{equation}
\label{eq:all-other-components-vanish}
    \|\Pi_\lambda\psi\|^2
    =
    \langle\psi|\Pi_\lambda|\psi\rangle
    =
    c_\lambda\|\psi\|^2
    =0.
\end{equation}
Therefore \(\Pi_\lambda\psi=0\) whenever
\(\lambda\neq\lambda_0\).  Using completeness of the Schur projectors,
\begin{align}
|\psi\rangle
&=
I_{\mathcal{H}^{\otimes t}_{d}}|\psi\rangle
\notag\\
&=
\left(
   \sum_{\lambda\in\Lambda(d,t)}\Pi_\lambda
\right)|\psi\rangle
\notag\\
&=
\Pi_{\lambda_0}|\psi\rangle.
\label{eq:S-contained-single-block}
\end{align}
Hence
\(
    \mathcal S\subseteq\operatorname{im}(\Pi_{\lambda_0})
    \cong V_{\lambda_0}\otimes W_{\lambda_0}
\).
After this identification, \(\Pi_{\lambda_0}\) acts as the identity on
\(\mathcal S\), so Condition~\ref{cond2} reduces to
\eqref{eq:single-block-fixed-marginal}.  Conversely, if
\(
    \mathcal S\subseteq V_{\lambda_0}\otimes W_{\lambda_0}
\)
and \eqref{eq:single-block-fixed-marginal} holds, then
Conditions~\ref{cond1} and~\ref{cond2} hold with
\(
    c_{\lambda_0}=1
\)
and \(c_\lambda=0\) for \(\lambda\neq\lambda_0\).

Finally, suppose that \(\mathcal S\) contains a nonzero repeated product
vector
\[
    |\Phi\rangle:=|\varphi\rangle^{\otimes t}\neq0.
\]
Let
\(
    |\widehat\Phi\rangle:=|\Phi\rangle/\|\Phi\|
\).
The vector \(|\widehat\Phi\rangle\) is invariant under every permutation of
its \(t\) tensor factors.  Consequently,
\begin{equation}
\label{eq:repeated-vector-in-symmetric-block}
    |\widehat\Phi\rangle
    \in
    \operatorname{Sym}^t(\mathbb C^d)
    =
    \operatorname{im}\Pi_{(t)},
\end{equation}
and hence
\begin{equation}
\label{eq:repeated-vector-projections}
    \Pi_{(t)}|\widehat\Phi\rangle
       =|\widehat\Phi\rangle,
    \qquad
    \Pi_\lambda|\widehat\Phi\rangle
       =0
    \quad\text{for every }\lambda\neq(t).
\end{equation}
Applying Condition~\ref{cond1} to
\(
    \psi=\phi=\widehat\Phi
\)
gives
\begin{align}
c_{(t)}
&=
\langle\widehat\Phi|\Pi_{(t)}|\widehat\Phi\rangle
=1,
\label{eq:c-symmetric-is-one}\\
c_\lambda
&=
\langle\widehat\Phi|\Pi_\lambda|\widehat\Phi\rangle
=0
\qquad
\text{for every }\lambda\neq(t).
\label{eq:other-c-are-zero}
\end{align}
Thus \eqref{eq:single-positive-block} is forced with
\(\lambda_0=(t)\), completing the proof.
\end{proof}

\begin{remark}[Repeated product vectors versus arbitrary product vectors]
\label{rem:repeated-versus-arbitrary-product}
The final assertion of Theorem~\ref{thm:qpb-classification} concerns a repeated product
vector \(|\varphi\rangle^{\otimes t}\).  It does not hold for a general
product vector
\(
    |\varphi_1\rangle\otimes\cdots\otimes|\varphi_t\rangle
\), which may have nonzero projections onto several Schur-Weyl blocks.
For example, if \(t=2\), \(d\geq2\), and
\(|0\rangle,|1\rangle\) are orthonormal, then for
\(
    |\Psi\rangle=|0\rangle\otimes|1\rangle
\)
and
\(
    \Pi_\pm=\tfrac12(I\pm F)
\), where \(F\) interchanges the two tensor factors, one has
\begin{align*}
\langle\Psi|\Pi_+|\Psi\rangle
&=
\frac12
\left(
   \langle\Psi|\Psi\rangle
   +\langle\Psi|F|\Psi\rangle
\right)
=\frac12,\\
\langle\Psi|\Pi_-|\Psi\rangle
&=
\frac12
\left(
   \langle\Psi|\Psi\rangle
   -\langle\Psi|F|\Psi\rangle
\right)
=\frac12,
\end{align*}
because
\(
    \langle\Psi|F|\Psi\rangle
    =\langle0|1\rangle\langle1|0\rangle=0
\).
Thus a general product vector need not force concentration in a single Schur
block.
\end{remark}

\begin{corollary}[Haar-twirl admissible full Schur-Weyl blocks]
\label{cor:active-block-injectivity}
Let
\(
    \lambda\in\Lambda(d,t)
\)
and let
\[
    \mathcal H_\lambda
    :=
    \iota_\lambda
    \bigl(
        V_\lambda\otimes W_\lambda
    \bigr)
    \subseteq
    \mathcal H_d^{\otimes t}
\]
be the full \(\lambda\)-th Schur-Weyl isotypic component.
Then
\(\mathcal H_\lambda\) is Haar-twirl admissible if and only if \(\dim_{\mathbb C}W_\lambda=1\). Consequently, for \(t\geq2\), the only Haar-twirl admissible full Schur-Weyl components are
\(
    \operatorname{Sym}^{t}(\mathcal H_d)
\)
and, when \(t\leq d\),
\(
    \bigwedge\nolimits^{t}\mathcal H_d.
\)
\end{corollary}

\begin{proof}
Suppose first that
\(
    \dim_{\mathbb C}W_\lambda=1
\).
Choose a unit vector
\(
    |s_\lambda\rangle
\)
spanning \(W_\lambda\).  Every vector in the \(\lambda\)-block has
the form \(|v\rangle\otimes|s_\lambda\rangle\), \(|v\rangle\in V_\lambda\).

For arbitrary
\(
    |\psi\rangle,|\phi\rangle\in V_\lambda
\),
we obtain
\begin{align}
    \operatorname{Tr}_{V_\lambda}
    \left[
        \bigl(
            |\psi\rangle\langle\phi|
        \bigr)
        \otimes
        |s_\lambda\rangle\langle s_\lambda|
    \right]
    &=
    \operatorname{Tr}
    \bigl(
        |\psi\rangle\langle\phi|
    \bigr)
    |s_\lambda\rangle\langle s_\lambda|
    \notag\\
    &=
    \langle\phi|\psi\rangle
    |s_\lambda\rangle\langle s_\lambda|.
\end{align}
Thus the single-block condition of
Theorem~\ref{thm:qpb-classification} holds with
\(
    c_\lambda=1,
    \;
    \omega_\lambda
    =
    |s_\lambda\rangle\langle s_\lambda|.
\)
Hence \(\mathcal H_\lambda\) is Haar-twirl admissible.

Conversely, suppose that
\(
    \dim_{\mathbb C}W_\lambda\geq2
\).
Choose a unit vector
\(
    |v\rangle\in V_\lambda
\)
and two orthonormal vectors
\(
    |w_1\rangle,|w_2\rangle\in W_\lambda.
\)
Define
\[
    |\psi_1\rangle
    :=
    \iota_\lambda
    \bigl(
        |v\rangle\otimes|w_1\rangle
    \bigr),
    \qquad
    |\psi_2\rangle
    :=
    \iota_\lambda
    \bigl(
        |v\rangle\otimes|w_2\rangle
    \bigr).
\]
Then
\begin{align}
    \operatorname{Tr}_{V_\lambda}
    \left[
        \iota_\lambda^\dagger
        |\psi_1\rangle\langle\psi_1|
        \iota_\lambda
    \right]
    &=
    |w_1\rangle\langle w_1|,
\\
    \operatorname{Tr}_{V_\lambda}
    \left[
        \iota_\lambda^\dagger
        |\psi_2\rangle\langle\psi_2|
        \iota_\lambda
    \right]
    &=
    |w_2\rangle\langle w_2|.
\end{align}
The two density operators are distinct. Therefore no single
\(
    \omega_\lambda\in\mathcal D(W_\lambda)
\)
can satisfy Condition~2 of Theorem~\ref{thm:qpb-classification} on the entire block.
Hence the full block is not Haar-twirl admissible.

Finally, the one-dimensional irreducible complex representations of \(\mathfrak{S}_t\), for \(t\geq2\), are the trivial representation and the sign representation. They correspond respectively to the partitions
\(
    (t)
    \;\text{and}\;
    (1^t).
\)
The latter belongs to \(\Lambda(d,t)\) when \(t\leq d\).
Under Schur-Weyl duality these blocks are respectively
\(
    \operatorname{Sym}^{t}(\mathcal H_d)
    \;\text{and}\;
    \bigwedge\nolimits^{t}\mathcal H_d.
\)
\end{proof}

\begin{corollary}
\label{cor:active-block-dimension-bound}
Let
\(
    \mathcal S\subset \mathcal H_d^{\otimes t}
\)
be a nonzero Haar-twirl admissible subspace, and let
\(
    \{c_\lambda\}_{\lambda\in\Lambda(d,t)}
\)
be the constants appearing in Theorem~\ref{thm:qpb-classification}.
For every
\(
    \lambda\in\Lambda(d,t)
\)
with
\(
    c_\lambda>0
\),
the restricted block projection
\[
    \Pi_\lambda|_{\mathcal S}:
    \mathcal S
    \longrightarrow
    \operatorname{im}(\Pi_\lambda)
    \subseteq
    \mathcal H_d^{\otimes t}
\]
is injective. Consequently,
\begin{equation}
\label{eq:active-block-dimension-bound}
    \dim_{\mathbb C}\mathcal S
    \leq
    \dim_{\mathbb C}(V_\lambda\otimes W_\lambda)
    =
    \dim_{\mathbb C}V_\lambda\,
    \dim_{\mathbb C}W_\lambda .
\end{equation}
\end{corollary}

\begin{proof}
Fix
\(
    \lambda\in\Lambda(d,t)
\)
with
\(
    c_\lambda>0
\),
and suppose that
\(
    |\psi\rangle\in\mathcal S
\)
satisfies
\(
    \Pi_\lambda|\psi\rangle=0.
\)
Applying Condition~\ref{cond1} of Theorem~\ref{thm:qpb-classification} with
\(
    |\phi\rangle=|\psi\rangle
\),
we obtain
\[
    \langle\psi|\Pi_\lambda|\psi\rangle
    =
    c_\lambda\langle\psi|\psi\rangle
    =
    c_\lambda\|\psi\|^2.
\]
Since
\(
    \Pi_\lambda=\Pi_\lambda^\dagger=\Pi_\lambda^2
\),
we also have
\[
    \langle\psi|\Pi_\lambda|\psi\rangle
    =
    \langle\psi|\Pi_\lambda^\dagger
                      \Pi_\lambda|\psi\rangle
    =
    \|\Pi_\lambda|\psi\rangle\|^2
    =
    0.
\]
Therefore
\(
    c_\lambda\|\psi\|^2=0.
\)
Because
\(
    c_\lambda>0
\),
it follows that
\(
    \|\psi\|=0
\),
and hence
\(
    |\psi\rangle=0
\).
Thus
\(
    \Pi_\lambda|_{\mathcal S}
\)
is injective.

Since
\[
    \Pi_\lambda(\mathcal S)
    \subseteq
    \operatorname{im}(\Pi_\lambda)
    =
    \iota_\lambda(V_\lambda\otimes W_\lambda),
\]
injectivity gives
\[
    \dim_{\mathbb C}\mathcal S
    =
    \dim_{\mathbb C}\Pi_\lambda(\mathcal S)
    \leq
    \dim_{\mathbb C}
    \operatorname{im}(\Pi_\lambda)
    =
    \dim_{\mathbb C}(V_\lambda\otimes W_\lambda).
\]
Finally,
\[
    \dim_{\mathbb C}(V_\lambda\otimes W_\lambda)
    =
    \dim_{\mathbb C}V_\lambda\,
    \dim_{\mathbb C}W_\lambda,
\]
which proves~\eqref{eq:active-block-dimension-bound}.
\end{proof}

\begin{remark}[Twirls of general compact groups]\label{rem:compact-group}
The proof of Theorem~\ref{thm:qpb-classification} uses only the blockwise form of the twirl in Lemma~\ref{lem:blockwise-collective-Haar-twirl}. That form rests on Schur's lemma and the inequivalence of distinct blocks, not on the specific structure of Schur-Weyl duality. Let $\mathrm G$ be a compact group with normalized Haar measure $\mu_{\text{Haar}}$, let $\pi_{\mathrm{G}}\colon\mathrm G\longrightarrow\mathrm U(\mathcal H)$ be a finite-dimensional unitary representation, and put
\[
  \mathcal T_\pi(X):=\int_{\mathrm G}\pi_{\mathrm{G}}(g)X\pi_{\mathrm{G}}(g)^\dagger\,
  \mathrm d\mu_{\mathrm G}(g),
  \qquad
  \mathcal H=\bigoplus_{\alpha}\iota_\alpha(V_\alpha\otimes M_\alpha),
\]
where the $V_\alpha$ are pairwise inequivalent irreducible representations occurring in $\pi_{\mathrm{G}}$ and the $M_\alpha$ are their multiplicity spaces. Then $\mathcal T_\pi(X)=\sum_\alpha\iota_\alpha\bigl[\tfrac{I_{V_\alpha}}{\dim_{\mathbb{C}} V_\alpha} \otimes\operatorname{Tr}_{V_\alpha}(\iota_\alpha^\dagger X\iota_\alpha)\bigr]
\iota_\alpha^\dagger$.
The proofs of Theorems~\ref{thm:qpb-classification} and
\ref{thm:general-multiblock-construction} therefore will be the same applied with $(V_\lambda,W_\lambda)$ replaced by $(V_\alpha,M_\alpha)$. A nonzero subspace $\mathcal S\subseteq\mathcal H$ satisfies $\mathcal T_\pi\circ\mathcal J_{\mathcal S}=\mathcal R_{\sigma_{\mathcal S},\mathcal S}$
for some state $\sigma_{\mathcal S}$ if and only if there are weights $c_\alpha\ge0$ with $\sum_\alpha c_\alpha=1$ and states $\omega_\alpha\in\mathcal D(M_\alpha)$ such that
\[
  \langle\phi|\Pi_\alpha|\psi\rangle=c_\alpha\langle\phi|\psi\rangle,
  \qquad
  \operatorname{Tr}_{V_\alpha}\bigl(\iota_\alpha^\dagger|\psi\rangle\langle\phi|
  \iota_\alpha\bigr)=c_\alpha\langle\phi|\psi\rangle\,\omega_\alpha
\]
for all $|\psi\rangle,|\phi\rangle\in\mathcal S$. The same normal form holds, with the bound $n\operatorname{rank}\omega_\alpha\le\dim V_\alpha$. The Schur-Weyl case adds two features that the rest of the paper exploits: the multiplicity spaces $W_\lambda$ are themselves irreducible modules of $\mathfrak S_t$, and they branch along the Young graph under deletion of a tensor factor. The heredity theory of Section~\ref{subsec:hereditary-qpb} rests on this additional structure and has no analogue for a general compact group.
\end{remark}

\subsubsection{The general Schur-Weyl multi-block construction}\label{subsubsec:general-multi-block-construction}
The following construction simultaneously allows an arbitrary finite collection of active Schur-Weyl blocks and an arbitrary fixed density operator on the multiplicity space of each active block. Also, in the following theorem we show that every Haar-twirl admissible subspace is obtained in this way, block by
block.  In particular, the rank of the fixed multiplicity-space state gives a sharp restriction on the dimension of the admissible subspace.

\begin{theorem}[General multi-block construction]
\label{thm:general-multiblock-construction}
Fix an integer
\(
    1 \leq n \leq d^t.
\)
Choose the following data.

\begin{enumerate}
\item[\textup{(D1)}]
A nonempty set of active partitions
\(
    \Lambda_0
    =
    \{\lambda_1,\ldots,\lambda_r\}
    \subseteq
    \Lambda(d,t).
\)
The construction is genuinely multi-block when
\(
    r\geq2
\).

\item[\textup{(D2)}]
Strictly positive real numbers
\(
    c_\lambda>0,\lambda\in\Lambda_0,
\)
such that
\(
\label{eq:general-weights}
    \sum_{\lambda\in\Lambda_0}
    c_\lambda
    =
    1.
\)

\item[\textup{(D3)}]
For every
\(
    \lambda\in\Lambda_0
\),
a density operator
\(
    \omega_\lambda
    \in
    \mathcal D(W_\lambda),
\)
and let
\(
\label{eq:r-lambda-def}
    r_\lambda
    :=
    \operatorname{rank}(\omega_\lambda).
\)

\item[\textup{(D4)}]
For every
\(
    \lambda\in\Lambda_0
\),
an auxiliary Hilbert space
\(
    K_\lambda
\)
with
\(
    \dim_{\mathbb C}K_\lambda=r_\lambda,
\)
together with a unit vector
\(
    |\Omega_\lambda\rangle
    \in
    K_\lambda\otimes W_\lambda
\)
which is a minimal purification of
\(
    \omega_\lambda
\), i.e.,
\begin{equation}
\label{eq:Omega-lambda-purification}
    \operatorname{Tr}_{K_\lambda}
    \bigl(
        |\Omega_\lambda\rangle
        \langle\Omega_\lambda|
    \bigr)
    =
    \omega_\lambda.
\end{equation}

\item[\textup{(D5)}]
For every
\(
    \lambda\in\Lambda_0
\),
assume the dimension condition
\begin{equation}
\label{eq:general-sharp-dimension-condition}
    n\,r_\lambda
    \leq
    \dim_{\mathbb C}V_\lambda,
\end{equation}
and choose a linear isometry
\begin{equation}
\label{eq:U-lambda-def}
    U_\lambda:
    \mathbb C^n\otimes K_\lambda
    \longrightarrow
    V_\lambda,
    \qquad
    U_\lambda^\dagger U_\lambda
    =
    I_{\mathbb C^n\otimes K_\lambda}.
\end{equation}
\end{enumerate}

For every active
\(
    \lambda\in\Lambda_0
\),
define
\begin{equation}
\label{eq:J-lambda-def}
    J_\lambda:
    \mathbb C^n
    \longrightarrow
    V_\lambda\otimes W_\lambda
\end{equation}
by
\begin{equation}
\label{eq:J-lambda-action}
    J_\lambda|x\rangle
    :=
    (U_\lambda\otimes I_{W_\lambda})
    \bigl(
        |x\rangle\otimes|\Omega_\lambda\rangle
    \bigr),
    \qquad
    |x\rangle\in\mathbb C^n.
\end{equation}
Define the ambient encoding isometry
\begin{equation}
\label{eq:J-total-def}
    J:
    \mathbb C^n
    \longrightarrow
    \mathcal H_d^{\otimes t}
\end{equation}
by
\begin{equation}
\label{eq:J-total-action}
    J|x\rangle
    :=
    \sum_{\lambda\in\Lambda_0}
    \sqrt{c_\lambda}\,
    \iota_\lambda
    J_\lambda|x\rangle,
\end{equation}
and define
\begin{equation}
\label{eq:general-S-def}
    \mathcal S
    :=
    J(\mathbb C^n)
    \subseteq
    \mathcal H_d^{\otimes t}.
\end{equation}

Then \(J\) is an isometry, and
\(
    \dim_{\mathbb C}\mathcal S=n.
\)
Moreover, \(\mathcal S\) is Haar-twirl admissible.  Its active Schur-Weyl weights are
\(
    \{c_\lambda:
      \lambda\in\Lambda_0\},
\)
while
\(
    c_\nu=0 \; \text{with} \; \nu\in\Lambda(d,t)\setminus\Lambda_0,
\)
and its replacement state is
\begin{equation}
\label{eq:general-multi-block-sigma}
    \sigma
    =
    \sum_{\lambda\in\Lambda_0}
    c_\lambda\,
    \iota_\lambda
    \left(
        \frac{I_{V_\lambda}}
             {\dim_{\mathbb C}V_\lambda}
        \otimes
        \omega_\lambda
    \right)
    \iota_\lambda^\dagger.
\end{equation}

Finally, for every
\(
    |\psi\rangle\in\mathcal S
\)
and every active
\(
    \lambda\in\Lambda_0
\),
\begin{equation}
\label{eq:general-active-block-norm}
    \|\Pi_\lambda|\psi\rangle\|^2
    =
    c_\lambda\|\ket{\psi}\|^2.
\end{equation}
Consequently, if
\(
    |\Lambda_0|\geq2
\),
then every nonzero vector in
\(
    \mathcal S
\)
has nonzero projection onto every active Schur-Weyl block and
therefore cannot lie entirely in a single isotypic component.

\smallskip
Conversely, if \(\mathcal{S}\) is Haar-twirl admissible then there exist data satisfying \textup{(D1)}-\textup{(D5)} that recover \(\mathcal{S}\). Combining both the conditions we get \(\mathcal{S}\) is Haar-twirl admissible if and only if \textup{(D1)}-\textup{(D5)} holds.
\end{theorem}

\begin{proof}
Fix
\(
    \lambda\in\Lambda_0
\)
and
\(
    |x\rangle,|y\rangle\in\mathbb C^n
\).
Using~\eqref{eq:J-lambda-action}, we get,
\begin{align}
    \langle J_\lambda y|J_\lambda x\rangle
    &=
    \bigl(
        \langle y|\otimes
        \langle\Omega_\lambda|
    \bigr)
    \bigl(
        U_\lambda^\dagger U_\lambda
        \otimes I_{W_\lambda}
    \bigr)
    \bigl(
        |x\rangle\otimes
        |\Omega_\lambda\rangle
    \bigr)
    \notag\\
    &=
    \bigl(
        \langle y|\otimes
        \langle\Omega_\lambda|
    \bigr)
    \bigl(
        |x\rangle\otimes
        |\Omega_\lambda\rangle
    \bigr)
    \notag\\
    &=
    \langle y|x\rangle
    \langle\Omega_\lambda|
           \Omega_\lambda\rangle
    \notag\\
    &=
    \langle y|x\rangle.
\end{align}
Hence
\begin{equation}
\label{eq:J-lambda-isometry}
    J_\lambda^\dagger J_\lambda
    =
    I_{\mathbb C^n}.
\end{equation}

\medskip
Now we show that \(J\) is an isometry. 

Let
\(
    |x\rangle,|y\rangle\in\mathbb C^n
\).
From~\eqref{eq:J-total-action}, we have,
\begin{align}
    \langle Jy|Jx\rangle
    &=
    \sum_{\lambda,\mu\in\Lambda_0}
    \sqrt{c_\lambda c_\mu}\,
    \langle J_\lambda y|
    \iota_\lambda^\dagger
    \iota_\mu
    |J_\mu x\rangle.
\end{align}
Using [Appendix A, ~\eqref{eq:appendix-cross-inclusion-relations}],
all terms with
\(
    \lambda\neq\mu
\)
vanish.  Therefore
\begin{align}
    \langle Jy|Jx\rangle
    &=
    \sum_{\lambda\in\Lambda_0}
    c_\lambda
    \langle J_\lambda y|J_\lambda x\rangle
    \notag\\
    &=
    \sum_{\lambda\in\Lambda_0}
    c_\lambda
    \langle y|x\rangle
    \notag\\
    &=
    \left(
        \sum_{\lambda\in\Lambda_0}
        c_\lambda
    \right)
    \langle y|x\rangle
    \notag\\
    &=
    \langle y|x\rangle,
\end{align}
where the last equality follows from \textup{(D2)}.
Thus
\(
\label{eq:J-total-isometry}
    J^\dagger J
    =
    I_{\mathbb C^n}.
\)
In particular, \(J\) is injective and
\(
    \dim_{\mathbb C}\mathcal S
    = n.
\)

\medskip
Fix
\(
    \lambda\in\Lambda_0
\).
Using
\(
    \Pi_\lambda
    =
    \iota_\lambda\iota_\lambda^\dagger
\),
we obtain
\begin{align}
    \Pi_\lambda J|x\rangle
    &=
    \sum_{\mu\in\Lambda_0}
    \sqrt{c_\mu}\,
    \iota_\lambda
    \iota_\lambda^\dagger
    \iota_\mu
    J_\mu|x\rangle
    \notag\\
    &=
    \sqrt{c_\lambda}\,
    \iota_\lambda
    J_\lambda|x\rangle.
\label{eq:general-block-projection}
\end{align}
Therefore,
\begin{equation}
\label{eq:general-block-coordinate}
    \iota_\lambda^\dagger
    J|x\rangle
    =
    \sqrt{c_\lambda}\,
    J_\lambda|x\rangle.
\end{equation}
If
\(
    \nu\notin\Lambda_0
\),
then
\(
    \Pi_\nu J|x\rangle=0.
\)

\medskip
Let
\(
    |\psi\rangle
    =
    J|x\rangle
\),
\(|\phi\rangle
    =
    J|y\rangle
\)
be arbitrary vectors in
\(
    \mathcal S
\).
For every active
\(
    \lambda\in\Lambda_0
\),
Equation~\eqref{eq:general-block-projection} and
\eqref{eq:J-lambda-isometry} give
\begin{align}
    \langle\phi|
    \Pi_\lambda
    |\psi\rangle
    &=
    c_\lambda
    \langle J_\lambda y|
            J_\lambda x\rangle
    \notag\\
    &=
    c_\lambda
    \langle y|x\rangle.
\end{align}
Since \(J\) is an isometry,
\[
    \langle\phi|\psi\rangle
    =
    \langle Jy|Jx\rangle
    =
    \langle y|x\rangle.
\]
Hence
\begin{equation}
\label{eq:general-condition-A}
    \langle\phi|
    \Pi_\lambda
    |\psi\rangle
    =
    c_\lambda
    \langle\phi|\psi\rangle.
\end{equation}
For an inactive
\(
    \nu\notin\Lambda_0
\),
\[
    \langle\phi|
    \Pi_\nu
    |\psi\rangle
    =
    0
    =
    c_\nu
    \langle\phi|\psi\rangle.
\]
Together with
\(
    \sum_\lambda c_\lambda=1
\),
this proves Condition~\ref{cond1}.

\medskip
Now we verify Condition~\ref{cond2}.

Fix
\(
    \lambda\in\Lambda_0
\).
By~\eqref{eq:general-block-coordinate}, we have,
\begin{align}
    \iota_\lambda^\dagger
    |\psi\rangle\langle\phi|
    \iota_\lambda
    &=
    c_\lambda
    J_\lambda
    |x\rangle\langle y|
    J_\lambda^\dagger
    \notag\\
    &=
    c_\lambda
    (U_\lambda\otimes I_{W_\lambda})
    \Bigl(
        |x\rangle\langle y|
        \otimes
        |\Omega_\lambda\rangle
        \langle\Omega_\lambda|
    \Bigr)
    (U_\lambda^\dagger\otimes I_{W_\lambda}).
\end{align}
Apply~\eqref{eq:isometric-partial-trace-identity} with
\[
    A
    =
    \mathbb C^n\otimes K_\lambda,
    \qquad
    B
    =
    V_\lambda,
    \qquad
    C
    =
    W_\lambda.
\]
We obtain
\begin{align}
    \operatorname{Tr}_{V_\lambda}
    \left(
        \iota_\lambda^\dagger
        |\psi\rangle\langle\phi|
        \iota_\lambda
    \right)
    =
    c_\lambda
    \operatorname{Tr}_{\mathbb C^n\otimes K_\lambda}
    \Bigl(
        |x\rangle\langle y|
        \otimes
        |\Omega_\lambda\rangle
        \langle\Omega_\lambda|
    \Bigr).
\end{align}
Applying ~\eqref{eq:Omega-lambda-purification}, we get
\begin{align}
    \operatorname{Tr}_{\mathbb C^n\otimes K_\lambda}
    \Bigl(
        |x\rangle\langle y|
        \otimes
        |\Omega_\lambda\rangle
        \langle\Omega_\lambda|
    \Bigr)
    &=
    \operatorname{Tr}_{\mathbb C^n}
    \bigl(
        |x\rangle\langle y|
    \bigr)
    \,
    \operatorname{Tr}_{K_\lambda}
    \bigl(
        |\Omega_\lambda\rangle
        \langle\Omega_\lambda|
    \bigr)
    \notag\\
    &=
    \langle y|x\rangle
    \omega_\lambda.
\end{align}

Therefore
\begin{equation}
\label{eq:general-condition-B}
    \operatorname{Tr}_{V_\lambda}
    \left(
        \iota_\lambda^\dagger
        |\psi\rangle\langle\phi|
        \iota_\lambda
    \right)
    =
    c_\lambda
    \langle y|x\rangle
    \omega_\lambda.
\end{equation}
Since
\(
    \langle y|x\rangle
    =
    \langle\phi|\psi\rangle
\),
Equation~\eqref{eq:general-condition-B} becomes
\[
    \operatorname{Tr}_{V_\lambda}
    \left(
        \iota_\lambda^\dagger
        |\psi\rangle\langle\phi|
        \iota_\lambda
    \right)
    =
    c_\lambda
    \langle\phi|\psi\rangle
    \omega_\lambda,
\]
which is the Condition~\ref{cond2}.

Thus Theorem~\ref{thm:qpb-classification} implies that
\(
    \mathcal S
\)
is Haar-twirl admissible.

\medskip
Substituting the active weights and multiplicity-space states into
the replacement-state formula of Theorem~\ref{thm:qpb-classification} gives
\[
    \sigma
    =
    \sum_{\lambda\in\Lambda_0}
    c_\lambda
    \iota_\lambda
    \left(
        \frac{I_{V_\lambda}}
             {\dim_{\mathbb C}V_\lambda}
        \otimes
        \omega_\lambda
    \right)
    \iota_\lambda^\dagger.
\]
Each summand is positive semidefinite.  Furthermore,
\begin{align}
    \operatorname{Tr}(\sigma)
    &=
    \sum_{\lambda\in\Lambda_0}
    c_\lambda
    \operatorname{Tr}
    \left(
        \frac{I_{V_\lambda}}
             {\dim_{\mathbb C}V_\lambda}
        \otimes
        \omega_\lambda
    \right)
    \notag\\
    &=
    \sum_{\lambda\in\Lambda_0}
    c_\lambda
    \operatorname{Tr}
    \left(
        \frac{I_{V_\lambda}}
             {\dim_{\mathbb C}V_\lambda}
    \right)
    \operatorname{Tr}(\omega_\lambda)
    \notag\\
    &=
    \sum_{\lambda\in\Lambda_0}
    c_\lambda
    \cdot1\cdot1
    \notag\\
    &=
    1.
\end{align}
Hence
\(
    \sigma\in
    \mathcal D(\mathcal H_d^{\otimes t})
\),
as required.

\medskip
Now we show that every nonzero encoded vector occupies every active block.

Let
\(
    |\psi\rangle=J|x\rangle\in\mathcal S
\).
For every active
\(
    \lambda\in\Lambda_0
\),
Equation~\eqref{eq:general-block-projection} gives
\begin{align}
    \|\Pi_\lambda|\psi\rangle\|^2
    &=
    c_\lambda
    \|\iota_\lambda J_\lambda|x\rangle\|^2
    \notag\\
    &=
    c_\lambda
    \|J_\lambda|x\rangle\|^2
    \notag\\
    &=
    c_\lambda
    \|x\|^2
    \notag\\
    &=
    c_\lambda
    \|J|x\rangle\|^2
    \notag\\
    &=
    c_\lambda
    \|\psi\|^2.
\end{align}
This proves~\eqref{eq:general-active-block-norm}.

If
\(
    |\psi\rangle\neq0
\),
then
\(
    \|\psi\|^2>0
\).
Since
\(
    c_\lambda>0
\)
for every
\(
    \lambda\in\Lambda_0
\),
we have
\(
    \|\Pi_\lambda|\psi\rangle\|^2>0
\)
for every \(\lambda\in\Lambda_0\).
Therefore every nonzero vector of
\(
    \mathcal S
\)
has a nonzero component in every active Schur-Weyl block.
If
\(
    |\Lambda_0|\geq2
\),
no such vector can be contained in a single block.

\smallskip
Conversely, let \(\mathcal S\subset \mathcal H_d^{\otimes t}\) be a nonzero Haar-twirl admissible subset, and let
\(n:=\dim_{\mathbb C}\mathcal S\).
Choose an isometry
\[
       J\colon \mathbb C^n\longrightarrow
                  \mathcal H_d^{\otimes t},
       \qquad
       J(\mathbb C^n)=\mathcal S,
       \qquad J^\dagger J=I_{\mathbb C^n}.
\]
By the classification theorem (Theorem~\ref{thm:qpb-classification}), there
are constants \(c_\lambda\geq0\),
\(\lambda\in\Lambda(d,t)\), with
\(\sum_{\lambda}c_\lambda=1\), such that
\begin{equation}\label{eq:c1}
 \langle Jy|\Pi_\lambda|Jx\rangle
       =c_\lambda\langle y|x\rangle,
       \qquad
       \; x,y\in\mathbb C^n.
\end{equation}
For each \(\lambda\), put
\[
 B_\lambda:=\iota_\lambda^\dagger J:
       \mathbb C^n\longrightarrow V_\lambda\otimes W_\lambda.
\]
Since \(\Pi_\lambda=\iota_\lambda\iota_\lambda^\dagger\),
equation~\ref{eq:c1} gives
\[
 \langle y|B_\lambda^\dagger B_\lambda|x\rangle
    =\langle Jy|\Pi_\lambda|Jx\rangle
    =c_\lambda\langle y|x\rangle.
\]
Nondegeneracy of the inner product therefore yields
\begin{equation}\label{eq:c2}
             B_\lambda^\dagger B_\lambda
                    =c_\lambda I_{\mathbb C^n}.
\end{equation}

Define the active set
\[
       \Lambda_0
          :=\{\lambda\in\Lambda(d,t):c_\lambda>0\}.
\]
It is nonempty because \(\sum_\lambda c_\lambda=1\);
its weights are strictly positive and sum to one.
If \(\lambda\notin\Lambda_0\), then equation~\ref{eq:c2} implies \(B_\lambda^\dagger B_\lambda=0\), hence \(B_\lambda=0\). For \(\lambda\in\Lambda_0\),
define
\begin{equation}\label{eq:c3}
       J_\lambda:=c_\lambda^{-1/2}B_\lambda .
\end{equation}
Equation~\ref{eq:c2} and Equation~\ref{eq:c3} show that
\(J_\lambda^\dagger J_\lambda=I_{\mathbb C^n}\).
Thus \(J_\lambda\) is an isometry. In particular,
\begin{equation}\label{eq:c4}
       \iota_\lambda^\dagger J
             =\sqrt{c_\lambda}\,J_\lambda,
       \qquad
       \iota_\lambda^\dagger J=0
       ,\quad\lambda\notin\Lambda_0.
\end{equation}

Fix an active \(\lambda\). Then from the second condition of Theorem~\ref{thm:qpb-classification} we  get a density operator
\(\omega_\lambda\in\mathcal D(W_\lambda)\) for which
\[
 \operatorname{Tr}_{V_\lambda}
       \bigl(B_\lambda|x\rangle\langle y|
                    B_\lambda^\dagger\bigr)
       =c_\lambda\langle y|x\rangle\omega_\lambda,
       \quad x,y\in\mathbb C^n.
\]
Substitution of \(B_\lambda=\sqrt{c_\lambda}J_\lambda\)
and division by \(c_\lambda>0\) give
\begin{equation}\label{eq:c5}
 \operatorname{Tr}_{V_\lambda}
       \bigl(J_\lambda|x\rangle\langle y|
                    J_\lambda^\dagger\bigr)
       =\langle y|x\rangle\omega_\lambda.
\end{equation}
Because rank-one operators span
\(\mathcal L(\mathbb C^n)\), the same identity,
by linearity, we have
\begin{equation}\label{eq:c6}
 \operatorname{Tr}_{V_\lambda}
       (J_\lambda XJ_\lambda^\dagger)
       =\operatorname{Tr}(X)\omega_\lambda,
       \qquad X\in\mathcal L(\mathbb C^n).
\end{equation}

Set \(r_\lambda:=\operatorname{rank}(\omega_\lambda)\).
Choose an orthonormal eigenbasis of \(W_\lambda\)
whose first \(r_\lambda\) vectors correspond to the
strictly positive eigenvalues. Thus
\begin{equation}\label{eq:c7}
 \omega_\lambda
   =\sum_{a=1}^{r_\lambda}
        p_{\lambda,a}
        |w_{\lambda,a}\rangle
        \langle w_{\lambda,a}|,
 \qquad
 p_{\lambda,a}>0,
 \qquad
 \sum_{a=1}^{r_\lambda}p_{\lambda,a}=1.
\end{equation}
Complete these vectors to an orthonormal basis
\((|w_{\lambda,a}\rangle)_{a=1}^{\dim_{\mathbf{C}} W_\lambda}\)
of \(W_\lambda\), and define linear operators
\(A_{\lambda,a}:\mathbb C^n \longrightarrow V_\lambda\) by
\[
 A_{\lambda,a}
    :=(I_{V_\lambda}\otimes
           \langle w_{\lambda,a}|)J_\lambda .
\]
For each \(\ket{x}\in\mathbb C^n\), expansion in this
basis gives
\begin{equation}\label{eq:c8}
 J_\lambda|x\rangle
   =\sum_{a=1}^{\dim_{\mathbf{C}} W_\lambda}
        A_{\lambda,a}|x\rangle
           \otimes|w_{\lambda,a}\rangle.
\end{equation}
Consequently, taking the partial trace in Equation~\ref{eq:c5} yields
\begin{align}\label{eq:c9}
 \operatorname{Tr}_{V_\lambda}
       \bigl(J_\lambda|x\rangle\langle y|
                    J_\lambda^\dagger\bigr)
 &=\sum_{a,b=1}^{\dim_{\mathbf{C}} W_\lambda}
       \langle A_{\lambda,b}y
                    |A_{\lambda,a}x\rangle
       |w_{\lambda,a}\rangle
          \langle w_{\lambda,b}|.
\end{align}
Comparing the \((a,b)\) matrix coefficients of Equation~\ref{eq:c9} with Equation~\ref{eq:c5} and Equation~\ref{eq:c7}, for
all \(x,y\in\mathbb C^n\), gives
\[
 \langle y|
       A_{\lambda,b}^\dagger A_{\lambda,a}
       |x\rangle
   =
   \begin{cases}
      p_{\lambda,a}\langle y|x\rangle,
          &a=b\leq r_\lambda,\\
      0,  &\text{otherwise}.
   \end{cases}
\]
Hence the following are operator identities:
\begin{equation}\label{eq:c10}
 A_{\lambda,b}^\dagger A_{\lambda,a}
   =
   \begin{cases}
      p_{\lambda,a} I_{\mathbb C^n},
          &a=b\leq r_\lambda,\\
      0,  &\text{otherwise}.
   \end{cases}
\end{equation}
In particular,
\(A_{\lambda,a}^\dagger A_{\lambda,a}=0\)
and therefore \(A_{\lambda,a}=0\) whenever
\(a>r_\lambda\).

For \(1\leq a\leq r_\lambda\), define
\[
       T_{\lambda,a}
          :=p_{\lambda,a}^{-1/2}
                         A_{\lambda,a}.
\]
Equation~\ref{eq:c10} implies
\begin{equation}\label{eq:c11}
       T_{\lambda,b}^\dagger T_{\lambda,a}
             =\delta_{ab} I_{\mathbb C^n},
       \qquad 1\leq a,b\leq r_\lambda.
\end{equation}
Thus each \(T_{\lambda,a}\) is an isometry, and
the ranges of distinct \(T_{\lambda,a}\) are
orthogonal. Their orthogonal direct sum has
dimension \(nr_\lambda\), so
\begin{equation}\label{eq:c12}
              nr_\lambda\leq\dim_{\mathbb C}V_\lambda .
\end{equation}
This proves the dimension condition in (D5).

Let \(\mathcal K_\lambda\cong\mathbb C^{r_\lambda}\)
have orthonormal basis
\(\{|a\rangle\}_{a=1}^{r_\lambda}\), and define
\begin{equation}\label{eq:c13}
 U_\lambda:\mathbb C^n\otimes\mathcal K_\lambda
                     \longrightarrow V_\lambda,
 \qquad
 U_\lambda(|x\rangle\otimes|a\rangle)
       :=T_{\lambda,a}|x\rangle ,
\end{equation}
extending linearly. For arbitrary \(\ket{x},\ket{y}\) and basis vectors \(\ket{a},\ket{b}\), Equation~\ref{eq:c11}
gives
\[
 \langle U_\lambda(y\otimes b)
             |U_\lambda(x\otimes a)\rangle
   =\langle y|
         T_{\lambda,b}^\dagger T_{\lambda,a}
      |x\rangle
   =\delta_{ab}\langle y|x\rangle
   =\langle y\otimes b|x\otimes a\rangle .
\]
By sesquilinearity this holds for all vectors in
\(\mathbb C^n\otimes\mathcal K_\lambda\); hence
\(U_\lambda^\dagger U_\lambda
 =I_{\mathbb C^n\otimes\mathcal K_\lambda}\).

Now put
\begin{equation}\label{eq:c14}
 |\Omega_\lambda\rangle
   :=\sum_{a=1}^{r_\lambda}
       \sqrt{p_{\lambda,a}}\,
       |a\rangle\otimes|w_{\lambda,a}\rangle
       \in\mathcal K_\lambda\otimes W_\lambda .
\end{equation}
It is a unit vector, and orthonormality of
\(\{|a\rangle\}_{i=1}^{r_{\lambda}}\) gives
\begin{equation}\label{eq:c15}
 \operatorname{Tr}_{\mathcal K_\lambda}
       (|\Omega_\lambda\rangle
          \langle\Omega_\lambda|)
   =\sum_{a=1}^{r_\lambda}
       p_{\lambda,a}
       |w_{\lambda,a}\rangle
          \langle w_{\lambda,a}|
   =\omega_\lambda .
\end{equation}
Because its Schmidt rank is \(r_\lambda\),
this is a minimal purification of
\(\omega_\lambda\).

Using Equation~\ref{eq:c8}, the vanishing of the terms
with \(a>r_\lambda\), and
\(A_{\lambda,a}
 =\sqrt{p_{\lambda,a}}T_{\lambda,a}\), we obtain,
for every \(x\in\mathbb C^n\),
\begin{align}\label{eq:c16}
 J_\lambda|x\rangle
 &=\sum_{a=1}^{r_\lambda}
       \sqrt{p_{\lambda,a}}\,
       T_{\lambda,a}|x\rangle
          \otimes|w_{\lambda,a}\rangle \notag\\
 &=(U_\lambda\otimes I_{W_\lambda})
       \bigl(|x\rangle\otimes
                    |\Omega_\lambda\rangle\bigr).
\end{align}
Thus the data
\((\Lambda_0,\{c_\lambda\},
  \{\omega_\lambda\},\{\mathcal K_\lambda,
  \Omega_\lambda,U_\lambda\})\)
satisfy (D1)--(D5).

\smallskip
Finally, completeness of the Schur-Weyl
projectors and Equation~\ref{eq:c4} show that, for every \(x\in\mathbb C^n\),
\begin{align}\label{eq:c17}
 J|x\rangle
 &=\sum_{\lambda\in\Lambda(d,t)}
       \Pi_\lambda J|x\rangle \notag\\
 &=\sum_{\lambda\in\Lambda_0}
       \iota_\lambda B_\lambda|x\rangle \notag\\
 &=\sum_{\lambda\in\Lambda_0}
       \sqrt{c_\lambda}\,
       \iota_\lambda
       (U_\lambda\otimes I_{W_\lambda})
       \bigl(|x\rangle\otimes
                    |\Omega_\lambda\rangle\bigr).
\end{align}
Equation~\ref{eq:c17} is the encoding
isometry in the stated construction, and its
range is the original subspace \(\mathcal S\).
The classification theorem also identifies its
replacement state as
\[
 \sigma_{\mathcal S}
   =\sum_{\lambda\in\Lambda_0}
       c_\lambda\,
       \iota_\lambda
       \left(
         \frac{I_{V_\lambda}}
              {\dim_{\mathbb C}V_\lambda}
             \otimes\omega_\lambda
       \right)\iota_\lambda^\dagger .
\]
This proves that every Haar-twirl admissible subspace
has the asserted multi-block normal form and
completes the converse.
\end{proof}

\begin{corollary}\label{cor:exact-maximal-dimension-fixed-active-data}
Fix a nonempty active set \(\Lambda_0\), positive weights
\(\{c_\lambda\}_{\lambda\in\Lambda_0}\) summing to one, and density operators \(\omega_\lambda\in\mathcal D(W_\lambda)\). Such data are realizable by a nonzero Haar-twirl admissible subspace if and only if
\begin{equation}
\label{eq:feasibility-fixed-active-data}
    \operatorname{rank}(\omega_\lambda)
    \leq \dim_{\mathbb C}V_\lambda
    \qquad\forall\lambda\in\Lambda_0.
\end{equation}
Whenever Equation~\eqref{eq:feasibility-fixed-active-data} holds, among all Haar-twirl admissible subspaces having these data with active blocks, the maximal possible
dimension is
\begin{equation}
\label{eq:exact-maximal-dimension-fixed-active-data}
    n_{\max}
    =
    \min_{\lambda\in\Lambda_0}
    \left\lfloor
       \frac{\dim_{\mathbb C}V_\lambda}
            {\operatorname{rank}(\omega_\lambda)}
    \right\rfloor.
\end{equation}
\end{corollary}

\begin{proof}
If a nonzero admissible subspace exists, then its dimension is at least one, so Equation~\eqref{eq:general-sharp-dimension-condition} implies Equation~\eqref{eq:feasibility-fixed-active-data}. Conversely, if Equation~\eqref{eq:feasibility-fixed-active-data} holds, the construction below with \(n=1\) produces a nonzero admissible subspace, proving feasibility.

The upper bound follows from Equation~\eqref{eq:general-sharp-dimension-condition} for every active \(\lambda\). Conversely, let \(n\) equal the right-hand side of Equation~\eqref{eq:exact-maximal-dimension-fixed-active-data}. Then
\[
    n \operatorname{rank}(\omega_\lambda)
    \leq\dim V_\lambda
    \qquad\forall\lambda\in\Lambda_0.
\]
Hence, for every active \(\lambda\), an isometry
\[
    U_\lambda:
       \mathbb C^n\otimes\mathbb C^{\operatorname{rank}(\omega_\lambda)}
       \longrightarrow V_\lambda
\]
exists. The forward construction in Theorem~\ref{thm:general-multiblock-construction} then produces an admissible subspace of dimension \(n\) with the data (D1)--(D5).
\end{proof}

\begin{corollary}[Pure multiplicity states]
\label{cor:pure-multiplicity-normal-form}
In Theorem~\ref{thm:general-multiblock-construction}, suppose that
\[
\omega_\lambda
=|s_\lambda\rangle\langle s_\lambda|,
\qquad \|s_\lambda\|=1,
\qquad \lambda\in\Lambda_0.
\]
Then \(r_\lambda
    =
    \operatorname{rank}(\omega_\lambda)
    =
    1\), so that the dimension condition 
\eqref{eq:general-sharp-dimension-condition} reduces to
\begin{equation}
\label{eq:pure-multi-block-dimension}
    n
    \leq
    \dim_{\mathbb C}V_\lambda
    \qquad
    \text{for every }
    \lambda\in\Lambda_0.
\end{equation}
For any isometries \(f_\lambda \colon \mathbb C^n \longrightarrow V_\lambda\), the map
\begin{equation}
Jx
=
\sum_{\lambda\in\Lambda_0}
\sqrt{c_\lambda}\,
\iota_\lambda(f_\lambda x\otimes s_\lambda)
\label{eq:pure-multiplicity-normal-form}
\end{equation}
is an isometry, and \(S:=J(\mathbb C^n)\) is Haar-twirl admissible with
replacement state
\begin{equation}
\sigma_S
=
\sum_{\lambda\in\Lambda_0}
c_\lambda\,
\iota_\lambda
\left(
\frac{I_{V_\lambda}}{\dim_{\mathbb C}V_\lambda}
\otimes |s_\lambda\rangle\langle s_\lambda|
\right)
\iota_\lambda^\dagger .
\label{eq:pure-multiplicity-replacement-state}
\end{equation}
Conversely, every admissible subspace whose active multiplicity-space states
are pure has the form \eqref{eq:pure-multiplicity-normal-form}, after an
isometric identification of its domain with \(\mathbb C^n\).
\end{corollary}

\begin{proof}
Apply Theorem~\ref{thm:general-multiblock-construction} with
\(K_\lambda=\mathbb C\),
\(\Omega_\lambda=1\otimes s_\lambda\), and
\(U_\lambda=f_\lambda\) under the identification
\(\mathbb C^n\otimes\mathbb C\cong\mathbb C^n\).
\end{proof}

\begin{remark}
\label{rem:one-dimensional-W-correct}
The hypothesis in Corollary~\ref{cor:pure-multiplicity-normal-form} is that each fixed
multiplicity-space state \(\omega_\lambda\) has rank one; it does not require \(\dim_{\mathbb{C}}W_\lambda=1\). If \(\dim_{\mathbb C}W_\lambda=1\), however, every density operator on \(W_\lambda\) is automatically pure, so the hypothesis of Corollary~\ref{cor:pure-multiplicity-normal-form} is automatic in that block.
If
\(
    W_\lambda
\)
is one-dimensional and
\(
    s_\lambda
\)
is a unit vector spanning it, then
\(
    V_\lambda\otimes W_\lambda
    \cong
    V_\lambda,
\)
and every linear map
\[
    F\colon
    \mathbb C^n
    \longrightarrow
    V_\lambda\otimes W_\lambda
\]
has a unique representation
\(
    F|x\rangle
    =
    f_\lambda|x\rangle
    \otimes
    |s_\lambda\rangle
\)
for a linear map
\(
    f_\lambda:\mathbb C^n\longrightarrow V_\lambda
\).
If \(F\) is an isometry, then
\(
    f_\lambda^\dagger f_\lambda
    =
    I_{\mathbb C^n}.
\)

For the one-row and one-column partitions,
\[
    \dim_{\mathbb C}W_{(t)}
    =
    1,
    \qquad
    \dim_{\mathbb C}W_{(1^t)}
    =
    1,
\]
the corresponding Specht modules being respectively the trivial and the sign representations of \(\mathfrak{S}_t\). In particular, when
\(t=2\),
\[
    \dim_{\mathbb C}W_{(2)}
    =
    \dim_{\mathbb C}W_{(1,1)}
    =
    1.
\]
On the other hand,
\[
    \dim_{\mathbb C}V_{(2)}
    =
    \dim_{\mathbb C}
    \operatorname{Sym}^2(\mathbb C^d)
    =
    \frac{d(d+1)}{2},
\]
whereas, for \(d\geq2\),
\[
    \dim_{\mathbb C}V_{(1,1)}
    =
    \dim_{\mathbb C}
    \bigwedge\nolimits^2\mathbb C^d
    =
    \frac{d(d-1)}{2}.
\]
Thus, in the \(t=2\) case, the fixed multiplicity-space factor is automatic because \(W_{(2)}\) and \(W_{(1,1)}\) are one-dimensional.
\end{remark}

\subsubsection{A two-block specialization}

The following construction is the pure multiplicity-state, two-block specialization of the multiblock normal form in Theorem~\ref{thm:general-multiblock-construction}. We present it separately as a specific construction.

\begin{proposition}[Two-block Haar-twirl admissible construction.]
\label{prop:two-block-construction}
Let
\(
    \lambda,\mu\in\Lambda(d,t)
\)
be distinct partitions.  Fix an integer
\(
    n\geq1
\)
satisfying
\begin{equation}
\label{eq:two-block-dimension-condition}
    n
    \leq
    \min
    \bigl\{
        \dim_{\mathbb C}V_\lambda,
        \dim_{\mathbb C}V_\mu
    \bigr\}.
\end{equation}
Choose orthonormal families
\(
    \{v_1,\ldots,v_n\}\subseteq V_\lambda
\), \(\{u_1,\ldots,u_n\}\subseteq V_\mu\),
and unit vectors
\(
    s_\lambda\in W_\lambda,
\)
\(
    s_\mu\in W_\mu.
\)
For a fixed real number
\(
    c\in(0,1)
\),
define
\begin{equation}
\label{eq:two-block-xi-def}
    |\xi_j\rangle
    :=
    \sqrt c\,
    \iota_\lambda(v_j\otimes s_\lambda)
    +
    \sqrt{1-c}\,
    \iota_\mu(u_j\otimes s_\mu),
    \qquad
    1\leq j\leq n,
\end{equation}
and let
\begin{equation}
\label{eq:two-block-S-def}
    \mathcal S
    :=
    \operatorname{span}_{\mathbb C}
    \bigl\{
        |\xi_1\rangle,\ldots,|\xi_n\rangle
    \bigr\}
    \subseteq
    \mathcal H_d^{\otimes t}.
\end{equation}
Then the following statements hold.

\begin{enumerate}
\item
The vectors
\(
    |\xi_1\rangle,\ldots,|\xi_n\rangle
\)
form an orthonormal basis of
\(
    \mathcal S
\),
hence,
\(
    \dim_{\mathbb C}\mathcal S=n.
\)

\item
The subspace
\(
    \mathcal S
\)
satisfies Conditions~\ref{cond1} and~\ref{cond2} of Theorem~\ref{thm:qpb-classification} with
\[
    c_\lambda=c,
    \qquad
    c_\mu=1-c,
    \qquad
    c_\nu=0
    \quad
    \text{for }
    \nu\notin\{\lambda,\mu\},
\]
and
\[
    \omega_\lambda
    =
    |s_\lambda\rangle\langle s_\lambda|,
    \qquad
    \omega_\mu
    =
    |s_\mu\rangle\langle s_\mu|.
\]

\item
Consequently,
\(
    \mathcal S
\)
is Haar-twirl admissible and
\begin{equation}
\label{eq:two-block-replacement-state}
\begin{split}
    \sigma
    =
    c\,
    \iota_\lambda
    \left(
        \frac{I_{V_\lambda}}
             {\dim_{\mathbb C}V_\lambda}
        \otimes
        |s_\lambda\rangle\langle s_\lambda|
    \right)
    \iota_\lambda^\dagger
    +
    (1-c)\,
    \iota_\mu
    \left(
        \frac{I_{V_\mu}}
             {\dim_{\mathbb C}V_\mu}
        \otimes
        |s_\mu\rangle\langle s_\mu|
    \right)
    \iota_\mu^\dagger .
\end{split}
\end{equation}

\item
Every nonzero vector in
\(
    \mathcal S
\)
has a nonzero component in each of the two active Schur-Weyl
blocks.  More precisely, for every
\(
    |\psi\rangle\in\mathcal S
\),
\begin{equation}
\label{eq:two-block-norm-identities}
    \|\Pi_\lambda|\psi\rangle\|^2
    =
    c\|\psi\|^2,
    \qquad
    \|\Pi_\mu|\psi\rangle\|^2
    =
    (1-c)\|\psi\|^2.
\end{equation}
Hence no nonzero vector of
\(
    \mathcal S
\)
lies entirely in a single Schur-Weyl isotypic component.
\end{enumerate}
\end{proposition}

\begin{proof}
Fix
\(
    j,k\in[n]
\).
Using~\eqref{eq:two-block-xi-def}, we get,
\begin{align}
    \langle\xi_k|\xi_j\rangle
    &=
    c\,
    \bigl\langle
        \iota_\lambda(v_k\otimes s_\lambda),
        \iota_\lambda(v_j\otimes s_\lambda)
    \bigr\rangle
    \notag\\
    &\quad+
    \sqrt{c(1-c)}\,
    \bigl\langle
        \iota_\lambda(v_k\otimes s_\lambda),
        \iota_\mu(u_j\otimes s_\mu)
    \bigr\rangle
    \notag\\
    &\quad+
    \sqrt{c(1-c)}\,
    \bigl\langle
        \iota_\mu(u_k\otimes s_\mu),
        \iota_\lambda(v_j\otimes s_\lambda)
    \bigr\rangle
    \notag\\
    &\quad+
    (1-c)\,
    \bigl\langle
        \iota_\mu(u_k\otimes s_\mu),
        \iota_\mu(u_j\otimes s_\mu)
    \bigr\rangle .
\label{eq:two-block-orth-expand}
\end{align}
Because
\(
    \lambda\neq\mu
\),
we have
\(
    \iota_\lambda^\dagger\iota_\mu
    =
    \iota_\mu^\dagger\iota_\lambda=0.
\)
Hence both cross terms in
\eqref{eq:two-block-orth-expand} vanish.  Since
\(
    \iota_\lambda
\)
and
\(
    \iota_\mu
\)
are isometries,
\begin{align}
    \langle\xi_k|\xi_j\rangle
    &=
    c\,
    \langle v_k|v_j\rangle
    \langle s_\lambda|s_\lambda\rangle
    +
    (1-c)\,
    \langle u_k|u_j\rangle
    \langle s_\mu|s_\mu\rangle
    \notag\\
    &=
    c\,\delta_{kj}
    +
    (1-c)\,\delta_{kj}
    \notag\\
    &=
    \delta_{kj}.
\end{align}
Therefore
\(
    |\xi_1\rangle,\ldots,|\xi_n\rangle
\)
are orthonormal.  In particular, they are linearly independent and
\(
    \dim_{\mathbb C}\mathcal S=n.
\)

\smallskip
Let
\(
    |\psi\rangle
    =
    \sum_{j=1}^n
    \alpha_j|\xi_j\rangle
    \in\mathcal S.
\)
Using
\(
    \Pi_\lambda
    =
    \iota_\lambda\iota_\lambda^\dagger
\),
we have,
\begin{align}
    \Pi_\lambda|\xi_j\rangle
    &=
    \sqrt c\,
    \iota_\lambda
    \iota_\lambda^\dagger
    \iota_\lambda
    (v_j\otimes s_\lambda)
    +
    \sqrt{1-c}\,
    \iota_\lambda
    \iota_\lambda^\dagger
    \iota_\mu
    (u_j\otimes s_\mu)
    \notag\\
    &=
    \sqrt c\,
    \iota_\lambda(v_j\otimes s_\lambda).
\end{align}
Similarly,
\(
    \Pi_\mu|\xi_j\rangle
    =
    \sqrt{1-c}\,
    \iota_\mu(u_j\otimes s_\mu).
\)
Therefore
\begin{equation}
\label{eq:two-block-projection-psi}
    \Pi_\lambda|\psi\rangle
    =
    \sqrt c
    \sum_{j=1}^n
    \alpha_j
    \iota_\lambda(v_j\otimes s_\lambda),
\end{equation}
and
\begin{equation}
\label{eq:two-block-projection-psi-mu}
    \Pi_\mu|\psi\rangle
    =
    \sqrt{1-c}
    \sum_{j=1}^n
    \alpha_j
    \iota_\mu(u_j\otimes s_\mu).
\end{equation}
For every
\(
    \nu\notin\{\lambda,\mu\}
\),
\(
    \Pi_\nu|\psi\rangle=0.
\)

\medskip
Let
\(
    |\psi\rangle
    =
    \sum_{j=1}^n
    \alpha_j|\xi_j\rangle
\),
\(|\phi\rangle
    =
    \sum_{k=1}^n
    \beta_k|\xi_k\rangle.
\)
Since
\(
    \{|\xi_j\rangle\}_{j=1}^n
\)
is orthonormal,
\begin{equation}
\label{eq:two-block-inner-product}
    \langle\phi|\psi\rangle
    =
    \sum_{j=1}^n
    \overline{\beta_j}\alpha_j.
\end{equation}
Using~\eqref{eq:two-block-projection-psi},
\begin{align}
    \langle\phi|\Pi_\lambda|\psi\rangle
    &=
    c
    \sum_{j,k=1}^n
    \overline{\beta_k}\alpha_j
    \langle v_k|v_j\rangle
    \langle s_\lambda|s_\lambda\rangle
    \notag\\
    &=
    c
    \sum_{j,k=1}^n
    \overline{\beta_k}\alpha_j
    \delta_{kj}
    \notag\\
    &=
    c
    \sum_{j=1}^n
    \overline{\beta_j}\alpha_j
    \notag\\
    &=
    c\langle\phi|\psi\rangle.
\end{align}
Similarly,
\[
    \langle\phi|\Pi_\mu|\psi\rangle
    =
    (1-c)\langle\phi|\psi\rangle.
\]
For
\(
    \nu\notin\{\lambda,\mu\}
\),
\[
    \langle\phi|\Pi_\nu|\psi\rangle=0.
\]
Thus Condition~\ref{cond1} holds with
\[
    c_\lambda=c,
    \qquad
    c_\mu=1-c,
    \qquad
    c_\nu=0
    \quad
    (\nu\notin\{\lambda,\mu\}).
\]
Moreover,
\(
    \sum_{\nu\in\Lambda(d,t)}c_\nu
    =1.
\)

\medskip
Now from~\eqref{eq:two-block-projection-psi}, or by applying
\(
    \iota_\lambda^\dagger
\)
directly, we have
\[
    \iota_\lambda^\dagger|\psi\rangle
    =
    \sqrt c
    \sum_{j=1}^n
    \alpha_j(v_j\otimes s_\lambda),
\]
and
\[
    \iota_\lambda^\dagger|\phi\rangle
    =
    \sqrt c
    \sum_{k=1}^n
    \beta_k(v_k\otimes s_\lambda).
\]
Hence
\begin{align}
    \iota_\lambda^\dagger
    |\psi\rangle\langle\phi|
    \iota_\lambda
    &=
    c
    \sum_{j,k=1}^n
    \alpha_j\overline{\beta_k}
    \bigl(
        |v_j\rangle\langle v_k|
    \bigr)
    \otimes
    \bigl(
        |s_\lambda\rangle
        \langle s_\lambda|
    \bigr).
\end{align}
Taking the partial trace over
\(
    V_\lambda
\),
and using
\(
    \operatorname{Tr}
    \bigl(
        |v_j\rangle\langle v_k|
    \bigr)
    =
    \langle v_k|v_j\rangle
    =
    \delta_{kj},
\)
we obtain
\begin{align}
    \operatorname{Tr}_{V_\lambda}
    \left(
        \iota_\lambda^\dagger
        |\psi\rangle\langle\phi|
        \iota_\lambda
    \right)
    &=
    c
    \sum_{j,k=1}^n
    \alpha_j\overline{\beta_k}
    \delta_{kj}
    |s_\lambda\rangle
    \langle s_\lambda|
    \notag\\
    &=
    c
    \sum_{j=1}^n
    \alpha_j\overline{\beta_j}
    |s_\lambda\rangle
    \langle s_\lambda|
    \notag\\
    &=
    c\langle\phi|\psi\rangle
    |s_\lambda\rangle
    \langle s_\lambda|.
\label{eq:two-block-condition-B-lambda}
\end{align}
Thus Condition~\ref{cond2} holds in the
\(
    \lambda
\)-block with
\(
    \omega_\lambda
    =
    |s_\lambda\rangle
    \langle s_\lambda|.
\)

The same calculation gives
\begin{equation}
\label{eq:two-block-condition-B-mu}
    \operatorname{Tr}_{V_\mu}
    \left(
        \iota_\mu^\dagger
        |\psi\rangle\langle\phi|
        \iota_\mu
    \right)
    =
    (1-c)
    \langle\phi|\psi\rangle
    |s_\mu\rangle
    \langle s_\mu|.
\end{equation}
Thus
\(
    \omega_\mu
    =
    |s_\mu\rangle
    \langle s_\mu|
\).
For every inactive
\(
    \nu\notin\{\lambda,\mu\}
\),
\(
    \iota_\nu^\dagger|\psi\rangle=0,
\)
and hence
\(
    \operatorname{Tr}_{V_\nu}
    \left(
        \iota_\nu^\dagger
        |\psi\rangle\langle\phi|
        \iota_\nu
    \right)
    =
    0.
\)
Therefore both conditions of Theorem~\ref{thm:qpb-classification} are satisfied.

\medskip
By Theorem~\ref{thm:qpb-classification}, the subspace
\(
    \mathcal S
\)
is Haar-twirl admissible, and its replacement state is
the state in the Equation~\eqref{eq:two-block-replacement-state}.

For arbitrary
\(
    |\psi\rangle,|\phi\rangle\in\mathcal S
\), substituting \eqref{eq:two-block-condition-B-lambda} and
\eqref{eq:two-block-condition-B-mu} into the Haar-twirl Equation~\ref{eq:appendix-twirl-direct-sum-form} gives, 
\begin{align}
    \mathcal T^{(t)}
    \bigl(
        |\psi\rangle\langle\phi|
    \bigr)
    &=
    \langle\phi|\psi\rangle\,\sigma.
\label{eq:two-block-rank-one-twirl}
\end{align}
In particular, if
\[
    \rho
    =
    \sum_{\ell=1}^{\operatorname{rank}(\rho)}
    p_\ell
    |x_\ell\rangle\langle x_\ell|
    \in\mathcal D(\mathcal S),
\]
where
\[
    p_\ell\geq0,
    \qquad
    \sum_{\ell}p_\ell=1,
    \qquad
    \|x_\ell\|=1,
\]
then
\begin{align}
    \mathcal T^{(t)}(\rho)
    &=
    \sum_\ell
    p_\ell
    \mathcal T^{(t)}
    \bigl(
        |x_\ell\rangle\langle x_\ell|
    \bigr)
    \notag\\
    &=
    \sum_\ell
    p_\ell
    \langle x_\ell|x_\ell\rangle
    \sigma
    \notag\\
    &=
    \left(
        \sum_\ell p_\ell
    \right)\sigma
    \notag\\
    &=
    \sigma.
\end{align}

\medskip
Let
\(
    |\psi\rangle\in\mathcal S
\).
Using~\eqref{eq:two-block-projection-psi},
orthonormality of
\(
    \{v_j\}_{j=1}^n
\),
and~\eqref{eq:two-block-inner-product},
we obtain
\begin{align}
    \|\Pi_\lambda|\psi\rangle\|^2
    &=
    c
    \left\|
        \sum_{j=1}^n
        \alpha_j(v_j\otimes s_\lambda)
    \right\|^2
    \notag\\
    &=
    c
    \sum_{j=1}^n|\alpha_j|^2
    \notag\\
    &=
    c\|\ket{\psi}\|^2,
\end{align}
and
\[
    \|\Pi_\mu|\psi\rangle\|^2
    =
    (1-c)\|\ket{\psi}\|^2.
\]
If
\(
    |\psi\rangle\neq0
\),
then
\(
    \|\ket{\psi}\|^2>0
\),
and since
\(
    0<c<1
\),
both quantities are strictly positive. This proves
\eqref{eq:two-block-norm-identities} and completes the proof.
\end{proof}

\subsection{Complementary-channel and Knill-Laflamme interpretation}\label{subsec:classification-knill-laflamme}

The conditions in Theorem~\ref{thm:qpb-classification}
have a direct quantum-error-correcting interpretation.
The relevant error-correction channel is not the collective Haar-twirl itself, but a \emph{complementary channel} of the Haar twirl.
We make this relation explicit in the present subsection.

This connection is a specialization of the general complementarity
between private and correctable quantum codes
\cite{KretschmannKribsSpekkens2008,
JochymOConnorKribsLaflammePlosker2013,
JochymOConnorKribsLaflammePlosker2014}.
For the collective Haar twirl, Schur-Weyl duality makes the corresponding Knill-Laflamme operator equations completely explicit and yields the two conditions of Theorem~\ref{thm:qpb-classification}.

\medskip
Let
\(
    \Phi \colon 
    \mathcal L(\mathcal H_A)
    \longrightarrow
    \mathcal L(\mathcal H_B)
\)
be a quantum channel.  Choose a Stinespring isometry
\(
    V\colon
    \mathcal H_A
    \longrightarrow
    \mathcal H_B\otimes\mathcal H_E
\)
such that
\begin{equation}
\label{eq:stinespring-main}
    \Phi(X)
    =
    \operatorname{Tr}_{E}
    \bigl(
        VXV^\dagger
    \bigr).
\end{equation}
The corresponding complementary channel is
\begin{equation}
\label{eq:complementary-channel-main}
    \Phi^{c}(X)
    :=
    \operatorname{Tr}_{B}
    \bigl(
        VXV^\dagger
    \bigr).
\end{equation}

Fix an orthonormal basis
\(
    \{|a\rangle\}_{a=1}^{\dim_{\mathbb{C}}\mathcal H_B}
\)
of \(\mathcal H_B\) and define
\begin{equation}
\label{eq:complement-kraus}
    F_a
    :=
    (\langle a|\otimes I_E)V
    :
    \mathcal H_A
    \longrightarrow
    \mathcal H_E.
\end{equation}
Then
\begin{equation}
\label{eq:complement-kraus-rep}
    \Phi^c(X)
    =
    \sum_a
        F_aXF_a^\dagger.
\end{equation}
Moreover,
\begin{equation}
\label{eq:FstarF-adjoint}
    F_a^\dagger F_b
    =
    V^\dagger
    \bigl(
        |a\rangle\langle b|
        \otimes I_E
    \bigr)
    V
    =
    \Phi^\dagger
    \bigl(
        |a\rangle\langle b|
    \bigr).
\end{equation}
Since the matrix units
\(
    |a\rangle\langle b|
\)
span
\(
    \mathcal L(\mathcal H_B)
\),
Equation~\eqref{eq:FstarF-adjoint} implies
\begin{equation}
\label{eq:KL-system-range}
    \operatorname{span}_{\mathbb C}
    \{
        F_a^\dagger F_b
    \}_{a,b}
    =
    \operatorname{im}\Phi^\dagger.
\end{equation}

We first derive the exact subspace version of privacy-correctability complementarity.

\begin{theorem}[Private subspaces and complementary error correction]\label{thm:privacy-complementarity}
Let
\(
    \mathcal S\subseteq \mathcal H_A
\)
be a nonzero subspace and let
\(
    \Pi_{\mathcal S}
\)
denote its orthogonal projector.
The following are equivalent.

\begin{enumerate}[label=\textup{(\roman*)}]
\item
There exists
\(
    \sigma_B\in\mathcal D(\mathcal H_B)
\)
such that
\begin{equation}
\label{eq:private-channel-condition-KL}
    \Phi\circ\mathcal J_{\mathcal S}
  =\mathcal R_{\sigma_B,\mathcal S};
    \qquad
    \forall
    X\in\mathcal L(\mathcal S).
\end{equation}

\item
The subspace \(\mathcal S\) satisfies the Knill-Laflamme
conditions for the complementary channel \(\Phi^c\):
\begin{equation}
\label{eq:KL-complement-condition}
    \Pi_{\mathcal S}
    F_a^\dagger F_b
    \Pi_{\mathcal S}
    =
    \lambda_{ab}
    \Pi_{\mathcal S}
    \qquad
    \forall a,b,
\end{equation}
where
\(
\label{eq:lambda-sigma-relation}
    \lambda_{ab}
    =
    \langle b|\sigma_B|a\rangle.
\)
\end{enumerate}

Consequently,
\[
    \mathcal S\text{ is private for }\Phi
    \iff
    \mathcal S\text{ is exactly correctable for }\Phi^c.
\]
\end{theorem}

\begin{proof}
Assume first that
\eqref{eq:private-channel-condition-KL}
holds.
Let
\(
    |\psi\rangle,|\phi\rangle\in\mathcal S
\).
For arbitrary \(a,b\), using
\eqref{eq:FstarF-adjoint},
\begin{align}
    \langle\phi|
        F_a^\dagger F_b
    |\psi\rangle
    &=
    \langle\phi|
        \Phi^\dagger
        (|a\rangle\langle b|)
    |\psi\rangle
    \notag\\
    &=
    \operatorname{Tr}
    \left[
        |a\rangle\langle b|\,
        \Phi
        \bigl(
            |\psi\rangle\langle\phi|
        \bigr)
    \right]
    \notag\\
    &=
    \operatorname{Tr}
    \left[
        |a\rangle\langle b|\,
        \langle\phi|\psi\rangle
        \sigma_B
    \right]
    \notag\\
    &=
    \langle\phi|\psi\rangle
    \operatorname{Tr}
    \left(
        |a\rangle\langle b|
        \sigma_B
    \right)
    \notag\\
    &=
    \langle b|\sigma_B|a\rangle
    \langle\phi|\psi\rangle.
\end{align}
Since this identity holds for every
\(
    |\psi\rangle,|\phi\rangle\in\mathcal S
\),
we obtain the operator identity
\[
    \Pi_{\mathcal S}
    F_a^\dagger F_b
    \Pi_{\mathcal S}
    =
    \langle b|\sigma_B|a\rangle
    \Pi_{\mathcal S},
\]
which is
\eqref{eq:KL-complement-condition}.

Conversely, suppose that
\eqref{eq:KL-complement-condition}
holds.
For arbitrary
\(
    |\psi\rangle,|\phi\rangle\in\mathcal S
\),
we have
\begin{align}
    \langle a|
        \Phi
        \bigl(
            |\psi\rangle\langle\phi|
        \bigr)
    |b\rangle
    &=
    \operatorname{Tr}
    \left[
        |b\rangle\langle a|\,
        \Phi
        \bigl(
            |\psi\rangle\langle\phi|
        \bigr)
    \right]
    \notag\\
    &=
    \langle\phi|
        \Phi^\dagger
        (|b\rangle\langle a|)
    |\psi\rangle
    \notag\\
    &=
    \langle\phi|
        F_b^\dagger F_a
    |\psi\rangle
    \notag\\
    &=
    \lambda_{ba}
    \langle\phi|\psi\rangle.
\end{align}
Choose a unit vector
\(
    |\psi_0\rangle\in\mathcal S
\)
and put
\[
    \sigma_B
    :=
    \Phi
    \bigl(
        |\psi_0\rangle\langle\psi_0|
    \bigr).
\]
Since \(\Phi\) is CPTP,
\(
    \sigma_B\in\mathcal D(\mathcal H_B)
\).
The preceding identity with
\(
    \psi=\phi=\psi_0
\)
gives
\[
    \langle a|\sigma_B|b\rangle
    =
    \lambda_{ba}.
\]
Hence, for arbitrary
\(
    |\psi\rangle,|\phi\rangle\in\mathcal S
\),
\[
    \langle a|
        \Phi
        (|\psi\rangle\langle\phi|)
    |b\rangle
    =
    \langle\phi|\psi\rangle
    \langle a|\sigma_B|b\rangle.
\]
Since this holds for every pair \(a,b\),
\begin{equation}
\label{eq:rank-one-private-from-KL}
    \Phi
    (|\psi\rangle\langle\phi|)
    =
    \langle\phi|\psi\rangle\sigma_B
    =
    \operatorname{Tr}
    (|\psi\rangle\langle\phi|)
    \sigma_B.
\end{equation}
Rank-one operators
\(
    |\psi\rangle\langle\phi|
\),
with
\(
    |\psi\rangle,|\phi\rangle\in\mathcal S
\),
span
\(
    \mathcal L(\mathcal S)
\).
Therefore linearity of \(\Phi\) extends
\eqref{eq:rank-one-private-from-KL} to every
\(
    X\in\mathcal L(\mathcal S)
\),
which proves
\eqref{eq:private-channel-condition-KL}.

The final equivalence follows from the Knill-Laflamme theorem applied
to the error operators
\(
    \{F_a\}_a
\)
of the complementary channel.
\end{proof}

\noindent \(\bullet\) \textbf{Specialization to the collective Haar-twirl.} We now apply Theorem~\ref{thm:privacy-complementarity} to
\(
    \Phi
    =
    \mathcal T^{(t)}.
\)

The collective Haar twirl is self-adjoint with respect to the
Hilbert--Schmidt inner product:
\begin{equation}
\label{eq:twirl-self-adjoint}
    \bigl(
        \mathcal T^{(t)}
    \bigr)^\dagger
    =
    \mathcal T^{(t)}.
\end{equation}
For
\(
    X,Y\in\mathcal L(\mathcal H_d^{\otimes t})
\),
\begin{align}
    \operatorname{Tr}
    \left[
        Y^\dagger
        \mathcal T^{(t)}(X)
    \right]
    &=
    \int_{\mathrm U(d)}
    \operatorname{Tr}
    \left[
        Y^\dagger
        U^{\otimes t}
        X
        (U^\dagger)^{\otimes t}
    \right]
    \mathrm d\mu_{\mathrm{Haar}}(U)
    \notag\\
    &=
    \int_{\mathrm U(d)}
    \operatorname{Tr}
    \left[
        \bigl(
            U^{\otimes t}
            Y
            (U^\dagger)^{\otimes t}
        \bigr)^\dagger
        X
    \right]
    \mathrm d\mu_{\mathrm{Haar}}(U)
    \notag\\
    &=
    \operatorname{Tr}
    \left[
        \mathcal T^{(t)}(Y)^\dagger
        X
    \right].
\end{align}
Together with idempotence,
\(
    \mathcal T^{(t)}
    \circ
    \mathcal T^{(t)}
    =
    \mathcal T^{(t)},
\)
this shows that
\(
    \mathcal T^{(t)}
\)
is the Hilbert--Schmidt orthogonal projection onto its range.

By Schur-Weyl duality, this range is the commutant of the collective
\(\mathrm U(d)\)-representation:
\begin{equation}
\label{eq:twirl-range-algebra}
\begin{split}
    \mathfrak A_t
    &:=
    \operatorname{im}
    \bigl(
        \mathcal T^{(t)}
    \bigr)
\\
    &=
    \operatorname{End}_{\mathrm U(d)}
    \left(
        \mathcal H_d^{\otimes t}
    \right)
\\
    &=
    \bigoplus_{\lambda\in\Lambda(d,t)}
    \iota_\lambda
    \left(
        I_{V_\lambda}
        \otimes
        \mathcal L(W_\lambda)
    \right)
    \iota_\lambda^\dagger.
\end{split}
\end{equation}

\begin{proposition}[Knill-Laflamme form of the Schur-Weyl classification]
\label{prop:classification-KL}
Let
\(
    \mathcal S\subset \mathcal H_d^{\otimes t}
\)
be nonzero, with orthogonal projector
\(
    \Pi_{\mathcal S}
\).
The following are equivalent.

\begin{enumerate}[label=\textup{(\roman*)}]
\item
\(\mathcal S\) is Haar-twirl admissible for the collective Haar twirl.

\item
\(\mathcal S\) is an exactly correctable subspace for a complementary
channel of
\(
    \mathcal T^{(t)}
\).

\item
There exists a linear functional
\(
    \alpha:\mathfrak A_t \longrightarrow \mathbb C
\)
such that
\begin{equation}
\label{eq:scalar-compression-commutant}
    \Pi_{\mathcal S}A\Pi_{\mathcal S}
    =
    \alpha(A)\Pi_{\mathcal S}
    \qquad
    \forall A\in\mathfrak A_t.
\end{equation}

\item
There exist constants
\(
    c_\lambda\geq0
\)
with
\(
    \sum_{\lambda\in\Lambda(d,t)}
        c_\lambda=1
\)
and, for every active block \(c_\lambda>0\), a density operator
\(
    \omega_\lambda\in\mathcal D(W_\lambda)
\)
such that
\begin{align}
    \langle\phi|\Pi_\lambda|\psi\rangle
    &=
    c_\lambda
    \langle\phi|\psi\rangle,
\label{eq:KL-condition1-recovered}
\\
    \operatorname{Tr}_{V_\lambda}
    \left(
        \iota_\lambda^\dagger
        |\psi\rangle\langle\phi|
        \iota_\lambda
    \right)
    &=
    c_\lambda
    \langle\phi|\psi\rangle
    \omega_\lambda
\label{eq:KL-condition2-recovered}
\end{align}
for every
\(
    |\psi\rangle,|\phi\rangle\in\mathcal S
\).
\end{enumerate}

Thus
Theorem~\ref{thm:qpb-classification}
is the Schur-Weyl resolution of the Knill-Laflamme operator equations for a complementary channel of the collective Haar twirl.
\end{proposition}

\begin{proof}
The equivalence between (1) and (2) is Theorem~\ref{thm:privacy-complementarity}.

To prove
\(
    \textup{(2)}
    \Longleftrightarrow
    \textup{(3)}
\),
let
\(
    \{F_a\}_a
\)
be the Kraus operators of a complementary channel constructed as
above. By \eqref{eq:complement-kraus-rep},
\begin{align}
    \operatorname{span}_{\mathbb{C}}
    \{
        F_a^\dagger F_b
    \}_{a,b}
    &=
    \operatorname{im}
    \left[
        \bigl(
            \mathcal T^{(t)}
        \bigr)^\dagger
    \right]
    \notag\\
    &=
    \operatorname{im}\mathcal T^{(t)}
    \notag\\
    &=
    \mathfrak A_t,
\end{align}
where the first equality follows from \eqref{eq:KL-system-range}, and the second follows from \eqref{eq:twirl-self-adjoint}.
The Knill-Laflamme equations
\[
    \Pi_{\mathcal S}
    F_a^\dagger F_b
    \Pi_{\mathcal S}
    =
    \lambda_{ab} \Pi_{\mathcal S}
\]
for all \(a,b\) are therefore equivalent, by linearity, to scalar compression of every element of \(\mathfrak A_t\), which is \eqref{eq:scalar-compression-commutant}.

Now we derive \eqref{eq:scalar-compression-commutant}
with the two conditions of Theorem~\ref{thm:qpb-classification}.
For
\(
    \lambda\in\Lambda(d,t)
\)
and
\(
    B\in\mathcal L(W_\lambda)
\),
define
\begin{equation}
\label{eq:A-lambda-B}
    A_{\lambda,B}
    :=
    \iota_\lambda
    \bigl(
        I_{V_\lambda}\otimes B
    \bigr)
    \iota_\lambda^\dagger
    \in
    \mathfrak A_t.
\end{equation}
The family of all
\(
    A_{\lambda,B}
\)
spans \(\mathfrak A_t\).

For arbitrary
\(
    |\psi\rangle,|\phi\rangle\in\mathcal S
\),
\begin{align}
    \langle\phi|
        A_{\lambda,B}
    |\psi\rangle
    &=
    \operatorname{Tr}
    \left[
        A_{\lambda,B}
        |\psi\rangle\langle\phi|
    \right]
    \notag\\
    &=
    \operatorname{Tr}_{V_\lambda\otimes W_\lambda}
    \left[
        (I_{V_\lambda}\otimes B)
        \iota_\lambda^\dagger
        |\psi\rangle\langle\phi|
        \iota_\lambda
    \right]
    \notag\\
    &=
    \operatorname{Tr}_{W_\lambda}
    \left[
        B\,
        \operatorname{Tr}_{V_\lambda}
        \left(
            \iota_\lambda^\dagger
            |\psi\rangle\langle\phi|
            \iota_\lambda
        \right)
    \right].
\label{eq:commutant-pairing-partial-trace}
\end{align}

Suppose first that the conditions of
Theorem~\ref{thm:qpb-classification}
hold.  Substituting
\eqref{eq:KL-condition2-recovered}
into
\eqref{eq:commutant-pairing-partial-trace}
gives
\begin{align}
    \langle\phi|
        A_{\lambda,B}
    |\psi\rangle
    &=
    c_\lambda
    \langle\phi|\psi\rangle
    \operatorname{Tr}
    (B\omega_\lambda).
\end{align}
Hence
\begin{equation}
\label{eq:compression-A-lambda-B}
    \Pi_{\mathcal S}
    A_{\lambda,B}
    \Pi_{\mathcal S}
    =
    c_\lambda
    \operatorname{Tr}
    (B\omega_\lambda)
    \Pi_{\mathcal S}.
\end{equation}
Since the operators
\(
    A_{\lambda,B}
\)
span \(\mathfrak A_t\),
Condition~\textup{(3)} follows.

Conversely, suppose that
\eqref{eq:scalar-compression-commutant}
holds.
Taking
\(
    B=I_{W_\lambda}
\)
in
\eqref{eq:A-lambda-B}
gives
\(
    A_{\lambda,I_{W_\lambda}}
    =
    \Pi_\lambda.
\)
Therefore there exists a scalar
\(
    c_\lambda
\)
such that
\(
    \Pi_{\mathcal S}
    \Pi_\lambda
    \Pi_{\mathcal S}
    =
    c_\lambda \Pi_{\mathcal S}.
\)
Since
\(
    \Pi_\lambda\succeq0
\),
we have
\(
    c_\lambda\geq0
\).
Moreover,
\(
    \sum_{\lambda}
        \Pi_\lambda
    =
    I,
\)
and hence
\begin{align}
    \Pi_{\mathcal S}
    &=
    \Pi_{\mathcal S}
    \left(
        \sum_\lambda\Pi_\lambda
    \right)
    \Pi_{\mathcal S}
    \notag\\
    &=
    \left(
        \sum_\lambda c_\lambda
    \right)
    \Pi_{\mathcal S}.
\end{align}
Since
\(
    \Pi_{\mathcal S}\neq0
\),
\(
    \sum_\lambda c_\lambda=1.
\)
Thus
\eqref{eq:KL-condition1-recovered}
holds.

\smallskip
Fix now an active
\(
    \lambda
\)
with
\(
    c_\lambda>0
\).
For every
\(
    B\in\mathcal L(W_\lambda)
\),
scalar compression gives a scalar
\(
    \beta_\lambda(B)
\)
such that
\[
    \langle\phi|
        A_{\lambda,B}
    |\psi\rangle
    =
    \beta_\lambda(B)
    \langle\phi|\psi\rangle.
\]
By
\eqref{eq:commutant-pairing-partial-trace},
\begin{align}
    &
    \operatorname{Tr}
    \left[
        B\,
        \operatorname{Tr}_{V_\lambda}
        \left(
            \iota_\lambda^\dagger
            |\psi\rangle\langle\phi|
            \iota_\lambda
        \right)
    \right]
    \notag\\
    &\qquad=
    \beta_\lambda(B)
    \langle\phi|\psi\rangle.
\label{eq:functional-beta}
\end{align}
The functional
\(
    B\mapsto\beta_\lambda(B)
\)
is positive and satisfies
\(
    \beta_\lambda(I_{W_\lambda})=c_\lambda.
\)
By the finite-dimensional Riesz representation theorem for the trace
pairing, there exists a unique positive operator
\(
    \tau_\lambda\in\mathcal L(W_\lambda)
\)
such that
\[
    \beta_\lambda(B)
    =
    \operatorname{Tr}
    (B\tau_\lambda)
    \qquad
    \forall B\in\mathcal L(W_\lambda).
\]
Furthermore,
\(
    \operatorname{Tr}(\tau_\lambda)
    =
    c_\lambda.
\)

\smallskip
Define
\(
    \omega_\lambda
    :=
    \frac{\tau_\lambda}{c_\lambda}
    \in
    \mathcal D(W_\lambda).
\)
Then
\eqref{eq:functional-beta}
becomes
\[
    \operatorname{Tr}
    \left[
        B\,
        \operatorname{Tr}_{V_\lambda}
        \left(
            \iota_\lambda^\dagger
            |\psi\rangle\langle\phi|
            \iota_\lambda
        \right)
    \right]
    =
    c_\lambda
    \langle\phi|\psi\rangle
    \operatorname{Tr}(B\omega_\lambda)
\]
for every
\(
    B\in\mathcal L(W_\lambda)
\).
Nondegeneracy of the trace pairing therefore implies
\[
    \operatorname{Tr}_{V_\lambda}
    \left(
        \iota_\lambda^\dagger
        |\psi\rangle\langle\phi|
        \iota_\lambda
    \right)
    =
    c_\lambda
    \langle\phi|\psi\rangle
    \omega_\lambda.
\]
This is
\eqref{eq:KL-condition2-recovered} and completes the proof.
\end{proof}

\begin{remark}
\label{rem:meaning-KL-classification}
The preceding proposition clarifies the information-theoretic meaning
of the two conditions in
Theorem~\ref{thm:qpb-classification}, and
\[
    \mathfrak A_t
    =
    \bigoplus_{\lambda}
        I_{V_\lambda}
        \otimes
        \mathcal L(W_\lambda)
\]
is the image after taking the collective Haar twirl. The Haar-twirl admissibility requires that every observable in this surviving algebra have scalar compression to the plaintext subspace.

The operators
\(
    \Pi_\lambda
\)
measure the Schur-Weyl block label.  Their scalar-compression
conditions give the state-independent probabilities
\(
    c_\lambda
\).
Within each active block, the operators
\(
    I_{V_\lambda}\otimes B
\)
probe the multiplicity system \(W_\lambda\). Their
scalar-compression conditions give the state-independent density
operator
\(
    \omega_\lambda
\).
Thus the two conditions of Theorem~\ref{thm:qpb-classification} are same as the blockwise Knill-Laflamme equations for the complementary channel.
\end{remark}

\subsection{Hereditary reduction of Haar-twirl admissible subspaces}\label{subsec:hereditary-qpb}

The classification theorem of the preceding subsection characterizes those
subspaces
\(
    \mathcal S\subset \mathcal H_d^{\otimes t}
\)
on which the collective Haar twirl acts as a state-replacement channel.
A natural question is whether this property is preserved when one tensor
factor, corresponding to one recipient, is removed.

There are two mathematically distinct issues involved in this question.
First, one has to define canonically what is meant by the reduction of a
linear subspace
\(
    \mathcal S\subset \mathcal H_d^{\otimes s}
\)
to \(s-1\) tensor factors.  The partial trace cannot be applied directly to a
subspace, because
\(
    \operatorname{Tr}_j
\)
is a linear map on operators rather than on vectors or subspaces. The appropriate reduced subspace is therefore obtained from the support of a partially traced orthogonal projection, equivalently from all one-factor contractions of vectors in \(\mathcal S\).

Second, two different hereditary questions arise.  One may ask whether the reduced subspace itself remains Haar-twirl admissible for the collective Haar twirl. Independently, one may ask whether every exact unitary design that reproduces the Haar twirl on the original subspace automatically reproduces the Haar twirl on the reduced subspace. We shall call these properties
\emph{privacy heredity} and \emph{design heredity}, respectively.

For the purposes of \textsf{QPB}, neither property by itself is sufficient. A dimensionally reduced \textsf{QPB} scheme constructed with the same key ensemble
requires both the Haar-twirl admissibility of the reduced plaintext space and the descent of the corresponding restricted unitary design. Accordingly, we shall reserve the unqualified term \emph{\textsf{QPB}-hereditary} for the simultaneous validity of both properties.

\medskip
Fix an integer \(s\) with \(2\leq s\leq t\). Let \(\mathcal H_d^{\otimes s}\) be the Hilbert space. Define the collective Haar twirl at level \(s\) by
\begin{equation}
\label{eq:hereditary-Haar-twirl}
    \mathcal T^{(s)}(X)
    :=
    \int_{\mathrm U(d)}
       U^{\otimes s}
       X
       (U^\dagger)^{\otimes s}
       \,\mathrm d\mu_{\mathrm{Haar}}(U),
    \qquad
    X\in\mathcal L(\mathcal H_d^{\otimes s}).
\end{equation}
Here
\(
    \mu_{\mathrm{Haar}}
\)
denotes normalized Haar probability measure on \(\mathrm U(d)\).

\subsubsection{The canonical reduction of a subspace}\label{subsubsec:canonical-reduction}

Let
\(
    \mathcal S \subset \mathcal H_d^{\otimes s}
\)
be a nonzero linear subspace, and let
\(
    \Pi_{\mathcal S}
\)
denote the orthogonal projection onto \(\mathcal S\).

\begin{definition}[One-factor contraction]
\label{def:one-factor-contraction}
Fix
\(
    j\in\{1,\ldots,s\}
\)
and
\(
    |v\rangle\in\mathcal H_d
\).
The contraction in the \(j\)-th tensor factor with respect to
\(|v\rangle\) is the linear map
\begin{equation}
\label{eq:contraction-map}
    C_{j,v}^{(s)}
    :=
    I_{\mathcal H_d}^{\otimes(j-1)}
    \otimes
    \langle v|
    \otimes
    I_{\mathcal H_d}^{\otimes(s-j)}
    \colon
    \mathcal H_d^{\otimes s}
    \longrightarrow
    \mathcal H_d^{\otimes(s-1)}.
\end{equation}
Thus, for a decomposable tensor
\(
    |\psi\rangle
    =
    |x_1\rangle\otimes\cdots\otimes|x_s\rangle,
\)
we have
\begin{equation}
\label{eq:contraction-decomposable}
\begin{split}
    C_{j,v}^{(s)}|\psi\rangle
    &=
    \langle v|x_j\rangle
    |x_1\rangle\otimes\cdots
    \otimes|x_{j-1}\rangle
    \otimes|x_{j+1}\rangle
    \otimes\cdots\otimes|x_s\rangle .
\end{split}
\end{equation}
\end{definition}

\begin{definition}[Canonical reduced subspace]
\label{def:canonical-reduced-subspace}
The reduction of
\(
    \mathcal S \subset \mathcal H_d^{\otimes s}
\)
in the \(j\)-th tensor coordinate is
\begin{equation}
\label{eq:Red-definition}
    \operatorname{Red}_j(\mathcal S)
    :=
    \operatorname{span}_{\mathbb C}
    \left\{
       C_{j,v}^{(s)}|\psi\rangle:
       |v\rangle\in\mathcal H_d,\,
       |\psi\rangle\in\mathcal S
    \right\}
    \subseteq
    \mathcal H_d^{\otimes(s-1)}.
\end{equation}
\end{definition}

Since all Hilbert spaces considered here are finite-dimensional, the span in~\eqref{eq:Red-definition} is automatically closed.

The following result shows that this vector-space definition is equivalent to the operator-theoretic definition obtained from partial
trace.

\begin{lemma}[Support characterization of the reduced subspace]
\label{lem:Red-support}
For every nonzero
\(
    \mathcal S \subset \mathcal H_d^{\otimes s}
\)
and every
\(
    j\in\{1,\ldots,s\}
\),
\begin{equation}
\label{eq:Red-support}
    \operatorname{Red}_j(\mathcal S)
    =
    \operatorname{supp}
    \left(
       \operatorname{Tr}_j \Pi_{\mathcal S}
    \right).
\end{equation}
In particular,
\(
    \operatorname{Red}_j(\mathcal S)
\)
is independent of every choice of basis of
\(\mathcal S\) and of \(\mathcal H_d\).
\end{lemma}

\begin{proof}
It is sufficient to prove the statement for \(j=1\), since the remaining
cases are obtained by relabelling the tensor coordinates.

Choose an orthonormal basis
\(
    \{|\psi_\alpha\rangle\}_{\alpha=1}^{m}
\)
of \(\mathcal S\), where
\(
    m=\dim_{\mathbb C}\mathcal S
\),
and an orthonormal basis
\(
    \{|e_a\rangle\}_{a=1}^{d}
\)
of \(\mathcal H_d\).
Then
\begin{equation}
\label{eq:PS-basis-expansion}
    \Pi_{\mathcal S}
    =
    \sum_{\alpha=1}^{m}
       |\psi_\alpha\rangle\langle\psi_\alpha|.
\end{equation}
Applying the definition of the partial trace gives
\begin{align}
    \operatorname{Tr}_1(\Pi_{\mathcal S})
    &=
    \sum_{\alpha=1}^{m}
    \sum_{a=1}^{d}
       \bigl(
          \langle e_a|\otimes
          I_{\mathcal H_d^{\otimes(s-1)}}
       \bigr)
       |\psi_\alpha\rangle\langle\psi_\alpha|
       \bigl(
          |e_a\rangle\otimes
          I_{\mathcal H_d^{\otimes(s-1)}}
       \bigr)
       \notag\\
    &=
    \sum_{\alpha=1}^{m}
    \sum_{a=1}^{d}
       |C_{1,e_a}^{(s)}\psi_\alpha\rangle
       \langle C_{1,e_a}^{(s)}\psi_\alpha|.
\label{eq:partial-trace-PS}
\end{align}

For an arbitrary finite collection of vectors
\(
    |z_1\rangle,\ldots,|z_N\rangle
\)
and the positive semidefinite operator
\(
    A
    :=
    \sum_{r=1}^{N}
       |z_r\rangle\langle z_r|,
\)
we have
\begin{align}
    \langle x|A|x\rangle
    &=
    \sum_{r=1}^{N}
       |\langle z_r|x\rangle|^2.
\end{align}
Therefore
\(
    A|x\rangle=0
\) if and only if \(\langle z_r|x\rangle=0\) for every \(r\), and hence
\begin{equation}
\label{eq:support-sum-rank-one}
    \operatorname{supp}(A)
    =
    (\ker A)^\perp
    =
    \operatorname{span}
    \{|z_1\rangle,\ldots,|z_N\rangle\}.
\end{equation}
Applying~\eqref{eq:support-sum-rank-one} to
\eqref{eq:partial-trace-PS}, we obtain
\[
    \operatorname{supp}
    \left(
       \operatorname{Tr}_{1} \Pi_{\mathcal S}
    \right)
    =
    \operatorname{span}
    \left\{
       C_{1,e_a}^{(s)}|\psi_\alpha\rangle:
       1\leq a\leq d,\,
       1\leq\alpha\leq m
    \right\}.
\]
Because every
\(
    |v\rangle\in\mathcal H_d
\)
is a linear combination of the
\(
    |e_a\rangle
\)'s and every
\(
    |\psi\rangle\in\mathcal S
\)
is a linear combination of the
\(
    |\psi_\alpha\rangle
\)'s, the last span is
\(
    \operatorname{Red}_1(\mathcal S)
\).
This proves~\eqref{eq:Red-support}.
\end{proof}

\begin{remark}\label{rem:fibre-not-reduction}
For a fixed nonzero
\(
    |v\rangle\in\mathcal H_d
\),
the set
\[
    \left\{
       |\eta\rangle\in\mathcal H_d^{\otimes(s-1)}:
       |v\rangle\otimes|\eta\rangle\in\mathcal S
    \right\}
\]
is the product-vector fibre of \(\mathcal S\) over
\(|v\rangle\).  It is not the partial-trace reduction of
\(\mathcal S\).

For example, let \(d\geq2\), \(s=2\), and
\[
    |\Omega\rangle
    :=
    \frac{1}{\sqrt2}
    \bigl(
       |e_1\rangle\otimes|e_1\rangle
       +
       |e_2\rangle\otimes|e_2\rangle
    \bigr),
    \qquad
    \mathcal S
    :=
    \operatorname{span}\{|\Omega\rangle\}.
\]
The vector \(|\Omega\rangle\) has Schmidt rank \(2\), and hence no
nonzero vector of \(\mathcal S\) is a product vector.  Thus every
product-vector fibre is zero.  On the other hand,
\(
    \Pi_{\mathcal S}
    =
    |\Omega\rangle\langle\Omega|,
\)
and
\begin{align}
    \operatorname{Tr}_1(\Pi_{\mathcal S})
    &=
    \frac12
    \left(
       |e_1\rangle\langle e_1|
       +
       |e_2\rangle\langle e_2|
    \right).
\end{align}
Consequently,
\[
    \operatorname{Red}_1(\mathcal S)
    =
    \operatorname{span}
    \{|e_1\rangle,|e_2\rangle\}
    \neq\{0\}.
\]
\end{remark}

\subsubsection{Compatibility of Haar twirling with partial trace}
\label{subsubsec:Haar-partial-trace}

The partial trace is covariant with respect to the collective unitary
action.

\begin{lemma}[Partial-trace covariance]
\label{lem:partial-trace-covariance}
For every
\(
    X\in\mathcal L(\mathcal H_d^{\otimes s})
\),
every
\(
    U\in\mathrm U(d)
\),
and every
\(
    j\in\{1,\ldots,s\}
\),
\begin{equation}
\label{eq:partial-trace-unitary-covariance}
\begin{split}
    &\operatorname{Tr}_j
    \left[
       U^{\otimes s}
       X
       (U^\dagger)^{\otimes s}
    \right]
    =
    U^{\otimes(s-1)}
    \operatorname{Tr}_j(X)
    (U^\dagger)^{\otimes(s-1)}.
\end{split}
\end{equation}
Consequently,
\begin{equation}
\label{eq:Haar-partial-trace-intertwining}
    \operatorname{Tr}_j
    \circ
    \mathcal T^{(s)}
    =
    \mathcal T^{(s-1)}
    \circ
    \operatorname{Tr}_j.
\end{equation}
\end{lemma}

\begin{proof}
It is enough to consider \(j=1\). Choose an orthonormal basis
\(
    \{|e_a\rangle\}_{a=1}^{d}
\)
of \(\mathcal H_d\).  

Then
\begin{align}
    \operatorname{Tr}_1
    \left[
       U^{\otimes s}
       X
       (U^\dagger)^{\otimes s}
    \right]
    &=
    \sum_{a=1}^{d}
       (\langle e_a|\otimes I)
       (U\otimes U^{\otimes(s-1)})
       X
       (U^\dagger\otimes(U^\dagger)^{\otimes(s-1)})
       (|e_a\rangle\otimes I)
    \notag\\
    &=
    U^{\otimes(s-1)}
    \left[
       \sum_{a=1}^{d}
       (\langle U^\dagger e_a|\otimes I)
       X
       (|U^\dagger e_a\rangle\otimes I)
    \right]
    (U^\dagger)^{\otimes(s-1)}.
\end{align}
Since
\(
    \{|U^\dagger e_a\rangle\}_{a=1}^{d}
\)
is again an orthonormal basis of \(\mathcal H_d\), the expression in square brackets is
\(
    \operatorname{Tr}_1(X)
\).
This proves~\eqref{eq:partial-trace-unitary-covariance}.
Integrating both sides with respect to normalized Haar measure proves \eqref{eq:Haar-partial-trace-intertwining}.
\end{proof}


The preceding lemma gives an automatic, but weaker, form of descent.

\begin{proposition}\label{prop:automatic-marginal-descent}
Suppose that
\(
    \mathcal S \subset \mathcal H_d^{\otimes s}
\)
is Haar-twirl admissible, with replacement state
\(
    \sigma_{\mathcal S}
\).
Then for every
\(
    \rho\in\mathcal D(\mathcal S)
\),
\begin{equation}
\label{eq:automatic-marginal-descent}
    \mathcal T^{(s-1)}
    \left(
       \operatorname{Tr}_j\rho
    \right)
    =
    \operatorname{Tr}_j
    \left(
       \sigma_{\mathcal S}
    \right).
\end{equation}
\end{proposition}

\begin{proof}
Since \(\mathcal S\) is Haar-twirl admissible,
\(
    \mathcal T^{(s)}(\rho)
    =
    \sigma_{\mathcal S}
\). Applying \(\operatorname{Tr}_j\) and using
\eqref{eq:Haar-partial-trace-intertwining} gives
\[
    \mathcal T^{(s-1)}
    \left(
       \operatorname{Tr}_j\rho
    \right)
    =
    \operatorname{Tr}_j
    \left(
       \mathcal T^{(s)}(\rho)
    \right)
    =
    \operatorname{Tr}_j
    \left(
       \sigma_{\mathcal S}
    \right).
\]
\end{proof}

\begin{remark} \label{rem:marginal-descent-too-weak}
Equation~\eqref{eq:automatic-marginal-descent} concerns only density
operators of the form
\(
    \operatorname{Tr}_j\rho,
    \;
    \rho\in\mathcal D(\mathcal S).
\)
In general,
\[
    \left\{
       \operatorname{Tr}_j\rho:
       \rho\in\mathcal D(\mathcal S)
    \right\}
    \subsetneq
    \mathcal D
    \left(
       \operatorname{Red}_j(\mathcal S)
    \right).
\]
Therefore~\eqref{eq:automatic-marginal-descent} does not imply that
every density operator supported on
\(
    \operatorname{Red}_j(\mathcal S)
\)
has the same Haar average.

For \textsf{QPB} at level \((s-1)\), the admissible plaintexts are all states in
\(
    \mathcal D(\operatorname{Red}_j(\mathcal S))
\),
not merely those which happen to arise as partial traces of states in
\(
    \mathcal D(\mathcal S)
\).
This is the reason that privacy heredity below is a stronger
condition.
\end{remark}

\subsubsection{Privacy heredity, Design heredity, and \textsf{QPB}-heredity}\label{subsubsec:two-heredity-notions}

We first derive an equivalent operator formulation of Haar-twirl admissibility.

\begin{lemma}[Operator formulation of Haar-twirl admissibility]
\label{lem:operator-form-qpb-admissible}
Let
\(
    \mathcal S \subset \mathcal H_d^{\otimes s}
\)
be nonzero.  The following are equivalent:
\begin{enumerate}[label=\textup{(\roman*)}]
\item there exists
\(
    \sigma_{\mathcal S}
    \in\mathcal D(\mathcal H_d^{\otimes s})
\)
such that
\(
    \mathcal T^{(s)}(\rho)
    =
    \sigma_{\mathcal S}
    \qquad
    \text{for every }
    \rho\in\mathcal D(\mathcal S);
\)

\item there exists
\(
    \sigma_{\mathcal S}
    \in\mathcal D(\mathcal H_d^{\otimes s})
\)
such that
\begin{equation}
\label{eq:operator-replacer-property}
    \mathcal T^{(s)}(X)
    =
    \operatorname{Tr}(X)
    \sigma_{\mathcal S}
    \qquad
    \text{for every }
    X\in\mathcal L(\mathcal S).
\end{equation}
\end{enumerate}
\end{lemma}

\begin{proof}
Condition \textup{(2)} immediately implies \textup{(1)}, because every density operator has trace one.

Conversely, suppose \textup{(1)} holds.  If
\(
    A\in\mathcal L(\mathcal S)
\)
is positive semidefinite and \(A\neq0\), then
\[
    \rho_A
    :=
    \frac{A}{\operatorname{Tr}(A)}
    \in
    \mathcal D(\mathcal S).
\]
Hence
\[
    \mathcal T^{(s)}(A)
    =
    \operatorname{Tr}(A)
    \mathcal T^{(s)}(\rho_A)
    =
    \operatorname{Tr}(A)
    \sigma_{\mathcal S}.
\]
Every Hermitian
\(
    H\in\mathcal L(\mathcal S)
\)
has a Jordan decomposition
\(
    H=H_+-H_-
\)
with \(H_\pm\geq0\), so linearity gives
\[
    \mathcal T^{(s)}(H)
    =
    \operatorname{Tr}(H)\sigma_{\mathcal S}.
\]
Finally, every operator
\(
    X\in\mathcal L(\mathcal S)
\)
can be written as
\[
    X
    =
    \left[\frac{X+X^\dagger}{2}\right]
    +
    i
    \left[\frac{X-X^\dagger}{2i}\right],
\]
a complex linear combination of two Hermitian operators.
Equation~\eqref{eq:operator-replacer-property} follows by linearity.
\end{proof}

Let
\(
    \mathfrak{U}
    =
    \{U_k\}_{k\in K}
    \subset
    \mathrm U(d)
\)
be a finite unitary ensemble with positive weights
\(
    w_k>0,
    \;
    \sum_{k\in K}w_k=1.
\)
At level \(s\), define the ensemble twirl
\begin{equation}
\label{eq:ensemble-twirl-heredity}
    \Phi_{\mathfrak{U},w}^{(s)}(X)
    :=
    \sum_{k\in K}
       w_k
       U_k^{\otimes s}
       X
       (U_k^\dagger)^{\otimes s}.
\end{equation}

\begin{definition}[Exact unitary
\((\mathcal S,s)\)-design]
\label{def:exact-restricted-design-heredity}
The weighted ensemble
\(
    (\mathfrak{U},w)
\)
is an exact unitary
\(
    (\mathcal S,s)
\)-design if
\begin{equation}
\label{eq:exact-S-design}
    \Phi_{\mathfrak{U},w}^{(s)}(X)
    =
    \mathcal T^{(s)}(X)
    \qquad
    \text{for every }
    X\in\mathcal L(\mathcal S).
\end{equation}
\end{definition}

The following definitions separate the two requirements needed for recipient deletion: privacy of the reduced plaintext space and descent of the same
exact design ensemble. Example~\ref{ex:top-exterior-failure}, and Example~\ref{ex:design-without-privacy} show that neither implies the other.

\begin{definition}\label{qpb-heriditary}
Let \(\mathcal{S}\subseteq\mathcal H_d^{\otimes s}\) be a nonzero subspace,
let \(j\in\{1,\ldots,s\}\), and put
\[
\mathcal{R}_j:=\operatorname{Red}_j(\mathcal{S}).
\]

\begin{enumerate}
\item The subspace \(\mathcal{S}\) is called \emph{design-hereditary in coordinate \(j\)} if every exact unitary \((\mathcal{S},s)\)-design \((\mathfrak U,w)\) is, with the same unitaries and weights, an exact unitary \((\mathcal{R}_j,s-1)\)-design.

\item If \(\mathcal{S}\) is Haar-twirl admissible, then \(\mathcal{S}\) is called \emph{privacy-hereditary in coordinate \(j\)} if \(\mathcal{R}_j\) is Haar-twirl admissible for the level-\((s-1)\) collective Haar twirl.

\item If \(\mathcal{S}\) is Haar-twirl admissible, then \(\mathcal{S}\) is called \emph{\textsf{QPB}-hereditary in coordinate \(j\)} if it is both privacy-hereditary and design-hereditary in coordinate \(j\).
\end{enumerate}
\end{definition}

If a family
\(
    \mathcal S_s,
    \mathcal S_{s-1},
    \ldots,
    \mathcal S_1
\)
satisfies
\[
    \mathcal S_{r-1}
    =
    \operatorname{Red}_r(\mathcal S_r)
    \qquad
    (2\leq r\leq s),
\]
and the relevant one-step hereditary property holds at every level,
we call the family \emph{fully privacy-hereditary},
\emph{fully design-hereditary}, or \emph{fully \textsf{QPB}-hereditary},
respectively.

\begin{remark}\label{rem:which-heredity}
For the \textsf{QPB} problem, the appropriate notion is
\emph{\textsf{QPB}-heredity}, i.e. the conjunction of privacy heredity and design heredity.

The privacy heredity guarantees that there exists a fixed
lower-level replacement state
\(
    \sigma_{\mathcal R_j}
\)
such that
\[
    \mathcal T^{(s-1)}(\rho)
    =
    \sigma_{\mathcal R_j}
    \qquad
    \forall
    \rho\in\mathcal D(\mathcal R_j).
\]
It does not imply that the finite ensemble used at level \(s\) actually
implements
\(
    \mathcal T^{(s-1)}
\).

Conversely, design heredity guarantees
\[
    \Phi_{\mathfrak{U},w}^{(s-1)}
    =
    \mathcal T^{(s-1)}
    \quad
    \text{on }\mathcal L(\mathcal R_j),
\]
but this equality is useful for \textsf{QPB} security only when the Haar map on
\(\mathcal R_j\) is itself a replacement channel.

Thus \textsf{QPB}-heredity requires both properties: the same key ensemble then implements the lower-level Haar twirl, and that Haar twirl is a replacement channel on the reduced plaintext space.
\end{remark}

\subsubsection{Criterion for privacy heredity}\label{subsubsec:privacy-heredity-criterion}

At level \(s-1\), consider the Schur-Weyl decomposition
\begin{equation}
\label{eq:SW-level-s-minus-one}
    \mathcal H_d^{\otimes(s-1)}
    \cong
    \bigoplus_{\mu\in\Lambda(d,s-1)}
       V_\mu\otimes W_\mu.
\end{equation}
For every
\(
    \mu\in\Lambda(d,s-1)
\),
let
\[
    \iota_\mu^{(s-1)}
    :
    V_\mu\otimes W_\mu
    \longrightarrow
    \mathcal H_d^{\otimes(s-1)}
\]
be the isometric inclusion and
\begin{equation}
\label{eq:lower-Schur-projector}
    \Pi_\mu^{(s-1)}
    :=
    \iota_\mu^{(s-1)}
    \bigl(\iota_\mu^{(s-1)}\bigr)^\dagger.
\end{equation}
For
\(
    Y\in
    \mathcal L(\mathcal H_d^{\otimes(s-1)})
\),
define its local \(\mu\)-diagonal block by
\begin{equation}
\label{eq:lower-local-block}
    Y_{\mu\mu}
    :=
    \bigl(\iota_\mu^{(s-1)}\bigr)^\dagger
    Y
    \iota_\mu^{(s-1)}
    \in
    \mathcal L(V_\mu\otimes W_\mu).
\end{equation}

The blockwise Haar-twirl formula is then
\begin{equation}
\label{eq:lower-Haar-block-form}
    \mathcal T^{(s-1)}(Y)
    =
    \sum_{\mu\in\Lambda(d,s-1)}
       \iota_\mu^{(s-1)}
       \left[
          \frac{I_{V_\mu}}
               {\dim_{\mathbb C}V_\mu}
          \otimes
          \operatorname{Tr}_{V_\mu}(Y_{\mu\mu})
       \right]
       \bigl(\iota_\mu^{(s-1)}\bigr)^\dagger.
\end{equation}

\begin{theorem}[Privacy-heredity criterion]
\label{thm:privacy-heredity}
Let \(\mathcal{S}\subseteq\mathcal H_d^{\otimes s}\) be Haar-twirl admissible, fix
\(j\in\{1,\ldots,s\}\), and let
\[
\mathcal{R}_j:=\operatorname{Red}_j(\mathcal{S}).
\]
Then \(\mathcal{S}\) is privacy-hereditary in coordinate \(j\) if and only if, for every
\(\mu\in\Lambda(d,s-1)\), there exists
\(\tau_\mu\in\mathcal L(W_\mu)\) with \(\tau_\mu\geq0\) such that, for every \(Y\in\mathcal L(\mathcal{R}_j)\),
\begin{equation}
\operatorname{Tr}_{V_\mu}
\!\left[
\bigl(\iota_\mu^{(s-1)}\bigr)^\dagger
\mathcal J_{\mathcal{R}_j}(Y)
\iota_\mu^{(s-1)}
\right]
=\operatorname{Tr}(Y)\tau_\mu,
\label{eq:lower-fixed-marginals}
\end{equation}
and
\begin{equation}
\sum_{\mu\in\Lambda(d,s-1)}\operatorname{Tr}(\tau_\mu)=1.
\label{eq:lower-normalization}
\end{equation}
In this case the replacement state of \(\mathcal{R}_j\) is
\begin{equation}
\sigma_{\mathcal{R}_j}
=
\sum_{\mu\in\Lambda(d,s-1)}
\iota_\mu^{(s-1)}
\left(
\frac{I_{V_\mu}}{\dim_{\mathbb C}V_\mu}
\otimes\tau_\mu
\right)
\bigl(\iota_\mu^{(s-1)}\bigr)^\dagger .
\label{eq:lower-replacement-state}
\end{equation}
\end{theorem}

\begin{proof}
Suppose first that \eqref{eq:lower-fixed-marginals} and \eqref{eq:lower-normalization} hold. Substituting \eqref{eq:lower-fixed-marginals} into \eqref{eq:lower-Haar-block-form}, we obtain, for every
\(
    Y\in\mathcal L(\mathcal R_j)
\),
\begin{align}
    \mathcal T^{(s-1)}\!\left(\mathcal J_{\mathcal{R}_j}(Y)\right)
&=
\sum_{\mu\in\Lambda(d,s-1)}
\iota_\mu^{(s-1)}
\left[
\frac{I_{V_\mu}}{\dim_{\mathbb C}V_\mu}
\otimes
\operatorname{Tr}_{V_\mu}
\!\left(
\bigl(\iota_\mu^{(s-1)}\bigr)^\dagger
\mathcal J_{\mathcal{R}_j}(Y)
\iota_\mu^{(s-1)}
\right)
\right]
\bigl(\iota_\mu^{(s-1)}\bigr)^\dagger .
    \notag\\
    &=
    \operatorname{Tr}(Y)
    \sum_{\mu}
       \iota_\mu^{(s-1)}
       \left[
          \frac{I_{V_\mu}}
               {\dim_{\mathbb C}V_\mu}
          \otimes\tau_\mu
       \right]
       \bigl(\iota_\mu^{(s-1)}\bigr)^\dagger
    \notag\\
    &=
    \operatorname{Tr}(Y)
       \sigma_{\mathcal R_j}.
\end{align}
Each summand of~\eqref{eq:lower-replacement-state} is positive semidefinite.  Moreover,
\begin{align}
    \operatorname{Tr}
    \left(
       \sigma_{\mathcal R_j}
    \right)
    &=
    \sum_\mu
       \operatorname{Tr}
       \left(
          \frac{I_{V_\mu}}
               {\dim_{\mathbb C}V_\mu}
       \right)
       \operatorname{Tr}(\tau_\mu)
    \notag\\
    &=
    \sum_\mu
       \operatorname{Tr}(\tau_\mu)
    \notag\\
    &=
    1.
\end{align}
Thus
\(
    \sigma_{\mathcal R_j}
\in\mathcal D(\mathcal H_d^{\otimes(s-1)})
\),
and
\(
    \mathcal R_j
\)
is Haar-twirl admissible.

\medskip
Conversely, suppose that
\(
    \mathcal R_j
\)
is Haar-twirl admissible. Let
\(
    m_j
    :=
    \dim_{\mathbb C}\mathcal R_j,
\)
and consider the maximally mixed state on
\(
    \mathcal R_j
\),
\[
    \rho_{\mathcal R_j}^{\,0}
    :=
    \frac{\Pi_{\mathcal R_j}}{m_j}.
\]
Then
\[
    \sigma_{\mathcal R_j}
    =
    \mathcal T^{(s-1)}
    \left(
       \rho_{\mathcal R_j}^{\,0}
    \right).
\]
Define
\begin{equation}
\label{eq:tau-from-maxmixed}
    \tau_\mu
    :=
    \operatorname{Tr}_{V_\mu}
    \left[
       \bigl(\iota_\mu^{(s-1)}\bigr)^\dagger
       \rho_{\mathcal R_j}^{\,0}
       \iota_\mu^{(s-1)}
    \right].
\end{equation}
Because
\(
    \rho_{\mathcal R_j}^{\,0}\geq0
\),
positivity of the partial trace gives
\(
    \tau_\mu\geq0
\).
Equation~\eqref{eq:lower-Haar-block-form} gives
\[
    \sigma_{\mathcal R_j}
    =
    \sum_\mu
       \iota_\mu^{(s-1)}
       \left[
          \frac{I_{V_\mu}}
               {\dim_{\mathbb C}V_\mu}
          \otimes\tau_\mu
       \right]
       \bigl(\iota_\mu^{(s-1)}\bigr)^\dagger.
\]
Taking the trace gives
\(
    \sum_\mu\operatorname{Tr}(\tau_\mu)=1
\).

For arbitrary
\(
    Y\in\mathcal L(\mathcal R_j)
\),
Haar-twirl admissibility gives
\[
    \mathcal T^{(s-1)}(\mathcal{J}_{\mathcal{R}_j}(Y))
    =
    \operatorname{Tr}(Y)
    \sigma_{\mathcal R_j}.
\]
Using the block formula on the left-hand side and the preceding
block decomposition of \(\sigma_{\mathcal R_j}\) on the right-hand
side, and then compressing to the \(\mu\)-th block, gives
\[
    \frac{I_{V_\mu}}
         {\dim_{\mathbb C}V_\mu}
    \otimes
    \operatorname{Tr}_{V_\mu}
    \left[
       \bigl(\iota_\mu^{(s-1)}\bigr)^\dagger
       \mathcal{J}_{\mathcal{R}_j}(Y)
       \iota_\mu^{(s-1)}
    \right]
\]
\[
    =
    \operatorname{Tr}(Y)
    \frac{I_{V_\mu}}
         {\dim_{\mathbb C}V_\mu}
    \otimes\tau_\mu.
\]
Then
\[
    \operatorname{Tr}_{V_\mu}
    \left[
       \bigl(\iota_\mu^{(s-1)}\bigr)^\dagger
       \mathcal{J}_{\mathcal{R}_j}(Y)
       \iota_\mu^{(s-1)}
    \right]
    =
    \operatorname{Tr}(Y)\tau_\mu.
\]
This proves the necessity.
\end{proof}

\begin{remark}\label{rem:privacy-criterion-classification}
Define
\(
    c_\mu
    :=
    \operatorname{Tr}(\tau_\mu).
\)
Then
\(
    c_\mu\geq0
\)
and
\(
    \sum_\mu c_\mu=1
\).
Whenever
\(
    c_\mu>0
\),
let
\[
    \omega_\mu
    :=
    \frac{\tau_\mu}{c_\mu}
    \in\mathcal D(W_\mu).
\]
Taking
\(
    Y=|\psi\rangle\langle\phi|
\)
in~\eqref{eq:lower-fixed-marginals} gives
\begin{equation}
\label{eq:privacy-heredity-rank-one}
\begin{split}
    &\operatorname{Tr}_{V_\mu}
    \left[
       \bigl(\iota_\mu^{(s-1)}\bigr)^\dagger
       |\psi\rangle\langle\phi|
       \iota_\mu^{(s-1)}
    \right]
    \\
    &\qquad=
    c_\mu
    \langle\phi|\psi\rangle
    \omega_\mu.
\end{split}
\end{equation}
Taking the trace over \(W_\mu\) yields
\begin{align}
    \langle\phi|
       \Pi_\mu^{(s-1)}
    |\psi\rangle
    &=
    c_\mu
    \langle\phi|\psi\rangle.
\end{align}
Thus Theorem~\ref{thm:privacy-heredity} is the classification theorem of the preceding subsection applied to the reduced subspace
\(
    \mathcal R_j
\).
\end{remark}

\subsubsection{Criterion for design heredity}\label{subsubsec:design-heredity-criterion}

Privacy heredity alone concerns the Haar channel.  To determine whether the
same finite ensemble descends, we must compare the restricted moment
information required at the two tensor levels.

\begin{definition}\label{def:restricted-coefficient-space}
Let
\(
    \mathcal S\subseteq\mathcal H_d^{\otimes s}
\)
be nonzero.  Define
\begin{equation}
\label{eq:coefficient-space}
\begin{split}
    \mathfrak C_s(\mathcal S)
    :=
    \operatorname{span}_{\mathbb C}
    \Bigl\{
       f_{a,b;\psi,\phi}^{(s)}
       :\;
       |a\rangle,|b\rangle
          \in\mathcal H_d^{\otimes s}, |\psi\rangle,|\phi\rangle\in\mathcal S
    \Bigr\},
\end{split}
\end{equation}
where
\begin{equation}
\label{eq:coefficient-function}
\begin{split}
    f_{a,b;\psi,\phi}^{(s)}(U)
    &:=
    \langle a|
       U^{\otimes s}
    |\psi\rangle
    \,
    \langle\phi|
       (U^\dagger)^{\otimes s}
    |b\rangle,
    \qquad U\in\mathrm U(d).
\end{split}
\end{equation}
\end{definition}

The space
\(
    \mathfrak C_s(\mathcal S)
\)
is finite-dimensional over \(\mathbb{C}\).  Moreover, it is invariant under complex-conjugation, because
\begin{equation}
\label{eq:coefficient-conjugation}
    \overline{
       f_{a,b;\psi,\phi}^{(s)}(U)
    }
    =
    f_{b,a;\phi,\psi}^{(s)}(U).
\end{equation}
It also contains the constant function \(1\). If
\(
    |\psi\rangle\in\mathcal S
\)
is a unit vector and
\(
    \{|a_r\rangle\}
\)
is an orthonormal basis of
\(
    \mathcal H_d^{\otimes s}
\),
then
\begin{align}
    \sum_r
       f_{a_r,a_r;\psi,\psi}^{(s)}(U)
    &=
    \sum_r
       \left|
          \langle a_r|
             U^{\otimes s}
          |\psi\rangle
       \right|^2
    \notag\\
    &=
    \left\|
       U^{\otimes s}|\psi\rangle
    \right\|^2
    \notag\\
    &=
    1.
\end{align}

\begin{lemma}
\label{lem:coefficient-design-equivalence}
A weighted finite ensemble
\(
    (\mathfrak{U},w)
\)
is an exact unitary
\(
    (\mathcal S,s)
\)-design if and only if
\begin{equation}
\label{eq:coefficient-cubature}
    \sum_{U_{k} \in \mathfrak{U}}
       w_k f(U_k)
    =
    \int_{\mathrm U(d)}
       f(U)\,
       \mathrm d\mu_{\mathrm{Haar}}(U)
\end{equation}
for every
\(
    f\in\mathfrak C_s(\mathcal S)
\).
\end{lemma}

\begin{proof}
Let \(|\psi\rangle,|\phi\rangle\in\mathcal S\). The rank-one operators
\(
    |\psi\rangle\langle\phi|
\)
span
\(
    \mathcal L(\mathcal S)
\).
Thus it is sufficient to test equality
\[
    \Phi_{\mathfrak{U},w}^{(s)}
    \left(
       |\psi\rangle\langle\phi|
    \right)
    =
    \mathcal T^{(s)}
    \left(
       |\psi\rangle\langle\phi|
    \right).
\]
For arbitrary
\(
    |a\rangle,|b\rangle
    \in\mathcal H_d^{\otimes s}
\),
the corresponding matrix element of a single conjugated rank-one
operator is
\begin{align}
    &\left\langle
       a\left|
       U^{\otimes s}
       |\psi\rangle\langle\phi|
       (U^\dagger)^{\otimes s}
       \right|b
    \right\rangle
    \notag\\
    &\qquad=
    \langle a|
       U^{\otimes s}
    |\psi\rangle
    \langle\phi|
       (U^\dagger)^{\otimes s}
    |b\rangle
    \notag\\
    &\qquad=
    f_{a,b;\psi,\phi}^{(s)}(U).
\end{align}
Therefore equality of the ensemble average and the Haar average on
all rank-one operators is equivalent to equality of the ensemble and
Haar integrals on the spanning functions
\(
    f_{a,b;\psi,\phi}^{(s)}
\),
and hence on their linear span
\(
    \mathfrak C_s(\mathcal S)
\).
\end{proof}

\begin{theorem}[Design-heredity criterion]
\label{thm:design-heredity-criterion}
Let
\(
    \mathcal S\subseteq\mathcal H_d^{\otimes s}
\)
be nonzero and let
\[
    \mathcal R_j
    :=
    \operatorname{Red}_j(\mathcal S).
\]
Then
\(
    \mathcal S
\)
is design-hereditary in coordinate \(j\) if and only if
\begin{equation}
\label{eq:coefficient-space-inclusion}
    \mathfrak C_{s-1}(\mathcal R_j)
    \subseteq
    \mathfrak C_s(\mathcal S)
\end{equation}
as subspaces of
\(
    C(\mathrm U(d)).
\)
\end{theorem}

\begin{proof}
Suppose first that
\eqref{eq:coefficient-space-inclusion} holds.
Let
\(
    (\mathfrak{U},w)
\)
be any exact unitary
\(
    (\mathcal S,s)
\)-design.  By
Lemma~\ref{lem:coefficient-design-equivalence},
\[
    \sum_{U_{k} \in \mathfrak{U} } w_k f(U_k)
    =
    \int_{\U(d)} f\,\mathrm d\mu_{\mathrm{Haar}}(U)
    \qquad
    \forall f\in\mathfrak C_s(\mathcal S).
\]
Since
\[
    \mathfrak C_{s-1}(\mathcal R_j)
    \subseteq
    \mathfrak C_s(\mathcal S),
\]
the same equality holds for every
\(
    f\in\mathfrak C_{s-1}(\mathcal R_j)
\).
Applying
Lemma~\ref{lem:coefficient-design-equivalence}
at level \(s-1\), we conclude that
\(
    (\mathfrak{U},w)
\)
is an exact unitary
\(
    (\mathcal R_j,s-1)
\)-design.

Conversely, assume
\(
\label{eq:coeff-not-inclusion}
    \mathfrak C_{s-1}(\mathcal R_j)
    \nsubseteq
    \mathfrak C_s(\mathcal S).
\)
Because both coefficient spaces are invariant under complex
conjugation, there exists a real-valued function
\(
    g
    \in
    \mathfrak C_{s-1}(\mathcal R_j)
\)
which does not belong to the real-valued part of
\(
    \mathfrak C_s(\mathcal S)
\).
Denote this latter finite-dimensional real vector space by
\(
    \mathfrak C_s(\mathcal S)_{\mathbb R}.
\)

Let
\[
    Q:
    L^2(\mathrm U(d),\mu_{\mathrm{Haar}};\mathbb R)
    \longrightarrow
    \mathfrak C_s(\mathcal S)_{\mathbb R}
\]
be the orthogonal projection and define
\begin{equation}
\label{eq:h-g-Qg}
    h
    :=
    g-Qg.
\end{equation}
Then
\(
    h\neq0
\)
and
\begin{equation}
\label{eq:h-orthogonality-final}
    \int_{\mathrm U(d)}
       f(U)h(U)
       \,\mathrm d\mu_{\mathrm{Haar}}(U)
    =
    0
\end{equation}
for every
\(
    f\in\mathfrak C_s(\mathcal S)_{\mathbb R}.
\)
Since the constant function \(1\) belongs to this space,
\[
    \int h\,\mathrm d\mu_{\mathrm{Haar}}
    =
    0.
\]
The function \(h\) is continuous and bounded.

Choose
\(
    0<\varepsilon<
       \frac{1}{\|h\|_\infty}
\)
and define a new probability measure by
\begin{equation}
\label{eq:perturbed-Haar}
    \mathrm d\nu(U)
    :=
    \bigl(
       1+\varepsilon h(U)
    \bigr)
    \mathrm d\mu_{\mathrm{Haar}}(U),
\end{equation}
where
\[
    1+\varepsilon h(U)
    \geq
    1-\varepsilon\|h\|_\infty
    >
    0,
\]
and
\[
    \int
       (1+\varepsilon h)
       \,\mathrm d\mu_{\mathrm{Haar}}
    =
    1.
\]

For every
\(
    f\in\mathfrak C_s(\mathcal S)_{\mathbb R}
\),
Equation~\eqref{eq:h-orthogonality-final} gives
\[
    \int f\,\mathrm d\nu
    =
    \int f\,\mathrm d\mu_{\mathrm{Haar}}.
\]
However,
\begin{align}
    \int g\,\mathrm d\nu
    -
    \int g\,\mathrm d\mu_{\mathrm{Haar}}
    &=
    \varepsilon
    \int gh
       \,\mathrm d\mu_{\mathrm{Haar}}
    \notag\\
    &=
    \varepsilon
    \int
       (Qg+h)h
       \,\mathrm d\mu_{\mathrm{Haar}}
    \notag\\
    &=
    \varepsilon
       \|h\|_{L^2(\mu_{\mathrm{Haar}})}^2 > 0.
\label{eq:g-separated}
\end{align}

 Now we replace \(\nu\) by a finite ensemble.

Define
\[
    F
    :=
    \mathfrak C_s(\mathcal S)_{\mathbb R}
    +
    \operatorname{span}_{\mathbb R}\{g\}.
\]
Let
\(
    f_1,\ldots,f_m
\)
be a real basis of \(F\), and define
\[
    \gamma:
    \mathrm U(d)\longrightarrow\mathbb R^m,
    \qquad
    \gamma(U)
    :=
    \bigl(
       f_1(U),\ldots,f_m(U)
    \bigr).
\]
Because \(\mathrm U(d)\) is compact and \(\gamma\) is continuous,
\(
    \gamma(\mathrm U(d))
\)
is compact.  The vector
\[
    \int_{\mathrm U(d)}
       \gamma(U)\,
       \mathrm d\nu(U)
\]
is a barycenter of this compact set and therefore belongs to its
convex hull.  By Carath\'eodory's theorem, there exist
\[
    U_1,\ldots,U_N\in\mathrm U(d),
    \qquad
    N\leq m+1,
\]
and nonnegative weights
\(
    w_1,\ldots,w_N
\)
with
\(
    \sum_kw_k=1
\),
such that
\[
    \sum_{k=1}^{N}
       w_k f(U_k)
    =
    \int f\,\mathrm d\nu
    \qquad
    \forall f\in F.
\]
After deleting zero weights and combining repeated unitaries if
necessary, we obtain a finite weighted unitary ensemble with positive
weights.

This ensemble agrees with Haar measure on
\(
    \mathfrak C_s(\mathcal S)
\),
and therefore is an exact unitary
\(
    (\mathcal S,s)
\)-design.  On the other hand,
Equation~\eqref{eq:g-separated} shows that it does not agree with
Haar measure on \(g\in\mathfrak C_{s-1}(\mathcal R_j)\).
Consequently it is not an exact
\(
    (\mathcal R_j,s-1)
\)-design.

Thus failure of
\eqref{eq:coefficient-space-inclusion}
implies failure of design heredity.
\end{proof}

Combining the preceding two theorems gives the exact general criterion.

\begin{corollary}[\textsf{QPB}-heredity criterion]
\label{cor:qpb-heredity}
Let \(\mathcal{S}\subseteq\mathcal H_d^{\otimes s}\) be Haar-twirl admissible, fix
\(j\in\{1,\ldots,s\}\), and put \(\mathcal{R}_j:=\operatorname{Red}_j(\mathcal{S})\). Then \(\mathcal{S}\) is \textsf{QPB}-hereditary in coordinate \(j\) if and only if 

\begin{enumerate}
\item both the conditions \eqref{eq:lower-fixed-marginals}--\eqref{eq:lower-normalization} hold; and 
\item \(\mathfrak{C}_{s-1}(R_j)\subseteq\mathfrak{ C}_s(S)\).
\end{enumerate}
\end{corollary}

\begin{proof}
By definition, \(\mathcal S\) is \textsf{QPB}-hereditary in coordinate \(j\) if and only if it is both privacy-hereditary and design-hereditary in that coordinate. By Theorem~\ref{thm:privacy-heredity}, privacy heredity is equivalent to Condition~\textup{(1)}, while by Theorem~\ref{thm:design-heredity-criterion}, design heredity is equivalent to Condition~\textup{(2)}. Hence the conjunction of Conditions~\textup{(1)} and~\textup{(2)} is necessary and sufficient for \textsf{QPB}-heredity.
\end{proof}

\subsubsection{A sufficient condition for design descent}
\label{subsubsec:covariant-lift}

Although
Theorem~\ref{thm:design-heredity-criterion}
is necessary and sufficient, the coefficient-space inclusion may be difficult to verify directly. Here we derive the following sufficient condition for the symmetric subspaces.

\begin{proposition}\label{prop:covariant-lifting}
Let \(\mathcal S \subseteq \mathcal H_d^{\otimes s}\), and \(\mathcal R \subseteq \mathcal H_d^{\otimes(s-1)}
\)
be two nonzero subspaces. Suppose that there exists a linear map
\[
    \mathscr L\colon
    \mathcal L(\mathcal R)
    \longrightarrow
    \mathcal L(\mathcal S)
\]
such that, for every
\(
    A\in\mathcal L(\mathcal R)
\)
and every
\(
    U\in\mathrm U(d)
\),
\begin{align}
    \operatorname{Tr}_j \circ \mathcal{J}_{\mathcal{S}} \circ \mathscr L
    &=
    \mathcal{J}_{\mathcal{R}}.
\label{eq:lift-covariance}
\end{align}
Then every exact unitary
\(
    (\mathcal S,s)
\)-design is an exact unitary
\(
    (\mathcal R,s-1)
\)-design.
\end{proposition}

\begin{proof}
Let $(\mathfrak U,w)$ be exact on $\mathcal S$. Covariance of partial
trace gives, for $A\in\mathcal L(\mathcal R)$,
\begin{align*}
 \Phi_{\mathfrak U,w}^{(s-1)}(\mathcal J_{\mathcal R}(A))
 &=\Tr_j\Phi_{\mathfrak U,w}^{(s)}
       (\mathcal J_{\mathcal S}(\mathscr L(A)))\\
 &=\Tr_j\mathcal T^{(s)}
       (\mathcal J_{\mathcal S}(\mathscr L(A)))\\
 &=\mathcal T^{(s-1)}(\Tr_j\mathcal J_{\mathcal S}(\mathscr L(A)))\\
 &=\mathcal T^{(s-1)}(\mathcal J_{\mathcal R}(A)).
\end{align*}
This is exactness on $\mathcal R$.
\end{proof}

\subsubsection{The symmetric tower}
\label{subsubsec:symmetric-heredity}

Let
\(
    \Pi_{\mathrm{sym}}^{(d,s)}
\)
denote the orthogonal projection onto
\(
    \operatorname{Sym}^s(\mathcal H_d)
\),
and define
\[
    D_s^+
    :=
    \dim_{\mathbb C}
       \operatorname{Sym}^s(\mathcal H_d)
    =
    \binom{d+s-1}{s}.
\]

Broadbent, Gonz\'alez-Guill\'en, and Schuknecht proved that every exact symmetric unitary \(s\)-design is an exact symmetric unitary \((s-1)\)-design~[Lemma~VI.3,\cite{BroadbentGonzalezGuillenSchuknecht2022}].
We present a short alternative proof because it realizes their descent theorem as an application of the equivariant-lifting criterion in Proposition~\ref{prop:covariant-lifting}.

\begin{proposition}[Reduction and privacy heredity of the symmetric tower]~\cite{BroadbentGonzalezGuillenSchuknecht2022}
\label{prop:symmetric-privacy-heredity}
For every
\(
    s\geq2
\)
and every
\(
    j\in\{1,\ldots,s\}
\),
\begin{equation}
\label{eq:symmetric-Red}
    \operatorname{Red}_j
    \left(
       \operatorname{Sym}^s(\mathcal H_d)
    \right)
    =
    \operatorname{Sym}^{s-1}(\mathcal H_d).
\end{equation}
Moreover,
\begin{equation}
\label{eq:symmetric-Haar-replacer}
    \mathcal T^{(s)}(X)
    =
    \operatorname{Tr}(X)
    \frac{
       \Pi_{\mathrm{sym}}^{(d,s)}
    }{
       D_s^+
    }
\end{equation}
for every
\(
    X\in
    \mathcal L(\operatorname{Sym}^s(\mathcal H_d)).
\)
Hence the symmetric tower is fully privacy-hereditary.
\end{proposition}

\begin{proof}
By permutation symmetry it suffices to consider \(j=1\).

Let
\(
    |\psi\rangle
    \in
    \operatorname{Sym}^s(\mathcal H_d)
\).
Every permutation \(\pi\in \mathfrak{S}_{s-1}\) acting on the remaining
\(s-1\) factors satisfies
\[
    (I_{\mathcal H_d}\otimes\Gamma_{\mathfrak{S}_{s-1}}(\pi))
    |\psi\rangle
    =
    |\psi\rangle.
\]
Since contraction in the first factor commutes with permutations of
the remaining factors,
\begin{align}
    \Gamma_{\mathfrak{S}_{s-1}}(\pi)
    C_{1,v}^{(s)}|\psi\rangle
    &=
    C_{1,v}^{(s)}
    (I\otimes\Gamma_{\mathfrak{S}_{s-1}}(\pi))
    |\psi\rangle
    \notag\\
    &=
    C_{1,v}^{(s)}|\psi\rangle.
\end{align}
Therefore
\[
    C_{1,v}^{(s)}|\psi\rangle
    \in
    \operatorname{Sym}^{s-1}(\mathcal H_d),
\]
and hence
\[
    \operatorname{Red}_1
    \left(
       \operatorname{Sym}^s(\mathcal H_d)
    \right)
    \subseteq
    \operatorname{Sym}^{s-1}(\mathcal H_d).
\]

For the reverse inclusion, let
\(
    |v\rangle\in\mathcal H_d
\)
be nonzero.  Since
\(
    |v\rangle^{\otimes s}
\in\operatorname{Sym}^s(\mathcal H_d)
\),
\begin{align}
    C_{1,v}^{(s)}
    \left(
       |v\rangle^{\otimes s}
    \right)
    &=
    \langle v|v\rangle
       |v\rangle^{\otimes(s-1)}.
\end{align}
The pure symmetric tensors
\[
    \left\{
       |v\rangle^{\otimes(s-1)}
       :
       |v\rangle\in\mathcal H_d
    \right\}
\]
span
\(
    \operatorname{Sym}^{s-1}(\mathcal H_d)
\).
Consequently,
\[
    \operatorname{Sym}^{s-1}(\mathcal H_d)
    \subseteq
    \operatorname{Red}_1
    \left(
       \operatorname{Sym}^s(\mathcal H_d)
    \right),
\]
proving~\eqref{eq:symmetric-Red}.

Finally,
\(
    \operatorname{Sym}^s(\mathcal H_d)
\)
is the irreducible \(\mathrm U(d)\)-module \(V_{(s)}\), and its
multiplicity space \(W_{(s)}\) is one-dimensional.  Schur averaging
therefore gives
\[
    \mathcal T^{(s)}(X)
    =
    \frac{\operatorname{Tr}(X)}{D_s^+}
    I_{\operatorname{Sym}^s(\mathcal H_d)}.
\]
After extending this operator by zero to
\(
    \mathcal H_d^{\otimes s}
\),
the identity on the symmetric subspace is
\(
    \Pi_{\mathrm{sym}}^{(d,s)}
\),
which proves~\eqref{eq:symmetric-Haar-replacer}.
\end{proof}

The symmetric tower satisfies the stronger design-descent property.

\begin{theorem}[Full \textsf{QPB}-heredity of the symmetric tower]
\label{thm:symmetric-tower-heredity}
For every \(s\geq2\),
\(
    \operatorname{Sym}^s(\mathcal H_d)
\)
is \textsf{QPB}-hereditary. More precisely, if
\(
    (\mathfrak{U},w)
\)
is an exact unitary
\(
    (\operatorname{Sym}^s(\mathcal H_d),s)
\)-design, then the same weighted ensemble is an exact unitary
\(
    \bigl(
       \operatorname{Sym}^r(\mathcal H_d),r
    \bigr)
\text{-design}
\)
for every
\(
    1\leq r\leq s-1
\).
\end{theorem}

\begin{proof}
Privacy descent follows from
Proposition~\ref{prop:symmetric-privacy-heredity}.
It remains to establish one-step design descent.

For
\(
    A
    \in
    \mathcal L
    \left(
       \operatorname{Sym}^{s-1}(\mathcal H_d)
    \right),
\)
identify \(A\) with its extension by zero to
\(
    \mathcal H_d^{\otimes(s-1)}
\),
and define
\begin{equation}
\label{eq:Jplus}
    \mathcal J_s^+(A)
    :=
    \Pi_{\mathrm{sym}}^{(d,s)}
    (I_{\mathcal H_d}\otimes A)
    \Pi_{\mathrm{sym}}^{(d,s)}.
\end{equation}
Then
\[
    \mathcal J_s^+(A)
    \in
    \mathcal L
    \left(
       \operatorname{Sym}^s(\mathcal H_d)
    \right).
\]
Define also
\begin{equation}
\label{eq:Kplus}
    \mathcal K_s^+(A)
    :=
    \operatorname{Tr}_1
    \left(
       \mathcal J_s^+(A)
    \right).
\end{equation}
By~\eqref{eq:symmetric-Red},
\[
    \mathcal K_s^+(A)
    \in
    \mathcal L
    \left(
       \operatorname{Sym}^{s-1}(\mathcal H_d)
    \right).
\]

We first prove that
\(
    \mathcal J_s^+
\)
is injective.  Suppose that
\(
    \mathcal J_s^+(A)=0
\).
For arbitrary
\(
    |v\rangle,|w\rangle\in\mathcal H_d
\),
the vectors
\(
    |v\rangle^{\otimes s}
\)
and
\(
    |w\rangle^{\otimes s}
\)
belong to the symmetric subspace, so
\begin{align}
    0
    &=
    \langle v|^{\otimes s}
       \mathcal J_s^+(A)
       |w\rangle^{\otimes s}
    \notag\\
    &=
    \langle v|^{\otimes s}
       (I\otimes A)
       |w\rangle^{\otimes s}
    \notag\\
    &=
    \langle v|w\rangle
    \,
    \langle v|^{\otimes(s-1)}
       A
       |w\rangle^{\otimes(s-1)}.
\label{eq:Jplus-injective}
\end{align}
Whenever
\(
    \langle v|w\rangle\neq0
\),
we obtain
\[
    \langle v|^{\otimes(s-1)}
       A
       |w\rangle^{\otimes(s-1)}
    =
    0.
\]
The left-hand side is continuous in \((v,w)\), while the set
\(
    \{(v,w):\langle v|w\rangle\neq0\}
\)
is dense.  Hence the same equality holds for all \(v,w\).
Since pure powers span
\(
    \operatorname{Sym}^{s-1}(\mathcal H_d)
\),
all matrix elements of \(A\) vanish, and therefore \(A=0\).
Thus
\(
    \mathcal J_s^+
\)
is injective.

Equip both operator spaces with the Hilbert--Schmidt inner product.
For
\(
    A,B
\in
\mathcal L(\operatorname{Sym}^{s-1}(\mathcal H_d))
\),
\begin{align}
    \langle A,\mathcal K_s^+(B)\rangle_{\mathrm{HS}}
    &=
    \operatorname{Tr}
    \left[
       A^\dagger
       \operatorname{Tr}_1
       \left(
          \Pi_{\mathrm{sym}}^{(d,s)}
          (I\otimes B)
          \Pi_{\mathrm{sym}}^{(d,s)}
       \right)
    \right]
    \notag\\
    &=
    \operatorname{Tr}
    \left[
       (I\otimes A^\dagger)
       \Pi_{\mathrm{sym}}^{(d,s)}
       (I\otimes B)
       \Pi_{\mathrm{sym}}^{(d,s)}
    \right]
    \notag\\
    &=
    \operatorname{Tr}
    \left[
       \bigl(
          \mathcal J_s^+(A)
       \bigr)^\dagger
       \mathcal J_s^+(B)
    \right]
    \notag\\
    &=
    \langle
       \mathcal J_s^+(A),
       \mathcal J_s^+(B)
    \rangle_{\mathrm{HS}}.
\end{align}
Hence
\[
    \mathcal K_s^+
    =
    (\mathcal J_s^+)^\ast
    \mathcal J_s^+.
\]
For \(A\neq0\),
\[
    \langle
       A,\mathcal K_s^+(A)
    \rangle_{\mathrm{HS}}
    =
    \|\mathcal J_s^+(A)\|_{\mathrm{HS}}^2
    >
    0.
\]
Therefore
\(
    \mathcal K_s^+
\)
is positive definite and, since its domain is finite-dimensional, is
invertible.

The symmetric projector commutes with \(U^{\otimes s}\), so
\begin{align}
    &\mathcal J_s^+
    \left(
       U^{\otimes(s-1)}
       A
       (U^\dagger)^{\otimes(s-1)}
    \right)
    \notag\\
    &\qquad=
    U^{\otimes s}
    \mathcal J_s^+(A)
    (U^\dagger)^{\otimes s}.
\end{align}
Taking the first partial trace and using covariance shows that
\(
    \mathcal K_s^+
\)
intertwines the level-\((s-1)\) conjugation action.  Consequently
\(
    (\mathcal K_s^+)^{-1}
\)
also intertwines it.

Define
\begin{equation}
\label{eq:symmetric-lift}
    \mathcal L_s^+
    :=
    \mathcal J_s^+
    \circ
    (\mathcal K_s^+)^{-1}.
\end{equation}
Then
\[
    \operatorname{Tr}_1
    \mathcal L_s^+(A)
    =
    A,
\]
and \(\mathcal L_s^+\) is covariant.  Therefore
Proposition~\ref{prop:covariant-lifting}
implies design descent from \(s\) to \(s-1\).
Iterating the same argument proves the statement for every
\(1\leq r<s\).
\end{proof}

\begin{corollary}\label{cor:symmetric-proper-subspaces}
Every nonzero subspace
\(
    \mathcal S
    \subset
    \operatorname{Sym}^s(\mathcal H_d)
\)
is Haar-twirl admissible and privacy-hereditary in every tensor coordinate.
More precisely,
\[
    \operatorname{Red}_j(\mathcal S)
    \subseteq
    \operatorname{Sym}^{s-1}(\mathcal H_d),
\]
and every nonzero reduced subspace is Haar-twirl admissible.
\end{corollary}

\subsubsection{The antisymmetric tower}\label{subsubsec:antisymmetric-heredity}

Assume
\(
    1\leq s\leq d
\).
Let
\(
    \Pi_{\mathrm{alt}}^{(d,s)}
\)
denote the orthogonal projection onto
\(
    \bigwedge^s\mathcal H_d
\),
and put
\[
    D_s^-
    :=
    \dim_{\mathbb C}
       \bigwedge^s\mathcal H_d
    =
    \binom ds.
\]

\begin{theorem}\label{thm:antisymmetric-privacy-heredity}
For
\(
    2\leq s\leq d
\)
and every tensor coordinate \(j\),
\begin{equation}
\label{eq:antisymmetric-Red}
    \operatorname{Red}_j
    \left(
       \bigwedge^s\mathcal H_d
    \right)
    =
    \bigwedge^{s-1}\mathcal H_d.
\end{equation}
Furthermore,
\begin{equation}
\label{eq:antisymmetric-replacer}
    \mathcal T^{(s)}(X)
    =
    \operatorname{Tr}(X)
    \frac{
       \Pi_{\mathrm{alt}}^{(d,s)}
    }{
       D_s^-
    }
\end{equation}
for every
\(
    X\in
    \mathcal L(\bigwedge^s\mathcal H_d).
\)
Thus the antisymmetric tower is fully privacy-hereditary.
\end{theorem}

\begin{proof}
Again it suffices to consider \(j=1\).
If
\(
    |\psi\rangle\in\bigwedge^s\mathcal H_d
\),
then every permutation \(\pi\in \mathfrak{S}_{s-1}\) of the remaining factors
satisfies
\[
    (I\otimes\Gamma_{\mathfrak{S}_{s-1}}(\pi))
    |\psi\rangle
    =
    \operatorname{sgn}(\pi)|\psi\rangle.
\]
Therefore
\begin{align}
    \Gamma_{\mathfrak{S}_{s-1}}(\pi)
       C_{1,v}^{(s)}|\psi\rangle
    &=
    C_{1,v}^{(s)}
       (I\otimes\Gamma_{\mathfrak{S}_{s-1}}(\pi))
       |\psi\rangle
    \notag\\
    &=
    \operatorname{sgn}(\pi)
       C_{1,v}^{(s)}|\psi\rangle.
\end{align}
Hence
\[
    \operatorname{Red}_1
    \left(
       \bigwedge^s\mathcal H_d
    \right)
    \subseteq
    \bigwedge^{s-1}\mathcal H_d.
\]

Conversely, decomposable exterior tensors span
\(
    \bigwedge^{s-1}\mathcal H_d
\).
Take orthonormal vectors
\(
    |v_2\rangle,\ldots,|v_s\rangle
\).
Since \(s\leq d\), choose a unit vector
\(
    |v_1\rangle
\)
orthogonal to their span.  With the normalization
\[
    |v_1\wedge\cdots\wedge v_s\rangle
    :=
    \frac{1}{\sqrt{s!}}
    \sum_{\pi\in \mathfrak{S}_s}
       \operatorname{sgn}(\pi)
       |v_{\pi(1)}\rangle
       \otimes\cdots\otimes
       |v_{\pi(s)}\rangle,
\]
only those permutations for which
\(
    \pi(1)=1
\)
survive contraction with
\(
    \langle v_1|
\).
Consequently,
\begin{align}
    C_{1,v_1}^{(s)}
    |v_1\wedge\cdots\wedge v_s\rangle &=
    \frac{1}{\sqrt{s!}}
    \sum_{\pi\in S_s:\,\pi(1)=1}
       \operatorname{sgn}(\pi)
       |v_{\pi(2)}\rangle
       \otimes\cdots\otimes
       |v_{\pi(s)}\rangle
    \notag\\
    &=
    \frac{\sqrt{(s-1)!}}{\sqrt{s!}}
       |v_2\wedge\cdots\wedge v_s\rangle
    \notag\\
    &=
    \frac{1}{\sqrt s}
       |v_2\wedge\cdots\wedge v_s\rangle.
\end{align}
Hence the reverse inclusion holds, proving
\eqref{eq:antisymmetric-Red}.

Finally,
\(
    \bigwedge^s\mathcal H_d
\)
is the irreducible \(\mathrm U(d)\)-module
\(
    V_{(1^s)}
\),
and
\(
    W_{(1^s)}
\)
is one-dimensional.  Schur averaging gives
\eqref{eq:antisymmetric-replacer}.
\end{proof}

As in the symmetric case, every nonzero subspace
\(
    \mathcal S
    \subseteq
    \bigwedge^s\mathcal H_d
\)
is privacy-hereditary.  However, unlike the symmetric tower, universal
design descent does not hold at every exterior degree.

To determine the precise threshold, let
\(
    \omega_1,\ldots,\omega_{d-1}
\)
denote the fundamental weights of \(\mathrm{SU}(d)\), and let
\(
    L_\gamma
\)
denote the irreducible \(\mathrm{SU}(d)\)-module of highest weight
\(\gamma\). 

\medskip
Let
\(
    \omega_0=\omega_d=0.
\)
\begin{lemma}\label{lem:exterior-conjugation-decomposition}
Let
\(
    m_s
    :=
    \min\{s,d-s\}.
\)
Then, as an \(\mathrm{SU}(d)\)-module,
\begin{equation}
\label{eq:End-wedge-decomposition}
    \mathcal L
    \left(
       \bigwedge^s\mathcal H_d
    \right)
    \cong
    \bigoplus_{r=0}^{m_s}
       L_{\omega_r+\omega_{d-r}}.
\end{equation}
Every irreducible summand occurs with multiplicity one.
\end{lemma}

\begin{proof}
There is a \(\mathrm U(d)\)-equivariant isomorphism
\begin{equation}
\label{eq:wedge-dual}
    \left(
       \bigwedge^s\mathcal H_d
    \right)^*
    \cong
    \bigwedge^{d-s}\mathcal H_d
    \otimes
    (\det)^{-1}.
\end{equation}
Therefore
\begin{align}
    \mathcal L
    \left(
       \bigwedge^s\mathcal H_d
    \right)
    &\cong
    \bigwedge^s\mathcal H_d
    \otimes
    \left(
       \bigwedge^s\mathcal H_d
    \right)^*
    \notag\\
    &\cong
    \left(
       \bigwedge^s\mathcal H_d
       \otimes
       \bigwedge^{d-s}\mathcal H_d
    \right)
    \otimes
    (\det)^{-1}.
\end{align}
The Littlewood--Richardson rule gives
\begin{equation}
\label{eq:wedge-Pieri}
    \bigwedge^s\mathcal H_d
    \otimes
    \bigwedge^{d-s}\mathcal H_d
    \cong
    \bigoplus_{r=0}^{m_s}
       \mathbb S_{(2^r,1^{d-2r})}
       (\mathcal H_d).
\end{equation}
We start from the column partition
\(
    (1^s)
\),
and add a vertical strip of size \(d-s\).
If \(r\) of the added boxes are placed in rows already occupied
by the original column, the resulting partition is
\[
    (2^r,1^{d-2r}),
\]
and the inequalities
\[
    0\leq r\leq s,
    \qquad
    0\leq r\leq d-s
\]
give
\[
    0\leq r\leq m_s.
\]
The Pieri rule also shows that every such summand occurs once.

After padding the highest weight
\(
    (2^r,1^{d-2r})
\)
to \(d\) coordinates, tensoring by
\(
    (\det)^{-1}
\)
subtracts \(1\) from each coordinate.  The resulting weight is
\[
    (\underbrace{1,\ldots,1}_{r},
     \underbrace{0,\ldots,0}_{d-2r},
     \underbrace{-1,\ldots,-1}_{r}),
\]
whose \(\mathrm{SU}(d)\)-highest weight is
\[
    \omega_r+\omega_{d-r}.
\]
For \(r=0\), this is the zero weight, corresponding to the trivial
representation.  This proves
\eqref{eq:End-wedge-decomposition}.
\end{proof}

\begin{theorem}[Exact design heredity of the antisymmetric tower]
\label{thm:antisymmetric-design-heredity}
Let
\(
    2\leq s\leq d
\).
Then
\(
    \bigwedge^s\mathcal H_d
\)
is design-hereditary from level \(s\) to level \(s-1\) if and only if
\begin{equation}
\label{eq:antisymmetric-threshold}
       s\leq\frac{d+1}{2}.
\end{equation}
Consequently,
\(
    \bigwedge^s\mathcal H_d
\)
is \textsf{QPB}-hereditary if and only if
\(
    s\leq(d+1)/2
\).
\end{theorem}

\begin{proof}
Since
\(
    \bigwedge^s\mathcal H_d
\)
is invariant under the collective \(\mathrm U(d)\)-action, its
restricted coefficient space is the coefficient space of the
conjugation representation on
\(
    \mathcal L(\bigwedge^s\mathcal H_d)
\).
The central subgroup
\(
    \mathrm U(1)\subseteq\mathrm U(d)
\)
acts by scalars on
\(
    \bigwedge^s\mathcal H_d
\)
and hence acts trivially by conjugation.  It is therefore sufficient to
use the \(\mathrm{SU}(d)\)-decomposition in
Lemma~\ref{lem:exterior-conjugation-decomposition}.

By Peter--Weyl orthogonality,
\begin{equation}
\label{eq:wedge-coeff-s}
    \mathfrak C_s
    \left(
       \bigwedge^s\mathcal H_d
    \right)
    =
    \bigoplus_{r=0}^{m_s}
       \operatorname{Coeff}
       \left(
          L_{\omega_r+\omega_{d-r}}
       \right),
\end{equation}
where
\[
    m_s
    =
    \min\{s,d-s\}.
\]
Similarly,
\begin{equation}
\label{eq:wedge-coeff-sminus1}
    \mathfrak C_{s-1}
    \left(
       \bigwedge^{s-1}\mathcal H_d
    \right)
    =
    \bigoplus_{r=0}^{m_{s-1}}
       \operatorname{Coeff}
       \left(
          L_{\omega_r+\omega_{d-r}}
       \right),
\end{equation}
where
\[
    m_{s-1}
    =
    \min\{s-1,d-s+1\}.
\]
Since matrix-coefficient spaces of inequivalent irreducible representations are orthogonal, the inclusion required by Theorem~\ref{thm:design-heredity-criterion} holds when
\begin{equation}
\label{eq:m-threshold}
    m_{s-1}
    \leq
    m_s.
\end{equation}

We now evaluate~\eqref{eq:m-threshold}.

If
\(
    s\leq d/2
\),
then
\[
    m_s=s,
    \qquad
    m_{s-1}=s-1,
\]
and hence
\(
    m_{s-1}\leq m_s
\).

If \(d=2q+1\) is odd and
\(
    s=q+1=(d+1)/2
\),
then
\[
    m_s=q
    =
    m_{s-1},
\]
and again the inclusion holds.

These two possibilities hold when
\[
    s\leq\frac{d+1}{2}.
\]

Finally, if
\(
    s>(d+1)/2
\),
then
\[
    m_s=d-s,
    \qquad
    m_{s-1}=d-s+1,
\]
and therefore
\(
    m_{s-1}>m_s.
\)
The lower coefficient space then contains the additional nonzero
Peter-Weyl summand
\(
    \operatorname{Coeff}
    \left(
       L_{\omega_{d-s+1}+\omega_{s-1}}
    \right)
\)
which is absent from the upper coefficient space.  Thus design
heredity fails.

Privacy heredity holds for all
\(
    2\leq s\leq d
\)
by
Theorem~\ref{thm:antisymmetric-privacy-heredity}.
The final assertion follows.
\end{proof}

\begin{example}[Failure of design heredity at the top exterior power]
\label{ex:top-exterior-failure}
Take \(s=d\).  Then
\[
    \dim_{\mathbb C}\bigwedge^d\mathcal H_d
    =
    1
\]
and
\[
    U^{\otimes d}
    \big|_{\bigwedge^d\mathcal H_d}
    =
    \det(U).
\]
Hence conjugation by
\(
    U^{\otimes d}
\)
acts trivially on
\(
    \mathcal L(\bigwedge^d\mathcal H_d)
\).
Consequently every weighted unitary ensemble, even the ensemble
consisting only of the identity, is an exact restricted design on
\(
    \bigwedge^d\mathcal H_d
\).

On the other hand,
\[
    \dim_{\mathbb C}
       \bigwedge^{d-1}\mathcal H_d
    =
    d.
\]
For \(d\geq2\), conjugation on this space is nontrivial, and the
identity-only ensemble does not reproduce its Haar twirl.  Thus
\(
    \bigwedge^d\mathcal H_d
\)
is privacy-hereditary but not design-hereditary.  This gives an
explicit demonstration that privacy heredity does not imply design
heredity.
\end{example}

\subsubsection{Canonical carriers in a general Schur-Weyl block}
\label{subsubsec:canonical-Schur-carriers}

We now consider the canonical single-copy carriers arising in an
arbitrary Schur-Weyl block.  At tensor level \(s\), let
\begin{equation}
\label{eq:SW-level-s-carrier}
    \mathcal H_d^{\otimes s}
    \cong
    \bigoplus_{\lambda\in\Lambda(d,s)}
      \left( V_\lambda\otimes W_\lambda \right).
\end{equation}

For a unit vector
\(
    |\eta\rangle\in W_\lambda
\),
define the ambient carrier
\begin{equation}
\label{eq:canonical-carrier}
    \mathcal S_{\lambda,\eta}
    :=
    \iota_\lambda^{(s)}
    \left(
       V_\lambda
       \otimes
       \operatorname{span}_{\mathbb{C}}\{|\eta\rangle\}
    \right)
    \subseteq
    \mathcal H_d^{\otimes s}.
\end{equation}
This subspace is Haar-twirl admissible, because the multiplicity-space state is
fixed and the \(\mathrm U(d)\)-action on \(V_\lambda\) is irreducible.

For the branching theorem, we contract the \emph{last} tensor factor.The remaining first \((s-1)\) tensor factors carry the standard subgroup
\(
    \mathfrak{S}_{s-1}
    \subset 
    \mathfrak{S}_s,
\)
which fixes the symbol \(s\).  The standard symmetric-group branching
rule therefore takes the form
\begin{equation}
\label{eq:Specht-branching}
    W_\lambda
    \downarrow_{\mathfrak{S}_{s-1}}
    \cong
    \bigoplus_{\mu\nearrow\lambda}
       W_\mu,
\end{equation}
where
\(
    \mu\nearrow\lambda
\)
means that the Young diagram of \(\mu\vdash s-1\) is obtained by removing one removable box from \(\lambda\).

Since the restriction is multiplicity-free, the \(W_\mu\)-isotypic
subspace in
\(
    W_\lambda\downarrow_{\mathfrak{S}_{s-1}}
\)
is canonical.  Let
\[
    q_{\mu\leftarrow\lambda}
    :
    W_\lambda
    \longrightarrow
    W_\mu
\]
denote the corresponding orthogonal projection under the fixed
branching identification, and define the active branch set
\begin{equation}
\label{eq:active-branch-set-correct}
    \mathcal B(\lambda,\eta)
    :=
    \left\{
       \mu\nearrow\lambda:
       q_{\mu\leftarrow\lambda}|\eta\rangle
       \neq0
    \right\}.
\end{equation}

\begin{theorem}[Branching classification for canonical Schur carriers]
\label{thm:canonical-carrier-branching}
For
\(
    \mathcal S_{\lambda,\eta}
\)
defined in~\eqref{eq:canonical-carrier},
\begin{equation}
\label{eq:carrier-Red}
\begin{split}
    \operatorname{Red}_s
    \left(
       \mathcal S_{\lambda,\eta}
    \right)
    &=
    \bigoplus_{\mu\in\mathcal B(\lambda,\eta)}
       \iota_\mu^{(s-1)}
       \left(
          V_\mu
          \otimes
          \operatorname{span}_{\mathbb{C}}
          \{
             q_{\mu\leftarrow\lambda}|\eta\rangle
          \}
       \right).
\end{split}
\end{equation}
Consequently:
\begin{enumerate}[label=\textup{(\roman*)}]
\item
\(
    \mathcal S_{\lambda,\eta}
\)
is privacy-hereditary under last-coordinate reduction if and only if
\begin{equation}
\label{eq:one-active-branch}
    |\mathcal B(\lambda,\eta)|=1.
\end{equation}

\item
If
\(
    \mathcal B(\lambda,\eta)=\{\mu\}
\),
then
\(
    \mathcal S_{\lambda,\eta}
\)
is design-hereditary if and only if
\begin{equation}
\label{eq:carrier-coeff-inclusion}
\begin{split}
    \operatorname{Coeff}
    \left(
       V_\mu\otimes\overline{V_\mu}
    \right)
    \subseteq
    \operatorname{Coeff}
    \left(
       V_\lambda\otimes\overline{V_\lambda}
    \right).
\end{split}
\end{equation}

\item
Thus
\(
    \mathcal S_{\lambda,\eta}
\)
is \textsf{QPB}-hereditary if and only if
\eqref{eq:one-active-branch} and
\eqref{eq:carrier-coeff-inclusion} both hold.
\end{enumerate}
\end{theorem}

\begin{proof}
Let
\(
   \pi_\alpha\colon \mathrm U(d) \longrightarrow \U(V_\alpha)
\)
for the irreducible representation carried by \(V_\alpha\).  We use the
fixed unitary Schur-Weyl identifications
\[
   \mathcal H_d^{\otimes r}
   =
   \bigoplus_{\alpha\in\Lambda(d,r)}
      \iota_\alpha^{(r)}(V_\alpha\otimes W_\alpha),
   \qquad r\in\{s-1,s\}.
\]

Define the joint last-coordinate contraction
\[
   C_s\colon
   \mathcal H_d^*\otimes\mathcal H_d^{\otimes s}
   \longrightarrow
   \mathcal H_d^{\otimes(s-1)},
   \qquad
   C_s(f\otimes\ket{\psi})
   :=
   (I_{\mathcal H_d^{\otimes(s-1)}}\otimes f)
   \ket{\psi}.
\]
Thus, if \(f=\bra{v}\), then
\(
   C_s(\bra{v}\otimes\ket{\psi})
   =
   C_{s,v}^{(s)}\ket{\psi}
\).
Let \(U^\vee f:=f\circ U^{-1}\) be the contragredient action on \(\mathcal H_d^*\). By definition,
\[
   C_s\!\left(U^\vee f\otimes U^{\otimes s}\ket{\psi}\right)
   =
   U^{\otimes(s-1)}C_s(f\otimes\ket{\psi}),
   \qquad U\in\mathrm U(d),
\]
so \(C_s\) is \(\mathrm U(d)\)-equivariant.  It is also
\(\mathfrak S_{s-1}\)-equivariant, where
\(\mathfrak S_{s-1}\) permutes the first \(s-1\) tensor factors.  Moreover,
the definition of the reduced subspace gives, for every
\(\mathcal S\subseteq\mathcal H_d^{\otimes s}\),
\begin{equation*}
   \operatorname{Red}_s(\mathcal S)
   =
   \operatorname{im}\big(C_s\big|_{\mathcal H_d^*\otimes\mathcal S}\big).
\end{equation*}
The image on the right is the linear span of all
last-coordinate contractions of vectors in \(\mathcal S\).

We next determine the restriction of \(C_s\) to the
\(\lambda\)-th Schur-Weyl block.  Under the multiplicity-free branching
decomposition, write
\[
   W_\lambda\!\downarrow_{\mathfrak S_{s-1}}
   =
   \bigoplus_{\mu\nearrow\lambda}
      j_{\mu\to\lambda}(W_\mu),
   \qquad
   j_{\mu\to\lambda}:=q_{\mu\leftarrow\lambda}^{\dagger}.
\]
The maps \(j_{\mu\to\lambda}\) are isometries, their ranges are mutually
orthogonal, and
\begin{equation*}
   \sum_{\mu\nearrow\lambda}
      j_{\mu\to\lambda}q_{\mu\leftarrow\lambda}
   =
 I_{W_\lambda}.
\end{equation*}

Fix a branch \(\mu\nearrow\lambda\).  The subspace
\[
   \iota_\lambda^{(s)}
   \bigl(V_\lambda\otimes j_{\mu\to\lambda}(W_\mu)\bigr)
   \subseteq
   \mathcal H_d^{\otimes(s-1)}\otimes\mathcal H_d
\]
is \(W_\mu\)-isotypic for the action of
\(\mathfrak S_{s-1}\).  Schur's lemma, applied to the
\(\mathfrak S_{s-1}\)-factor, therefore yields a linear map
\[
   b_{\lambda\leftarrow\mu}\colon
   V_\lambda\longrightarrow V_\mu\otimes\mathcal H_d
\]
such that, for every \(x\in V_\lambda\) and \(\zeta\in W_\mu\),
\begin{equation*}
   \iota_\lambda^{(s)}
   \bigl(x\otimes j_{\mu\to\lambda}\zeta\bigr)
   =
   \sum_{a=1}^{d}
      \iota_\mu^{(s-1)}
      \bigl(B_a^{\mu\leftarrow\lambda}x\otimes\zeta\bigr)
      \otimes e_a,
\end{equation*}
where \(\{e_a\}_{a=1}^{d}\) is an orthonormal basis of
\(\mathcal H_d\) and
\[
   b_{\lambda\leftarrow\mu}x
   =
   \sum_{a=1}^{d}
      B_a^{\mu\leftarrow\lambda}x\otimes e_a.
\]
Because the Schur-Weyl identifications and \(j_{\mu\to\lambda}\) are
isometric, \(b_{\lambda\leftarrow\mu}\) is an isometry.  Equivariance of
the Schur-Weyl identifications also shows that it is a
\(\mathrm U(d)\)-intertwiner.  It is therefore the nonzero Pieri
inclusion corresponding to the branch \(\mu\nearrow\lambda\). Therefore
\[
   \dim_{\mathbb C}
   \operatorname{Hom}_{\mathrm U(d)}
   \bigl(
      \mathcal H_d^*\otimes V_\lambda,V_\mu
   \bigr)
   =1.
\]

Define the Pieri map
\[
   A_{\mu\leftarrow\lambda}\colon
   \mathcal H_d^*\otimes V_\lambda
   \longrightarrow V_\mu
\]
by
\begin{equation*}
   A_{\mu\leftarrow\lambda}(f\otimes x)
   :=
   (\mathbf I_{V_\mu}\otimes f)
   b_{\lambda\leftarrow\mu}x
   =
   \sum_{a=1}^{d}
      f(e_a)B_a^{\mu\leftarrow\lambda}x.
\end{equation*}
The equivariance of \(b_{\lambda\leftarrow\mu}\) and of the evaluation
pairing implies that \(A_{\mu\leftarrow\lambda}\) is a
\(\mathrm U(d)\)-intertwiner.  It is nonzero. If
\(A_{\mu\leftarrow\lambda}=0\), then, for every \(a\) and every
\(x\in V_\lambda\),
\[
   B_a^{\mu\leftarrow\lambda}x
   =
   A_{\mu\leftarrow\lambda}(e_a^*\otimes x)
   =0,
\]
where \(\{e_a^*\}_{a=1}^{d}\) is the dual basis.  This would imply
\(b_{\lambda\leftarrow\mu}=0\), contradicting the fact that
\(b_{\lambda\leftarrow\mu}\) is an isometry.  Since \(V_\mu\) is
irreducible, the nonzero invariant subspace
\(\operatorname{im}A_{\mu\leftarrow\lambda}\subseteq V_\mu\) must be all
of \(V_\mu\); hence
\begin{equation*}
   \operatorname{im}A_{\mu\leftarrow\lambda}=V_\mu.
\end{equation*}

Applying \(C_s\) to the preceding Schur-Weyl identity gives
\begin{align*}
   C_s\!\left[
      f\otimes
      \iota_\lambda^{(s)}
      \bigl(x\otimes j_{\mu\to\lambda}\zeta\bigr)
   \right]
   &=
   \sum_{a=1}^{d}
      f(e_a)
      \iota_\mu^{(s-1)}
      \bigl(B_a^{\mu\leftarrow\lambda}x\otimes\zeta\bigr)
   \\
   &=
   \iota_\mu^{(s-1)}
   \bigl(
      A_{\mu\leftarrow\lambda}(f\otimes x)
      \otimes\zeta
   \bigr).
\end{align*}
Because
\(
   \sum_{\mu\nearrow\lambda}
   j_{\mu\to\lambda}q_{\mu\leftarrow\lambda}
   =I_{W_\lambda}
\), it follows that, as a map on
\(\mathcal H_d^*\otimes V_\lambda\otimes W_\lambda\),
\begin{equation*}
\begin{split}
   C_s\circ
   \bigl(
      I_{\mathcal H_d^*}
      \otimes\iota_\lambda^{(s)}
   \bigr)
   =
   \sum_{\mu\nearrow\lambda}
      \iota_\mu^{(s-1)}
      \circ
      \bigl(
         A_{\mu\leftarrow\lambda}
         \otimes q_{\mu\leftarrow\lambda}
      \bigr).
\end{split}
\end{equation*}
This is the branch decomposition of the contraction map.

Now set
\begin{equation*}
\begin{split}
   \mathcal I_{\lambda,\eta}
   &:=
   \operatorname{im}\!\left(
      \left.
      C_s\circ
      \bigl(
         I_{\mathcal H_d^*}
         \otimes\iota_\lambda^{(s)}
      \bigr)
      \right|_{
         \mathcal H_d^*\otimes
         (V_\lambda\otimes\operatorname{span}\{\ket{\eta}\})
      }
   \right).
\end{split}
\end{equation*}
By the contraction characterization of reduction,
\begin{equation*}
   \mathcal I_{\lambda,\eta}
   =
   \operatorname{Red}_s(\mathcal S_{\lambda,\eta}).
\end{equation*}
For \(f\in\mathcal H_d^*\) and \(x\in V_\lambda\), the branch
decomposition yields
\begin{equation*}
\begin{split}
   &C_s\!\left[
      f\otimes
      \iota_\lambda^{(s)}(x\otimes\ket{\eta})
   \right]
   =
   \sum_{\mu\nearrow\lambda}
      \iota_\mu^{(s-1)}
      \left(
         A_{\mu\leftarrow\lambda}(f\otimes x)
         \otimes q_{\mu\leftarrow\lambda}\ket{\eta}
      \right).
\end{split}
\end{equation*}
Consequently,
\begin{equation*}
   \mathcal I_{\lambda,\eta}
   \subseteq
   \mathcal D_{\lambda,\eta}
   :=
   \bigoplus_{\mu\in\mathcal B(\lambda,\eta)}
      \mathcal E_\mu,
\end{equation*}
where
\begin{equation*}
   \mathcal E_\mu
   :=
   \iota_\mu^{(s-1)}
   \left(
      V_\mu\otimes
      \operatorname{span}_{\mathbb{C}}\!\left\{
         q_{\mu\leftarrow\lambda}\ket{\eta}
      \right\}
   \right).
\end{equation*}
Notice that \(\mathcal B(\lambda,\eta)\neq\emptyset\); otherwise all the vectors \(q_{\mu\leftarrow\lambda}\ket{\eta}\) would vanish, and the
orthogonal branching resolution of the identity would give
\(\ket{\eta}=0\), contrary to the hypothesis that \(\ket{\eta}\) is a
unit vector.

Both \(\mathcal I_{\lambda,\eta}\) and
\(\mathcal D_{\lambda,\eta}\) are \(\mathrm U(d)\)-invariant. Each \(\mathcal E_\mu\) is an irreducible \(\mathrm U(d)\)-module equivalent to \(V_\mu\), and the modules \(\mathcal E_\mu\) for distinct
\(\mu\in\mathcal B(\lambda,\eta)\) are pairwise inequivalent. We now justify that the preceding inclusion is an equality.

Let \(Q\) be the orthogonal projection of
\(\mathcal D_{\lambda,\eta}\) onto
\(\mathcal I_{\lambda,\eta}\).  Since the representation is unitary and
\(\mathcal I_{\lambda,\eta}\) is invariant, its orthogonal complement
inside \(\mathcal D_{\lambda,\eta}\) is invariant as well; hence \(Q\)
commutes with the \(\mathrm U(d)\)-action.  If \(\Pi_\rho\) denotes the
orthogonal projection onto \(\mathcal E_\rho\), then
\[
   \Pi_\rho Q\big|_{\mathcal E_\sigma}
   \in
   \operatorname{Hom}_{\mathrm U(d)}
   (\mathcal E_\sigma,\mathcal E_\rho).
\]
Schur's lemma therefore gives
\[
   \Pi_\rho Q\big|_{\mathcal E_\sigma}=0
   \quad(\rho\neq\sigma),
   \qquad
   \Pi_\rho Q\big|_{\mathcal E_\rho}
   =c_\rho I_{\mathcal E_\rho}
\]
for some \(c_\rho\in\mathbb C\).  Since \(Q=Q^\dagger=Q^2\), we have
\(c_\rho\in\{0,1\}\).  Thus
\begin{equation*}
   \mathcal I_{\lambda,\eta}
   =
   \bigoplus_{\rho:\,c_\rho=1}\mathcal E_\rho.
\end{equation*}
For every active branch
\(\rho\in\mathcal B(\lambda,\eta)\), both
\(A_{\rho\leftarrow\lambda}\) and
\(q_{\rho\leftarrow\lambda}\ket{\eta}\) are nonzero.  Hence there exist
\(f\in\mathcal H_d^*\) and \(x\in V_\lambda\) for which
\[
   \iota_\rho^{(s-1)}
   \left(
      A_{\rho\leftarrow\lambda}(f\otimes x)
      \otimes q_{\rho\leftarrow\lambda}\ket{\eta}
   \right)
   \neq0.
\]
This vector is the \(\rho\)-component of an element of
\(\mathcal I_{\lambda,\eta}\), so \(c_\rho=1\).  As this holds for every
active branch,
\begin{equation*}
   \mathcal I_{\lambda,\eta}
   =
   \mathcal D_{\lambda,\eta}
   =
   \bigoplus_{\mu\in\mathcal B(\lambda,\eta)}
      \iota_\mu^{(s-1)}
      \left(
         V_\mu\otimes
         \operatorname{span}\!\left\{
            q_{\mu\leftarrow\lambda}\ket{\eta}
         \right\}
      \right).
\end{equation*}
Together with
\(
   \mathcal I_{\lambda,\eta}
   =\operatorname{Red}_s(\mathcal S_{\lambda,\eta})
\), this proves~\eqref{eq:carrier-Red}.

We next prove part~\textup{(i)}.  For every active branch \(\rho\), put
\[
   \ket{\widehat\eta_\rho}
   :=
   \frac{q_{\rho\leftarrow\lambda}\ket{\eta}}
        {\lVert q_{\rho\leftarrow\lambda}\ket{\eta}\rVert}.
\]
Suppose first that
\(\mathcal B(\lambda,\eta)=\{\mu\}\).  Then
\[
   \operatorname{Red}_s(\mathcal S_{\lambda,\eta})
   =
   \iota_\mu^{(s-1)}
   \bigl(
      V_\mu\otimes\operatorname{span}\{\ket{\widehat\eta_\mu}\}
   \bigr).
\]
Every operator \(Y\) supported on this reduced subspace has a unique
representation
\[
   Y
   =
   \iota_\mu^{(s-1)}
   \bigl(
      X\otimes
      \ket{\widehat\eta_\mu}\!\bra{\widehat\eta_\mu}
   \bigr)
   \bigl(\iota_\mu^{(s-1)}\bigr)^\dagger
\]
with \(X\in\mathcal L(V_\mu)\).  Schur averaging gives
\begin{align*}
   \mathcal T^{(s-1)}(Y)
   &=
   \iota_\mu^{(s-1)}
   \left[
      \left(
         \int_{\mathrm U(d)}
            \pi_\mu(U)X\pi_\mu(U)^\dagger
            \,\mathrm d\mu_{\mathrm{Haar}}(U)
      \right)
      \otimes
      \ket{\widehat\eta_\mu}\!\bra{\widehat\eta_\mu}
   \right]
   \bigl(\iota_\mu^{(s-1)}\bigr)^\dagger
   \\
   &=
   \operatorname{Tr}(X)
   \iota_\mu^{(s-1)}
   \left[
      \frac{I_{V_\mu}}{\dim_{\mathbb C}V_\mu}
      \otimes
      \ket{\widehat\eta_\mu}\!\bra{\widehat\eta_\mu}
   \right]
   \bigl(\iota_\mu^{(s-1)}\bigr)^\dagger
   \\
   &=
   \operatorname{Tr}(Y)\,\sigma_\mu,
\end{align*}
where
\[
   \sigma_\mu
   :=
   \iota_\mu^{(s-1)}
   \left[
      \frac{I_{V_\mu}}{\dim_{\mathbb C}V_\mu}
      \otimes
      \ket{\widehat\eta_\mu}\!\bra{\widehat\eta_\mu}
   \right]
   \bigl(\iota_\mu^{(s-1)}\bigr)^\dagger.
\]
Thus the Haar twirl is a replacement channel on the reduced subspace,
so \(\mathcal S_{\lambda,\eta}\) is privacy-hereditary.

Conversely, suppose that
\(
   |\mathcal B(\lambda,\eta)|\geq2
\), and choose distinct
\(
   \mu,\nu\in\mathcal B(\lambda,\eta)
\).  Choose unit vectors \(v_\mu\in V_\mu\) and \(v_\nu\in V_\nu\), and
let
\[
   \ket{x_\mu}
   :=
   \iota_\mu^{(s-1)}
   \bigl(v_\mu\otimes\ket{\widehat\eta_\mu}\bigr),
   \qquad
   \ket{x_\nu}
   :=
   \iota_\nu^{(s-1)}
   \bigl(v_\nu\otimes\ket{\widehat\eta_\nu}\bigr).
\]
By~\eqref{eq:carrier-Red}, both vectors belong to
\(\operatorname{Red}_s(\mathcal S_{\lambda,\eta})\).  A second
application of Schur averaging yields
\begin{align*}
   \mathcal T^{(s-1)}
   \bigl(\ket{x_\mu}\!\bra{x_\mu}\bigr)
   &=
   \iota_\mu^{(s-1)}
   \left[
      \frac{I_{V_\mu}}{\dim_{\mathbb C}V_\mu}
      \otimes
      \ket{\widehat\eta_\mu}\!\bra{\widehat\eta_\mu}
   \right]
   \bigl(\iota_\mu^{(s-1)}\bigr)^\dagger,
   \\
   \mathcal T^{(s-1)}
   \bigl(\ket{x_\nu}\!\bra{x_\nu}\bigr)
   &=
   \iota_\nu^{(s-1)}
   \left[
      \frac{I_{V_\nu}}{\dim_{\mathbb C}V_\nu}
      \otimes
      \ket{\widehat\eta_\nu}\!\bra{\widehat\eta_\nu}
   \right]
   \bigl(\iota_\nu^{(s-1)}\bigr)^\dagger.
\end{align*}
The first state is supported in the \(\mu\)-th Schur-Weyl block and the second in the orthogonal \(\nu\)-th block.  They are nonzero trace-one operators and therefore cannot be equal.  Hence the Haar twirl is not constant on the density operators supported on \(\operatorname{Red}_s(\mathcal S_{\lambda,\eta})\).  The reduced subspace is not Haar-twirl admissible, and privacy heredity fails. This proves part~\textup{(i)}.

\smallskip
Now we prove part~\textup{(ii)}. Assume
\(
   \mathcal B(\lambda,\eta)=\{\mu\}
\).  Since \(\mathcal S_{\lambda,\eta}\) is invariant under the
collective action, the restriction of \(U^{\otimes s}\) to this carrier
is unitarily equivalent to \(\pi_\lambda(U)\), while the restriction of
\(U^{\otimes(s-1)}\) to the reduced carrier is unitarily equivalent to
\(\pi_\mu(U)\).

Let
\(
   \psi=\iota_\lambda^{(s)}(x\otimes\ket{\eta})
\)
and
\(
   \varphi=\iota_\lambda^{(s)}(y\otimes\ket{\eta})
\).  Because the upper carrier is invariant, only the orthogonal
projections of the ambient test vectors
\(a,b\in\mathcal H_d^{\otimes s}\) onto that carrier contribute to the
defining coefficient function.  If these projections are
\(
   \iota_\lambda^{(s)}(a_\lambda\otimes\ket{\eta})
\)
and
\(
   \iota_\lambda^{(s)}(b_\lambda\otimes\ket{\eta})
\), respectively, then
\begin{align*}
   f_{a,b;\psi,\varphi}^{(s)}(U)
   &=
   \langle a_\lambda,
      \pi_\lambda(U)x\rangle
   \langle y,
      \pi_\lambda(U)^\dagger b_\lambda\rangle
   \\
   &=
   \langle a_\lambda,
      \pi_\lambda(U)x\rangle
   \overline{
      \langle b_\lambda,
         \pi_\lambda(U)y\rangle
   }.
\end{align*}
The linear span of these products is the coefficient space of
the conjugation representation
\(
   \pi_\lambda\otimes\overline{\pi_\lambda}
\).  Conversely, every decomposable matrix coefficient of
\(
   \pi_\lambda\otimes\overline{\pi_\lambda}
\)
is obtained by choosing the four vectors
\(a_\lambda,b_\lambda,x,y\), and decomposable tensors span the tensor
product.  Hence
\begin{equation*}
   \mathfrak C_s(\mathcal S_{\lambda,\eta})
   =
   \operatorname{Coeff}
   \bigl(
      V_\lambda\otimes\overline{V_\lambda}
   \bigr).
\end{equation*}
Applying the same argument to the lower carrier gives
\begin{equation*}
   \mathfrak C_{s-1}
   \bigl(
      \operatorname{Red}_s(\mathcal S_{\lambda,\eta})
   \bigr)
   =
   \operatorname{Coeff}
   \bigl(
      V_\mu\otimes\overline{V_\mu}
   \bigr).
\end{equation*}
By Theorem~\ref{thm:design-heredity-criterion}, design heredity is
equivalent to the inclusion of the lower coefficient space in the upper
coefficient space.  Substituting the preceding two identities gives
\[
   \operatorname{Coeff}
   \bigl(
      V_\mu\otimes\overline{V_\mu}
   \bigr)
   \subseteq
   \operatorname{Coeff}
   \bigl(
      V_\lambda\otimes\overline{V_\lambda}
   \bigr),
\]
which proves part~\textup{(ii)}.

Finally, by definition, \textsf{QPB}-heredity is the conjunction of privacy
heredity and design heredity.  Part~\textup{(i)} shows that privacy
heredity is equivalent to
\(
   |\mathcal B(\lambda,\eta)|=1
\), and, under this condition, part~\textup{(ii)} gives the exact
criterion for design heredity.  This proves part~\textup{(iii)} and
completes the proof.
\end{proof}

\begin{corollary}[Hereditary paths in the Young graph]
\label{cor:young-yamanouchi-heredity}
Let
\(
    |\eta\rangle\in W_\lambda
\)
be a Young--Yamanouchi basis vector.  Under the standard chain
\[
    \mathfrak{S}_1\subset \mathfrak{S}_2\subset\cdots\subset \mathfrak{S}_s,
\]
the vector \(|\eta\rangle\) determines a unique path
\[
    \lambda^{(s)}
    =
    \lambda
    \searrow
    \lambda^{(s-1)}
    \searrow
    \cdots
    \searrow
    \lambda^{(1)}
    =
    (1)
\]
in the Young graph.  The associated sequence of canonical carriers,
obtained by successively contracting the last tensor factor, is fully
privacy-hereditary along this path.

It is fully \textsf{QPB}-hereditary if and only if
\begin{equation}
\label{eq:Young-path-coeff-condition}
\begin{split}
    \operatorname{Coeff}
    \left(
       V_{\lambda^{(r-1)}}
       \otimes
       \overline{V_{\lambda^{(r-1)}}}
    \right)
    \subseteq
    \operatorname{Coeff}
    \left(
       V_{\lambda^{(r)}}
       \otimes
       \overline{V_{\lambda^{(r)}}}
    \right)
\end{split}
\end{equation}
for every
\(
    2\leq r\leq s
\).
\end{corollary}

\begin{proof}
In a Young--Yamanouchi basis, restriction from \(\mathfrak{S}_r\) to \(\mathfrak{S}_{r-1}\) removes the box containing \(r\). Hence the basis vector has exactly one nonzero component in the multiplicity-free branching decomposition at every step. The privacy assertion follows by successive application of
Theorem~\ref{thm:canonical-carrier-branching}.
The design assertion follows from the coefficient-space condition in the same theorem at every step.
\end{proof}

\subsubsection{Deletion of an arbitrary recipient}
\label{subsubsec:any-recipient}

Let $s\ge2$ and use
\[
 \Gamma_{\mathfrak{S}_{t}}(\pi)(x_1\otimes\cdots\otimes x_s)
 =x_{\pi^{-1}(1)}\otimes\cdots\otimes x_{\pi^{-1}(s)}
 \qquad(\pi\in\mathfrak S_s).
\]
For $1\le j<s$, put $\tau_j=(j\,s)$, and put $\tau_s=e$.
Define the permutation unitary $Q_j$ on
$\mathcal H_d^{\otimes(s-1)}$ by
\[
 Q_j(y_1\otimes\cdots\otimes y_{s-1})
 =y_1\otimes\cdots\otimes y_{j-1}\otimes y_{s-1}
  \otimes y_j\otimes\cdots\otimes y_{s-2}
 \qquad(1\le j<s),
\]
where empty tensor strings are omitted, and set
$Q_s=I_{\mathcal H_d^{\otimes(s-1)}}$.
Thus $Q_{s-1}$ is also the identity. Let
\[
 \mathsf J_s=\sum_{q=1}^{s-1}(q\,s),\qquad
 \mathsf J^{(j)}=\tau_j\mathsf J_s\tau_j
               =\sum_{q\ne j}(q\,j)
 \in\mathbb C[\mathfrak S_s].
\]
The representation $W_\lambda$ is extended linearly to the group algebra.
For a removable box $\lambda\setminus\mu$ in row $r$ and column $c$,
write $\kappa_{\lambda/\mu}=c-r$ for its content.

\begin{theorem}[Deletion of the $j$th recipient]
\label{thm:any-recipient}
Let $\lambda\in\Lambda(d,s)$, let $|\eta\rangle\in W_\lambda$
be a unit vector, and let $1\le j\le s$. Let
\[
 |\eta_j\rangle=W_\lambda(\tau_j)|\eta\rangle,
 \qquad
 \mathcal S_{\lambda,\eta}
 =\iota^{(s)}_\lambda
   \bigl(V_\lambda\otimes\operatorname{span}_{\mathbb{C}}\{|\eta\rangle\}\bigr).
\]
Then the following assertions hold.
\begin{enumerate}
\item The canonical reduction is
\begin{align}
 \operatorname{Red}_j(\mathcal S_{\lambda,\eta})
 &=Q_j^\dagger\operatorname{Red}_s
        (\mathcal S_{\lambda,\eta_j})\notag\\
 &=Q_j^\dagger
   \bigoplus_{\mu\in\mathcal B(\lambda,\eta_j)}
   \iota^{(s-1)}_\mu
   \bigl(V_\mu\otimes
     \operatorname{span}_{\mathbb{C}}\{q_{\mu\leftarrow\lambda}|\eta_j\rangle\}\bigr).
 \label{eq:any-recipient-reduction}
\end{align}
\item The following conditions are equivalent:
\begin{enumerate}
\item $\mathcal S_{\lambda,\eta}$ is privacy-hereditary in coordinate $j$;
\item $|\mathcal B(\lambda,\eta_j)|=1$;
\item $|\eta\rangle$ is an eigenvector of
      $W_\lambda(\mathsf J^{(j)})$.
\end{enumerate}
If the unique active branch is $\mu$, the eigenvalue in the last condition
is $\kappa_{\lambda/\mu}$.
\item If $\mathcal B(\lambda,\eta_j)=\{\mu\}$, then
$\mathcal S_{\lambda,\eta}$ is design-hereditary in coordinate $j$
if and only if
\[
 \operatorname{Coeff}(V_\mu\otimes\overline{V_\mu})
 \subseteq
 \operatorname{Coeff}(V_\lambda\otimes\overline{V_\lambda}).
\]
Under this single-branch hypothesis, the same condition characterizes \textsf{QPB}-heredity in coordinate $j$.
\item The carrier $\mathcal S_{\lambda,\eta}$ is privacy-hereditary in every coordinate if and only if $|\eta\rangle$ is a common eigenvector of $W_\lambda(\mathsf J^{(1)}),\ldots,W_\lambda(\mathsf J^{(s)})$.
\end{enumerate}
\end{theorem}

\begin{proof}
Schur--Weyl equivariance gives
\[
 \Gamma_{\mathfrak{S}_{t}}(\tau_j)\iota^{(s)}_\lambda
 =\iota^{(s)}_\lambda
   \bigl(I_{V_\lambda}\otimes W_\lambda(\tau_j)\bigr),
 \qquad
 \Gamma_{\mathfrak{S}_{t}}(\tau_j)\mathcal S_{\lambda,\eta}
 =\mathcal S_{\lambda,\eta_j}.
\]
For $j<s$, applying the two sides of the following identity to a
decomposable tensor gives
\begin{align*}
 C^{(s)}_{s,v}\Gamma_{\mathfrak{S}_{t}}(\tau_j)(x_1\otimes\cdots\otimes x_s)
 &=\langle v|x_j\rangle\,
   x_1\otimes\cdots\otimes x_{j-1}\otimes x_s
   \otimes x_{j+1}\otimes\cdots\otimes x_{s-1}\\
 &=Q_j C^{(s)}_{j,v}(x_1\otimes\cdots\otimes x_s).
\end{align*}
Consequently, by linearity in the tensor input,
\[
 C^{(s)}_{s,v}\Gamma_{\mathfrak{S}_{t}}(\tau_j)=Q_j C^{(s)}_{j,v}.
\]
The same identity is immediate for $j=s$. Taking the linear spans of all
contractions of vectors in $\mathcal S_{\lambda,\eta}$ yields
\[
 \operatorname{Red}_s(\mathcal S_{\lambda,\eta_j})
 =Q_j\operatorname{Red}_j(\mathcal S_{\lambda,\eta}).
\]
The last-coordinate branching theorem now gives
\eqref{eq:any-recipient-reduction}.

Now we show that the permutation conjugation preserves the relevant privacy and design properties. If a unitary $Q$ on $\mathcal H_d^{\otimes r}$ commutes with every $U^{\otimes r}$, then
\begin{align*}
 \mathcal T^{(r)}(QXQ^\dagger)
 &=Q\mathcal T^{(r)}(X)Q^\dagger,\\
 \Phi^{(r)}_{\mathcal U,w}(QXQ^\dagger)
 &=Q\Phi^{(r)}_{\mathcal U,w}(X)Q^\dagger
\end{align*}
for every weighted unitary ensemble $(\mathcal U,w)$ and every ambient operator $X$. Thus a subspace $\mathcal R$ is Haar-twirl admissible if and only if $Q\mathcal R$ is, with replacement states related by
$\sigma_{Q\mathcal R}=Q\sigma_{\mathcal R}Q^\dagger$.
Moreover, a fixed ensemble is exact on $\mathcal R$ if and only if the same ensemble is exact on $Q\mathcal R$.

The corresponding coefficient spaces are also equal. For the
coefficient functions defining $\mathfrak C_r$,
\[
 f^{(r)}_{a,b;Q\psi,Q\phi}(U)
 =f^{(r)}_{Q^\dagger a,Q^\dagger b;\psi,\phi}(U),
\]
and the ambient test vectors $a,b$ range over the whole tensor power.
Applying these observations to $\Gamma_{\mathfrak{S}_{t}}(\tau_j)$ at level $s$ and $Q_j$ at level $s-1$ proves that privacy heredity and design heredity in coordinate $j$ for $\mathcal S_{\lambda,\eta}$ are equivalent to the respective last-coordinate properties for $\mathcal S_{\lambda,\eta_j}$. The canonical-carrier branching theorem therefore proves the equivalence of the first two conditions in item~2 and proves item~3, including its \textsf{QPB} assertion.

Under the fixed multiplicity-free branching identification, the last Jucys-Murphy element satisfies
\[
 W_\lambda(\mathsf J_s)
 =\sum_{\mu\nearrow\lambda}
   \kappa_{\lambda/\mu}\,
   q_{\mu\leftarrow\lambda}^{\dagger}q_{\mu\leftarrow\lambda}.
\]
In fact, on a Young-Yamanouchi vector the box occupied by $s$ is
$\lambda\setminus\mu$, so its eigenvalue is the content of that box.
The contents of distinct removable boxes are distinct: if their rows are $r<r'$, their contents differ by
\[
 (\lambda_r-r)-(\lambda_{r'}-r')
 =\lambda_r-\lambda_{r'}+(r'-r)>0.
\]
Therefore the eigenspaces of $W_\lambda(\mathsf J_s)$ are exactly its individual removable-box branch subspaces. It follows that
\[
 |\mathcal B(\lambda,\eta_j)|=1
 \quad\Longleftrightarrow\quad
 |\eta_j\rangle\text{ is an eigenvector of }W_\lambda(\mathsf J_s).
\]
Since $\tau_j^{-1}=\tau_j$, conjugating this eigenvector equation gives
\[
 W_\lambda(\mathsf J^{(j)})|\eta\rangle
 =\kappa_{\lambda/\mu}|\eta\rangle
\]
when the unique branch is $\mu$. Here
\[
 \tau_j\mathsf J_s\tau_j
 =\sum_{q<s}(\tau_j(q)\,\tau_j(s))
 =\sum_{q\ne j}(q\,j).
\]
This proves item~2. Item~4 follows by applying item~2 in every coordinate.
\end{proof}

\begin{corollary}[Rectangular shapes and the shape $(2,1)$]
\label{cor:rectangular}
Let $s\ge2$, let $\lambda\in\Lambda(d,s)$, and let
$|\eta\rangle\in W_\lambda$ be a unit vector.
\begin{enumerate}
\item Suppose $\lambda=(a^b)$ with $a,b\ge1$ and $ab=s$.
Its unique removable-box shape is
\[
 \mu=(a^{b-1},a-1),
\]
where a zero final part is omitted. Then
$\mathcal S_{\lambda,\eta}$ is privacy-hereditary in every coordinate and
\[
 \operatorname{Red}_j(\mathcal S_{\lambda,\eta})
 =Q_j^\dagger\iota^{(s-1)}_\mu
   \bigl(V_\mu\otimes
     \operatorname{span}_{\mathbb{C}}\{q_{\mu\leftarrow\lambda}|\eta_j\rangle\}\bigr).
\]
It is design-hereditary in at least one coordinate if and only if it is
design-hereditary in every coordinate, and these conditions hold
when
\[
 \operatorname{Coeff}(V_\mu\otimes\overline{V_\mu})
 \subseteq\operatorname{Coeff}(V_\lambda\otimes\overline{V_\lambda}).
\]
The same criterion characterizes \textsf{QPB}-heredity in any or all coordinates.
Furthermore,
\[
 W_\lambda(\mathsf J^{(j)})=(a-b) I_{W_\lambda}
 \qquad(1\le j\le s).
\]
The cases $b=1$ and $a=1$ are the symmetric and exterior towers, respectively.
\item Suppose $\lambda=(2,1)$, so $s=3$ and $d\ge2$.
No canonical carrier $\mathcal S_{\lambda,\eta}$ is privacy-hereditary in all three coordinates. Each Young-Yamanouchi carrier is privacy-hereditary in coordinate $3$, and in neither coordinate $1$
nor coordinate $2$.
\end{enumerate}
\end{corollary}

\begin{proof}
For a rectangle, only its bottom-right box is removable. Hence restriction
to $\mathfrak S_{s-1}$ has just the branch $W_\mu$, and
$q_{\mu\leftarrow\lambda}$ is a unitary identification with that branch.
In particular, its value on every unit $|\eta_j\rangle$ is nonzero.
The theorem gives the reduced-space formula and privacy heredity in
every coordinate. Its design criterion depends only on $\lambda$ and
$\mu$, not on $j$ or $\eta$, proving the remaining equivalences.
The unique removable box has content $a-b$, so
$W_\lambda(\mathsf J_s)=(a-b) I$; conjugating by
$W_\lambda(\tau_j)$ proves the displayed scalar identity.

For $\lambda=(2,1)$, the central element
\[
 Z_3=(1\,2)+(1\,3)+(2\,3)=\mathsf J_2+\mathsf J_3
\]
acts by the total content of the diagram, namely $0+1-1=0$.
Consequently
\[
 W_\lambda(\mathsf J^{(1)})=-W_\lambda((2\,3)),\quad
 W_\lambda(\mathsf J^{(2)})=-W_\lambda((1\,3)),\quad
 W_\lambda(\mathsf J^{(3)})=-W_\lambda((1\,2)).
\]
A vector witnessing privacy heredity in every coordinate would therefore
span a common eigenvector of $W_\lambda((1\,2))$ and
$W_\lambda((2\,3))$. Since these transpositions generate $\mathfrak S_3$, that vector would be invariant under the entire group, contradicting irreducibility of the two-dimensional $W_{(2,1)}$.

Let $T_1$ have first row $1,2$ and second row $3$, and let $T_2$ have
first row $1,3$ and second row $2$. In a compatible Young orthogonal
phase convention,
\[
 W_\lambda((1\,2))=
 \begin{pmatrix}1&0\\0&-1\end{pmatrix},\qquad
 W_\lambda((2\,3))=
 \begin{pmatrix}-\tfrac12&\tfrac{\sqrt3}{2}\\
                 \tfrac{\sqrt3}{2}&\tfrac12\end{pmatrix}.
\]
Thus each $|T_i\rangle$ is an eigenvector of $W_\lambda((1\,2))$
but not of $W_\lambda((2\,3))$. Since
\[
 W_\lambda((1\,3))=
 W_\lambda((1\,2))W_\lambda((2\,3))W_\lambda((1\,2)),
\]
neither is an eigenvector of $W_\lambda((1\,3))$. The theorem now gives
the assertions about the three individual coordinates.
\end{proof}

\begin{example}[Design heredity without privacy heredity]\label{ex:design-without-privacy}
Let $d=3$, $s=3$, $\lambda=(2,1)$, and
\[
 |\eta\rangle=a|T_1\rangle+b|T_2\rangle,\qquad
 |a|^2+|b|^2=1,\qquad ab\ne0,
\]
where $T_1,T_2$ are the tableaux specified above.
Their penultimate shapes are $(2)$ and $(1,1)$, respectively. Hence
$\mathcal B(\lambda,\eta)=\{(2),(1,1)\}$, and the canonical-carrier
branching theorem gives
\begin{align*}
 \mathcal R:=\operatorname{Red}_3(\mathcal S_{\lambda,\eta})
 &=\iota^{(2)}_{(2)}
   \bigl(V_{(2)}\otimes
      \operatorname{span}_{\mathbb{C}}\{q_{(2)\leftarrow\lambda}|\eta\rangle\}\bigr)\\
 &\quad\oplus\iota^{(2)}_{(1,1)}
   \bigl(V_{(1,1)}\otimes
      \operatorname{span}_{\mathbb{C}}\{q_{(1,1)\leftarrow\lambda}|\eta\rangle\}\bigr)
 =\mathcal H_3^{\otimes2}.
\end{align*}
The last equality holds because both level-two multiplicity spaces are one dimensional, both branch vectors are nonzero, and $\dim_{\mathbb{C}} V_{(2)}+\dim_{\mathbb{C}} V_{(1,1)}=6+3=9$.
In particular $\mathcal R$ is not Haar-twirl admissible: the normalized Haar outputs of a symmetric and an antisymmetric unit vector have orthogonal supports. Thus privacy heredity fails.

We verify design heredity as an equality of function spaces on
$\mathrm U(3)$. Write $[\gamma_1,\gamma_2,\gamma_3]$ for the irreducible
$\mathrm U(3)$-module of that highest weight, and set
\[
 A=[1,0,-1],\quad D=[2,-1,-1],\quad
 \overline D=[1,1,-2],\quad E=[2,0,-2],
 \qquad\mathbf1=[0,0,0].
\]
Here $A$ is the eight-dimensional conjugation representation on
traceless matrices, so $\overline A\cong A$, and
$V_{(2,1)}\cong\det\otimes A$.

The four level-two tensor products decompose as
\begin{align*}
 V_{(2)}\otimes\overline{V_{(2)}}
   &\cong\mathbf1\oplus A\oplus E,\\
 V_{(1,1)}\otimes\overline{V_{(1,1)}}
   &\cong\mathbf1\oplus A,\\
 V_{(2)}\otimes\overline{V_{(1,1)}}
   &\cong D\oplus A,\\
 V_{(1,1)}\otimes\overline{V_{(2)}}
   &\cong\overline D\oplus A.
\end{align*}
These follow from Pieri's rule and determinant twists.
We have 
\[
 \overline{V_{(2)}}\cong(\det)^{-2}\otimes\mathbb S_{(2,2)}(\mathcal H_3),
 \qquad
 \overline{V_{(1,1)}}\cong(\det)^{-1}\otimes\mathcal H_3,
\]
where $\mathbb S_\nu(\mathcal H_3)$ is the Schur module. The relevant Pieri decompositions in dimension three are
\begin{align*}
 \mathbb S_{(2)}\otimes\mathbb S_{(2,2)}
 &\cong\mathbb S_{(4,2)}\oplus\mathbb S_{(3,2,1)}
       \oplus\mathbb S_{(2,2,2)},\\
 \mathbb S_{(1,1)}\otimes\mathcal H_3
 &\cong\mathbb S_{(2,1)}\oplus\mathbb S_{(1,1,1)},\\
 \mathbb S_{(2)}\otimes\mathcal H_3
 &\cong\mathbb S_{(3)}\oplus\mathbb S_{(2,1)},
\end{align*}
where Schur functors without an argument are evaluated at $\mathcal H_3$.
Subtracting $(2,2,2)$ in the first line and $(1,1,1)$ in the latter
two lines yields the first three tensor-product decompositions above;
the fourth is the conjugate of the third.

Since $\mathcal R=\mathcal H_3^{\otimes2}$, these decompositions give
\[
 \mathcal R\otimes\overline{\mathcal R}
 \cong\mathbf1^{\oplus2}\oplus A^{\oplus4}
       \oplus D\oplus\overline D\oplus E.
\]
They also determine the upper conjugation representation. We have $\mathcal H_3\otimes\overline{\mathcal H_3}\cong\mathbf1\oplus A$, so rearranging tensor factors gives
\[
 \mathcal R\otimes\overline{\mathcal R}
 \cong(\mathbf1\oplus A)\otimes(\mathbf1\oplus A)
 \cong\mathbf1\oplus A^{\oplus2}\oplus(A\otimes A).
\]
Uniqueness of irreducible multiplicities therefore yields
\[
 V_{(2,1)}\otimes\overline{V_{(2,1)}}
 \cong A\otimes A
 \cong\mathbf1\oplus A^{\oplus2}\oplus D\oplus\overline D\oplus E.
\]

For an invariant subspace carrying a representation $\rho$, its
restricted coefficient space is $\operatorname{Coeff}
(\rho\otimes\overline\rho)$. Moreover, repeated copies of an
irreducible representation do not enlarge its coefficient space.
Consequently the two representations just computed have the same
coefficient space, and
\[
 \mathfrak C_2(\mathcal R)
 =\mathfrak C_3(\mathcal S_{\lambda,\eta})
 =\operatorname{Coeff}(\mathbf1)
  \oplus\operatorname{Coeff}(A)
  \oplus\operatorname{Coeff}(D)
  \oplus\operatorname{Coeff}(\overline D)
  \oplus\operatorname{Coeff}(E).
\]
These sums are orthogonal in $L^2(\mathrm U(3),\mu_{\mathrm{Haar}})$.
For $\gamma_1\ge\gamma_2\ge\gamma_3$, Weyl's dimension formula is
\[
 \dim_{\mathbb{C}}[\gamma_1,\gamma_2,\gamma_3]
 =\frac{(\gamma_1-\gamma_2+1)(\gamma_2-\gamma_3+1)
              (\gamma_1-\gamma_3+2)}2.
\]
It gives $\dim A=8$, $\dim D=\dim\overline D=10$, and $\dim E=27$.
By Peter-Weyl orthogonality,
\[
 \dim_{\mathbb{C}}\mathfrak C_2(\mathcal R)
 =\dim_{\mathbb{C}}\mathfrak C_3(\mathcal S_{\lambda,\eta})
 =1^2+8^2+10^2+10^2+27^2=994.
\]
The equality of coefficient spaces along with dimensions implies by Theorem~\ref{thm:design-heredity-criterion} that
$\mathcal S_{\lambda,\eta}$ is design-hereditary in coordinate $3$.
It is therefore design-hereditary but not privacy-hereditary in that
coordinate. In fact the argument applies to every coordinate: the
upper coefficient space is independent of the multiplicity vector,
and each reduced coefficient space is contained in
$\mathfrak C_2(\mathcal H_3^{\otimes2})$.
\end{example}

\subsubsection{Replacement states under privacy-hereditary reduction}
\label{subsubsec:replacement-state-reduction}

The replacement state of the reduced subspace is related canonically to
the replacement state of the original subspace.

\begin{proposition}[Partial trace of the replacement state]
\label{prop:replacement-state-partial-trace}
Let
\(
    \mathcal S
    \subset
    \mathcal H_d^{\otimes s}
\)
be Haar-twirl admissible with replacement state
\(
    \sigma_{\mathcal S}
\).
Assume that
\(
    \mathcal S
\)
is privacy-hereditary in coordinate \(j\), and let
\(
    \mathcal R_j
    =
    \operatorname{Red}_j(\mathcal S)
\)
have replacement state
\(
    \sigma_{\mathcal R_j}
\).
Then
\begin{equation}
\label{eq:replacement-state-partial-trace}
       \sigma_{\mathcal R_j}
       =
       \operatorname{Tr}_j
       \left(
          \sigma_{\mathcal S}
       \right).
\end{equation}
\end{proposition}

\begin{proof}
Let
\(
    m
    :=
    \dim_{\mathbb C}\mathcal S.
\)
Then
\(
    \rho_{\mathcal S}^{\,0}
    :=
    \frac{\Pi_{\mathcal S}}{m}
    \in
    \mathcal D(\mathcal S),
\)
and hence
\(
    \mathcal T^{(s)}
    \left(
       \rho_{\mathcal S}^{\,0}
    \right)
    =
    \sigma_{\mathcal S}.
\)
Applying the \(j\)-th partial trace and using
Lemma~\ref{lem:partial-trace-covariance}, we obtain
\begin{align}
    \operatorname{Tr}_j
       (\sigma_{\mathcal S})
    &=
    \operatorname{Tr}_j
    \left[
       \mathcal T^{(s)}
       \left(
          \frac{\Pi_{\mathcal S}}{m}
       \right)
    \right]
    \notag\\
    &=
    \mathcal T^{(s-1)}
    \left(
       \frac{
          \operatorname{Tr}_j \Pi_{\mathcal S}
       }{
          m
       }
    \right).
\label{eq:replacement-trace-computation}
\end{align}
Now
\[
    \operatorname{Tr}
    \left(
       \frac{
          \operatorname{Tr}_j \Pi_{\mathcal S}
       }{
          m
       }
    \right)
    =
    \frac{
       \operatorname{Tr}(\Pi_{\mathcal S})
    }{
       m
    }
    =
    1,
\]
and
\[
    \operatorname{supp}
    \left(
       \frac{
          \operatorname{Tr}_j \Pi_{\mathcal S}
       }{
          m
       }
    \right)
    =
    \operatorname{Red}_j(\mathcal S)
    =
    \mathcal R_j.
\]
Thus
\[
    \frac{
       \operatorname{Tr}_j \Pi_{\mathcal S}
    }{
       m
    }
    \in
    \mathcal D(\mathcal R_j).
\]
Since
\(
    \mathcal R_j
\)
is Haar-twirl admissible,
\[
    \mathcal T^{(s-1)}
    \left(
       \frac{
          \operatorname{Tr}_j \Pi_{\mathcal S}
       }{
          m
       }
    \right)
    =
    \sigma_{\mathcal R_j}.
\]
Combining this with
\eqref{eq:replacement-trace-computation}
proves~\eqref{eq:replacement-state-partial-trace}.
\end{proof}

\subsubsection{\textsf{QPB}-security under hereditary dimensional reduction}\label{subsubsec:QPB-dimensional-reduction}

We finally state the operational consequence of the preceding definitions. The reduced protocol is defined directly on the canonical reduced subspace.

\begin{theorem}[\textsf{QPB}-security preservation under hereditary reduction]\label{thm:QPB-hereditary-security}
Let
\(
    \mathcal S\subset \mathcal H_d^{\otimes s}
\)
be \textsf{QPB}-hereditary in coordinate \(j\), and let
\[
    \mathcal R_j
    :=
    \operatorname{Red}_j(\mathcal S).
\]
Let
\(
    (\mathfrak{U},w)
\),
with
\(
    \mathfrak{U}=\{U_k\}_{k\in K}
\),
be an exact unitary
\(
    (\mathcal S,s)
\)-design.

For the reduced \((s-1)\)-recipient protocol, define
\begin{equation}
\label{eq:reduced-encryption-correct}
    \operatorname{Enc}_k^{(s-1)}(\rho)
    :=
    U_k^{\otimes(s-1)}
    \mathcal{J}_{\mathcal{R}_{j}}(\rho)
    (U_k^\dagger)^{\otimes(s-1)},
    \qquad
    \rho\in\mathcal D(\mathcal R_j),
\end{equation}
with the same local decryption map
\begin{equation}
\label{eq:reduced-local-decryption}
    \operatorname{Dec}_k(\gamma)
    :=
    U_k^\dagger\gamma U_k.
\end{equation}
Then:

\begin{enumerate}[label=\textup{(\roman*)}]
\item
the reduced protocol is perfectly correct:
\begin{equation}
\label{eq:reduced-perfect-correctness}
\operatorname{Dec}_k^{\otimes(s-1)}
\circ
\operatorname{Enc}_k^{(s-1)}
=
\mathcal{J}_{\mathcal{R}_{j}}.
\end{equation}

\item the averaged reduced encryption is the replacement channel
\begin{equation}
\label{eq:reduced-perfect-security}
     \mathbb{E}_{k\in \mathcal{K}}\left[\operatorname{Enc}^{(s-1)}_{k}(\rho)\right]
    =
    \sigma_{\mathcal R_j}
    \qquad
    \forall
    \rho\in\mathcal D(\mathcal R_j),
\end{equation}
where
\(
\label{eq:reduced-sigma-trace}
    \sigma_{\mathcal R_j}
    =
    \operatorname{Tr}_j
       (\sigma_{\mathcal S});
\)

\item
the same unitary ensemble, the same weights, and therefore the same
key set are used at levels \(s\) and \(s-1\).
\end{enumerate}
\end{theorem}

\begin{proof}
For every
\(
    \rho\in\mathcal D(\mathcal R_j)
\),
\begin{align}
    \operatorname{Dec}_k^{\otimes(s-1)}
    \left(
       \operatorname{Enc}_k^{(s-1)}(\rho)
    \right)
    &=
    (U_k^\dagger)^{\otimes(s-1)}
    \left[
       U_k^{\otimes(s-1)}
       \rho
       (U_k^\dagger)^{\otimes(s-1)}
    \right]
    U_k^{\otimes(s-1)}
    \notag\\
    &=
    \rho,
\end{align}
which proves perfect correctness.

Since \(\mathcal S\) is design-hereditary,
the exact
\(
    (\mathcal S,s)
\)-design
\(
    (\mathfrak{U},w)
\)
is also an exact
\(
    (\mathcal R_j,s-1)
\)-design. Therefore
\begin{align}
     \mathbb{E}_{k\in \mathcal{K}}\left[\operatorname{Enc}^{(s-1)}_{k}(\rho)\right]
     &=\sum_{k\in K}w_k\operatorname{Enc}_k^{(s-1)}(\rho)\notag\\
    &=\Phi_{\mathfrak{U},w}^{(s-1)}(\rho)
    \notag\\
    &=\mathcal T^{(s-1)}(\rho).
\end{align}
Since \(\mathcal S\) is also privacy-hereditary,
\(
    \mathcal R_j
\)
is Haar-twirl admissible, and hence
\[
    \mathcal T^{(s-1)}(\rho)
    =
    \sigma_{\mathcal R_j}
    \qquad
    \forall
    \rho\in\mathcal D(\mathcal R_j).
\]
Finally,
Proposition~\ref{prop:replacement-state-partial-trace}
gives
\[
    \sigma_{\mathcal R_j}
    =
    \operatorname{Tr}_j
       (\sigma_{\mathcal S}).
\]
This proves all assertions.
\end{proof}

\begin{remark}\label{rem:approximate-heredity}
It is not always possible to convert an exact hereditary statement into an \(\varepsilon\)-approximate statement by keeping the same \(\varepsilon\).

Suppose, for example, that design descent is obtained from a covariant
lifting map
\[
    \mathscr{L}\colon
    \mathcal L(\mathcal R_j)
    \longrightarrow
    \mathcal L(\mathcal S)
\]
as in
Proposition~\ref{prop:covariant-lifting}.
If the restricted error is measured in the induced trace norm, then
the factorization
\[
    \Phi_{\mathfrak{U},w}^{(s-1)}
    -
    \mathcal T^{(s-1)}
    =
    \operatorname{Tr}_j
    \circ
    \left(
       \Phi_{\mathfrak{U},w}^{(s)}
       -
       \mathcal T^{(s)}
    \right)
    \circ
    \mathscr{L}
\]
gives the estimate
\begin{equation}
\label{eq:approximate-heredity-bound}
\begin{split}
    \left\|
       \Phi_{\mathfrak{U},w}^{(s-1)}
       -
       \mathcal T^{(s-1)}
    \right\|_{1\to1,\mathcal R_j}
    \leq
    \|\mathscr{L}\|_{1\to1}
    \left\|
       \Phi_{\mathfrak{U},w}^{(s)}
       -
       \mathcal T^{(s)}
    \right\|_{1\to1,\mathcal S},
\end{split}
\end{equation}
because the partial trace has induced trace norm at most one. Thus the lower approximation error is controlled by the norm of the lifting map. Preservation of the same error requires an additional estimate, for example
\(
    \|\mathscr{L}\|_{1\to1}\leq1
\).
\end{remark}

\section{Conclusion and future directions}
\label{sec:conclusion}
The purpose of this paper has been to understand quantum private
broadcast through the invariant structure of the collective unitary
twirl.  The main structural result,
Theorem~\ref{thm:qpb-classification},
gives a complete characterization of the subspaces
\(
    \mathcal S
    \subset
    \mathcal H_d^{\otimes t}
\)
on which
\(
    \mathcal T^{(t)}
\)
acts as a replacement channel.  In Schur-Weyl coordinates, privacy is equivalent to the simultaneous state-independence of the Schur-block probabilities and of the multiplicity-space marginals. The resulting class is substantially larger than the collection of subspaces contained in one isotypic component: the multi-block constructions of Proposition~\ref{prop:two-block-construction} and Theorem~\ref{thm:general-multiblock-construction} exhibit coherent admissible subspaces with support across several Schur-Weyl sectors.

Proposition~\ref{prop:classification-KL}
places this classification in the general theory of quantum error
correction.  The collective Haar twirl is a conditional expectation
onto
\[
    \mathfrak A_t
    =
    \bigoplus_{\lambda}
        I_{V_\lambda}\otimes\mathcal L(W_\lambda),
\]
and the Haar-twirl admissible subspaces are those on which every element of this algebra has scalar compression. They are the exactly correctable subspaces of a complementary channel of the twirl. The fixed block weights and fixed multiplicity marginals are therefore the explicit Schur-Weyl form of the corresponding Knill-Laflamme equations.

The second part of the paper develops a hereditary theory under deletion of tensor factors. The distinction between privacy heredity and design heredity is essential: the former is a property of the Haar channel, whereas the latter is a property of the moment constraints imposed on a finite ensemble.  Their conjunction is the condition required for the same key ensemble to remain a \textsf{QPB}-scheme after a recipient is removed. We obtained exact criteria for both properties, established full heredity of the symmetric tower, determined the sharp exterior-power threshold
\[
    s\leq\frac{d+1}{2},
\]
and described hereditary canonical carriers through paths in the Young graph and inclusions of Peter-Weyl coefficient spaces.

Several problems remain open.

\smallskip

\begin{enumerate}
\item Recent work~\cite{SchusterHaferkampHuang2025,CuiSchusterBrandaoHuang2025} studies low-depth constructions of full unitary designs, and
Schuster-Haferkamp-Huang prove a gluing lemma for full designs. This leaves open the problem of gluing local symmetric unitary designs directly.

The Cauchy decomposition
\[
    \operatorname{Sym}^{t}
    (\mathcal H_A\otimes\mathcal H_B)
    \cong
    \bigoplus_{\lambda\vdash t}
        \mathbb S_\lambda(\mathcal H_A)
        \otimes
        \mathbb S_\lambda(\mathcal H_B)
\]
shows why the known full-design proof does not immediately descend: global symmetry contains correlated non-one-row sectors that are not controlled by local symmetric moments.  A symmetric gluing theorem would have to control these sectors collectively or find a mechanism that bypasses their explicit resolution.

\smallskip
\item  Classification of arbitrary multi-block hereditary subspaces Theorem~\ref{thm:canonical-carrier-branching} gives a complete one-step classification for canonical single-copy carriers. A substantially harder problem is to describe
\(
    \operatorname{Red}_j(\mathcal S)
\)
for a general multi-block Haar-twirl admissible subspace of Theorem~\ref{thm:qpb-classification}. This requires to understand how the fixed multiplicity states in several \(\mathfrak{S}_s\)-blocks transform simultaneously under restriction to \(\mathfrak{S}_{s-1}\).
\end{enumerate}

\begin{appendices}
\section{Haar-twirl on Schur-Weyl decomposition}\label{secA1}

Let \(d,t\) be two positive integers, and 
\(
    \mathcal{H}^{\otimes t}_{d}:=(\mathbb C^d)^{\otimes t}
\). Consider a Schur-Weyl decomposition
\begin{equation}
\label{eq:appendix-Schur-Weyl-decomposition}
    \mathcal{H}^{\otimes t}_{d}
    \cong
    \bigoplus_{\lambda\in\Lambda(d,t)}
    (V_\lambda\otimes W_\lambda).
\end{equation}
For each \(\lambda\in\Lambda(d,t)\), let
\[
    \iota_\lambda \colon
    (V_\lambda\otimes W_\lambda)\longrightarrow\mathcal{H}^{\otimes t}_{d}
\]
be the canonical isometric inclusion and let
\begin{equation}
\label{eq:appendix-Schur-projector}
    \Pi_\lambda
    :=
    \iota_\lambda\iota_\lambda^\dagger
\end{equation}
be the orthogonal projection onto the \(\lambda\)-th Schur-Weyl block.
Thus
\begin{equation}
\label{eq:appendix-inclusion-relations}
    \iota_\lambda^\dagger\iota_\lambda
       =I_{(V_\lambda\otimes W_\lambda)},
    \qquad
    \Pi_\lambda\Pi_\mu
       =\delta_{\lambda\mu}\Pi_\lambda,
    \qquad
    \sum_{\lambda\in\Lambda(d,t)}\Pi_\lambda
       =I_{\mathcal{H}^{\otimes t}_{d}}.
\end{equation}
More explicitly, for arbitrary \(\lambda,\mu\in\Lambda(d,t)\),
\begin{equation}
\label{eq:appendix-cross-inclusion-relations}
    \iota_\lambda^\dagger\iota_\mu
    =
    \begin{cases}
       I_{(V_\lambda\otimes W_\lambda)},&\lambda=\mu,\\
       \mathbf{0},&\lambda\neq\mu.
    \end{cases}
\end{equation}

For \(U\in\mathrm U(d)\), the collective representation decomposes as
\begin{equation}
\label{eq:appendix-collective-representation-blocks}
    \Gamma_{\U(d)}(U)
    \cong
    \bigoplus_{\lambda\in\Lambda(d,t)}
    \left(V_{\lambda}(U)\otimes I_{W_\lambda}\right),
\end{equation}
where \(V_{\lambda}(U)\) denotes the irreducible representation of
\(\mathrm U(d)\) on \(V_\lambda\), and
\begin{equation}
\label{eq:appendix-intertwining-inclusion}
    \Gamma_{\U(d)}(U)\iota_\lambda
    =
    \iota_\lambda
    \left(V_{\lambda}(U)\otimes I_{W_\lambda}\right).
\end{equation}
The representations on \(V_\lambda\) and \(V_\mu\) are inequivalent when
\(\lambda\neq\mu\).

\begin{corollary}\label{cor:schur-integral}
Let
\[
\Gamma_{\U(d),V_\lambda} \colon \U(d)\longrightarrow\U(V_\lambda)
\]
be a finite-dimensional irreducible unitary representation, and let
\[
V_\lambda(U):=\Gamma_{\U(d),V_\lambda}(U).
\]
Let \(d\mu_{\mathrm{Haar}}\) be normalized Haar probability measure on
\(\U(d)\). Then, for every \(A\in\End(V_\lambda)\),
\begin{equation}
\int_{\U(d)}
V_\lambda(U)A V_\lambda(U)^\dagger\,
d\mu_{\mathrm{Haar}}(U)
=
\frac{\Tr(A)}{\dim V_\lambda}I_{V_\lambda}.
\label{eq:schur-integral}
\end{equation}
\end{corollary}

\begin{proof}
Define
\[
\Phi(A)
:=
\int_{\U(d)}
V_\lambda(U)A V_\lambda(U)^\dagger\,
d\mu_{\mathrm{Haar}}(U).
\]
For every \(W\in\U(d)\), left invariance of Haar measure gives
\begin{align*}
V_\lambda(W)\Phi(A)V_\lambda(W)^\dagger
&=
\int_{\U(d)}
V_\lambda(WU)A V_\lambda(WU)^\dagger\,
d\mu_{\mathrm{Haar}}(U) \\
&=
\int_{\U(d)}
V_\lambda(U)A V_\lambda(U)^\dagger\,
d\mu_{\mathrm{Haar}}(U) \\
&=
\Phi(A).
\end{align*}
Thus
\[
\Phi(A)\in\End_{\U(d)}(V_\lambda).
\]
Since \(V_\lambda\) is irreducible, Schur's lemma implies that
\[
\Phi(A)=c(A)I_{V_\lambda}
\]
for some \(c(A)\in\CC\). Taking traces and using the cyclicity of trace,
\[
c(A)\dim_{\mathbb{C}} V_\lambda
=
\Tr(\Phi(A))
=
\int_{\U(d)}
\Tr\bigl(V_\lambda(U)A V_\lambda(U)^\dagger\bigr)\,
d\mu_{\mathrm{Haar}}(U)
=
\Tr(A).
\]
Hence
\[
c(A)=\frac{\Tr(A)}{\dim V_\lambda},
\]
which proves \eqref{eq:schur-integral}.
\end{proof}

\begin{lemma}\label{lem:blockwise-collective-Haar-twirl}
Let
\[
    \mathcal T^{(t)}(X)
    :=
    \int_{\mathrm U(d)}
    U^{\otimes t}X(U^\dagger)^{\otimes t}\,d\mu_{\mathrm{Haar}}(U),
    \qquad
    X\in\mathcal L(\mathcal{H}^{\otimes t}_{d}),
\]
where \(d\mu_{\mathrm{Haar}}(U)\) is normalized Haar probability measure on \(\mathrm U(d)\).
For \(\lambda,\mu\in\Lambda(d,t)\), define the local block
\begin{equation}
\label{eq:appendix-X-lambda-mu}
    X_{\lambda\mu}
    :=
    \iota_\lambda^\dagger X\iota_\mu
    \in
    \mathcal L(V_\mu\otimes W_\mu,V_\lambda\otimes W_\lambda).
\end{equation}
In particular,
\[
    X_{\lambda\lambda}
    \in
    \mathcal L(V_\lambda\otimes W_\lambda).
\]
Then
\begin{equation}
\label{eq:appendix-twirl-direct-sum-form}
   \mathcal T^{(t)}(X)
    =
    \sum_{\lambda\in\Lambda(d,t)}
    \iota_\lambda
    \left[
       \frac{I_{V_\lambda}}{\dim_{\mathbb C}V_\lambda}
       \otimes
       \operatorname{Tr}_{V_\lambda}
       \!\left(X_{\lambda\lambda}\right)
    \right]
    \iota_\lambda^\dagger.
\end{equation}
Moreover,
\begin{equation}
\label{eq:appendix-compressed-block-identification}
    \Pi_\lambda X\Pi_\lambda
    =
    \iota_\lambda X_{\lambda\lambda}\iota_\lambda^\dagger,
\end{equation}
Consequently, \(\mathcal T^{(t)}\) is the trace-preserving, Hilbert--Schmidt orthogonal conditional expectation onto
\begin{equation}\label{eq:twirl-range}
  \operatorname{im}\mathcal T^{(t)}
  =\bigoplus_{\lambda\in\Lambda(d,t)}
    \iota_\lambda\bigl(I_{V_\lambda}\otimes
    \mathcal L(W_\lambda)\bigr)\iota_\lambda^\dagger.
\end{equation}
\end{lemma}

\begin{proof}
Using the mutually orthogonal projections in
\eqref{eq:appendix-inclusion-relations}, we have
\begin{align}
X
&=
I_{\mathcal{H}^{\otimes t}_{d}}X I_{\mathcal{H}^{\otimes t}_{d}}
\notag\\
&=
\left(
   \sum_{\lambda\in\Lambda(d,t)}\Pi_\lambda
\right)
X
\left(
   \sum_{\mu\in\Lambda(d,t)}\Pi_\mu
\right)
\notag\\
&=
\sum_{\lambda,\mu\in\Lambda(d,t)}
\Pi_\lambda X\Pi_\mu.
\label{eq:appendix-X-expanded-by-projectors}
\end{align}
For every pair \((\lambda,\mu)\), equations
\eqref{eq:appendix-Schur-projector} and
\eqref{eq:appendix-X-lambda-mu} give
\begin{align}
\Pi_\lambda X\Pi_\mu
&=
\left(\iota_\lambda\iota_\lambda^\dagger\right)
X
\left(\iota_\mu\iota_\mu^\dagger\right)
\notag\\
&=
\iota_\lambda
\left(
   \iota_\lambda^\dagger X\iota_\mu
\right)
\iota_\mu^\dagger
\notag\\
&=
\iota_\lambda X_{\lambda\mu}\iota_\mu^\dagger.
\label{eq:appendix-projector-block-as-inclusion}
\end{align}
Substituting \eqref{eq:appendix-projector-block-as-inclusion} into
\eqref{eq:appendix-X-expanded-by-projectors}, we obtain the full block-matrix
expansion
\begin{equation}
\label{eq:appendix-X-full-block-expansion}
    X
    =
    \sum_{\lambda,\mu\in\Lambda(d,t)}
    \iota_\lambda X_{\lambda\mu}\iota_\mu^\dagger.
\end{equation}
Taking \(\mu=\lambda\) in
\eqref{eq:appendix-projector-block-as-inclusion} proves
\eqref{eq:appendix-compressed-block-identification}.  Conversely,
\begin{align}
\iota_\lambda^\dagger
\left(\Pi_\lambda X\Pi_\lambda\right)
\iota_\lambda
&=
\iota_\lambda^\dagger
\left(
   \iota_\lambda X_{\lambda\lambda}
   \iota_\lambda^\dagger
\right)
\iota_\lambda
\notag\\
&=
\left(\iota_\lambda^\dagger\iota_\lambda\right)
X_{\lambda\lambda}
\left(\iota_\lambda^\dagger\iota_\lambda\right)
\notag\\
&=
I_{V_{\lambda} \otimes W_{\lambda}}
X_{\lambda\lambda}
I_{V_{\lambda} \otimes W_{\lambda}}
\notag\\
&=
X_{\lambda\lambda}.
\label{eq:appendix-compressed-block-recovered}
\end{align}

\smallskip
From \eqref{eq:appendix-intertwining-inclusion},
\begin{equation}
\label{eq:appendix-action-on-inclusion}
    U^{\otimes t}\iota_\lambda
    =
    \iota_\lambda
    \left(V_{\lambda}(U)\otimes I_{W_\lambda}\right).
\end{equation}
Taking adjoints of the corresponding identity for \(\mu\) gives
\begin{equation}
\label{eq:appendix-action-on-inclusion-adjoint}
    \iota_\mu^\dagger(U^\dagger)^{\otimes t}
    =
    \left(V_\mu(U)^\dagger\otimes I_{W_\mu}\right)
    \iota_\mu^\dagger.
\end{equation}
Therefore, using \eqref{eq:appendix-X-full-block-expansion},
\begin{align}
&U^{\otimes t}X(U^\dagger)^{\otimes t}
\notag\\
&\quad=
U^{\otimes t}
\left(
   \sum_{\lambda,\mu}
   \iota_\lambda X_{\lambda\mu}\iota_\mu^\dagger
\right)
(U^\dagger)^{\otimes t}
\notag\\
&\quad=
\sum_{\lambda,\mu}
\left(U^{\otimes t}\iota_\lambda\right)
X_{\lambda\mu}
\left(\iota_\mu^\dagger(U^\dagger)^{\otimes t}\right)
\notag\\
&\quad=
\sum_{\lambda,\mu}
\iota_\lambda
\left(V_{\lambda}(U)\otimes I_{W_\lambda}\right)
X_{\lambda\mu}
\left(V_{\mu}(U)^{\dagger}\otimes I_{W_\mu}\right)
\iota_\mu^\dagger.
\label{eq:appendix-conjugated-X-block-expansion}
\end{align}
Because the set \(\Lambda(d,t)\) is finite, we may integrate the finite sum
term by term.  Thus
\begin{align}
\mathcal T^{(t)}(X)
&=
\int_{\mathrm U(d)}
U^{\otimes t}X(U^\dagger)^{\otimes t}\,d\mu_{\mathrm{Haar}}U
\notag\\
&=
\sum_{\lambda,\mu}
\iota_\lambda
\Upsilon_{\lambda\mu}(X_{\lambda\mu})
\iota_\mu^\dagger,
\label{eq:appendix-twirl-before-Schur-lemma}
\end{align}
where
\begin{equation}
\label{eq:appendix-A-lambda-mu-definition}
\begin{split}
\Upsilon_{\lambda\mu}(X_{\lambda\mu})
:=
\int_{\mathrm U(d)}
&\left(V_{\lambda}(U)\otimes I_{W_\lambda}\right)
X_{\lambda\mu}\left(V_{\mu}(U)^{\dagger}\otimes I_{W_\mu}\right)
\,d\mu_{\mathrm{Haar}}(U).
\end{split}
\end{equation}

\smallskip
Fix \(\lambda\neq\mu\).  Choose orthonormal bases
\[
    \{e_a^{(\lambda)}\}_{a=1}^{\dim_{\mathbb{C}} W_\lambda}
    \quad\text{and}\quad
    \{e_b^{(\mu)}\}_{b=1}^{\dim_{\mathbb{C}} W_\mu}
\]
of \(W_\lambda\) and \(W_\mu\), respectively.  Then there exist uniquely
determined operators
\(
    Y_{ab}^{(\lambda,\mu)}
    \in\mathcal L(V_\mu,V_\lambda)
\)
such that
\begin{equation}
\label{eq:appendix-off-diagonal-multiplicity-expansion}
    X_{\lambda\mu}
    =
    \sum_{a=1}^{\dim_{\mathbb{C}} W_\lambda}
    \sum_{b=1}^{\dim_{\mathbb{C}} W_\mu}
    Y_{ab}^{(\lambda,\mu)}
    \otimes
    |e_a^{(\lambda)}\rangle
    \langle e_b^{(\mu)}|.
\end{equation}
Substitution into \eqref{eq:appendix-A-lambda-mu-definition} gives
\begin{align}
\Upsilon_{\lambda\mu}(X_{\lambda\mu})
&=
\sum_{a,b}
M_{ab}^{(\lambda,\mu)}
\otimes
|e_a^{(\lambda)}\rangle
\langle e_b^{(\mu)}|,
\label{eq:appendix-off-diagonal-average-components}
\end{align}
where
\begin{equation}
\label{eq:appendix-M-ab-definition}
    M_{ab}^{(\lambda,\mu)}
    :=
    \int_{\mathrm U(d)}
    V_\lambda(U)
    Y_{ab}^{(\lambda,\mu)}
    V_\mu(U)^\dagger\,d \mu_{\mathrm{Haar}}(U)
    \in\mathcal L(V_\mu,V_\lambda).
\end{equation}

We show that \(M_{ab}^{(\lambda,\mu)}\) intertwines the representations on
\(V_\mu\) and \(V_\lambda\).  Fix \(G\in\mathrm U(d)\).  By left invariance
of Haar measure, the change of variables \(U'=GU\) gives
\begin{align}
&V_\lambda(G)
M_{ab}^{(\lambda,\mu)}
V_\mu(G)^\dagger
\notag\\
&\quad=
\int_{\mathrm U(d)}
V_\lambda(G) V_\lambda(U)
Y_{ab}^{(\lambda,\mu)}
V_\mu(U)^\dagger V_\mu^\dagger(G)\,d\mu_{\mathrm{Haar}}(U)
\notag\\
&\quad=
\int_{\mathrm U(d)}
V_\lambda(GU)
Y_{ab}^{(\lambda,\mu)}
V_\mu^\dagger(GU)\,d\mu_{\mathrm{Haar}}U
\notag\\
&\quad=
\int_{\mathrm U(d)}
V_\lambda(U')
Y_{ab}^{(\lambda,\mu)}
V_\mu(U')^\dagger\,d_{\mathrm{Haar}}(U')
\notag\\
&\quad=
M_{ab}^{(\lambda,\mu)}.
\label{eq:appendix-M-ab-invariant}
\end{align}
Multiplying \eqref{eq:appendix-M-ab-invariant} on the right by \(V_\mu(G)\)
gives
\begin{equation}
\label{eq:appendix-M-ab-intertwiner}
    V_\lambda(G) M_{ab}^{(\lambda,\mu)}
    =
    M_{ab}^{(\lambda,\mu)}V_\mu(G)
    \qquad
    \text{for every }G\in\mathrm U(d).
\end{equation}
Thus \(M_{ab}^{(\lambda,\mu)}\) is an intertwining operator from the
irreducible representation \(V_\mu\) to the irreducible representation
\(V_\lambda\).  Since \(\lambda\neq\mu\), these representations are
inequivalent.  Schur's lemma therefore implies
\begin{equation}
\label{eq:appendix-M-ab-zero}
    M_{ab}^{(\lambda,\mu)}=0
    \qquad
    \text{for all }a,b.
\end{equation}
Using \eqref{eq:appendix-off-diagonal-average-components}, we conclude that
\begin{equation}
\label{eq:appendix-off-diagonal-average-zero}
    \Upsilon_{\lambda\mu}(X_{\lambda\mu})=0
    \qquad
    \text{whenever }\lambda\neq\mu.
\end{equation}

\smallskip
Fix \(\lambda\in\Lambda(d,t)\), and choose an orthonormal basis
\(
    \{e_a^{(\lambda)}\}_{a=1}^{\dim_{\mathbb{C}} W_\lambda}
\)
of \(W_\lambda\).  There exist uniquely determined operators
\(
    Z_{ab}^{(\lambda)}\in\mathcal L(V_\lambda)
\)
such that
\begin{equation}
\label{eq:appendix-diagonal-multiplicity-expansion}
    X_{\lambda\lambda}
    =
    \sum_{a,b=1}^{\dim_{\mathbb{C}} W_\lambda}
    Z_{ab}^{(\lambda)}
    \otimes
    |e_a^{(\lambda)}\rangle
    \langle e_b^{(\lambda)}|.
\end{equation}
Consequently,
\begin{align}
\Upsilon_{\lambda\lambda}(X_{\lambda\lambda})
&=
\int_{\mathrm U(d)}
\left(V_{\lambda}(U)\otimes I_{W_\lambda}\right)
X_{\lambda\lambda}
\left(V_{\lambda}(U)^{\dagger}\otimes I_{W_\lambda}\right)
\,d\mu_{\mathrm{Haar}}(U)
\notag\\
&=
\sum_{a,b}
\left(
   \int_{\mathrm U(d)}
   V_\lambda(U) Z_{ab}^{(\lambda)}V_\lambda^\dagger(U)\,d\mu_{\mathrm{Haar}}(U)
\right)
\otimes
|e_a^{(\lambda)}\rangle
\langle e_b^{(\lambda)}|.
\label{eq:appendix-diagonal-average-expanded}
\end{align}
Define
\begin{equation}
\label{eq:appendix-N-ab-definition}
    N_{ab}^{(\lambda)}
    :=
    \int_{\mathrm U(d)}
   V_\lambda(U) Z_{ab}^{(\lambda)}V_\lambda^\dagger(U)\,d\mu_{\mathrm{Haar}}(U).
\end{equation}
For every \(G\in\mathrm U(d)\), left invariance of Haar measure gives
\begin{align}
V_\lambda(G) N_{ab}^{(\lambda)}V_\lambda(G)^\dagger
&=
\int_{\mathrm U(d)}
V_\lambda(GU) Z_{ab}^{(\lambda)}V_\lambda(GU)^\dagger\,d\mu_{\mathrm{Haar}}(U)
\notag\\
&=
\int_{\mathrm U(d)}
V_\lambda(U') Z_{ab}^{(\lambda)}V_\lambda(U')^\dagger\,d\mu_{\mathrm{Haar}}(U')
\notag\\
&=
N_{ab}^{(\lambda)}.
\label{eq:appendix-N-ab-commutation}
\end{align}
Hence
\[
    V_\lambda(G) N_{ab}^{(\lambda)}
    =
    N_{ab}^{(\lambda)}V_\lambda(G)
    \qquad
    \text{for all }G\in\mathrm U(d).
\]
Since the representation on \(V_\lambda\) is irreducible, Schur's lemma
implies that there exists a scalar
\(\alpha_{ab}^{(\lambda)}\in\mathbb C\) such that
\begin{equation}
\label{eq:appendix-N-ab-scalar}
    N_{ab}^{(\lambda)}
    =
    \alpha_{ab}^{(\lambda)}I_{V_\lambda}.
\end{equation}

We determine \(\alpha_{ab}^{(\lambda)}\) by taking traces.  By linearity of
trace and invariance of trace under unitary conjugation,
\begin{align}
\operatorname{Tr}_{V_\lambda}
\!\left(N_{ab}^{(\lambda)}\right)
&=
\operatorname{Tr}_{V_\lambda}
\!\left(
   \int_{\mathrm U(d)}
   V_\lambda(U) Z_{ab}^{(\lambda)}V_\lambda(U)^\dagger\,d\mu_{\mathrm{Haar}}(U)
\right)
\notag\\
&=
\int_{\mathrm U(d)}
\operatorname{Tr}_{V_\lambda}
\!\left(
   V_\lambda(U) Z_{ab}^{(\lambda)}V_\lambda(U)^\dagger
\right)d\mu_{\mathrm{Haar}}(U)
\notag\\
&=
\int_{\mathrm U(d)}
\operatorname{Tr}_{V_\lambda}
\!\left(Z_{ab}^{(\lambda)}\right)d\mu_{\mathrm{Haar}}(U)
\notag\\
&=
\operatorname{Tr}_{V_\lambda}
\!\left(Z_{ab}^{(\lambda)}\right)
\int_{\mathrm U(d)}d\mu_{\mathrm{Haar}}(U)
\notag\\
&=
\operatorname{Tr}_{V_\lambda}
\!\left(Z_{ab}^{(\lambda)}\right),
\label{eq:appendix-N-ab-trace-first-way}
\end{align}
because Haar measure is normalized.  On the other hand,
\eqref{eq:appendix-N-ab-scalar} gives
\begin{align}
\operatorname{Tr}_{V_\lambda}
\!\left(N_{ab}^{(\lambda)}\right)
&=
\alpha_{ab}^{(\lambda)}
\operatorname{Tr}_{V_\lambda}(I_{V_\lambda})
\notag\\
&=
\alpha_{ab}^{(\lambda)}
\dim_{\mathbb C}V_\lambda.
\label{eq:appendix-N-ab-trace-second-way}
\end{align}
Comparing \eqref{eq:appendix-N-ab-trace-first-way} and
\eqref{eq:appendix-N-ab-trace-second-way}, we obtain
\begin{equation}
\label{eq:appendix-alpha-ab-value}
    \alpha_{ab}^{(\lambda)}
    =
    \frac{
       \operatorname{Tr}_{V_\lambda}
       \!\left(Z_{ab}^{(\lambda)}\right)
    }{
       \dim_{\mathbb C}V_\lambda
    }.
\end{equation}
Substituting \eqref{eq:appendix-alpha-ab-value} into
\eqref{eq:appendix-diagonal-average-expanded} yields
\begin{align}
\Upsilon_{\lambda\lambda}(X_{\lambda\lambda})
&=
\sum_{a,b}
\frac{
   \operatorname{Tr}_{V_\lambda}
   \!\left(Z_{ab}^{(\lambda)}\right)
}{
   \dim_{\mathbb C}V_\lambda
}
I_{V_\lambda}
\otimes
|e_a^{(\lambda)}\rangle
\langle e_b^{(\lambda)}|
\notag\\
&=
\frac{I_{V_\lambda}}{\dim_{\mathbb C}V_\lambda}
\otimes
\left[
   \sum_{a,b}
   \operatorname{Tr}_{V_\lambda}
   \!\left(Z_{ab}^{(\lambda)}\right)
   |e_a^{(\lambda)}\rangle
   \langle e_b^{(\lambda)}|
\right].
\label{eq:appendix-diagonal-average-before-partial-trace}
\end{align}

By the definition of the partial trace over \(V_\lambda\), the expansion
\eqref{eq:appendix-diagonal-multiplicity-expansion} satisfies
\begin{equation}
\label{eq:appendix-partial-trace-expanded}
    \operatorname{Tr}_{V_\lambda}
    \!\left(X_{\lambda\lambda}\right)
    =
    \sum_{a,b}
    \operatorname{Tr}_{V_\lambda}
    \!\left(Z_{ab}^{(\lambda)}\right)
    |e_a^{(\lambda)}\rangle
    \langle e_b^{(\lambda)}|.
\end{equation}
Combining \eqref{eq:appendix-diagonal-average-before-partial-trace} and
\eqref{eq:appendix-partial-trace-expanded}, we conclude that
\begin{equation}
\label{eq:appendix-diagonal-average-final}
    \Upsilon_{\lambda\lambda}(X_{\lambda\lambda})
    =
    \frac{I_{V_\lambda}}{\dim_{\mathbb C}V_\lambda}
    \otimes
    \operatorname{Tr}_{V_\lambda}
    \!\left(X_{\lambda\lambda}\right).
\end{equation}

\smallskip
Substituting the off-diagonal vanishing
\eqref{eq:appendix-off-diagonal-average-zero} and the diagonal identity
\eqref{eq:appendix-diagonal-average-final} into
\eqref{eq:appendix-twirl-before-Schur-lemma}, we obtain
\begin{align}
\mathcal T^{(t)}(X)
&=
\sum_{\lambda\in\Lambda(d,t)}
\iota_\lambda
\Upsilon_{\lambda\lambda}(X_{\lambda\lambda})
\iota_\lambda^\dagger
\notag\\
&=
\sum_{\lambda\in\Lambda(d,t)}
\iota_\lambda
\left[
   \frac{I_{V_\lambda}}{\dim_{\mathbb C}V_\lambda}
   \otimes
   \operatorname{Tr}_{V_\lambda}
   \!\left(X_{\lambda\lambda}\right)
\right]
\iota_\lambda^\dagger.
\label{eq:appendix-twirl-derived}
\end{align}
This proves the ambient-space formula
\eqref{eq:appendix-twirl-direct-sum-form}.

To prove that this is same as the direct-sum formula
\eqref{eq:appendix-twirl-direct-sum-form}, define
\begin{equation}
\label{eq:appendix-A-lambda-local-output}
    A_\lambda(X)
    :=
    \frac{I_{V_\lambda}}{\dim_{\mathbb C}V_\lambda}
    \otimes
    \operatorname{Tr}_{V_\lambda}
    \!\left(X_{\lambda\lambda}\right)
    \in\mathcal L(V_{\lambda} \otimes W_{\lambda}).
\end{equation}
Let \(|h\rangle\in\mathcal{H}^{\otimes t}_{d}\).  By the orthogonal direct-sum
decomposition, there exist unique vectors
\(
    |h_\mu\rangle\in V_\mu \otimes W_{\mu}
\)
such that
\begin{equation}
\label{eq:appendix-h-direct-sum-decomposition}
    |h\rangle
    =
    \sum_{\mu\in\Lambda(d,t)}
    \iota_\mu|h_\mu\rangle.
\end{equation}
Then
\begin{align}
&\left(
   \sum_{\lambda\in\Lambda(d,t)}
   \iota_\lambda A_\lambda(X)\iota_\lambda^\dagger
\right)|h\rangle
\notag\\
&\quad=
\sum_{\lambda,\mu}
\iota_\lambda A_\lambda(X)
\iota_\lambda^\dagger\iota_\mu
|h_\mu\rangle
\notag\\
&\quad=
\sum_{\lambda}
\iota_\lambda A_\lambda(X)|h_\lambda\rangle.
\label{eq:appendix-zero-extension-action}
\end{align}
Thus \(\iota_\lambda A_\lambda(X)\iota_\lambda^\dagger\) acts as
\(A_\lambda(X)\) on the \(\lambda\)-th summand and as zero on every other
summand.  Consequently,
\begin{equation}
\label{eq:appendix-direct-sum-equals-embedded-sum}
    \bigoplus_{\lambda\in\Lambda(d,t)}A_\lambda(X)
    \quad\longleftrightarrow\quad
    \sum_{\lambda\in\Lambda(d,t)}
    \iota_\lambda A_\lambda(X)\iota_\lambda^\dagger
\end{equation}
under the fixed Schur-Weyl identification.  This proves
\eqref{eq:appendix-twirl-direct-sum-form}.

Finally, \eqref{eq:appendix-compressed-block-recovered} shows that the blockwise partial trace is
\begin{align}
\operatorname{Tr}_{V_\lambda}
\!\left[
   \iota_\lambda^\dagger
   \left(\Pi_\lambda X\Pi_\lambda\right)
   \iota_\lambda
\right]
=
\operatorname{Tr}_{V_\lambda}
\!\left(X_{\lambda\lambda}\right).
\label{eq:appendix-type-correct-partial-trace}
\end{align}
The tensor product
\begin{equation}
\label{eq:appendix-local-twirled-output}
    \frac{I_{V_\lambda}}{\dim_{\mathbb C}V_\lambda}
    \otimes
    \operatorname{Tr}_{V_\lambda}
    \!\left(X_{\lambda\lambda}\right)
\end{equation}
is therefore an operator on the local space
\(V_\lambda\otimes W_\lambda\), and on the full space
\(\mathcal{H}^{\otimes t}_{d}\), we have
\begin{equation}
\label{eq:appendix-local-output-zero-extension}
    \iota_\lambda
    \left[
       \frac{I_{V_\lambda}}{\dim_{\mathbb C}V_\lambda}
       \otimes
       \operatorname{Tr}_{V_\lambda}
       \!\left(X_{\lambda\lambda}\right)
    \right]
    \iota_\lambda^\dagger.
\end{equation}
This completes the proof.
\end{proof}

\begin{lemma}\label{lem:trace-embedded-normalized-block}
For \(\lambda\in\Lambda(d,t)\) and
\(Y_\lambda\in\mathcal L(W_\lambda)\), define the ambient operator
\begin{equation}
\label{eq:embedded-normalized-block}
    \widetilde Y_\lambda
    :=
    \iota_\lambda
    \left(
       \frac{I_{V_\lambda}}
            {\dim_{\mathbb C}V_\lambda}
       \otimes Y_\lambda
    \right)
    \iota_\lambda^\dagger
    \in\mathcal L(\mathcal H_d^{\otimes t}).
\end{equation}
Then, for every \(\mu\in\Lambda(d,t)\),
\begin{equation}
\label{eq:embedded-normalized-block-compression}
    \Pi_\mu\widetilde Y_\lambda\Pi_\mu
    =\delta_{\mu\lambda}\widetilde Y_\lambda
\end{equation}
and
\begin{equation}
\label{eq:trace-embedded-normalized-block}
    \operatorname{Tr}_{\mathcal H_d^{\otimes t}}
    (\Pi_\mu\widetilde Y_\lambda)
    =
    \delta_{\mu\lambda}
    \operatorname{Tr}_{W_\lambda}(Y_\lambda).
\end{equation}
\end{lemma}

\begin{proof}
Using \(\Pi_\mu=\iota_\mu\iota_\mu^\dagger\) and
\(\iota_\mu^\dagger\iota_\lambda
=\delta_{\mu\lambda} I\), we obtain
\begin{align*}
 \Pi_\mu\widetilde Y_\lambda\Pi_\mu
 &=
 \iota_\mu\iota_\mu^\dagger
 \iota_\lambda
 \left(
  \frac{I_{V_\lambda}}{\dim V_\lambda}\otimes Y_\lambda
 \right)
 \iota_\lambda^\dagger
 \iota_\mu\iota_\mu^\dagger\\
 &=\delta_{\mu\lambda}\widetilde Y_\lambda.
\end{align*}
This proves Equation~\eqref{eq:embedded-normalized-block-compression}.
If \(\mu\neq\lambda\), the trace in
Equation~\eqref{eq:trace-embedded-normalized-block} is zero.  If
\(\mu=\lambda\), invariance of trace under isometric insertion gives
\begin{align*}
 \operatorname{Tr}(\Pi_\lambda\widetilde Y_\lambda)
 &=
 \operatorname{Tr}_{V_\lambda\otimes W_\lambda}
 \left(
  \frac{I_{V_\lambda}}{\dim V_\lambda}\otimes Y_\lambda
 \right)\\
 &=
 \frac{\operatorname{Tr}(I_{V_\lambda})}
      {\dim V_\lambda}
 \operatorname{Tr}(Y_\lambda)
 =\operatorname{Tr}(Y_\lambda).
\end{align*}
\end{proof}

\begin{proposition}[Trace preservation of the collective Haar twirl]
\label{prop:appendix-twirl-trace-preserving}
For every \(X\in\mathcal L(\mathcal H_d^{\otimes t})\),
\begin{equation}
\label{eq:appendix-twirl-trace-preserving}
    \operatorname{Tr}(\mathcal T^{(t)}(X))=\operatorname{Tr}(X).
\end{equation}
\end{proposition}

\begin{proof}
By Lemma~\ref{lem:blockwise-collective-Haar-twirl},
\begin{align*}
 \operatorname{Tr}(\mathcal T^{(t)}(X))
 &=
 \sum_{\lambda\in\Lambda(d,t)}
 \operatorname{Tr}_{V_\lambda\otimes W_\lambda}
 \left[
  \frac{\mathbf I_{V_\lambda}}{\dim V_\lambda}
  \otimes\operatorname{Tr}_{V_\lambda}(X_{\lambda\lambda})
 \right]\\
 &=
 \sum_{\lambda\in\Lambda(d,t)}
 \frac{\operatorname{Tr}(\mathbf I_{V_\lambda})}
      {\dim V_\lambda}
 \operatorname{Tr}_{W_\lambda}
 \left[\operatorname{Tr}_{V_\lambda}(X_{\lambda\lambda})\right]\\
 &=
 \sum_{\lambda\in\Lambda(d,t)}
 \operatorname{Tr}_{V_\lambda\otimes W_\lambda}
 (X_{\lambda\lambda}).
\end{align*}
Furthermore,
\begin{align*}
 \operatorname{Tr}(X_{\lambda\lambda})
 &=\operatorname{Tr}(\iota_\lambda^\dagger X\iota_\lambda)\\
 &=\operatorname{Tr}(\iota_\lambda\iota_\lambda^\dagger X)\\
 &=\operatorname{Tr}(\Pi_\lambda X).
\end{align*}
Therefore
\[
 \operatorname{Tr}(\mathcal T^{(t)}(X))
 =\sum_\lambda\operatorname{Tr}(\Pi_\lambda X)
 =\operatorname{Tr}
 \left[
  \left(\sum_\lambda\Pi_\lambda\right)X
 \right]
 =\operatorname{Tr}(X).
\]
\end{proof}

\begin{proposition}[Idempotence of the collective Haar twirl]
\label{prop:appendix-twirl-idempotent}
The collective Haar twirl is idempotent:
\begin{equation}
\label{eq:appendix-twirl-idempotent}
    \mathcal T^{(t)}\circ\mathcal T^{(t)}=\mathcal T^{(t)}.
\end{equation}
\end{proposition}

\begin{proof}
Let
\(
    Y:=\mathcal T^{(t)}(X).
\)
By Lemma~\ref{lem:blockwise-collective-Haar-twirl},
\[
 Y
 =
 \sum_{\nu\in\Lambda(d,t)}
 \iota_\nu
 \left[
   \frac{\mathbf I_{V_\nu}}{\dim V_\nu}
   \otimes\operatorname{Tr}_{V_\nu}(X_{\nu\nu})
 \right]
 \iota_\nu^\dagger.
\]
Its local \(\lambda\)-diagonal block is therefore
\begin{align*}
 Y_{\lambda\lambda}
 &=\iota_\lambda^\dagger Y\iota_\lambda\\
 &=
 \frac{\mathbf I_{V_\lambda}}{\dim V_\lambda}
 \otimes\operatorname{Tr}_{V_\lambda}(X_{\lambda\lambda}).
\end{align*}
Taking the partial trace over \(V_\lambda\),
\begin{align*}
 \operatorname{Tr}_{V_\lambda}(Y_{\lambda\lambda})
 &=
 \operatorname{Tr}_{V_\lambda}
 \left[
  \frac{\mathbf I_{V_\lambda}}{\dim V_\lambda}
  \otimes\operatorname{Tr}_{V_\lambda}(X_{\lambda\lambda})
 \right]\\
 &=
 \frac{\operatorname{Tr}(\mathbf I_{V_\lambda})}
      {\dim V_\lambda}
 \operatorname{Tr}_{V_\lambda}(X_{\lambda\lambda})\\
 &=\operatorname{Tr}_{V_\lambda}(X_{\lambda\lambda}).
\end{align*}
Applying Lemma~\ref{lem:blockwise-collective-Haar-twirl} once more gives
\begin{align*}
 \mathcal T^{(t)}(Y)
 &=
 \sum_\lambda
 \iota_\lambda
 \left[
  \frac{\mathbf I_{V_\lambda}}{\dim V_\lambda}
  \otimes\operatorname{Tr}_{V_\lambda}(Y_{\lambda\lambda})
 \right]
 \iota_\lambda^\dagger\\
 &=
 \sum_\lambda
 \iota_\lambda
 \left[
  \frac{\mathbf I_{V_\lambda}}{\dim V_\lambda}
  \otimes\operatorname{Tr}_{V_\lambda}(X_{\lambda\lambda})
 \right]
 \iota_\lambda^\dagger\\
 &=\mathcal T^{(t)}(X).
\end{align*}
\end{proof}

\begin{lemma}\label{lem:appendix-isotypic-weight-conservation}
For every \(X\in\mathcal L(\mathcal H_d^{\otimes t})\) and every
\(\mu\in\Lambda(d,t)\),
\begin{equation}
\label{eq:appendix-isotypic-weight-conservation}
    \operatorname{Tr}
    \bigl(\Pi_\mu\mathcal T^{(t)}(X)\bigr)
    =
    \operatorname{Tr}(\Pi_\mu X).
\end{equation}
In particular, if \(\rho\) is a density operator, both sides are the
\(\mu\)-isotypic weight of \(\rho\); if
\(\rho=|\psi\rangle\langle\psi|\), then
\begin{equation}
\label{eq:appendix-pure-isotypic-weight}
    \operatorname{Tr}(\Pi_\mu\rho)
    =\langle\psi|\Pi_\mu|\psi\rangle.
\end{equation}
\end{lemma}

\begin{proof}
By Lemma~\ref{lem:blockwise-collective-Haar-twirl} and orthogonality of the
isotypic inclusions,
\begin{align*}
 \operatorname{Tr}
 \bigl(\Pi_\mu\mathcal T^{(t)}(X)\bigr)
 &=
 \operatorname{Tr}_{V_\mu\otimes W_\mu}
 \left[
  \frac{\mathbf I_{V_\mu}}{\dim V_\mu}
  \otimes\operatorname{Tr}_{V_\mu}(X_{\mu\mu})
 \right]\\
 &=
 \operatorname{Tr}_{W_\mu}
 \left[\operatorname{Tr}_{V_\mu}(X_{\mu\mu})\right]\\
 &=\operatorname{Tr}(X_{\mu\mu})\\
 &=\operatorname{Tr}(\iota_\mu^\dagger X\iota_\mu)\\
 &=\operatorname{Tr}(\Pi_\mu X).
\end{align*}
For \(\rho=|\psi\rangle\langle\psi|\), cyclicity of trace gives
\[
    \operatorname{Tr}(\Pi_\mu\rho)
    =\operatorname{Tr}
      (\Pi_\mu|\psi\rangle\langle\psi|)
    =\langle\psi|\Pi_\mu|\psi\rangle.
\]
\end{proof}

\section{Representation-theoretic lemmas}
\begin{proposition}~\cite{fulton2013representation,simon1996representations,BroadbentGonzalezGuillenSchuknecht2022}\label{prop:appendix-symmetric-sector-specialization}
Let
\[
    \lambda_{\mathrm{sym}}:=(t),
    \qquad
    \Pi^{(d,t)}_{\mathrm{sym}}:=\Pi_{(t)},
    \qquad
    D_{\mathrm{sym}}
    :=\dim_{\mathbb C}\operatorname{Sym}^t(\mathcal H_d)
    =\binom{d+t-1}{t}.
\]
Then, for every \(X\in\mathcal L(\mathcal H_d^{\otimes t})\),
\begin{equation}
\label{eq:appendix-symmetric-compression-twirl}
 \Pi^{(d,t)}_{\mathrm{sym}}
 \mathcal T^{(t)}(X)
 \Pi^{(d,t)}_{\mathrm{sym}}
 =
 \operatorname{Tr}(\Pi^{(d,t)}_{\mathrm{sym}}X)
 \frac{\Pi^{(d,t)}_{\mathrm{sym}}}{D_{\mathrm{sym}}}.
\end{equation}
In particular, if
\(
\rho\in\mathcal D(\operatorname{Sym}^t(\mathcal H_d))
\), then
\begin{equation}
\label{eq:appendix-symmetric-state-twirl}
    \mathcal T^{(t)}(\rho)
    =
    \frac{\Pi^{(d,t)}_{\mathrm{sym}}}{D_{\mathrm{sym}}}.
\end{equation}
\end{proposition}

\section{Partial-trace lemma}
\begin{lemma}
Let
\[
    \vartheta \colon \mathcal{H}_{A} \longrightarrow \mathcal{H}_{B}
\]
be an isometry between finite-dimensional Hilbert spaces, and let
\(\mathcal{H}_{C}\) be another finite-dimensional Hilbert space.  Then, for every
\(
    Z\in\mathcal L(\mathcal{H}_{A} \otimes \mathcal{H}_{C})
\),
\begin{equation}
\label{eq:isometric-partial-trace-identity}
    \operatorname{Tr}_{B}
    \left[
        (\vartheta\otimes I_C)
        Z
        (\vartheta^\dagger\otimes I_C)
    \right]
    =
    \operatorname{Tr}_{A}(Z).
\end{equation}

Let
\(
    R\in\mathcal L(\mathcal{H}_C)
\)
be arbitrary.  By the defining property of the partial trace and
cyclicity of the ordinary trace,
\begin{align}
    &\operatorname{Tr}
    \left[
        R\,
        \operatorname{Tr}_{B}
        \left(
            (\vartheta\otimes I_C)
            Z
            (\vartheta^\dagger\otimes I_C)
        \right)
    \right]
    \notag\\
    &\qquad=
    \operatorname{Tr}
    \left[
        (I_B\otimes R)
        (\vartheta\otimes I_C)
        Z
        (\vartheta^\dagger\otimes I_C)
    \right]
    \notag\\
    &\qquad=
    \operatorname{Tr}
    \left[
        (\vartheta^\dagger\otimes I_C)
        (I_B\otimes R)
        (\vartheta\otimes I_C)
        Z
    \right]
    \notag\\
    &\qquad=
    \operatorname{Tr}
    \left[
        (\vartheta^\dagger \vartheta\otimes R)
        Z
    \right]
    \notag\\
    &\qquad=
    \operatorname{Tr}
    \left[
        (I_A\otimes R)
        Z
    \right]
    \notag\\
    &\qquad=
    \operatorname{Tr}
    \left[
        R\,\operatorname{Tr}_{A}(Z)
    \right].
\end{align}
Since the trace pairing on
\(
    \mathcal L(C)
\)
is nondegenerate and this equality holds for every
\(
    R\in\mathcal L(C)
\),
Equation~\eqref{eq:isometric-partial-trace-identity} follows.
\end{lemma}

\end{appendices}


\begin{thebibliography}{24}
\ifx \bisbn   \undefined \def \bisbn  #1{ISBN #1}\fi
\ifx \binits  \undefined \def \binits#1{#1}\fi
\ifx \bauthor  \undefined \def \bauthor#1{#1}\fi
\ifx \batitle  \undefined \def \batitle#1{#1}\fi
\ifx \bjtitle  \undefined \def \bjtitle#1{#1}\fi
\ifx \bvolume  \undefined \def \bvolume#1{\textbf{#1}}\fi
\ifx \byear  \undefined \def \byear#1{#1}\fi
\ifx \bissue  \undefined \def \bissue#1{#1}\fi
\ifx \bfpage  \undefined \def \bfpage#1{#1}\fi
\ifx \blpage  \undefined \def \blpage #1{#1}\fi
\ifx \burl  \undefined \def \burl#1{\textsf{#1}}\fi
\ifx \doiurl  \undefined \def \doiurl#1{\url{https://doi.org/#1}}\fi
\ifx \betal  \undefined \def \betal{\textit{et al.}}\fi
\ifx \binstitute  \undefined \def \binstitute#1{#1}\fi
\ifx \binstitutionaled  \undefined \def \binstitutionaled#1{#1}\fi
\ifx \bctitle  \undefined \def \bctitle#1{#1}\fi
\ifx \beditor  \undefined \def \beditor#1{#1}\fi
\ifx \bpublisher  \undefined \def \bpublisher#1{#1}\fi
\ifx \bbtitle  \undefined \def \bbtitle#1{#1}\fi
\ifx \bedition  \undefined \def \bedition#1{#1}\fi
\ifx \bseriesno  \undefined \def \bseriesno#1{#1}\fi
\ifx \blocation  \undefined \def \blocation#1{#1}\fi
\ifx \bsertitle  \undefined \def \bsertitle#1{#1}\fi
\ifx \bsnm \undefined \def \bsnm#1{#1}\fi
\ifx \bsuffix \undefined \def \bsuffix#1{#1}\fi
\ifx \bparticle \undefined \def \bparticle#1{#1}\fi
\ifx \barticle \undefined \def \barticle#1{#1}\fi
\bibcommenthead
\ifx \bconfdate \undefined \def \bconfdate #1{#1}\fi
\ifx \botherref \undefined \def \botherref #1{#1}\fi
\ifx \url \undefined \def \url#1{\textsf{#1}}\fi
\ifx \bchapter \undefined \def \bchapter#1{#1}\fi
\ifx \bbook \undefined \def \bbook#1{#1}\fi
\ifx \bcomment \undefined \def \bcomment#1{#1}\fi
\ifx \oauthor \undefined \def \oauthor#1{#1}\fi
\ifx \citeauthoryear \undefined \def \citeauthoryear#1{#1}\fi
\ifx \endbibitem  \undefined \def \endbibitem {}\fi
\ifx \bconflocation  \undefined \def \bconflocation#1{#1}\fi
\ifx \arxivurl  \undefined \def \arxivurl#1{\textsf{#1}}\fi
\csname PreBibitemsHook\endcsname

\bibitem[\protect\citeauthoryear{Broadbent et~al.}{2022}]{BroadbentGonzalezGuillenSchuknecht2022}
\begin{barticle}
\bauthor{\bsnm{Broadbent}, \binits{A.}},
\bauthor{\bsnm{Gonz{\'a}lez-Guill{\'e}n}, \binits{C.E.}},
\bauthor{\bsnm{Schuknecht}, \binits{C.}}:
\batitle{Quantum private broadcasting}.
\bjtitle{Physical Review A}
\bvolume{105}(\bissue{2}),
\bfpage{022606}
(\byear{2022})
\doiurl{10.1103/PhysRevA.105.022606}
\end{barticle}
\endbibitem

\bibitem[\protect\citeauthoryear{Kretschmann et~al.}{2008}]{KretschmannKribsSpekkens2008}
\begin{barticle}
\bauthor{\bsnm{Kretschmann}, \binits{D.}},
\bauthor{\bsnm{Kribs}, \binits{D.W.}},
\bauthor{\bsnm{Spekkens}, \binits{R.W.}}:
\batitle{Complementarity of private and correctable subsystems in quantum cryptography and error correction}.
\bjtitle{Physical Review A}
\bvolume{78},
\bfpage{032330}
(\byear{2008})
\doiurl{10.1103/PhysRevA.78.032330}
\end{barticle}
\endbibitem

\bibitem[\protect\citeauthoryear{Knill and Laflamme}{1997}]{KnillLaflamme1997}
\begin{barticle}
\bauthor{\bsnm{Knill}, \binits{E.}},
\bauthor{\bsnm{Laflamme}, \binits{R.}}:
\batitle{Theory of quantum error-correcting codes}.
\bjtitle{Physical Review A}
\bvolume{55},
\bfpage{900}--\blpage{911}
(\byear{1997})
\doiurl{10.1103/PhysRevA.55.900}
\end{barticle}
\endbibitem

\bibitem[\protect\citeauthoryear{Kribs et~al.}{2005}]{KribsLaflammePoulin2005}
\begin{barticle}
\bauthor{\bsnm{Kribs}, \binits{D.}},
\bauthor{\bsnm{Laflamme}, \binits{R.}},
\bauthor{\bsnm{Poulin}, \binits{D.}}:
\batitle{Unified and generalized approach to quantum error correction}.
\bjtitle{Physical Review Letters}
\bvolume{94},
\bfpage{180501}
(\byear{2005})
\doiurl{10.1103/PhysRevLett.94.180501}
\end{barticle}
\endbibitem

\bibitem[\protect\citeauthoryear{Jochym-O'Connor et~al.}{2013}]{JochymOConnorKribsLaflammePlosker2013}
\begin{barticle}
\bauthor{\bsnm{Jochym-O'Connor}, \binits{T.}},
\bauthor{\bsnm{Kribs}, \binits{D.W.}},
\bauthor{\bsnm{Laflamme}, \binits{R.}},
\bauthor{\bsnm{Plosker}, \binits{S.}}:
\batitle{Private quantum subsystems}.
\bjtitle{Physical Review Letters}
\bvolume{111},
\bfpage{030502}
(\byear{2013})
\doiurl{10.1103/PhysRevLett.111.030502}
\end{barticle}
\endbibitem

\bibitem[\protect\citeauthoryear{Jochym-O'Connor et~al.}{2014}]{JochymOConnorKribsLaflammePlosker2014}
\begin{barticle}
\bauthor{\bsnm{Jochym-O'Connor}, \binits{T.}},
\bauthor{\bsnm{Kribs}, \binits{D.W.}},
\bauthor{\bsnm{Laflamme}, \binits{R.}},
\bauthor{\bsnm{Plosker}, \binits{S.}}:
\batitle{Quantum subsystems: Exploring the complementarity of quantum privacy and error correction}.
\bjtitle{Physical Review A}
\bvolume{90},
\bfpage{032305}
(\byear{2014})
\doiurl{10.1103/PhysRevA.90.032305}
\end{barticle}
\endbibitem

\bibitem[\protect\citeauthoryear{Bartlett et~al.}{2004}]{BartlettRudolphSpekkens2004}
\begin{barticle}
\bauthor{\bsnm{Bartlett}, \binits{S.D.}},
\bauthor{\bsnm{Rudolph}, \binits{T.}},
\bauthor{\bsnm{Spekkens}, \binits{R.W.}}:
\batitle{Decoherence-full subsystems and the cryptographic power of a private shared reference frame}.
\bjtitle{Physical Review A}
\bvolume{70},
\bfpage{032307}
(\byear{2004})
\doiurl{10.1103/PhysRevA.70.032307}
{\href{https://arxiv.org/abs/quant-ph/0403161}{{arXiv:quant-ph/0403161}}}
\end{barticle}
\endbibitem

\bibitem[\protect\citeauthoryear{Church et~al.}{2011}]{ChurchKribsPereiraPlosker2011}
\begin{barticle}
\bauthor{\bsnm{Church}, \binits{A.}},
\bauthor{\bsnm{Kribs}, \binits{D.W.}},
\bauthor{\bsnm{Pereira}, \binits{R.}},
\bauthor{\bsnm{Plosker}, \binits{S.}}:
\batitle{Private quantum channels, conditional expectations, and trace vectors}.
\bjtitle{Quantum Information and Computation}
\bvolume{11}(\bissue{9--10}),
\bfpage{774}--\blpage{783}
(\byear{2011})
\doiurl{10.26421/QIC11.9-10-4}
{\href{https://arxiv.org/abs/1210.6625}{{arXiv:1210.6625}}}
{[math.OA]}
\end{barticle}
\endbibitem

\bibitem[\protect\citeauthoryear{Gross et~al.}{2007}]{GrossAudenaertEisert2007}
\begin{barticle}
\bauthor{\bsnm{Gross}, \binits{D.}},
\bauthor{\bsnm{Audenaert}, \binits{K.}},
\bauthor{\bsnm{Eisert}, \binits{J.}}:
\batitle{Evenly distributed unitaries: On the structure of unitary designs}.
\bjtitle{Journal of Mathematical Physics}
\bvolume{48},
\bfpage{052104}
(\byear{2007})
\doiurl{10.1063/1.2716992}
\end{barticle}
\endbibitem

\bibitem[\protect\citeauthoryear{Dankert et~al.}{2009}]{dankert2009exact}
\begin{barticle}
\bauthor{\bsnm{Dankert}, \binits{C.}},
\bauthor{\bsnm{Cleve}, \binits{R.}},
\bauthor{\bsnm{Emerson}, \binits{J.}},
\bauthor{\bsnm{Livine}, \binits{E.}}:
\batitle{Exact and approximate unitary 2-designs and their application to fidelity estimation}.
\bjtitle{Physical Review A}
\bvolume{80}(\bissue{1}),
\bfpage{012304}
(\byear{2009})
\end{barticle}
\endbibitem

\bibitem[\protect\citeauthoryear{Roy and Scott}{2009}]{RoyScott2009}
\begin{barticle}
\bauthor{\bsnm{Roy}, \binits{A.}},
\bauthor{\bsnm{Scott}, \binits{A.J.}}:
\batitle{Unitary designs and codes}.
\bjtitle{Designs, Codes and Cryptography}
\bvolume{53},
\bfpage{13}--\blpage{31}
(\byear{2009})
\doiurl{10.1007/s10623-009-9290-2}
\end{barticle}
\endbibitem

\bibitem[\protect\citeauthoryear{Brandao et~al.}{2016}]{brandao2016local}
\begin{barticle}
\bauthor{\bsnm{Brandao}, \binits{F.G.}},
\bauthor{\bsnm{Harrow}, \binits{A.W.}},
\bauthor{\bsnm{Horodecki}, \binits{M.}}:
\batitle{Local random quantum circuits are approximate polynomial-designs}.
\bjtitle{Communications in Mathematical Physics}
\bvolume{346},
\bfpage{397}--\blpage{434}
(\byear{2016})
\end{barticle}
\endbibitem

\bibitem[\protect\citeauthoryear{Serre}{1977}]{serre1977linear}
\begin{bbook}
\bauthor{\bsnm{Serre}, \binits{J.-P.}}:
\bbtitle{Linear Representations of Finite Groups}
vol. \bseriesno{42}.
\bpublisher{Springer},
\blocation{New York}
(\byear{1977})
\end{bbook}
\endbibitem

\bibitem[\protect\citeauthoryear{Fulton and Harris}{2013}]{fulton2013representation}
\begin{bbook}
\bauthor{\bsnm{Fulton}, \binits{W.}},
\bauthor{\bsnm{Harris}, \binits{J.}}:
\bbtitle{Representation Theory: A First Course}.
\bsertitle{Graduate Texts in Mathematics},
vol. \bseriesno{129}.
\bpublisher{Springer},
\blocation{New York}
(\byear{2013})
\end{bbook}
\endbibitem

\bibitem[\protect\citeauthoryear{Simon}{1996}]{simon1996representations}
\begin{bbook}
\bauthor{\bsnm{Simon}, \binits{B.}}:
\bbtitle{Representations of Finite and Compact Groups}.
\bsertitle{Graduate Studies in Mathematics},
vol. \bseriesno{10}.
\bpublisher{American Mathematical Society},
\blocation{Providence, RI}
(\byear{1996})
\end{bbook}
\endbibitem

\bibitem[\protect\citeauthoryear{Sagan}{2001}]{sagan2001symmetric}
\begin{bbook}
\bauthor{\bsnm{Sagan}, \binits{B.E.}}:
\bbtitle{The Symmetric Group: Representations, Combinatorial Algorithms, and Symmetric Functions}.
\bsertitle{Graduate Texts in Mathematics},
vol. \bseriesno{203}.
\bpublisher{Springer},
\blocation{New York}
(\byear{2001})
\end{bbook}
\endbibitem

\bibitem[\protect\citeauthoryear{Macdonald}{1998}]{macdonald1998symmetric}
\begin{bbook}
\bauthor{\bsnm{Macdonald}, \binits{I.G.}}:
\bbtitle{Symmetric Functions and Hall Polynomials},
\bedition{2nd} edn.
\bpublisher{Oxford University Press},
\blocation{Oxford}
(\byear{1998})
\end{bbook}
\endbibitem

\bibitem[\protect\citeauthoryear{Goodman and Wallach}{1998}]{GoodmanWallach1998}
\begin{bbook}
\bauthor{\bsnm{Goodman}, \binits{R.}},
\bauthor{\bsnm{Wallach}, \binits{N.R.}}:
\bbtitle{Representations and Invariants of the Classical Groups}.
\bsertitle{Encyclopedia of Mathematics and its Applications},
vol. \bseriesno{68}.
\bpublisher{Cambridge University Press},
\blocation{Cambridge}
(\byear{1998})
\end{bbook}
\endbibitem

\bibitem[\protect\citeauthoryear{Folland}{2016}]{folland2016course}
\begin{bbook}
\bauthor{\bsnm{Folland}, \binits{G.B.}}:
\bbtitle{A Course in Abstract Harmonic Analysis},
\bedition{2nd} edn.
\bpublisher{CRC Press},
\blocation{Boca Raton, FL}
(\byear{2016})
\end{bbook}
\endbibitem

\bibitem[\protect\citeauthoryear{Br\"{o}cker and tom Dieck}{1985}]{BrockerTomDieck1985}
\begin{bbook}
\bauthor{\bsnm{Br\"{o}cker}, \binits{T.}},
\bauthor{\bsnm{Dieck}, \binits{T.}}:
\bbtitle{Representations of Compact Lie Groups}.
\bsertitle{Graduate Texts in Mathematics},
vol. \bseriesno{98}.
\bpublisher{Springer},
\blocation{Berlin}
(\byear{1985}).
\doiurl{10.1007/978-3-662-12918-0}
\end{bbook}
\endbibitem

\bibitem[\protect\citeauthoryear{Kaznatcheev}{2010}]{kaznatcheev2010structure}
\begin{botherref}
\oauthor{\bsnm{Kaznatcheev}, \binits{A.}}:
Structure of exact and approximate unitary t-designs.
Citeseer
(2010)
\end{botherref}
\endbibitem

\bibitem[\protect\citeauthoryear{Seymour and Zaslavsky}{1984}]{seymour1984averaging}
\begin{barticle}
\bauthor{\bsnm{Seymour}, \binits{P.D.}},
\bauthor{\bsnm{Zaslavsky}, \binits{T.}}:
\batitle{Averaging sets: a generalization of mean values and spherical designs}.
\bjtitle{Advances in Mathematics}
\bvolume{52}(\bissue{3}),
\bfpage{213}--\blpage{240}
(\byear{1984})
\end{barticle}
\endbibitem

\bibitem[\protect\citeauthoryear{Schuster et~al.}{2025}]{SchusterHaferkampHuang2025}
\begin{barticle}
\bauthor{\bsnm{Schuster}, \binits{T.}},
\bauthor{\bsnm{Haferkamp}, \binits{J.}},
\bauthor{\bsnm{Huang}, \binits{H.-Y.}}:
\batitle{Random unitaries in extremely low depth}.
\bjtitle{Science}
\bvolume{389}(\bissue{6755}),
\bfpage{92}--\blpage{96}
(\byear{2025})
\doiurl{10.1126/science.adv8590}
\end{barticle}
\endbibitem

\bibitem[\protect\citeauthoryear{Cui et~al.}{2025}]{CuiSchusterBrandaoHuang2025}
\begin{botherref}
\oauthor{\bsnm{Cui}, \binits{L.}},
\oauthor{\bsnm{Schuster}, \binits{T.}},
\oauthor{\bsnm{Brand{\~a}o}, \binits{F.}},
\oauthor{\bsnm{Huang}, \binits{H.-Y.}}:
Unitary designs in nearly optimal depth.
arXiv preprint arXiv:2507.06216
(2025)
\end{botherref}
\endbibitem

\end{thebibliography}


\end{document}